\documentclass[11pt, reqno]{amsart}
\usepackage[a4paper, margin=2cm]{geometry} 

\usepackage[sc]{mathpazo}

\usepackage[subpreambles=true]{standalone}

\usepackage[T1]{fontenc}

\usepackage{fancyhdr} 
\usepackage{xcolor}

\usepackage[foot]{amsaddr}

\usepackage{comment}
\usepackage{verbatim}

\fancypagestyle{suppheader}{
    \fancyhf{} 
    \fancyhead[C]{\uppercase{Supplementary Material}} 
    \fancyfoot[C]{\thepage} 
}

\usepackage{titletoc}
\newcommand{\hidesuppinmaintoc}{\setcounter{tocdepth}{0}}
\makeatletter
\newcommand{\sectionnotoc}[1]{{\let\addcontentsline\@gobblethree\section*{#1}}}
\makeatother

\usepackage{enumitem} 
\usepackage{multicol}

\usepackage{tikz}
\usetikzlibrary{trees, positioning}
\usepackage{caption}

\usepackage{algorithm}
\usepackage{algpseudocode}

\usetikzlibrary{
  arrows.meta,
  positioning,
  calc,
  fit,
  backgrounds,
  decorations.pathmorphing
}

\usepackage{hyperref}

\hypersetup{
    colorlinks=true,
    citecolor=blue,   
    pdftitle={Approximating Fixed Size Quantum Correlations in Polynomial Time},
    pdfpagemode=FullScreen,
}

\usepackage[style=alphabetic]{biblatex}

\usepackage{amsmath}
\usepackage{amsthm,amssymb,mathtools,amsfonts}

\allowdisplaybreaks[4]

\numberwithin{equation}{section}

\usepackage{thmtools}

\declaretheorem[name=Theorem,numberwithin=section, Refname={Thm.}, refname={Thm.}]{theorem}

\declaretheorem[name=Corollary,sibling=theorem, Refname={Cor.}, refname={Cor.}]{corollary}
\declaretheorem[name=Lemma,sibling=theorem, Refname={Lem.}, refname={Lem.}]{lemma}
\declaretheorem[name=Proposition,sibling=theorem, Refname={Prop.}, refname={Prop.}]{proposition}
\declaretheorem[name=Definition,sibling=theorem, Refname={Def.}, refname={Def.}]{definition}
\declaretheorem[name=Example,sibling=theorem, Refname={Ex.}, refname={Ex.}]{example}

\declaretheorem[name=Remark,sibling=theorem, Refname={Rem.}, refname={Rem.}]{remark}

\usepackage{ytableau}

\usepackage{multirow}

\usepackage{braket}   

\newcommand{\HS}{\mathcal{H}}
\newcommand{\HSn}{\mathcal{H}^{\otimes n}}
\newcommand{\BS}{\mathcal{B}}
\newcommand{\Al}{\mathcal{A}}

\newcommand{\Id}{\mathcal{I}}
\newcommand{\D}{\mathcal{D}}
\newcommand{\An}{\lrbracket{A_2Q_2\hat{T}}_1^n}

\newcommand{\ESNH}{\text{End}_{\CSn}\lrbracket{\HS}}
\newcommand{\ESN}{\text{End}_{\CSn}}

\newcommand{\R}{\mathbb{R}}
\newcommand{\CC}{\mathbb{C}}

\newcommand{\dHS}{d_{\HS}}

\newcommand{\SymH}{\vee^n\lrbracket{\HS}}
\newcommand{\Op}[1]{\mathcal{B}\lrbracket{#1}}

\newcommand{\Hom}[3]{\mathrm{Hom}_{#1}\lrbracket{#2, #3}}
\newcommand{\End}[2]{\mathrm{End}_{#1}\lrbracket{#2}}

\newcommand{\Specht}{V_{\lambda}}
\newcommand{\SpechtH}{\HS_{\lambda}}
\newcommand{\CSn}{\CC\lrrec{S_n}}

\newcommand{\Vn}{V^{\otimes n}}
\newcommand{\GLV}{\text{GL}\lrbracket{V}}
\newcommand{\GLH}{\text{GL}\lrbracket{\HS}}
\newcommand{\KG}{\mathbb{K}\lrrec{G}}
\newcommand{\KK}{\mathbb{K}}
\newcommand{\Irr}[1]{\text{Irr}\lrbracket{#1}}
\newcommand{\CG}{\mathbb{C}\lrrec{G}}

\newcommand{\EndBose}[1]{\End{\CC\lrrec{S_{#1}}}{A\otimes\lor^{#1}\lrbracket{B\bar{B}}}}

\newcommand{\CBose}[3]{C_{\vec{#1},\vec{#2}}^{\vee^{#3}}}
\newcommand{\KBose}[3]{K_{\vec{#1},\vec{#2}}^{\vee^{#3}}}

\newcommand{\norm}[1]{\left\lVert#1\right\rVert}
\newcommand{\lrbrace}[1]{\left\lbrace{#1}\right\rbrace}
\newcommand{\lrbracket}[1]{\left({#1}\right)}
\newcommand{\lrvert}[1]{\left\lvert{#1}\right\rvert}
\newcommand{\lrceil}[1]{\left\lceil{#1}\right\rceil}
\newcommand{\lrrec}[1]{\left[{#1}\right]}
\newcommand{\block}[1]{\left[\!\left[{#1}\right]\!\right]}

\newcommand{\rhoge}{\rho_{(A_1Q_1T)(A_2Q_2\hat{T})_1^n}}
\newcommand{\rhog}{\rho_{(A_1Q_1T)(A_2Q_2\hat{T})}}
\newcommand{\Bose}{\rho_{A\left(B\Bar{B}\right)_1^n}}
\newcommand{\BoseB}{\rho_{\left(B\Bar{B}\right)_1^n}}
\newcommand{\Sdpbose}{\mathrm{SDP}_n^{\mathrm{Bose}}}

\newcommand{\Trr}[2]{\Tr_{#1}\left[{#2}\right]}
\newcommand{\Tr}{\mathrm{tr}}
\newcommand{\Schurf}[1]{\mathbb{S}_{#1}}

\title{\large Approximating fixed size quantum correlations in polynomial time}

\author{Julius A. Zeiss$^{1,*}$}
\address{$^1$ Institute for Quantum Information, RWTH Aachen University, Aachen, Germany}
\author{Gereon Kossmann$^{1}$}
\author{Omar Fawzi$^{2}$}
\address{$^2$ Univ Lyon, Inria, ENS Lyon, UCBL, LIP, Lyon, France}
\author{Mario Berta$^{1,3}$}
\address{$^3$ Department of Computing, Imperial College London, London, UK}
\address{$^*$ Corresponding author: jzeiss@physik.rwth-aachen.de}

\begin{document}

\begin{abstract}
    We show that $\varepsilon$-additive approximations of the optimal value of fixed-size two-player free games with fixed-dimensional entanglement assistance can be computed in time $\mathrm{poly}(1/\varepsilon)$. This stands in contrast to previous analytic approaches, which focused on scaling with the number of questions and answers, but yielded only strict $\mathrm{exp}(1/\varepsilon)$ guarantees. Our main result is based on novel Bose-symmetric quantum de Finetti theorems tailored for constrained quantum separability problems. These results give rise to semidefinite programming (SDP) outer hierarchies for approximating the entangled value of such games. By employing representation-theoretic symmetry reduction techniques, we demonstrate that these SDPs can be formulated and solved with computational complexity $\mathrm{poly}(1/\varepsilon)$, thereby enabling efficient $\varepsilon$-additive approximations. In addition, we introduce a measurement-based rounding scheme that translates the resulting outer bounds into certifiably good inner sequences of entangled strategies. These strategies can, for instance, serve as warm starts for see-saw optimization methods. We believe that our techniques are of independent interest for broader classes of constrained separability problems in quantum information theory.
\end{abstract}

\maketitle

\tableofcontents


\section{Introduction}

\subsection{Motivation}

The EPR paradox \cite{einstein1935n} and Bell's resolution of it \cite{bell64} initiated the study of non-local correlations and, more broadly, non-local games. In this work, we focus exclusively on the standard setting of two-player free non-local games. In these games $G$, two spatially separated cooperative players receive questions $q_1\in Q_1$, $q_2\in Q_2$, respectively, chosen independently by a referee according to the probability distributions $\pi_1(q_1),\, \pi_2(q_2)$. We assume that $\lrvert{Q_1}=\lrvert{Q_2}$. The two players are able to agree on an answering strategy beforehand, but once the referee starts to hand out the questions, no further adaptation of said strategy is allowed. As opposed to the classical setting, where a strategy consists of a set of conditional answers, in the quantum counterpart, conditional measurements are performed on a shared entangled quantum state. To conclude the game, the referee judges the collection of answers $a_1\in A_1,\, a_2\in A_2$ with $\lrvert{A_1}=\lrvert{A_2}$ according to a game-specific rule function $V$. The main quantities of interest are the maximal winning probability $0 \leq w_{Q}(V, \pi) \leq 1$ and a strategy attaining this optimum. In this context, a strategy corresponds to the collection of the parties' respective sets of complete POVMs and their shared quantum state. Importantly, quantum resources allow for higher winning probabilities in certain games, see, e.g.\ the CHSH game \cite{cirel1980quantum}.

Our interest in these games is twofold. From a physics perspective, as in Bell's work, the potential gap in the maximal winning probability obtained from games with shared quantum rather than classical resources helps in delineating the boundary between the convex set of quantum correlations and the polytope of classical correlations \cite{Navascu_s_2007, navascues2008convergent}. From a complexity-theoretic viewpoint, even in the classical variant, approximately determining $w_{C}(V, \pi)$ within a constant multiplicative factor is NP-hard (see the PCP theorem \cite{arora1998proof, arora1998probabilistic}). In addition, computing $w_Q(V, \pi)$ is, in general, even harder \cite{cleve2004consequences, ito2009oracularization, kempe2011entangled, brandao_harrow_2014}. More precisely, in the general setting the corresponding approximation problem is undecidable \cite{ji2021mip} (equivalently, RE-complete in a suitable promise-gap formulation). However, for specific games \cite{cleve2004consequences, kempe2011entangled}, polynomial-time approximations are possible. We follow \cite{navascues2014characterization, navascues2015characterizing, navascues2015bounding, Donohue_2015, jee2020quasi} and consider only quantum resources of fixed local dimension $\lrvert{T}$ to be shared among the players and ask how well, and at what computational cost, we can approximate $w_{Q(T)}(V, \pi)$ by classical algorithms.

We will next give a brief overview of our results. The precise definitions of the quantities involved are given in the subsequent sections.


\subsection{Contributions}

Leveraging a synthesis of information-theoretic techniques and structural tools from group theory, we resolve several outstanding problems related to the approximation of quantum correlations, as articulated in \cite{jee2020quasi} and \cite{berta2021semidefinite}. 

Our primary contribution concerns free non-local games with $\lrvert{A}$ answers, $\lrvert{Q}$ questions, and quantum assistance of fixed local dimension $\lrvert{T}$. By exploiting symmetry arguments, we construct two efficient algorithmic variants that approximate the value of such games up to additive error $\epsilon$, via the solution of a semidefinite program (SDP) of size \begin{align}\label{eq:main-complexity} \left[\epsilon^{-4} \cdot \mathrm{poly}\left(|A|, |Q|, |T|\right)\right]^{\left(|A||Q||T|\right)^2}.\, \end{align}
In particular, for constant-sized games, this yields a $\text{poly}\lrbracket{\epsilon^{-1}}$ scaling. The result is a consequence of several key ideas, which we detail below:
    
	\begin{enumerate}[label=(\Alph*)]
	\item \label{first_result} We mitigate the dependence on the quantum dimension $\lrvert{T}$ by leveraging methods from quantum steering (see also \cite{Quintino2014, mateus_araujo, tavakoli2024semidefinite}), thereby reformulating the problem of determining the maximum winning probability of a two-player free non-local game with fixed quantum assistance as an optimization over a constrained bipartite separability set (see \cite{berta2021semidefinite}). A permutation-symmetry-based SDP hierarchy then provides converging upper bounds on the value of the game. The rate of convergence is certified by a new constrained de Finetti representation theorem, which extends \cite[Theorem 2.3]{berta2021semidefinite} to accommodate a broader class of constraints than those treated in \cite{berta2021semidefinite, jee2020quasi}, thereby addressing the constrained structure of the underlying separability problem. Namely, we show that for  $A=A_L\otimes A_R$, $B=B_L\otimes B_R$, and a quantum state $\rho_{AB_1^n}$ symmetric over $B_1^n:=B_1\otimes \ldots \otimes B_n$ w.r.t.\ $A$ satisfying
    \begin{align}
        \begin{array}{cc}
            \Theta_{A_L\rightarrow C_{A_L}}\lrbracket{\rho_{AB_1^n}} = W_{C_{A_L}}\otimes\rho_{A_RB_1^n},\, &  \Omega_{A\rightarrow A}\lrbracket{\rho_{AB_1^n}}=\rho_{AB_1^n},\,\\
            \Upsilon_{\lrbracket{B_L}_n\rightarrow C_{B_L}}\lrbracket{\rho_{B_1^n}} = K_{C_{B_L}}\otimes \rho_{B_1^{n-1}\lrbracket{B_R}_n},\, & \Xi_{B_n\rightarrow B_n}\lrbracket{\rho_{B_1^n}}=\rho_{B_1^n},\,\\
        \end{array}
    \end{align}
    for linear maps $\Theta_{A_L\rightarrow C_{A_L}}$, $\Omega_{A\rightarrow A}$, $ \Upsilon_{\lrbracket{B_L}_n\rightarrow C_{B_L}}$, $\Xi_{B_n\rightarrow B_n}$ and constant operators $W_{C_{A_L}}$, $K_{C_{B_L}}$, we have
    \begin{align}\label{eq:intro-definettiI_original}
        \left\lVert \rho_{AB} - \sum_{x\in\mathcal{X}}p(x)\rho_A^x\otimes\rho_{B}^x \right\rVert_1\leq\mathcal{O}\left(\frac{|B|}{\sqrt{n}}\cdot\sqrt{\ln{|A|}}\right),\,
    \end{align}
    where $\forall x\in\mathcal{X}:$
    \begin{align}
        \begin{split}
            \begin{array}{cc}
                  \Theta_{A_L\rightarrow C_{A_L}}\lrbracket{\rho_A^x}=W_{C_{A_L}}\otimes \rho_{A_R}^x,\, &  \Omega_{A\rightarrow A}\lrbracket{\rho_A^x}=\rho_A^x,\\
                \Upsilon_{B_L\rightarrow C_{B_L}}\lrbracket{\rho_B^x}=K_{C_{B_L}}\otimes \rho_{B_R}^x,\, &  \Xi_{B\rightarrow B}\lrbracket{\rho_B^x}=\rho_B^x\,.
            \end{array}
        \end{split}
    \end{align}
    Compared to previous work (see \autoref{sec:previous}), this already improves the error dependence on the dimension of the shared quantum assistance.\\
    
	\item \label{intro_bose_de_finetti} Now, for unconstrained optimizations over separable states, so-called Bose-symmetric de Finetti representation theorems based on symmetric subspace methods are commonly used in symmetry-based approaches. While symmetric states may be supported on the entire Hilbert space $\HSn$, Bose-symmetric states (cf.\ \cite{harrow2013church} or Dicke states in e.g.\ \cite{gulati2026entanglementdickesubspace}) are characterized by having both their support and range confined strictly to the symmetric subspace of $\HSn$, which, in turn, is of polynomial dimension w.r.t.\ $n$.
While it has previously been unclear how to incorporate linear constraints for Bose-symmetric considerations, we achieve exactly that and derive novel Bose-symmetric de Finetti representation theorems subject to a variety of linear constraints. Namely, we show that for quantum states $\rho_{A(B\bar{B})_1^n}$ Bose-symmetric on $\lrbracket{B\bar{B}}_1^n$ w.r.t.\ $A$ and
    \begin{align}
    \begin{split}
     \begin{array}{cc}
            \Theta_{A_L\rightarrow C_{A_L}}\lrbracket{\rho_{A\lrbracket{B\Bar{B}}_1^n}} = W_{C_{A_L}}\otimes\rho_{A_R\lrbracket{B\Bar{B}}_1^n},\, &
            \Omega_{A\rightarrow A}\lrbracket{\rho_{A\lrbracket{B\Bar{B}}_1^n}}=\rho_{A\lrbracket{B\Bar{B}}_1^n},\,\\
            \Upsilon_{\lrbracket{B_L}_n\rightarrow C_{B_L}}\circ \Tr_{\Bar{B}_n}\lrbracket{\rho_{\lrbracket{B\Bar{B}}_1^n}} = K_{C_{B_L}}\otimes \rho_{\lrbracket{B\Bar{B}}_1^{n-1}\lrbracket{B_R}_n},\, & \Xi_{B_n\rightarrow B_n}\circ \Tr_{\Bar{B}_n}\lrbracket{\rho_{\lrbracket{B\Bar{B}}_1^n}}=\rho_{\lrbracket{B\Bar{B}}_1^{n-1}B_n},\,\\
        \end{array}
    \end{split}
    \end{align}
    we have
     \begin{align}\label{eq:intro-definettiI}
        \left\lVert \rho_{AB\Bar{B}} - \sum_{x\in\mathcal{X}}p(x)\rho_A^x\otimes\rho_{B\Bar{B}}^x \right\rVert_1\leq\mathcal{O}\left(\frac{|B\bar{B}|}{\sqrt{n}}\cdot\sqrt{\ln{|A|}}\right),\,
    \end{align}
    where $\forall x\in\mathcal{X}:$
    \begin{align}\label{eq:intro-definettiII}
        \begin{split}
            \begin{array}{cc}
                  \Theta_{A_L\rightarrow C_{A_L}}\lrbracket{\rho_A^x}=W_{C_{A_L}}\otimes \rho_{A_R}^x,\, &  \Omega_{A\rightarrow A}\lrbracket{\rho_A^x}=\rho_A^x,\\
                \Upsilon_{B_L\rightarrow C_{B_L}}\circ\Trr{\bar{B}}{\rho_{B\bar{B}}^x}=K_{C_{B_L}}\otimes \rho_{B_R}^x,\, &  \Xi_{B\rightarrow B}\circ\Trr{\bar{B}}{\rho_{B\bar{B}}^x}=\rho_{B}^x\,.
            \end{array}
        \end{split}
    \end{align}
   Together with a tailor-made purification scheme based on the pretty-good purification, the constrained Bose-symmetric de Finetti theorem enables the formulation of a novel Bose-symmetric SDP hierarchy that yields converging outer approximations for any constrained separability problem. We refer to \autoref{sec:bose-symmetric} for precise definitions and full details.\\

	\item \label{intro_bose_algo} Achieving an overall polynomial runtime requires not only that the SDP be of polynomial size, but also that it can be constructed in polynomial time. We establish this by leveraging Bose-symmetric methods to construct the SDP directly in a symmetry-adapted basis. Based on these methods, the resulting algorithm for approximating the value of the game up to an additive $\epsilon$-error relies on an isomorphism to a full matrix algebra over $\CC$ of dimension $\text{poly}\lrbracket{\epsilon^{-1}}$. Crucially, we show that this isomorphism  \,---\, and the corresponding SDP \,---\, can be computed in time polynomial in $\epsilon^{-1}$, thereby yielding the complexity bound stated in \autoref{eq:main-complexity}.\\
    
    \item \label{intro_sdp_sym} Alternatively, starting from the permutation-symmetry-based methods and using the conceptually distinct mathematical framework for standard symmetry reduction from e.g.\ \sloppy \cite{chee2023efficient, litjens2017semidefinite, fawzi2022hierarchy, gijswijt2009block}, we develop a similar polynomial-time algorithm based on an isomorphism into blocks of full matrix algebras over $\CC$ of size $\text{poly}\lrbracket{\epsilon^{-1}}$. In regimes where the underlying system dimensions are large, this algorithm, although structurally more involved than its Bose-symmetric counterpart, is computationally more efficient in practice. We refer to \autoref{sec:previous} for earlier work in this direction.\\
    
	\item Finally, we leverage the outer approximations from the aforementioned SDP hierarchies to develop an effective rounding procedure with rigorous convergence guarantees (cf.\ our concurrent work \cite{kossmann2025symmetric}). This produces feasible strategies in terms of shared entangled states with corresponding conditional measurements, which provide lower bounds on the game's value.
	\end{enumerate}
	
	Although we developed and applied techniques to address a game-theoretic instance of optimization over the constrained separability set, our methods hold promise for future applications to other problems within this class. For example, in the context of approximate quantum error correction \cite{berta2021semidefinite}, a symmetry reduction method based on the constrained Bose-symmetry could be explored as an alternative to the standard SDP symmetry reduction approach in \cite{chee2023efficient} (cf.\ our concurrent work \cite{kossmann2025symmetric}). 


\subsection{Methods}

Our approach integrates concepts from quantum information theory and representation theory to derive the key results, utilizing ideas from quantum steering, symmetric subspace methods, Schur-Weyl duality, Weyl's tensorial construction, matrix block decomposition, and ring theory. Below, we outline the mathematical framework and procedural steps employed in our analysis.
	
We reformulate non-local games as constrained bipartite separability problems and, using an adapted Doherty--Parrilo--Spedalieri (DPS) hierarchy \cite{Doherty_2004, Doherty_2005}, construct a hierarchy of outer approximations with convergence rates quantified by an information-theoretic de Finetti representation theorem adapted to our setting. A tailored SDP rounding procedure, supported by convergence guarantees derived from a reverse application of the aforementioned de Finetti theorem, is employed to obtain suitable lower bounds and feasible strategies. To address the computational complexity of determining the value of fixed-size two-player free non-local games with bounded quantum assistance, we leverage the inherent symmetry of states in DPS hierarchies. The algorithmic procedure mentioned in \ref{intro_bose_algo} extends well-known symmetric subspace methods to accommodate additional linear constraints. While existing techniques rely on representing unconstrained separable states as a convex mixture of pure product states, one cannot in general choose such a decomposition of constrained separable states so that each term individually respects the additional linear constraints. We address these challenges using a local purification scheme in conjunction with a constrained Bose-symmetric de Finetti representation theorem as sketched in \autoref{eq:intro-definettiI} and \autoref{eq:intro-definettiII}. The resulting polynomial-size SDP representation is then obtained from a weight-space decomposition of the symmetric subspace. The procedure mentioned in \ref{intro_sdp_sym} implements an SDP symmetry reduction in the spirit of e.g.\ \sloppy \cite{chee2023efficient, litjens2017semidefinite, fawzi2022hierarchy, gijswijt2009block, vallentin2009symmetry} by decomposing the support space under the action of the exchange symmetry group.

The remainder of this manuscript is structured as follows. We first discuss previous related work in \autoref{sec:previous}. Then, we present in \autoref{sec:results} an overview of our results in the context of the quantum information literature and techniques used to derive them. The section concludes with an outlook on potential future applications and open problems. Thereafter, we present the corresponding proofs in detail. This includes the following steps: In \autoref{sec:non_lacal_games_as_cbos}, we reformulate fixed-size non-local games as a bipartite constrained separability problem. Then, in \autoref{sec:inner_sequence}, we construct a sequence of inner approximations to these games. \autoref{sec:bose-symmetric} introduces a constrained Bose-symmetric hierarchy, along with a de Finetti theorem, which serves to approximate the values of the games from above. The efficient computation of this hierarchy is demonstrated in \autoref{sec:symmetric_subspace_methods}, leading to our main contribution: a polynomial-time algorithm for approximating the value of fixed-size free two-player non-local games. Finally, we conclude with an SDP symmetry reduction in \autoref{sec:sdp_symmetry_reduction}, providing an alternative polynomial-time algorithm. Further material, including supplementary derivations and extended analyses, is provided in \autoref{sec:sdp} -- \autoref{sec:examples}.


\subsection{Previous work}
\label{sec:previous}

Most relevant for the present work, \cite{jee2020quasi} provide insights into how well $w_{Q(T)}(G)$ can be approximated by formulating an equivalent constrained tripartite separability problem. The authors formulate a corresponding SDP relaxation\,---\,inspired by but also crucially different from the plain DPS-hierarchy (cf.\ the discussion in our concurrent work \cite{kossmann2025symmetric}) \cite{Doherty_2005, Doherty_2004, navascues2009power, navascues2009complete}\,---\,in terms of strategies with a certain exchange symmetry rather than separability, obtaining a hierarchy of improving approximations converging to the true optimal value. The proof of convergence employs a de Finetti representation theorem based on information-theoretic arguments (see e.g.\ \sloppy \cite{caves2002unknown, konig2005finetti,christandl2007one, koenig2009most, renner2009finetti, chiribella2010quantum, brandao_harrow_2014, berta2021semidefinite, harrow2013church, ligthart2023convergent, gross2021schur} for variants of de Finetti theorems in quantum physics), bounding the trace distance between these symmetric approximate strategies subject to linear constraints and genuinely feasible strategies. \cite{jee2020quasi} provide an SDP of size 
\begin{align}
	\exp\lrrec{\epsilon^{-2} \cdot\mathcal{O}\lrbracket{\text{poly}\lrvert{T}\cdot\log\lrvert{A}\lrbracket{\log\lrvert{A}+\log\lrvert{Q}}}}
\end{align}
\,---\,quasi-polynomial in $A$, polynomial in $Q$, but exponential in $\epsilon^{-1}$\,---\,to estimate $w_{Q(T)}(G)$ up to an additive $\epsilon$-error. Their result differs from those of \cite{navascues2015characterizing, navascues2015bounding}, which focused more on practical implementations than on rigorous analytical guarantees. Nevertheless, the $\epsilon$-dependence introduced by symmetry constraints poses significant challenges for practical applications where one is typically not interested in the scaling in $A$ and $Q$, but rather has fixed-size $A$ and $Q$ (e.g.\, equal to two) and would then like to approximate the optimal game value up to some prespecified additive approximation error $\epsilon>0$.

This motivates us to focus our study on these symmetries in order to effectively reduce the computational cost. Similarly, \cite{Doherty_2004, litjens2017semidefinite, ioannou2021noncommutative, Klerk2007ReductionOS, rosset2018symdpoly, gijswijt2009block, vallentin2009symmetry} exploited the inherent symmetry in various optimization problems to formulate classical algorithms providing solutions in polynomial time. However, the symmetric objects in these works need not adhere to any linear constraints. This is a subtlety that renders any direct adaptation to our scenario ineffective. Recently, \cite{chee2023efficient} provided a novel symmetry-reduced hierarchy approximating the quantum channel fidelity in a coding task. To our knowledge, this constitutes the first time that such a symmetry reduction in the presence of additional linear constraints was carried out. Lastly, although Bose-symmetry-based reductions (e.g.\ \cite{christandl2007one, navascues2009power}) are available, they are not applicable to the constrained case.  

\section{Overview of results}\label{sec:results}

\subsection{Two-player non-local games}\label{sec:results_games_as_cbo}

This work focuses on the standard setting of two-player (provers) one-round free non-local games. In these games $G$, two spatially separated cooperative players, Alice and Bob, play against a third party, the referee (verifier). The referee chooses questions $q_1,\,q_2$ with probabilities $\pi_1(q_1),\,\pi_2(q_2)$ from given finite alphabets $Q_1, \, Q_2$ and distributes them to Alice and Bob, respectively. We assume the description of $\pi$ is polynomial in the size of $Q_i$ \cite{kempe2011entangled}. Conditioned on the questions the players receive, they return answers $a_1,\, a_2$ from finite alphabets $A_1,\, A_2$, respectively. To systematically judge the success of the players in answering the referee's questions, the \textit{rule function}
\begin{align}
    V:A_1\times A_2\times Q_1\times Q_2 \rightarrow \lbrace 0,1\rbrace
\end{align}
takes the value $V(a_1,a_2,q_1,q_2)=1$ whenever the players win the game, and $0$ otherwise. While the players cannot communicate once the referee has distributed the questions, they can strategize beforehand. 

To avoid the pathologies exhibited by general quantum-assisted games, including undecidability of value approximation and the distinction between tensor-product and commuting-operator strategies \cite{ji2021mip}, we follow \cite{navascues2014characterization, navascues2015characterizing, navascues2015bounding, jee2020quasi} and consider only games in which the players’ quantum assistance is dimensionally bounded. In these games, both parties agree on an arbitrary shared bipartite quantum state $\rho_{T\hat{T}} \in \mathcal{S}(\mathcal{H}_T \otimes \mathcal{H}_{\hat{T}})$\footnote{$\HS_{\hat{T}}$ is an isomorphic copy of $\HS_{T}$, thus $\dim_{\CC}(\HS_T)=\dim_{\CC}(\HS_{\hat{T}})=\lrvert{T}$.} beforehand. Upon receiving questions $q_1,\, q_2$, they perform measurements from a complete set of local POVMs $\{E_T(a_1 \vert q_1)\}_{a_1} \subset \mathcal{B}(\mathcal{H}_T)$ and $\{D_{\hat{T}}(a_2 \vert q_2)\}_{a_2} \subset \mathcal{B}(\mathcal{H}_{\hat{T}})$ and then return the answers to the referee. The completeness condition on Alice's and Bob's respective sets of POVMs, imposed for each question separately, ensures that the resulting strategy is non-signaling. Neither Alice nor Bob can communicate their respective question to the other party via measurements. Thus, the value of the game $w_{Q(T)}(V,\pi)$ is determined by solving the following optimization problem:
     \begin{align}\label{eqn:numerical_free_games}
        \begin{split}
            w_{Q(T)}(V,\pi) = &\max_{(\rho,\, E, \,D)}\hspace{0.5cm}\sum_{q_1,q_2}\pi_1(q_1)\pi_2(q_2)\sum_{a_1,a_2}V(a_1,a_2,q_1,q_2)\Tr\left[\left(E_T(a_1\vert q_1)\otimes D_{\hat{T}}(a_2\vert q_2)\right)\rho_{T\hat{T}}\right]\\
            \\
            \text{s.t.}\quad &\rho_{T\hat{T}} \succcurlyeq 0,\, \quad\Tr\left[\rho_{T\hat{T}}\right]=1,\quad\text{state condition,}\\
            & E_{T}(a_1\vert q_1) \succcurlyeq 0\quad\forall a_1,q_1,\,\hspace{1cm} \sum_{a_1}E_{T}(a_1\vert q_1)=\mathbb{1}_{T}\quad\forall q_1,\quad\text{POVMs,}\\
            & D_{\hat{T}}(a_2\vert q_2)\succcurlyeq 0\quad \forall a_2,q_2,\,\hspace{1cm} \sum_{a_2}D_{\hat{T}}(a_2\vert q_2) =\mathbb{1}_{\hat{T}}\quad \forall q_2,\quad\text{POVMs.} 
        \end{split}
    \end{align}
     In this paper, we derive new outer and inner bounds on $w_{Q(T)}(V, \pi)$, along with algorithms that, for games of fixed size, compute $w_{Q(T)}(V, \pi)$ up to an additive $\epsilon$-error in time and space polynomial in $\epsilon^{-1}$.

     \begin{figure}[!htbp]
    \centering
    \begin{tikzpicture}[
  >=Latex,
  font=\small,
  box/.style={
    draw,
    rounded corners=3pt,
    align=center,
    inner sep=6pt
  },
  referee/.style={
    box,
    fill=orange!10,
    minimum width=4.8cm,
    minimum height=1.2cm
  },
  player/.style={
    box,
    fill=blue!5,
    minimum width=3.7cm,
    minimum height=1.55cm
  },
  state/.style={
    box,
    fill=green!8,
    minimum width=4.4cm,
    minimum height=1.0cm
  },
  smallbox/.style={
    box,
    fill=gray!5,
    font=\footnotesize
  },
  arr/.style={->, thick},
  quantum/.style={
    ->,
    thick,
    decorate,
    decoration={snake, amplitude=.45mm, segment length=2.8mm}
  }
]

\node[referee] (R) at (0,3.2) {
  \textbf{Referee}\\
  samples $(q_1,q_2)\sim \pi_1\times \pi_2$
};

\node[smallbox, right=1.2cm of R] (V) {
  accepts iff $V(a_1,a_2,q_1,q_2)=1$
};

\draw[arr] (R) to (V);

\node[player] (A) at (-4,0) {
  \textbf{Alice}\\[1mm]
  input $q_1$\\
  output $a_1$
};

\node[player] (B) at (4,0) {
  \textbf{Bob}\\[1mm]
  input $q_2$\\
  output $a_2$
};

\node[state] (rho) at (0,-2.2) {
  shared state $\rho_{T\hat T}$
};

\node[smallbox] (NC) at (0,-0.9) {
  no communication
};

\draw[arr]
  (R.south west) to[bend right=12]
  node[left, pos=.55] {$q_1$}
  (A.north);

\draw[arr]
  (A.north east) to[bend right=12]
  node[below right, pos=.48] {$a_1$}
  (R.south);

\draw[arr]
  (R.south east) to[bend left=12]
  node[right, pos=.55] {$q_2$}
  (B.north);

\draw[arr]
  (B.north west) to[bend left=12]
  node[below left, pos=.48] {$a_2$}
  (R.south);

\draw[quantum]
  (rho.north west) to[bend right=8]
  node[left, pos=.25] {measurement $E_{T}(a_1\vert q_1)$ $\quad$ }
  (A.south);

\draw[quantum]
  (rho.north east) to[bend left=8]
  node[right, pos=.25] {$\quad$ 
  measurement $D_{\hat T}(a_2\vert q_2)$}
  (B.south);

\draw[dashed, thick, red!70!black] (A.east) -- (B.west);
\node[red!70!black, fill=white, inner sep=1pt] at (0,0) {$\not\leftrightarrow$};

\end{tikzpicture}
    \caption{Two-player free non-local game with fixed-dimensional entanglement assistance. The referee samples questions $q_i\sim \pi_i$ from finite question sets $Q_i$, sends $q_1$ to Alice and $q_2$ to Bob, and evaluates the returned answers using the predicate $V: A_1\times A_2\times Q_1\times Q_2\to\{0,1\}$.Before the questions are sent, the spatially separated players share a bipartite state $\rho_{T\hat T}$ of fixed local dimension. Upon receiving $q_1$, Alice measures her part of the state with the POVM element $E_T(a_1\vert q_1)$; upon receiving $q_2$, Bob measures his part of the state with the POVM element $D_{\hat T}(a_2\vert q_2)$. The outcomes $a_1$ and $a_2$ of these measurements are the answers returned to the referee, who accepts iff $V(a_1,a_2,q_1,q_2)=1$. No communication is allowed between the players after they receive their questions.}
    \label{fig:non_local_games}
\end{figure}
     
\paragraph{\textbf{Notation}} 
Unless stated otherwise, all Hilbert spaces considered in this work are finite-dimensional. In particular, we write $\HS$ for a Hilbert space of dimension $\lrvert{\HS}$, and identify it with $\CC^{\lrvert{\HS}}$ when convenient. We assume $\lrvert{A_1}=\lrvert{A_2}=\lrvert{A}$ and  $\lrvert{Q_1}=\lrvert{Q_2}=\lrvert{Q}$ and abbreviate $\lrvert{A}\lrvert{Q}\lrvert{T}=\lrvert{AQT}$. We often abbreviate $\HS_{T}$ to $T$ of dimension $\lrvert{T}$, denote the identity channel by $\mathcal{I}_T$ and the identity matrix by $\mathbb{1}_T$. We write $\End{}{\HS}$ for the algebra of linear endomorphisms on $\HS$, and $\mathcal{B}(\HS)$ for the algebra of bounded linear operators on $\HS$. In finite dimensions these notions coincide. We use $\succeq$ and $\succcurlyeq$ for the Loewner (partial) order on operators. The binary logarithm is denoted by $\log$ and the natural logarithm by $\ln$. By poly$(k)$, we denote a scaling that is at most polynomial in $k\in\mathbb{R}$. For $x\in\mathbb{R}$, we write $\lceil x\rceil$ for the smallest integer greater than or equal to $x$. By a standard abuse of notation, we use the same symbol for an optimization problem and for its optimal value.


\subsection{Non-local games as constrained separability problems}

In our first contribution, we avoid the complexities arising from the tripartite formulation in \cite{jee2020quasi}, such as multipartite entanglement  \cite{navascues2020genuine, walter2016multipartite, szalay2015multipartite, horodecki2009quantum}. We leverage techniques from quantum steering to reformulate the problem of determining the value of a non-local game with bounded quantum assistance as a linear optimization over bipartite separable states subject to additional linear constraints.

In this work, we consider the constrained separability problem (cSEP) (see also \cite{ohst2024characterising}) of the following form. Consider Hilbert spaces $A=A_L\otimes A_R$ and $B=B_L\otimes B_R$ and let $G_{AB}\in\mathcal{B}\lrbracket{AB}$ with $\left\lVert G_{AB} \right\rVert_{\infty}\leq 1$ be a problem-specific Hermitian operator. We define
\begin{align}\label{eqn:cSEP_Main}
    \begin{split}
        \begin{array}{ccc}
              &  \mathrm{cSEP}(G) := \displaystyle\max_{\rho_{AB}\in \mathcal{B}(AB)} \Tr[G_{AB}\rho_{AB}] &  \\
             &&\\
             \text{with} & \rho_{AB}=\sum_{x\in\mathcal{X}}p(x)\rho^x_A\otimes\rho^x_B & \text{s.th. }\forall x\in\mathcal{X}:\\
             &&\\
             & \rho^x_A\succeq 0,\, \rho^x_B\succeq 0,\,&   \Tr\left[\rho^x_A\right]=1,\, \Tr\left[\rho^x_B\right]=1,\,\\
             &&\\
             &p(x)\geq 0, & \sum_{x\in\mathcal{X}}p(x)=1, \\
             &&\\
             \multicolumn{2}{c}{\text{Local constraints }\quad \forall x\in\mathcal{X}:}&\\
             &&\\
             &\Theta_{A_L\rightarrow C_{A_L}}\lrbracket{\rho^x_A} = W_{C_{A_L}}\otimes \rho^x_{A_R},\, & \Upsilon_{B_L\rightarrow C_{B_L}}\lrbracket{\rho^x_B} = K_{C_{B_L}}\otimes \rho^x_{B_R} ,\, \\
             &&\\
             \multicolumn{2}{c}{\text{Fixed point constraints }\quad \forall x\in\mathcal{X}:}&\\
             &&\\
             &\Omega_{A\rightarrow A}\lrbracket{\rho^x_{A}}=\rho^x_{A} ,\, & \Xi_{B\rightarrow B}\lrbracket{\rho^x_{B}}=\rho^x_{B},\,\\
        \end{array}
    \end{split}
\end{align}
where 
\begin{align}
\begin{array}{cc}
      \Theta_{A_L\rightarrow C_{A_L}}\,:\, \mathcal{B}\lrbracket{A_L} \rightarrow \mathcal{B}\lrbracket{C_{A_L}}  &  \Upsilon_{B_L\rightarrow C_{B_L}}\,:\, \mathcal{B}\lrbracket{B_L} \rightarrow \mathcal{B}\lrbracket{C_{B_L}}  \\
\end{array}
\end{align}
are linear maps locally imposing a constraint on Alice's and Bob's marginal state via the constant operators $W_{C_{A_L}}\in\mathcal{B}\lrbracket{C_{A_L}},\, K_{C_{B_L}}\in\mathcal{B}\lrbracket{C_{B_L}}$ and 
\begin{align}
\begin{array}{cc}
      \Omega_{A\rightarrow A}\,:\, \mathcal{B}\lrbracket{A} \rightarrow \mathcal{B}\lrbracket{A}  &  \Xi_{B\rightarrow B}\,:\, \mathcal{B}\lrbracket{B} \rightarrow \mathcal{B}\lrbracket{B}\\
\end{array}
\end{align}
are linear maps enforcing fixed point constraints on those marginal states. This notation will be used throughout the remainder of the paper. Problems of this form represent strict generalizations of cSEP problems studied in \cite{berta2021semidefinite, jee2020quasi, ohst2024characterising}.

Although fixing the local dimension a priori circumvents the computational-hardness results associated with optimization over the set of separable states (see, e.g., \cite{gurvits2003proceedings, ioannou2007computationalcomplexityquantumseparability, gharibian2008strong, grotschel2012geometric, shi2015epsilon, fawzi2021set}), the tensor-product structure required in a separable decomposition gives rise to bilinear constraints, which are not directly amenable to standard convex-optimization frameworks. Concretely, while $\mathrm{cSEP}(G)$ is a convex program, it is in general not semidefinite-representable \cite{fawzi2021set}. To obtain an algorithmic approximation scheme, we adapt the DPS hierarchy \cite{Doherty_2004} to the constrained setting, in the spirit of \cite{berta2021semidefinite,jee2020quasi}. We then combine this hierarchy with a new information-theoretic finite de Finetti representation theorem in \autoref{lem:approximate_quantum_de_finetti}, tailored to the class of constrained separability problems introduced above. This yields the following approximation guarantee.

\begin{lemma}\label{lem:convergence_non_local_games_as_cbo_results}
    The value of a two-player free non-local game with $\lrvert{A}$ answers, $\lrvert{Q}$ questions, and quantum assistance of size $\lrvert{T}$ can be approximated up to an additive $\epsilon$-error by solving an SDP with at most
\begin{align}\label{eqn:exponential_growth_SDP_non_local}
    \exp\lrrec{\epsilon^{-2}\cdot\mathcal{O}\lrbracket{\lrvert{T}^4\log^2\lrvert{AQT}}}
\end{align}
real degrees of freedom.
\end{lemma}

Thus, the maximal winning probability with bounded quantum assistance can be approximated within an additive error $\epsilon > 0$ in \textit{quasi-polynomial time}\footnote{Following, e.g. \cite{jee2020quasi, ioannou2007computationalcomplexityquantumseparability, fang2025efficientapproximationregularizedrelative}, and disregarding lower-order contributions from iterative algorithms, we assume that a semidefinite program with $m$ variables and matrix size $k$ has time complexity $\mathcal{O}\lrbracket{m^2k^2}$.} with respect to the cardinalities of the answer and question alphabets. In comparison to the scaling result of \cite{jee2020quasi},
\begin{align}\label{eqn:hailey_complexity}
\exp\lrrec{\epsilon^{-2}\cdot\mathcal{O}\lrbracket{\lrvert{T}^{12}\log\lrvert{AT}\lrbracket{\log\lrvert{AT}+\log\lrvert{Q}}}},\,
\end{align}
we obtain improved bounds w.r.t.\ the quantum dimension $\lrvert{T}$ at the cost of a worse scaling in terms of $\lrvert{Q}$. Thus, our first contribution joins a sequence of works improving the upper bound on $w_{Q(T)}(V, \pi)$ for two-player free non-local games in a relevant regime. In this work, our primary contribution is the use of symmetry-based arguments to refine and simplify the approximation algorithm w.r.t.\ $\epsilon$.


\subsection{Approximating constrained separability problems via Bose-symmetry}

In the fixed-size game setting, we derive the following result on the efficiency of approximation.

\begin{lemma}\label{lem:Bose_symmetry_first_step}
	For a fixed-size game, the value $w_{Q(T)}(V, \pi)$ can be approximated up to an additive $\epsilon$-error by solving an SDP, denoted $\mathrm{SDP}_n^{\mathrm{Bose}}(V, \pi)$, with at most $\text{poly}\lrbracket{\epsilon^{-1}}$ real degrees of freedom. 
\end{lemma}

The reduction in complexity is fundamentally based on a symmetry argument intrinsic to $n$-extendable states, which are central to the DPS hierarchy. We lift the concept of Bose-symmetric states \cite{hudson1976locally, werner1989quantum, christandl2007one}, i.e.\ those with support and range in the symmetric subspace \cite{harrow2013church}, to optimization problems where additional linear constraints are present. On a technical level, these constraints preclude a straightforward adaptation of the techniques from \cite{Doherty_2004, ioannou2007computationalcomplexityquantumseparability, navascues2009power}, which rely on expressing separable states as convex mixtures of pure product states. Concretely, in the game setting, the DPS hierarchy works to approximate a convex mixture of POVM elements in a tensor product relation with state assemblages. While a further decomposition into pure objects is possible, these individual terms are in general no longer POVM elements or state assemblages. Through a tailored local purification, we construct an embedding of permutation-invariant states subject to linear constraints into the symmetric subspace of a higher-dimensional space, intersected with the corresponding constraints. We demonstrate that for any constrained separability problem, there exists a converging hierarchy of $n$-Bose-symmetric states that approximates the optimal value from above. This convergence is rigorously certified through a corresponding finite and constrained Bose-symmetric de Finetti representation theorem.

\begin{lemma}\label{lem:bose_sym_sdp_results}
	For $\mathrm{cSEP}(G)$ as in \autoref{eqn:cSEP_Main}, there exists a convergent SDP hierarchy $\mathrm{SDP}_n^{\mathrm{Bose}}(G)$ of outer approximations over Bose-symmetric states, such that:
	\begin{align}
		0 \leq \mathrm{SDP}_n^{\mathrm{Bose}}(G) - \mathrm{cSEP}(G) \leq \mathcal{O}\lrbracket{\text{poly}\lrbracket{\lrvert{B}}\sqrt{\frac{\ln\lrvert{A}}{n}}}.\,
	\end{align}
\end{lemma}
We then show that the optimization in $\mathrm{SDP}_n^{\mathrm{Bose}}(G) $ can be restricted to the spaces of $n$- and $(n-1)$-Bose-symmetric states. Although these spaces have polynomial dimension in $n$, the computational cost of determining an explicit representation of $\mathrm{SDP}_n^{\mathrm{Bose}}(G) $ in terms of objects of polynomial size in $n$ cannot be neglected. In implementations \cite{Doherty_2004, ioannou2007computationalcomplexityquantumseparability, navascues2009power}, this bottleneck is not directly discussed, as the complexity of the SDP is primarily dictated by the number of free variables. While this line of argument is perfectly viable for most practical purposes, it lacks the mathematical rigor necessary to analyze the computational complexity of algorithmic procedures to determine the value of a non-local game. Following e.g.\ \cite{gijswijt2009block, polak2020new, chee2023efficient}, and as detailed in \autoref{sec:results_efficient_Bose}, we address this gap by providing a rigorous analysis of the computational costs associated with transforming $\mathrm{SDP}_n^{\mathrm{Bose}}(G) $ into a representation involving only objects of size polynomial in $n$. To the best of our knowledge, this work constitutes the first use of the symmetric subspace to approximate constrained separability problems, thereby resolving an open question raised in \cite{jee2020quasi, berta2021semidefinite}, while also presenting the first detailed complexity analysis of methods based on symmetric subspaces.

Crucially, the results of the previous subsection enable the approximation of $w_{Q(T)}(V,\pi)$ via the proposed Bose-symmetric hierarchy, yielding an approximation algorithm with runtime polynomial in the inverse of the approximation error.


\subsection{Convergent sequence of inner bounds to constrained separability problems}

While there exists a variety of outer approximations to relevant optimization problems obtained from suitable relaxations of the corresponding constraints, obtaining good inner approximations is often not as straightforward. Due to the non-convexity of the optimization landscape, see-saw methods may become trapped in local optima, and the approaches in \cite{navascues2009complete, ohst2023certifying} are not compatible with the additional linear constraints imposed in our setting. Using an effective rounding procedure based on information-theoretic tools, we construct a sequence of inner approximations for constrained separability problems, accompanied by rigorous analytical convergence guarantees, which provide lower bounds on $\mathrm{cSEP}(G)$. With respect to non-local games, the rounding procedure turns a solution of an SDP relaxation into a solution of the actual game.

\begin{theorem}\label{thm:short_version_inner_sequence}
There exists a rounding procedure that converts an $\epsilon$-good approximation of $w_{Q(T)}(V, \pi)$ from above, obtained from either the SDP in \autoref{lem:convergence_non_local_games_as_cbo_results} or $\mathrm{SDP}_n^{\mathrm{Bose}}(G)$ in \autoref{lem:bose_sym_sdp_results}, into an $\epsilon$-good approximation from below, given by a feasible strategy consisting of Alice's and Bob's sets of complete POVMs and a shared quantum state of fixed dimension.
\end{theorem}

We establish an achievability result via an explicit rounding procedure, shown schematically in \autoref{fig:schematic_rounding}. The obtained inner points not only serve as warm starts for see-saw procedures but also, thanks to the rigorous convergence guarantees, enable us to establish a complexity result for an algorithmic procedure that determines $w_{Q(T)}(V, \pi)$ with an additive $\epsilon$-error from both above and below.

\begin{figure}[htbp!]
    \centering
    \begin{tikzpicture}[x=1cm,y=1cm]

\node[Asys] (A) at (0,0) {$A$};
\node[Bsys] (b1)  at (1.8,0) {$B_1$};
\node[fill=white,inner sep=1.5pt] (bd1) at (2.6,0) {$\cdots$};
\node[Bsys] (bk)  at (3.4,0) {$B_k$};
\node[Ksys] (bk1) at (4.9,0) {$B_{k+1}$};
\node[Dsys] (bk2) at (6.3,0) {$B_{k+2}$};
\node[fill=white,inner sep=1.5pt] (bd2) at (7.1,0) {$\cdots$};
\node[Dsys] (bn)  at (7.9,0) {$B_n$};

\begin{scope}[on background layer]
  \node[draw=Aline!55,dashed,rounded corners=6pt,line width=.8pt,
        inner sep=5.5mm,fill=Aline!4,fit=(A)(b1)(bn)] (state) {};
\end{scope}
\begin{scope}[on background layer]
  \draw[Aline!50,decorate,decoration={snake,amplitude=.5mm,segment length=2.4mm},
        line width=.6pt] (A.east) -- (bn.west);
\end{scope}
\node[note,anchor=south] at (state.north)
   {$\rho\st_{AB_1^{\,n}}$ : optimizer of $\mathrm{SDP}_n(G)$, \ permutation--symmetric over $B_1,\dots,B_n$};

\def\ym{-3.2}
\draw[qwire] (b1.south) -- (1.8,\ym+0.35);  \pic at (1.8,\ym) {meter};
\draw[qwire] (bk.south) -- (3.4,\ym+0.35);  \pic at (3.4,\ym) {meter};
\node at (2.6,\ym) {$\cdots$};
\node[note,text=Bline,fill=white,inner sep=3pt,text width=4.5cm] at (2.6,-1.5)
   {measure copies $B_1,\dots,B_k$ with an\\ informationally--complete POVM};

\node[lbl,text=Kline,anchor=west] at (5.2,-1.4) {retain};

\draw[qwire,Dline] (bk2.south) -- (6.3,\ym+0.35);  \pic[Dline] at (6.3,\ym) {trace};
\draw[qwire,Dline] (bn.south)  -- (7.9,\ym+0.35);  \pic[Dline] at (7.9,\ym) {trace};
\node[Dline] at (7.1,\ym+0.15) {$\cdots$};
\node[note,Dline] at (7.1,\ym-0.55) {trace out (discard)};

\node[sys,draw=Sline!55,fill=Sline!5,rounded corners=5pt,line width=.8pt,
      minimum width=6.2cm,minimum height=1.25cm] (cbox) at (2.6,-5.4) {};
\node[Asys] (cA) at (1.2,-5.4) {$\rho\st_{A\,|\,\zk}$};
\node[font=\large]     at (2.6,-5.4) {$\otimes$};
\node[Ksys] (cB) at (4.15,-5.4) {$\rho\st_{B_{k+1}\,|\,\zk}$};

\draw[qwire] (A.south)   -- (0,-4.35)   -- (1.2,-4.35)  -- (1.2,-4.775);
\draw[qwire] (bk1.south) -- (4.9,-4.35) -- (4.15,-4.35) -- (4.15,-4.775);
\draw[cwire] (1.8,\ym-0.35) -- (1.8,-4.775);
\node[lbl,anchor=west] at (1.92,-4.0) {$z_1$};
\draw[cwire] (3.4,\ym-0.35) -- (3.4,-4.775);
\node[lbl,anchor=west] at (3.52,-4.0) {$z_k$};

\node[note,anchor=west,Sline,text width=3.2cm] at (5.85,-5.4)
   {conditional \emph{product} state; outcome $\zk$ with prob.\ $p(\zk)$; feasible for $\mathrm{cSEP}(G)$};

\node[sys,draw=Sline,fill=Sfill,line width=1pt,rounded corners=4pt,
      minimum width=10.6cm,minimum height=1.2cm] (mix) at (2.6,-7.5)
   {$\displaystyle \sigma^{(n,k)}_{AB_{k+1}}
        \;=\; \sum_{\zk\in Z^{k}} p(\zk)\,
              \rho\st_{A\,|\,\zk}\otimes\rho\st_{B_{k+1}\,|\,\zk}$};
\draw[qwire] (cbox.south) -- (mix.north);
\node[note,Sline,anchor=north] at (mix.south)
   {feasible \emph{separable} candidate for $\mathrm{cSEP}(G)$};

\node[draw=Gline!45,rounded corners=3pt,fill=black!2,line width=.6pt,
      note,anchor=north,text width=10.5cm,inner sep=5pt] (dfbox) at (2.6,-9.05)
   {\textbf{Convergence (constrained de Finetti).} There is an index
    $k_n\in\{1,\dots,n-1\}$ with
    \[0\;\le\; \mathrm{cSEP}(G)-\operatorname{tr}\!\big[G_{AB}\,\sigma^{(n,k_n)}\big]\;\le\;\epsilon(|A|,|B|,n),\]
    where $\epsilon\to0$ as $n\to\infty$.};

\end{tikzpicture}
    \caption{\textbf{Measurement-based rounding scheme.}
    Schematic depiction of the measurement-based rounding scheme of \autoref{lem:inner_sequence_general}. The optimizer $\rho^{\star}_{AB_1^{n}}$ of the level-$n$ symmetric-extension relaxation $\mathrm{SDP}_n(G)$ is a state on system $A$ together with $n$ permutation-symmetric copies $B_1,\dots,B_n$. The first $k$ copies are measured with a fixed informationally complete POVM $\mathcal{M}_{B}$, giving an outcome string $z_1^{k}$; the copy $B_{k+1}$ is retained while $B_{k+2},\dots,B_n$ are traced out. Conditioned on $z_1^{k}$ the retained systems are in the \emph{product} state $\rho^{\star}_{A\mid z_1^{k}}\otimes\rho^{\star}_{B_{k+1}\mid z_1^{k}}$, which is feasible and separable, and averaging over the outcomes yields the feasible separable candidate $\sigma^{(n,k)}_{AB_{k+1}}=\sum_{z_1^{k}\in Z^{k}}p(z_1^{k})\,
  \rho^{\star}_{A\mid z_1^{k}}\otimes\rho^{\star}_{B_{k+1}\mid z_1^{k}}$ for $\mathrm{cSEP}(G)$. The constrained de Finetti bound provides an index $k_n\in\{0,\dots,n-1\}$ with
  $0\le\mathrm{cSEP}(G)-\operatorname{tr}\!\big[G_{AB}\,\sigma^{(n,k_n)}\big]
  \le\epsilon(|A|,|B|,n)$ and $\epsilon\to0$ as $n\to\infty$; already the single best outcome $z^{\star}=\arg\max_{z_1^{k}}\operatorname{tr}\!\big[
  G_{AB}\big(\rho^{\star}_{A\mid z_1^{k}}\otimes
  \rho^{\star}_{B_{k+1}\mid z_1^{k}}\big)\big]$ attains it.}
    \label{fig:schematic_rounding}
\end{figure}


\subsection{Polynomial-time algorithm approximating quantum non-local games via symmetry reductions}\label{sec:results_efficient_Bose}

We provide an efficient classical transformation realizing \autoref{lem:Bose_symmetry_first_step} by representing $\mathrm{SDP}_n^{\mathrm{Bose}}(G) $ in terms of a Schur basis, where objects are of size polynomial in $n$. Specifically, we develop a polynomial-time algorithm for approximating the value of a two-player free non-local game with bounded quantum assistance, leveraging Bose-symmetry arguments.
\begin{theorem}\label{thm:Bose_sym_complexity_results}
    For a two-player free non-local game with $\lrvert{A}$ answers, $\lrvert{Q}$ questions, and quantum assistance of size $\lrvert{T}$, there exists a poly$\lrbracket{\epsilon^{-1}}$-costly transformation (see \autoref{sec:bose-symmetric} and \autoref{sec:symmetric_subspace_methods} for further details) to an SDP with at most 
	 \begin{align}
		\lrbracket{\epsilon^{-4}\cdot \mathrm{poly}\lrbracket{\lrvert{A}, \lrvert{Q}, \lrvert{T}}}^{\lrbracket{\lrvert{A}\lrvert{Q}\lrvert{T}}^2},
	\end{align} real degrees of freedom, which approximates the value of the game up to an additive $\epsilon$-error.
\end{theorem}

In particular, for fixed-size games, we can approximate the optimal value of the game in time $\text{poly}\lrbracket{\epsilon^{-1}}$.

For any $\mathrm{cSEP}(G)$, at the core of an efficient (in $n$) representation of $\mathrm{SDP}_n^{\mathrm{Bose}}(G) $ lies the key observation that permutation symmetry \,---\, in particular, Bose symmetry \,---\, induces significant redundancy in representations that are not adapted to the underlying symmetry. Specifically, a Bose-symmetric state expressed in the canonical basis results in a matrix of exponential size in $n$, yet contains only polynomially many distinct entries. Consequently, identifying a basis that eliminates these redundancies enables a compact, symmetry-adapted, and efficient representation of $\mathrm{SDP}_n^{\mathrm{Bose}}(G)$, denoted by $\Psi\lrbracket{\mathrm{SDP}_n^{\mathrm{Bose}}(G)}$. It is crucial to avoid first constructing $\mathrm{SDP}_n^{\mathrm{Bose}}(G)$, followed by a basis transformation to obtain $\Psi\lrbracket{\mathrm{SDP}_n^{\mathrm{Bose}}(G)}$, as this would require handling exponentially large objects in $n$. Instead, we show how $\Psi\lrbracket{\mathrm{SDP}_n^{\mathrm{Bose}}(G)}$ can be constructed directly, in a manner that is efficient with respect to $n$.

For the mathematically inclined reader, we will briefly elaborate on the mathematical concepts that motivate \autoref{thm:Bose_sym_complexity_results}. While the symmetric subspace has been the subject of extensive study (see, e.g.\, \cite{harrow2013church, renner2008security, christandl2006structure}), comparatively less attention has been given to the space of (constrained) Bose-symmetric states and its elements (cf.\ \cite{gulati2026entanglementdickesubspace}). The representation theory of the symmetric group $S_n$ and the general linear group $\GLH$ together with Schur-Weyl duality \cite{fulton2013representation, stevens2016schur, christandl2006structure} offer a robust mathematical framework for advancing our understanding of the space of Bose-symmetric states. Firstly, we show that the space of these states is isomorphic to a full matrix algebra over $\CC$ with vector space dimension of polynomial size in $n$. This result leverages the $S_n\times\GLH$-bimodule structure of the symmetric subspace $\SymH\subseteq \HSn$, as established through Schur-Weyl duality. As an $S_n\times\GLH$-bimodule, $\SymH$ decomposes as the tensor product of a simple Specht module and a simple Weyl (or Schur) module.  However, as an $S_n$-module it decomposes into a direct sum of isomorphic trivial Specht modules. The multiplicity of this isotypic block is given by the $\CC$-vector space dimension of the corresponding Weyl module in the bimodule decomposition
\begin{align}
	m_{(n)}=\binom{n+\lrvert{\HS}-1}{n}.\,
\end{align}
As a $\GLH$-module, $\SymH$ is isomorphic to a highest weight module given by the Weyl module in the bimodule decomposition. The multiplicity is one, as this module pairs with the trivial Specht module in $\SymH$, which has $\mathbb{C}$-vector space dimension one. Since the span of the image of $\GLH$ in $\End{}{\SymH}$ is $\End{\CSn}{\SymH}$, i.e.\ the space of $\CSn$-invariant endomorphisms on the symmetric subspace, we can reason about Bose-symmetric states via highest-weight theory, i.e.\ the representation theory of $\GLH$. Furthermore, the double-centralizer theorem and Maschke's theorem certify that the image of the complex group algebra $\CSn$ in $\End{}{\SymH}$ is $\End{\GLH}{\SymH}$ and the centralizer of $\End{\CSn}{\SymH}$. As $\SymH$ is simple as a $\GLH$-module, by Schur's lemma, the $\GLH$-commutant in $\End{}{\SymH}$ is a division ring and therefore $\CC\cdot\text{Id}\vert_{\SymH}$, i.e.\ $\End{\GLH}{\SymH}\simeq \CC$. The corresponding centralizer algebra $\End{\CSn}{\SymH}$ is then just $\End{}{\SymH}$, and thus
 \begin{align}\label{eqn:isomorphism_results}
	\End{\CSn}{\SymH} \simeq \CC^{m_{(n)}\times m_{(n)}}.\,
\end{align}
To concretely realize the isomorphism, we consider a weight space decomposition of the highest weight module $\SymH$. The weight spaces are isomorphic as $S_n$-modules. We utilize the correspondence with semistandard Young labelings of the single-row Young shape containing $n$ boxes. This weight structure is subsequently employed to construct a canonical basis for the space of Bose-symmetric states. Given that these weight spaces are spanned by $S_n$-orbit averages, the isomorphism is realized as a change of basis from a representation in terms of the canonical basis of $\HSn$ to a basis of weight vectors spanning simple $S_n$-modules. Based on these insights, for any Weyl-Specht module pairing in $\HSn$, i.e.\ for any such $S_n\times\GLH$-bimodule, we provide an explicit bijection, preserving positive semidefiniteness, that realizes a corresponding isomorphism as in \autoref{eqn:isomorphism_results}. To complete the proof, we demonstrate that this change of basis transformation can be performed efficiently with respect to $n$ whenever we restrict the constraint maps to be partial trace maps. To achieve this, we avoid explicitly constructing the canonical basis for the space of Bose-symmetric states, as these objects grow exponentially in size with $n$. Using symmetry arguments, we demonstrate that it is sufficient to restrict attention to a number of entries polynomial in $n$ for any Bose-symmetric state. We show that the transformation in \autoref{thm:Bose_sym_complexity_results} can be efficiently computed using only these entries. Thus, we provide a polynomial-time algorithm to determine the value of a two-player non-local game up to an additive error.

\begin{figure}[!htbp]
    \centering
    \input{bose_vs_symmetric.tex}
    \caption{\textbf{Symmetric versus Bose-symmetric operators on $B_1^n$.}
In the Schur--Weyl decomposition of $B_1^n=B^{\otimes n}$, an operator that is
\emph{symmetric} (commutes with the permutation action,
$U(\pi)\,X\,U(\pi)^{\top}=X$ for all $\pi\in S_n$) is block diagonal across the
$S_n$-isotypic sectors indexed by Young diagrams $\lambda\vdash n$, i.e.
$X=\bigoplus_{\lambda}X_\lambda\in\operatorname{End}_{\mathbb{C}[S_n]}(B_1^n)$
(left). A \emph{Bose-symmetric} operator ($P_{\vee^n}\,X\,P_{\vee^n}=X$) is in
addition supported on the symmetric subspace $\vee^n(B)$ --- the
trivial-representation sector $\lambda=(n)$ --- so only the block $X_{(n)}$ survives
and $X\in\operatorname{End}(\vee^n(B))$ (right). Diagonal dots denote the remaining
sectors.}
    \label{fig:symmetry_vs_bose_symmetry_block_decomp}
\end{figure}

As an alternative procedure, following e.g.\ \cite{chee2023efficient, litjens2017semidefinite, fawzi2022hierarchy, gijswijt2009block, vallentin2009symmetry}, we extend the notions of the previous section to all of $\HSn$. Instead of limiting the discussion to the symmetric subspace, corresponding to Specht and Weyl modules associated with the partition of $n$ into a single row, we now extend our analysis to the full decomposition of $\HSn$ into Weyl modules $\Schurf{\lambda}\HS$ and Specht modules $\HS_{\lambda}$ indexed by partitions of $n$ of height at most $d_{\HS}$, i.e.\ $\HSn$ decomposes into $\bigoplus_{\lambda\vdash_{d_{\HS}} n}\lrbracket{\Schurf{\lambda}\HS}^{\oplus\dim_{\CC}\lrbracket{\SpechtH}}$,
which yields a decomposition of $\End{\CSn}{\HSn}$ into $\bigoplus_{\lambda\vdash_{\,d_{\HS}}n} \CC^{m_{\lambda}\times m_{\lambda}}$, where $m_{\lambda}$ is polynomial in $n$ for all $\lambda$, and the direct sum contains polynomially many terms in $n$. See \autoref{fig:symmetry_vs_bose_symmetry_block_decomp} for a comparison of the block matrix structures of symmetric and Bose-symmetric operators. As a result, we avoid the purification scheme and obtain polynomial-time algorithms for non-local games based on a more complex block-decomposition. Specifically, we establish the following theorem. 
\begin{theorem}\label{thm:SDP_sym_complexity_results}
 For a two-player free non-local game with $\lrvert{A}$ answers, $\lrvert{Q}$ questions, and quantum assistance of size $\lrvert{T}$, there exists a poly$\lrbracket{\epsilon^{-1}}$-costly transformation (see \autoref{sec:sdp_symmetry_reduction} for further details) to an SDP with
	 \begin{align}
	 	\mathcal{O}\lrbracket{\lrvert{AQT}^2\lrbracket{\lrvert{T}^6\frac{\log\lrvert{AQT}}{\epsilon^2}}^{\lrvert{AQT}^2}}
	 \end{align}
	 many variables and 
	 \begin{align}
	 	\mathcal{O}\lrbracket{\lrbracket{\lrvert{T}^6\frac{\log\lrvert{AQT}}{\epsilon^2}}^{\lrvert{AQT}}}
	 \end{align}
	 many positive semidefinite constraints and involving matrices of size at most
	 \begin{align}
	 	\mathcal{O}\lrbracket{\lrvert{AQT}^4\lrbracket{\lrvert{T}^6\frac{\log\lrvert{AQT}}{\epsilon^2}}^{\lrbracket{\lrvert{AQT}\lrbracket{\lrvert{AQT}-1}}}},\,
	 \end{align}
      that approximates the value of the game up to an additive $\epsilon$-error.
\end{theorem}
This represents a polynomial-time algorithm with respect to $\epsilon^{-1}$ for determining the value of a fixed-size free non-local game up to an additive $\epsilon$-error, thereby offering an alternative efficient procedure. In contrast to the Bose-symmetry-based algorithm, the symmetry variant does not require a purification step and therefore outperforms the former when the local dimensions \,---\, i.e.,\ $\lrvert{A}, \lrvert{Q}$ and $\lrvert{T}$ \,---\, are large. However, this exchange-symmetry-based algorithm can be implemented only after constructing an orthonormal basis for each Weyl module in the decomposition, a cumbersome step absent in the Bose-symmetric counterpart.

For readers’ convenience, \autoref{fig:overview_optimization_problems} provides an overview of the optimization problems considered in this work.
 \begin{figure}
\centering
\begin{tikzpicture}[
  edge from parent fork down,
  level distance=2cm,
  sibling distance=7cm,
  every node/.style={
    rectangle,
    draw,
    minimum width=2.5cm,
    minimum height=1.2cm,
    align=center,
    font=\small
  },
  edge from parent/.style={
    draw,
  },
  shorten >=2pt
]

\node {Non-local games \\ $\displaystyle w_{Q(T)}(V,\pi)$ \\ \autoref{eqn:numerical_free_games}, p.~\pageref{eqn:numerical_free_games}} 
  child [
  edge from parent path={
    (\tikzparentnode.south) edge[<->, double] (\tikzchildnode.north)
  }
] {node (csep) {$w_{Q(T)}(V,\pi)$ as cSEP \\ \autoref{eqn:cSEP_non_local_main}, p.~\pageref{eqn:cSEP_non_local_main}} 
    child [
  edge from parent path={
    (\tikzparentnode.south) edge[->, double] (\tikzchildnode.north)
  }
] {node (sdp) {$\mathrm{SDP}_n(V,\pi)$ \\ \autoref{eqn:sdp_free_game_definetti_1}, p.~\pageref{eqn:sdp_free_game_definetti_1}}
      child [
  edge from parent path={
    (\tikzparentnode.south) edge[<->, double] (\tikzchildnode.north)
  }
] {node {$\Psi\lrbracket{\mathrm{SDP}_n(V,\pi)}$ \\ \autoref{eqn:efficient_sdp_free_game}, p.~\pageref{eqn:efficient_sdp_free_game}}}
    }
    child [
  edge from parent path={
    (\tikzparentnode.south) edge[->, double] (\tikzchildnode.north)
  }
] {node (bose_sdp) {$\mathrm{SDP}^{\mathrm{Bose}}_n(V,\pi)$ \\ \autoref{eqn:Bose_SDP_non-local_games}, p.~\pageref{eqn:Bose_SDP_non-local_games}}
      child [
  edge from parent path={
    (\tikzparentnode.south) edge[<->, double] (\tikzchildnode.north)
  }
] {node {$\Psi\lrbracket{\mathrm{SDP}^{\mathrm{Bose}}_n(V,\pi)}$ \\ \autoref{eqn:reduced_Bose_SDP}, p.~\pageref{eqn:reduced_Bose_SDP}}}
    }
  };

\draw[->, double]
  (sdp.west) to[out=150, in=210, looseness=1.8]
  node[pos=0.15, above, rectangle, draw, fill=white, font=\scriptsize, inner sep=2pt, xshift=-25pt, yshift=15pt] {De Finetti\\ \autoref{lem:approximate_quantum_de_finetti}, p. \pageref{lem:approximate_quantum_de_finetti}}
  (csep.west);

\draw[->, double]
  (bose_sdp.east) to[out=30, in=-30, looseness=1.8]
  node[pos=0.15, above, rectangle, draw, fill=white, font=\scriptsize, inner sep=2pt, xshift=25pt, yshift=15pt] {Bose de Finetti\\
   \autoref{lem:bose-symmetric_deFinetti}, p. \pageref{lem:bose-symmetric_deFinetti}}
  (csep.east);

\end{tikzpicture}
\caption{Overview of optimization problems arising from non-local games and their semidefinite programming (SDP) relaxations. We formulate the value $w_{Q(T)}(V, \pi)$ as a constrained quantum separability problem. Two outer approximation hierarchies are introduced: $\mathrm{SDP}_n(V,\pi)$, based on optimizations over symmetric (i.e., $n$-extendable) states, and $\mathrm{SDP}_n^{\mathrm{Bose}}(V, \pi)$, based on Bose-symmetric states. Convergence of both hierarchies to  $w_{Q(T)}(V, \pi)$ is guaranteed by corresponding constrained de Finetti theorems. Each hierarchy admits an efficient reduction to a symmetry-adapted form. Arrows indicate the direction of logical implication between objects.}
\label{fig:overview_optimization_problems}
\end{figure}


\subsection{Outlook}

Future applications of the techniques developed in this work extend to other problems in quantum information theory that can be formulated as instances of the constrained separability problem. In the context of approximate quantum error correction \cite{berta2021semidefinite}, the novel constrained Bose-symmetric hierarchy introduced here may be compared to existing symmetry reduction approaches, such as those proposed in \cite{chee2023efficient}. Conversely, following our concurrent work \cite{kossmann2025symmetric}, we ask whether an improved complexity for the approximation algorithm can be achieved by combining specific game symmetries \,---\, arising from the structure of the rule function \,---\, with the exchange symmetry inherent to $n$-extendable states. Furthermore, systematically benchmarking competing numerical methods for constrained separability problems against our Bose-symmetric hierarchy in various concrete applications will offer valuable insights into their relative performance. 

In quantum information theory, other variants of Schur-Weyl duality arise. For instance, in \cite{gross2021schur}, the authors established a duality between the Clifford group and the stochastic orthogonal group, along with a corresponding de Finetti theorem that enables a converging SDP hierarchy \cite{belzig2024studying}, akin in spirit to those considered in this work. While these methods are particularly useful in the study of stabilizer codes, the associated procedures remain computationally expensive. Although some symmetry reduction techniques exist \cite{heinrich2019robustness}, it would be of interest to identify an analogous concept to Bose-symmetry, along with a corresponding hierarchy, that leads to improved complexity results.

	
\subsection*{Acknowledgments}

JZ thanks Tobias Rippchen, Steven Kim and Nikolaos Louloudis for insightful discussions at numerous stages of this work, and Hoang Ta for discussions regarding SDP symmetry reduction. The authors thank Thies Ohst, Mateus Araújo and Maarten Wegewijs for discussions on quantum steering. JZ thanks Max Schrauwen for helpful comments on an earlier version of the manuscript, especially for identifying typos and minor errors. JZ, GK and MB acknowledge support from the Excellence Cluster - Matter and Light for Quantum Computing (ML4Q) and funding by the European Research Council (ERC Grant Agreement No. 948139). OF acknowledges funding by the European Research Council (ERC Grant AlgoQIP, Agreement No. 851716). The authors acknowledge the use of Anthropic’s Claude Fable 5 \cite{claude} to improve the presentation and exposition of the final manuscript. All proof ideas and the initial draft of the manuscript were developed by the authors.

\section{Non-local games as constrained separability problems}
\label{sec:non_lacal_games_as_cbos}\label{sec:non_local_games_as_cbos} 

In this section, we reformulate free two-player non-local games as constrained bipartite separability problems. See \autoref{sec:cbo} for an overview of problems of that form. The structure of this section is as follows: First, we demonstrate that the problem of determining the value of a non-local game with bounded quantum assistance can be reformulated as a bipartite constrained separability problem. Next, similar to \cite{berta2021semidefinite}, we outline the construction of a corresponding semidefinite programming hierarchy of outer approximations and establish its convergence, leveraging a new bipartite de Finetti theorem, which we state in \autoref{sec:cbo}.

\subsection{Notation and preliminaries}
Let $V$ be a finite-dimensional complex vector space. We write $V^*$ for its dual space. Unless stated otherwise, matrix transposition with respect to the standard (computational) basis is denoted by $(\cdot)^T$. Henceforth, we denote the $n$-fold tensor product of a finite-dimensional Hilbert space $\HS$ by $\HS^n$, or equivalently, $\HS_1^n$. Furthermore, $\HS_i^j:= \HS_i\otimes\ldots\otimes \HS_j$ with $\HS_i\simeq\HS_j$ for any $i,j\in\lrrec{n}$. The Banach algebra of bounded linear operators from $\HS$ to $\HS$ is denoted as $\mathcal{B}\lrbracket{\HS}$ and, correspondingly, $\mathcal{S}\lrbracket{\HS}\subseteq \mathcal{B}\lrbracket{\HS}$ denotes state space. When not specified otherwise, we denote by $(\cdot)^\dagger$ the adjoint operation with respect to the standard Hilbert space inner product. For $X\in\mathcal{B}(\HS)$, let 
\begin{align}
    \operatorname{ran}(X):= X\HS = \lrbrace{ X\ket{\psi} \,:\, \ket{\psi} \in \HS} \subseteq \HS
\end{align}
denote its range. The kernel of $X$ is 
\begin{align}
    \ker(X):=\lrbrace{\ket{\psi} \in \HS\,:\, X\ket{\psi}= 0}.
\end{align}
For the support of $X$ we write $\operatorname{supp}(X):=\ker(X)^\perp$, where $(\cdot)^\perp$ denotes the orthogonal complement w.r.t.\ the inner product on $\HS$. Moreover, for $X\succcurlyeq 0$,
\begin{align}
    \operatorname{supp}(X) =\HS \Leftrightarrow \ker(X)=\lrbrace{0} \Leftrightarrow X\succ 0.
\end{align}
Thus, $X\succ 0$ has full support on $\HS$. By $\rho\in \mathcal{S}\lrbracket{\HS}$ we always denote a quantum state, i.e.\ a normalized ($\Tr\lrrec{\rho}=1$) positive semidefinite ($\rho\succcurlyeq 0$) operator. Let $A\otimes B$ be a bipartite Hilbert space. We write $\mathrm{SEP}(A:B)$ for the set of separable states on $A\otimes B$, and $\mathrm{cSEP}(A:B)$ for the set of constrained separable states. Analogously, $\mathrm{PROD}(A:B)$ and $\mathrm{cPROD}(A:B)$ denote the sets of product states and constrained product states, respectively. For any operator $X_{\HS_1^2}\in\mathcal{B}\lrbracket{\HS_1^2}$, $X_{\HS_1}:=\Tr_{\HS_2}\lrbracket{X_{\HS_1^2}}$ with partial trace $\Tr_{\HS_2}\lrbracket{\cdot}=\sum_{i\in\lrrec{\lrvert{\HS_2}}}\lrbracket{\mathcal{I}_{\HS_1}\otimes\bra{i}_{\HS_2}}\lrbracket{\cdot}\lrbracket{\mathcal{I}_{\HS_1}\otimes\ket{i}_{\HS_2}}$. When clear from context, system indices are omitted, and identity operations on unaffected subsystems are suppressed. A classical-quantum state is of the form $\rho_{AZ} = \sum_{z \in Z} \ket{z}\bra{z} \otimes \rho_A^z$, 
where $\lrbrace{\ket{z}}_{z\in Z}$ is an orthonormal basis for the classical register $Z$. For a quantum state $\rho_{AB}$ and a POVM element $M_{B}^{z}$ corresponding to outcome $z\in Z$ of the POVM 
\begin{align}
\begin{split}
     \mathcal{M}_{B\rightarrow Z} \,:\, \mathcal{B}(\HS_B) &\rightarrow \mathcal{B}(\HS_Z)\\
      \rho_B &\mapsto \sum_{z\in Z}\Trr{}{M_{B}^{z}\rho_B}\ket{z}\bra{z}_Z,\,
\end{split}
\end{align}
we denote by 
\begin{align}
    \rho_{A\vert z}=\frac{\Trr{B}{\lrbracket{\mathbb{1}_A\otimes M_{B}^{z}}\rho_{AB}}}{\Trr{}{\lrbracket{\mathbb{1}_A\otimes M_{B}^{z}}\rho_{AB}}}
\end{align}
the post-measurement state obtained after measuring $\rho_{AB}$ and obtaining outcome $z$. When the probability distribution $p(z)$ is clear from context, we adopt the shorthand notation $\mathbb{E}_{z}$ to denote the expectation over $z$, i.e.\, $\mathbb{E}_{z}\lrbrace{\cdot}=\sum_{z\in Z} p(z)(\cdot)$. In particular, the marginal state on system $A$ can be expressed as $\rho_A=\mathbb{E}_{z}\lrbrace{\rho_{A\vert z}}$. The Schatten $p$-norm is denoted by $\lVert\cdot \rVert_p$ and we regularly refer to $\lVert\cdot \rVert_1$ as the trace norm and to $\lVert\cdot \rVert_\infty$ as the operator norm. Hölder's inequality states that for $A,B$ in a suitable Schatten class,
\begin{align}
    \lrvert{\Trr{}{AB}}\leq \lVert A \rVert_p\lVert B \rVert_q\,, \quad \frac{1}{p}+ \frac{1}{q}=1,\, 1\le p,q \le \infty.
\end{align}
The symmetric group on $n\in\mathbb{N}$ symbols is denoted by $S_n$. For notational convenience, we often do not distinguish between abstract group elements and their associated representations on vector spaces; the context will make the meaning unambiguous. $\HSn$ admits a natural (left) action of $S_n$ given by 
\begin{align}
    \pi\cdot \lrbracket{\bigotimes_{i=1}^nh_i}=\bigotimes_{i=1}^nh_{\pi^{-1}(i)},\,\hspace{0.5cm} h_i\in\HS, \forall\pi\in S_n.\, 
\end{align}
We denote the unitary permutation representation
\begin{align}
    \begin{split}
        U_{\HS_1^n} \,:\, S_n&\rightarrow \mathcal{B}\lrbracket{\HS_1^n}\\
        \pi &\mapsto U_{\HS_1^n}\lrbracket{\pi} =  \sum_{i_1,\ldots,i_n\in\lrrec{d_{\HS}}}\ket{i_{\pi^{-1}(1)}, \ldots, i_{\pi^{-1}(n)}}\bra{i_1,\ldots, i_n}\,,
    \end{split}
\end{align}
mapping permutations to permutation matrices. Since the $U_{\HS_1^n}\lrbracket{\pi}$ are real-valued we have $U_{\HS_1^n}\lrbracket{\pi}^\dagger = U_{\HS_1^n}\lrbracket{\pi}^T$. We reserve this notation for permutation representations of the symmetric group. Let $A\otimes B$ be a Hilbert space. For $n\in\mathbb{N}$ copies, we identify 
\begin{align}
    \lrbracket{A\otimes B}^{\otimes n}\cong A^{\otimes n}\otimes B^{\otimes n} = A_1^n\otimes B_1^n =(AB)_1^n. 
\end{align}
Under this identification, the natural permutation representation of $S_n$ factorises as
\begin{align}
    U_{\lrbracket{AB}_1^n}(\pi)=U_{A_1^n}(\pi)\otimes U_{B_1^n}(\pi).
\end{align}
Indeed, this follows from the tensor-product structure and from the fact that permutation unitaries act by permuting tensor factors on computational basis vectors. An operator $X_{B_1^n}\in\mathcal{B}\lrbracket{B_1^n}$ is symmetric if 
\begin{align}
    U_{B_1^n}\lrbracket{\pi} X_{B_1^n}U^T_{B_1^n}\lrbracket{\pi} = X_{B_1^n},\,\hspace{0.5cm} \forall\pi\in S_n.\,
\end{align}
Furthermore, $X_{AB_1^n}\in\mathcal{B}\lrbracket{AB_1^n}$ is symmetric w.r.t.\ $A$ if
\begin{align}
    \lrbracket{\mathbb{1}_{A}\otimes U_{B_1^n}\lrbracket{\pi}}X_{AB_1^n} \lrbracket{\mathbb{1}_{A}\otimes U^T_{B_1^n}\lrbracket{\pi}}=X_{AB_1^n},\,\hspace{0.5cm} \forall\pi\in S_n.\,
\end{align}
In general for a finite or compact group $G$, the algebra of all $G$-equivariant maps from $\HS$ to $\HS$ is 
\begin{align}
    \End{\CC\lrrec{G}}{\HS}=\lrbrace{T\in \End{}{\HS} \,\mid\, T(g\cdot v)= g\cdot T(v),\, \forall g\in G,\, \forall v\in \HS}.
\end{align}
We will often abbreviate a group representation by $g\cdot v$ for any $g$ and $v$. When the action on $A$ is trivial, we suppress it and simply write $\End{\CSn}{A\otimes B_1^n}$ for the space of endomorphisms on $A\otimes B_1^n$ which are $S_n$-equivariant w.r.t.\ $A$. For an introduction to additional relevant mathematical concepts, we refer the reader to \autoref{sec:representation_theory}.
Furthermore, we define the minimal distortion 
        \begin{align}\label{eqn:measurement_distortion_AB}
            \begin{split}
                f(A,B) &:= \inf_{\mathcal{M}_A,\mathcal{M}_B} \,\max_{\substack{\zeta_{AB}\in \mathcal{B}\lrbracket{\HS_{AB}}\\\zeta_{AB}^\dagger = \zeta_{AB} \\ \zeta_A = 0,\ \zeta_B = 0}} \frac{\Vert \zeta_{AB} \Vert_1 }{\Vert \lrbracket{\mathcal{M}_A\otimes \mathcal{M}_B}\zeta_{AB}\Vert_1}\stackrel{\text{\cite{watrous2018theory}}}\leq 18\sqrt{\lrvert{A}\lrvert{B}},
            \end{split}
        \end{align}
        where $\mathcal{M}_A,\, \mathcal{M}_B$ are appropriate informationally complete measurements. Correspondingly, the minimal distortion with side information is defined as
        \begin{align}\label{eqn:measurement_distortion_A}
            \begin{split}
                f(B\vert \cdot)&:= \inf_{\mathcal{M}_B} \,\max_{\substack{\zeta_{AB}\in \mathcal{B}\lrbracket{\HS_{AB}}\\ \zeta_{AB}^\dagger = \zeta_{AB} \\ \zeta_A = 0,\ \zeta_B = 0}} \frac{\Vert \zeta_{AB} \Vert_1 }{\Vert \lrbracket{\mathcal{I}_A\otimes \mathcal{M}_B}\zeta_{AB}\Vert_1}\stackrel{\text{\cite{jee2020quasi}}}{\leq} 2\lrvert{B}.\,
            \end{split}
        \end{align}
        
\subsection{Approximating non-local games}

Our first step is to derive a block-diagonal formulation equivalent to $w_{Q(T)}(V,\pi)$ in \autoref{eqn:numerical_free_games}. To this end, for each $i=1,2$, we introduce classical systems $A_i$ and $Q_i$ equipped, respectively, with orthonormal bases
\begin{align}
    \lbrace \ket{a_i}\rbrace_{a_i\in A_i} \quad \text{ and }\quad \lbrace \ket{q_i}\rbrace_{q_i\in Q_i}    
\end{align}
and tensor these classical registers with the corresponding quantum systems. We choose to normalize $D_{A_2Q_2\hat{T}}$, but not $E_{A_1Q_1T}$, and obtain
    \begin{align}\label{eqn:deFinetti_CBO_free_games}
        \begin{split}
             w_{Q(T)}(V,\pi) = &\lvert T\rvert\max_{(\rho,\, E,\,D)}\hspace{0.5cm}\Tr\left[\left(V_{A_1A_2Q_1Q_2}\otimes\rho_{T\hat{T}}\right)\left(E_{A_1Q_1T}\otimes D_{A_2Q_2\hat{T}}\right)\right]\\
              \text{s.t.}\hspace{1.5cm} &\rho_{T\hat{T}} \succcurlyeq 0,\, \hspace{1cm}\Tr\left[\rho_{T\hat{T}}\right]=1,\,\\
              &E_{A_1Q_1T} =\sum_{a_1,q_1}\pi_1(q_1)\ket{a_1q_1}\bra{a_1q_1}_{A_1Q_1}\otimes E_T(a_1\vert q_1)\succcurlyeq 0,\,\\
              &E_{Q_1T} =\sum_{q_1}\pi_1(q_1)\ket{q_1}\bra{q_1}_{Q_1}\otimes\mathbb{1}_{T},\,\\
              &D_{A_2Q_2\hat{T}} =\sum_{a_2,q_2}\pi_2(q_2)\ket{a_2q_2}\bra{a_2q_2}_{A_2Q_2}\otimes\frac{D_{\hat{T}}(a_2\vert q_2)}{\lvert \hat{T}\rvert}\succcurlyeq 0,\,\\
              &D_{Q_2\hat{T}} =\sum_{q_2}\pi_2(q_2)\ket{q_2}\bra{q_2}_{Q_2}\otimes\frac{\mathbb{1}_{\hat{T}}}{\lvert \hat{T}\rvert},\,\\
        \end{split}
    \end{align}
where $V_{A_1A_2Q_1Q_2}:=\sum_{q_1,q_2}\sum_{a_1,a_2}V(a_1,a_2,q_1,q_2)\ket{a_1,a_2,q_1,q_2}\bra{a_1,a_2,q_1,q_2}_{A_1A_2Q_1Q_2}$ is a diagonal matrix with entries given by the rule function. We reformulate the problem above into the following variant of a cSEP problem.
\begin{lemma}\label{lem:non_local_games_as_cbo_sym}
The value of a two-player free non-local game with $\lrvert{A}$ many answers and $\lrvert{Q}$ many questions and quantum assistance of size $\lrvert{T}$ is given by the following bipartite constrained separability problem
\begin{align}\label{eqn:cSEP_non_local_main}
    \begin{split}
                 w_{Q(T)}(V, \pi) = &\lvert T\rvert\max_{(\alpha, D)}\hspace{0.5cm}\Tr\left[\left(V_{A_1A_2Q_1Q_2}\otimes S_{\tilde{T}\hat{T}}\right)\left(\alpha_{A_1Q_1\tilde{T}}\otimes D_{A_2Q_2\hat{T}}\right)\right]\\
                 \\
                 \text{s.t. } &\,\alpha_{A_1Q_1\tilde{T}} = \sum_{a_1,q_1}\pi_1(q_1)\ket{a_1q_1}\bra{a_1q_1}_{A_1Q_1}\otimes \alpha_{\tilde{T}}(a_1\vert q_1)\succcurlyeq 0,\,\\
        &  \alpha_{Q_1\tilde{T}}=\sum_{q_1}\pi_1(q_1)\ket{q_1}\bra{q_1}_{Q_1}\otimes\alpha_{\tilde{T}},\,\hspace{1cm} \Trr{}{\alpha_{A_1Q_1\tilde{T}}}=1\\
        & D_{A_2Q_2\hat{T}} =\sum_{a_2,q_2}\pi_2(q_2)\ket{a_2q_2}\bra{a_2q_2}_{A_2Q_2}\otimes\frac{D_{\hat{T}}(a_2\vert q_2)}{\lvert T\rvert}\succcurlyeq 0,\,\\
              &D_{Q_2\hat{T}} =\sum_{q_2}\pi_2(q_2)\ket{q_2}\bra{q_2}_{Q_2}\otimes\frac{\mathbb{1}_{\hat{T}}}{\lvert T\rvert},\,\\
            \end{split}
\end{align}
 where $S_{\tilde{T}\hat{T}}$ is the swap operator and $\lrvert{T}=\lrvert{\tilde{T}}=\lrvert{\hat{T}}$.
\end{lemma}

The proof is deferred to \autoref{sec:quantum_steering}. We observe that the optimization over $\alpha_{A_1Q_1\tilde{T}}\otimes D_{A_2Q_2\hat{T}}$ can be equivalently reformulated as an optimization over a finite convex mixture of these terms. It is important to emphasize that this yields an equivalent optimization problem only if the constraints are enforced individually for each summand in the mixture (see \autoref{sec:cbo} for additional details). Crucially, our reformulation as a bipartite problem circumvents several challenges inherent to the tripartite formulation analyzed in \cite{jee2020quasi}. However, the constraints appearing in \autoref{eqn:cSEP_non_local_main} exhibit a more intricate structure than those considered in \cite{jee2020quasi}, or indeed in any previously studied instances within the cSEP framework. To handle these more complex constraints, we introduce a broader class of cSEP problems and extend the methodological framework originally developed in \cite{berta2021semidefinite} accordingly. A detailed discussion of this extension is provided in \autoref{sec:cbo}.

Motivated by \cite{berta2021semidefinite}, we introduce an SDP hierarchy (see  \autoref{sec:cbo} and \autoref{sec:sdp}) of outer approximations for our formulation and prove its convergence. To this end, we use a game-theoretic reformulation that goes beyond the frameworks considered in \cite{Quintino2014, mateus_araujo, tavakoli2024semidefinite}, since those formulations do not directly yield the convergence guarantees required here. As any separable state admits an $n$-symmetric extension for any $n\in\mathbb{N}$, we can formulate a hierarchy of necessary conditions in terms of symmetric states. Concretely, the bipartite state $\rho_{\lrbracket{A_1Q_1T}\lrbracket{A_2Q_2\hat{T}}}$ is said to be $n$-extendable if there exists a state $\rho_{(A_1Q_1T)(A_2Q_2\hat{T})_1^n}$, which is symmetric over $(A_2Q_2\hat{T})_1^n$ w.r.t.\ $(A_1Q_1T)$ and which, upon tracing out any $n-1$ of the copy-equivalent $(A_2Q_2\hat{T})_i$, recovers $\rho_{\lrbracket{A_1Q_1T}\lrbracket{A_2Q_2\hat{T}}}$. We can identify $(A_1Q_1T)$ with Alice and $(A_2Q_2\hat{T})_1^n$ with $n$ Bobs. Furthermore, $\rho_{(A_1Q_1T)(A_2Q_2\hat{T})_1^n}$ satisfies adapted constraints that encode the positivity requirement and the non-signaling condition between Alice and each of the $n$ Bobs. In summary, we obtain \begin{align}\label{eqn:sdp_free_game_definetti_1}
            \begin{split}
                \mathrm{SDP}_{n}&\lrbracket{T, V, \pi} = \lvert T\rvert\cdot \max_{\rho}   \Tr\left[\left(V_{A_1A_2Q_1Q_2}\otimes S_{T\hat{T}}\right)\rho_{(A_1Q_1T)(A_2Q_2\hat{T})}\right],\,\\
            \end{split}
        \end{align}
        subject to the $n$-extendability conditions
        \begin{align}\label{eqn:sdp_free_game_definetti_2}
            \begin{split}
                \rho_{(A_1Q_1T)(A_2Q_2\hat{T})} &= \Tr_{(A_2Q_2\hat{T})_2^n}\left[\rho_{(A_1Q_1T)(A_2Q_2\hat{T})_1^n}\right],\,\hspace{0.5cm}
                 \rho_{(A_1Q_1T)(A_2Q_2\hat{T})_1^n} \succeq 0,\, 
                \hspace{0.5cm} \Tr\left[\rho_{(A_1Q_1T)(A_2Q_2\hat{T})_1^n}\right] = 1,\,\\
                \rho_{(A_1Q_1T)(A_2Q_2\hat{T})_1^n} &=\lrbracket{\mathbb{1}_{A_1Q_1T}\otimes U_{\lrbracket{A_2Q_2\hat{T}}_1^n}\lrbracket{\pi}}\left(\rho_{(A_1Q_1T)(A_2Q_2\hat{T})_1^n}\right)\lrbracket{\mathbb{1}_{A_1Q_1T}\otimes U_{\lrbracket{A_2Q_2\hat{T}}_1^n}^T\lrbracket{\pi}}\hspace{0.5cm} \forall \pi\in S_n,\,\\
            \end{split}
        \end{align}
        and additional linear constraints
        \begin{align}\label{eqn:sdp_constraints}
            \begin{split}
                \Trr{A_1}{\rho_{(A_1Q_1T)(A_2Q_2\hat{T})_1^n}}&=\sum_{q_1}\pi_1(q_1)\ket{q_1}\bra{q_1}\otimes\rho_{T\lrbracket{A_2Q_2\hat{T}}_1^n},\,\\
                \Tr_{(A_2)_1}\left[\rho_{(A_2Q_2\hat{T})_1^n}\right] &= \left(\sum_{q_2}\pi_2(q_2)\ket{q_2}\bra{q_2}_{(Q_2)_1} \otimes\frac{\mathbb{1}_{\hat{T}_1}}{\lvert T\rvert}\right)\otimes \rho_{(A_2Q_2\hat{T})_2^n},\,\\
            \end{split}
\end{align}
where $\rho_{(A_2Q_2\hat{T})_1^n}:=\Trr{A_1Q_1T}{\rho_{(A_1Q_1T)(A_2Q_2\hat{T})_1^n}}$. Moreover, we have the following proposition.

\begin{proposition}[Restriction to classical--quantum states]\label{prop:restriction_to_cq_states}
Let $n\in\mathbb{N}_{\geq 1}$. For a register $C$ with distinguished orthonormal basis $\lrbrace{\ket{c}}_{c\in\lrrec{\lrvert{C}}}$, let
\begin{align}\label{eqn:pinching_channel}
    \mathcal{P}_{C}\lrbracket{\,\cdot\,}:=\sum_{c\in\lrrec{\lrvert{C}}}\lrbracket{\ket{c}\bra{c}_C\otimes\mathbb{1}}\lrbracket{\,\cdot\,}\lrbracket{\ket{c}\bra{c}_C\otimes\mathbb{1}}
\end{align}
denote the associated pinching (dephasing) channel, acting as the identity on all remaining tensor factors, and set
\begin{align}\label{eqn:pinching_channel_cl}
    \mathcal{P}_{\mathrm{cl}}:=\mathcal{P}_{A_1Q_1}\circ\mathcal{P}_{(A_2Q_2)_1}\circ\ldots\circ\mathcal{P}_{(A_2Q_2)_n}
\end{align}
on $\mathcal{B}\lrbracket{\lrbracket{A_1Q_1T}\otimes\An}$. If $\rhoge$ is feasible for $\mathrm{SDP}_n\lrbracket{T,V,\pi}$ in \autoref{eqn:sdp_free_game_definetti_1}, then $\mathcal{P}_{\mathrm{cl}}\lrbracket{\rhoge}$ is feasible and attains the same objective value. In particular, the optimization in $\mathrm{SDP}_n\lrbracket{T,V,\pi}$ may be restricted to states that are classical-quantum with respect to all question and answer registers, i.e.\ to states of the form
\begin{align}\label{eqn:cq_states_form}
    \rhoge=\sum_{\substack{a_1\in\lrrec{\lrvert{A_1}},\, q_1\in\lrrec{\lrvert{Q_1}},\,\\ \vec{a}_2\in\lrrec{\lrvert{A_2}}^n,\, \vec{q}_2\in\lrrec{\lrvert{Q_2}}^n}}\ket{a_1q_1}\bra{a_1q_1}_{A_1Q_1}\otimes\ket{\vec{a}_2\vec{q}_2}\bra{\vec{a}_2\vec{q}_2}_{(A_2Q_2)_1^n}\otimes\rho^{\lrbracket{a_1,q_1,\vec{a}_2,\vec{q}_2}}_{T\hat{T}_1^n}
\end{align}
with $\rho^{\lrbracket{a_1,q_1,\vec{a}_2,\vec{q}_2}}_{T\hat{T}_1^n}\succcurlyeq 0$, without changing the optimal value. Here $\ket{\vec{a}_2\vec{q}_2}:=\bigotimes_{i=1}^{n}\ket{a_{2,i}\,q_{2,i}}_{(A_2Q_2)_i}$, and tensor factors are regrouped according to the identification $\An\cong\lrbracket{A_2Q_2}_1^n\otimes\hat{T}_1^n$.
\end{proposition}

The proof is deferred to \autoref{sec:proof_cq_states}. 

\begin{remark}\label{rem:cq_restriction}
\autoref{prop:restriction_to_cq_states} formalizes the observation that the relaxations inherit the classicality of the question and answer registers from \autoref{eqn:cSEP_non_local_main}: states arising from strategies are of the form \autoref{eqn:cq_states_form} by construction, and coherences on the classical registers, although admitted by the constraints of $\mathrm{SDP}_n\lrbracket{T,V,\pi}$, never affect its value.
\end{remark}

The rate at which $\mathrm{SDP}_{n}(T,V,\pi)$ approximates $w_{Q(T)}(V,\pi)$ from above follows from a de Finetti-type argument. However, due to the new constraints, we need an extension of \cite[Theorem 2.3]{berta2021semidefinite} to accommodate our setting. See \autoref{lem:approximate_quantum_de_finetti} for this new variant of a constrained de Finetti theorem. 

\begin{proof}[Proof of \autoref{lem:convergence_non_local_games_as_cbo_results}]
For a given $n\in\mathbb{N}$, let $\rho_{(A_1Q_1T)(A_2Q_2\hat{T})}$ be the marginal of a state optimal for $\mathrm{SDP}_{n}\lrbracket{T, V, \pi}$. Applying the constrained de Finetti representation theorem in \autoref{lem:approximate_quantum_de_finetti} with
\begin{align}
\begin{array}{cccccc}
     A_L= A_1Q_1,& C_{A_L}=Q_1,& A_R=T,& B_L= A_2Q_2\hat{T},& C_{B_L}=Q_2\hat{T},& B_R=\CC,\\
\end{array}
\end{align}
and
\begin{align}
\begin{split}
     \begin{array}{cc}
     \Theta_{A_L\rightarrow C_{A_L}}(\cdot)=\Trr{A_1}{\cdot}, & \Upsilon_{B_L\rightarrow C_{B_L}}(\cdot) = \Trr{A_2}{\cdot},\\
      W_{C_{A_L}}=\sum_{q_1\in Q_1}\pi_1(q_1)\ket{q_1}\bra{q_1}_{Q_1},   &  K_{C_{B_L}}= \sum_{q_2\in Q_2}\pi_2(q_2)\ket{q_2}\bra{q_2}_{Q_2}\otimes\frac{\mathbb{1}_{\hat{T}}}{\lrvert{T}},
    \end{array}
\end{split}
\end{align}
certifies the existence of a separable state $\sum_{x\in\mathcal{X}}p(x)\,\rho^x_{\lrbracket{A_1Q_1T}}\otimes\rho^x_{\lrbracket{A_2Q_2\hat{T}}}$ feasible for $w_{Q(T)}(V, \pi)$ in \autoref{eqn:cSEP_non_local_main}, such that
\begin{align}
    \left\lVert \rho_{\lrbracket{A_1Q_1T}\lrbracket{A_2Q_2\hat{T}}} - \sum_{x\in\mathcal{X}}p(x)\rho^x_{\lrbracket{A_1Q_1T}}\otimes\rho^x_{\lrbracket{A_2Q_2\hat{T}}} \right\rVert_1\leq 2\cdot\lrvert{T}\sqrt{2\ln 2}\sqrt{\frac{\log\lrbracket{\lrvert{A_1}\lrvert{Q_1}\lrvert{T}}}{n}}\,.
\end{align}
By \autoref{prop:restriction_to_cq_states}, to any marginal $\rho_{\lrbracket{A_1Q_1T}\lrbracket{A_2Q_2\hat{T}}}$ of an optimizer to  $\mathrm{SDP}_{n}\lrbracket{T, V, \pi}$, there exists a classical-quantum state attaining the same objective value. Thus, w.l.o.g.\ $\rho_{\lrbracket{A_1Q_1T}\lrbracket{A_2Q_2\hat{T}}}$ is a classical-quantum state and classical systems do not contribute to the measurement distortion bound. Moreover,
\begin{align}
\begin{split}
        \mathrm{SDP}_{n}&\lrbracket{T, V, \pi} = \lvert T\rvert\cdot \max_{\rho}   \Tr\left[\left(V_{A_1A_2Q_1Q_2}\otimes S_{T\hat{T}}\right)\rho_{(A_1Q_1T)(A_2Q_2\hat{T})}\right]\\
        &\leq w_{Q(T)}(V,\pi) + \lvert T\rvert \Tr\left[\left(V_{A_1A_2Q_1Q_2}\otimes S_{T\hat{T}}\right)\lrbracket{\rho_{(A_1Q_1T)(A_2Q_2\hat{T})}-\sum_{x\in\mathcal{X}}p(x)\,\rho^x_{\lrbracket{A_1Q_1T}}\otimes\rho^x_{\lrbracket{A_2Q_2\hat{T}}}}\right]
\end{split}
\end{align}
and by Hölder's inequality together with $\lVert V_{A_1A_2Q_1Q_2}\otimes S_{T\hat{T}}\rVert_\infty\leq 1$, 
    \begin{align}
        \begin{split}
            \left\lvert  \mathrm{SDP}_{n}(T,V,\pi) -w_{Q(T)}(V,\pi)\right\rvert
            &\leq \lrvert{T}\left\lVert \rho_{\lrbracket{A_1Q_1T}\lrbracket{A_2Q_2\hat{T}}} - \sum_{x\in\mathcal{X}}p(x)\rho^x_{\lrbracket{A_1Q_1T}}\otimes\rho^x_{\lrbracket{A_2Q_2\hat{T}}} \right\rVert_1\\
            &\leq 2\cdot\lvert  T\rvert ^2 \sqrt{2\ln 2}\sqrt{\frac{\log\lrbracket{\lrvert{A_1}\lrvert{Q_1}\lrvert{T}}}{n}}.
        \end{split}
    \end{align}
        Hence, we have $w_{Q(T)}(V,\pi) = \lim_{n\rightarrow\infty}  \mathrm{SDP}_{n}(T,V,\pi)$. Define
\begin{align}
    \mathbb{N}\ni n_{\epsilon} :=\max\lrbrace{1, \lrceil{\frac{8\ln 2\cdot\lrvert{T}^4 \log\lrbracket{\lrvert{A_1}\lrvert{Q_1}\lrvert{T}}}{\epsilon^2}}}.
\end{align}
Then, if $n\geq n_{\epsilon}$ we have
\begin{align}
    \left\lvert  \mathrm{SDP}_{n}(T,V,\pi) -w_{Q(T)}(V,\pi)\right\rvert
            &\leq \epsilon.
\end{align}
For $\lrvert{A_1}=\lrvert{A_2}=\lrvert{A}$ and $\lrvert{Q_1}=\lrvert{Q_2}=\lrvert{Q}$, in $\mathrm{SDP}_{n_{\epsilon}}(T,V,\pi)$ the optimization variable is a complex square matrix of order $\lrvert{AQT}^{n_{\epsilon} +1}$, with at most $\lrvert{AQT}^{2\lrbracket{n_{\epsilon} +1}}$ real degrees of freedom. In summary, to determine the value of a non-local free game with fixed quantum assistance up to an additive $\epsilon$-error requires solving an SDP with at most
\begin{align}
    \lrceil{\lrvert{AQT}^{2\lrbracket{\frac{8\ln 2 \lrvert{T}^4\log\lrvert{AQT}}{\epsilon^2} +1}}}, \text{ which is } \exp\lrrec{\frac{\mathcal{O}\lrbracket{\lrvert{T}^4\log^2\lrvert{AQT}}}{\epsilon^2}} \text{ as }\epsilon \rightarrow 0
\end{align}
real degrees of freedom.
\end{proof}

Rather than expressing the hierarchy in terms of the swap operator, one may follow \cite{jee2020quasi} and formulate an analogous SDP hierarchy with objective
\begin{align}
    \mathrm{SDP}^{(\Phi)}_{n}\lrbracket{T,V,\pi}=\lvert T\rvert\cdot\max_{\rho}\Trr{}{\left(V_{A_1A_2Q_1Q_2}\otimes\Phi_{T\vert \hat{T}}\right)\rho_{(A_1Q_1T)(A_2Q_2\hat{T})}},
\end{align}
where $\Phi_{T\vert \hat{T}}$ denotes the non-normalized canonical maximally entangled state. The constraints are otherwise unchanged. This formulation is often more favorable in numerical implementations. Indeed, the swap operator has eigenvalues $\pm 1$ and decomposes as the difference of the projectors onto the symmetric and antisymmetric subspaces. By contrast, $\Phi_{T\vert \hat{T}}$ is proportional to a rank-one projector onto a one-dimensional subspace of the symmetric subspace. The distinction is relevant because, while the symmetric subspace contains both product and entangled states, the antisymmetric subspace contains only entangled states. In particular, for any state $\rho_{T\hat{T}}$, the condition
\begin{align}
    \Tr\!\left[S_{T\hat{T}}\rho_{T\hat{T}}\right] < 0
\end{align}
certifies entanglement, so that $S_{T\hat{T}}$ acts as an entanglement witness \cite{Chruciski2014}. Consequently, an optimizer for $\mathrm{SDP}_n\lrbracket{T,V,\pi}$ may assign states with negative swap expectation values to losing question--answer combinations, thereby offsetting them against larger positive contributions from winning combinations. This mechanism is absent in the $\Phi$-based hierarchy, and hence one obtains
\begin{align}
\mathrm{SDP}_n\lrbracket{T,V,\pi}\geq\mathrm{SDP}^{(\Phi)}_n\lrbracket{T,V,\pi}.    
\end{align}

There is, however, a tradeoff at the level of the analytical convergence analysis. Since
\begin{align}
     \left\lVert \Phi_{T\vert \hat{T}}\right\rVert_{\infty}=\lvert T\rvert,
\end{align}
under the normalization used here, the same proof technique yields a convergence rate with a worse dependence on the local dimension $\lvert T\rvert$ than the rate obtained in \autoref{lem:convergence_non_local_games_as_cbo_results}. Moreover, although
\begin{align}
    \mathrm{SDP}^{(\Phi)}_{n}\lrbracket{T,V,\pi}\leq 1
\end{align}
for every $n\in\mathbb{N}$, the corresponding bound for the swap-based hierarchy is only
\begin{align}
        \mathrm{SDP}_{n}\lrbracket{T,V,\pi}\leq \lvert T\rvert .
\end{align}
The problem $\mathrm{SDP}_n\lrbracket{T,V,\pi}$ is, by construction, an $S_n$-invariant SDP and thus fits the framework described in \autoref{sec:symmetry_invariant_sdp}. Numerical comparisons of the SDP hierarchies introduced in this paper will be presented elsewhere.

\section{Convergent sequence of inner bounds to constrained separability problems}\label{sec:inner_sequence} 

To establish the rounding result in the context of non-local games, we first prove it within the more general cSEP framework given in \autoref{eqn:cSEP_Main} and discussed in more detail in \autoref{sec:cbo}. A schematic depiction of the rounding procedure is given in \autoref{fig:schematic_rounding}.

\begin{lemma}\label{lem:inner_sequence_general}
Consider Hilbert spaces $A=A_L\otimes A_R$ and $B=B_L\otimes B_R$ and a finite alphabet $\mathcal{X}$. Let $\mathrm{cSEP}(G)$ denote the optimal value of the constrained separability
problem
\begin{align}
    \mathrm{cSEP}(G) = \displaystyle\max_{\rho_{AB}\in \mathcal{B}(AB)} \Tr[G_{AB}\,\rho_{AB}]
\end{align}
where the optimization is over separable states
\begin{align}
    \rho_{AB}=\sum_{x\in\mathcal{X}}p(x)\rho^x_A\otimes\rho^x_B, \quad p(x)\geq 0\, \forall x\in\mathcal{X},\, \sum_{x\in\mathcal{X}}p(x)=1
\end{align}
such that, for every $x\in\mathcal{X}$,
\begin{align}
\begin{split}
\begin{array}{cc}
    \Theta_{A_L\rightarrow C_{A_L}}\lrbracket{\rho^x_A} = W_{C_{A_L}}\otimes \rho^x_{A_R}, & \Upsilon_{B_L\rightarrow C_{B_L}}\lrbracket{\rho^x_B} = K_{C_{B_L}}\otimes \rho^x_{B_R}, \\
      \Omega_{A\rightarrow A}\lrbracket{\rho^x_{A}}=\rho^x_{A}, & \Xi_{B\rightarrow B}\lrbracket{\rho^x_{B}}=\rho^x_{B}.
\end{array}
\end{split}
\end{align}
For any $n\in\mathbb{N}$, let $\mathrm{SDP}_n(G)$ (see \autoref{prop:hierarchy_outer_approx}) be the $n$-th symmetric-extension SDP relaxation of this
problem, and let $\rho^\star_{AB_1^n}$ be optimal for $\mathrm{SDP}_n(G)$. Then, for every $k\in \lrbrace{0, 1, \ldots, n-1}$ one can construct a separable state
\begin{align}\label{eqn:candidate_rounded_convex_mixture}
\sigma^{(n, k)}_{AB}=\sum_{z_1^k} p\lrbracket{z_1^k}\,\rho^\star_{A\vert z_1^k}\otimes\rho^\star_{B_{k+1}\vert z_1^k}
\end{align}
feasible for $\mathrm{cSEP}(G)$, obtained by measuring systems $B_1,\ldots,B_k$ of $\rho^\star_{AB_1^n}$ with an informationally complete measurement; for $k=0$, we define the product of marginals
\begin{align}
    \sigma^{(n, 0)}_{AB}:=\rho_A^\star\otimes \rho_{B_1}^\star=\Trr{B_1^n}{\rho^\star_{AB_1^n}}\otimes \Trr{AB_2^n}{\rho^\star_{AB_1^n}}.
\end{align}
Furthermore, for each $n\in \mathbb{N}$ there exists a $k_n\in \lrbrace{0, 1, \ldots, n-1}$ such that
\begin{align}\label{eqn:bound_inner_sequence}
    0 \leq \mathrm{cSEP}(G)-\Tr\lrrec{G_{AB}\sigma^{(n, k_n)}_{AB}} &\le \epsilon\lrbracket{\lrvert{A}, \lrvert{B}, n},
\end{align}
with 
\begin{align}\label{eqn:bound_epsilon_inner_sequence}
    \epsilon\lrbracket{\lrvert{A}, \lrvert{B}, n} \leq \left\lVert G_{AB}\right\rVert_\infty \min\lrbrace{2\lrbracket{1-\frac{1}{\min\lrbrace{\lrvert{A}, \lrvert{B}}^2}} ,\, 
\min\lrbrace{f(A,B),f(B\vert\cdot)}
\sqrt{\frac{2\ln 2\,\log \lrvert{A}}{n}}},
\end{align}
where $f(A,B)$ and $f(B\vert\cdot)$ denote the measurement distortions of \autoref{eqn:measurement_distortion_AB} and \autoref{eqn:measurement_distortion_A}. Moreover, for any $k\in\lrbrace{0, 1, \ldots, n-1}$ and any
\begin{align}
z_\star\in\operatorname{argmax}_{z_1^k}\Trr{}{G_{AB}\,\lrbracket{\rho^\star_{A\vert z_1^k}\otimes\rho^\star_{B_{k+1}\vert z_1^k}}},
\end{align}
the product state $\rho^\star_{A\vert z_\star}\otimes\rho^\star_{B_{k+1}\vert z_\star}$ is itself feasible for $\mathrm{cSEP}(G)$ and attains an objective value at least that of $\sigma^{(n, k)}_{AB}$; in particular, it satisfies \autoref{eqn:bound_inner_sequence} whenever $\sigma^{(n, k)}_{AB}$ does. In particular, if $\left\lVert G_{AB}\right\rVert_\infty\leq 1$, the factor $\left\lVert G_{AB}\right\rVert_\infty$ may be omitted. Moreover, the maximum $\mathrm{cSEP}(G)$ is attained. If, for every $n\in\mathbb{N}$,  $k_n\in\lrbrace{ 0,1,\ldots,n-1}$ is chosen such that $\sigma^{(n, k_n)}_{AB}$ satisfies \autoref{eqn:bound_inner_sequence}
\begin{align}
     \left(\sigma_{AB}^{(n,k_n)}\right)_{n\in\mathbb{N}}
\end{align}
admits a subsequence converging to an optimal solution of $\mathrm{cSEP}(G)$.
\end{lemma}
\begin{proof}
For any $n\in\mathbb{N}$, let $\rho_{AB_1^n}^{\star}$ be optimal for $\mathrm{SDP}_n(G)$; in particular, it satisfies the lifted constraints
\begin{align}\label{eqn:inner_sequence_lifted_constraints}
    \begin{split}
        \begin{array}{cc}
            \Theta_{A_L\rightarrow C_{A_L}}\lrbracket{\rho_{AB_1^n}^\star} = W_{C_{A_L}}\otimes\rho_{A_RB_1^n}^\star,  &   \Upsilon_{\lrbracket{B_L}_n\rightarrow \lrbracket{C_{B_L}}_n}\lrbracket{\rho_{B_1^n}^\star} = K_{C_{B_L}}\otimes \rho^\star_{B_1^{n-1}\lrbracket{B_R}_n},\\
           \Omega_{A\rightarrow A}\lrbracket{\rho_{AB_1^n}^\star}=\rho_{AB_1^n}^\star, & \Xi_{B_n\rightarrow B_n}\lrbracket{\rho_{B_1^n} ^\star}=\rho_{B_1^n}^\star.
        \end{array}
    \end{split}
\end{align}
By the permutation invariance of $\rho^\star_{B_1^n}$ --- inherited from the symmetry of $\rho^\star_{AB_1^n}$ w.r.t.\ $A$ by tracing out $A$ --- the $\Upsilon$- and $\Xi$-constraints in \autoref{eqn:inner_sequence_lifted_constraints}, stated for the $n$-th subsystem, hold verbatim with $B_{k+1}$ in place of $B_n$ for every $k\in\lrbrace{0, 1, \ldots, n-1}$ (swap the subsystems $k+1$ and $n$; cf.\ \autoref{prop:symmetric_states_have_symmetric_reduced_states}).
\begin{align}
    \Trr{B_1^n}{\rho_{AB_1^n}^{\star}}\otimes \Trr{AB_2^n}{\rho_{AB_1^n}^{\star}}=\rho_A^\star\otimes\rho_B^\star.
\end{align}
 Moreover, by linearity and \autoref{eqn:inner_sequence_lifted_constraints} the candidate satisfies
\begin{align}
    \begin{split}
        \Theta_{A_L\rightarrow C_{A_L}}\lrbracket{\rho_A^\star} &= \Theta_{A_L\rightarrow C_{A_L}}\lrbracket{\Trr{B_1^n}{\rho^\star_{AB_1^n}}}= \Trr{B_1^n}{\Theta_{A_L\rightarrow C_{A_L}}\lrbracket{\rho^\star_{AB_1^n}}}=\Trr{B_1^n}{W_{C_{A_L}}\otimes\rho_{A_RB_1^n}^\star},\\
        \Upsilon_{B_L\rightarrow C_{B_L}}\lrbracket{\rho_B^\star} &= \Upsilon_{\lrbracket{B_L}_1\rightarrow C_{B_L}}\lrbracket{\Trr{B_2^n}{\rho^\star_{B_1^n}}} = \Trr{B_2^n}{\Upsilon_{\lrbracket{B_L}_1\rightarrow C_{B_L}}\lrbracket{\rho^\star_{B_1^n}}}= \Trr{B_2^n}{K_{C_{B_L}}\otimes\rho^\star_{\lrbracket{B_R}_1B_2^n}},\\
         \Omega_{A\rightarrow A}\lrbracket{\rho_A^\star}&= \Omega_{A\rightarrow A}\lrbracket{\Trr{B_1^n}{\rho^\star_{AB_1^n}}} = \Trr{B_1^n}{\Omega_{A\rightarrow A}\lrbracket{\rho^\star_{AB_1^n}}}= \Trr{B_1^n}{\rho^\star_{AB_1^n}},\\
         \Xi_{B\rightarrow B}\lrbracket{\rho_B^\star}&=  \Xi_{B_1\rightarrow B_1}\lrbracket{\Trr{AB_2^n}{\rho_{AB_1^n}^{\star}}}=\Trr{B_2^n}{\Xi_{B_1\rightarrow B_1}\lrbracket{\rho_{B_1^n}^{\star}}}= \Trr{B_2^n}{\rho_{B_1^n}^\star}.
    \end{split}
\end{align}
In summary, 
\begin{align}\label{eqn:constaints_k_zero_rounded_candidate}
    \begin{split}
\begin{array}{cc}
    \Theta_{A_L\rightarrow C_{A_L}}\lrbracket{\rho_A^\star} = W_{C_{A_L}}\otimes \rho^\star_{A_R}, & \Upsilon_{B_L\rightarrow C_{B_L}}\lrbracket{\rho_B^\star} = K_{C_{B_L}}\otimes \rho^\star_{B_R}, \\
      \Omega_{A\rightarrow A}\lrbracket{\rho_A^\star}=\rho_A^\star, & \Xi_{B\rightarrow B}\lrbracket{\rho_B^\star}=\rho_B^\star,
\end{array}
\end{split}
\end{align}
and hence, $\rho_A^\star\otimes\rho_B^\star$ is feasible for $\mathrm{cSEP}(G)$. By \cite{Hall2013}, we have in general
\begin{align}
    \max_{\sigma_{AB}\in\mathcal{B}(AB)}\left\lVert \sigma_{AB}-\sigma_A\otimes\sigma_B\right\rVert_1 = 2\cdot \lrbracket{1-\frac{1}{\min\lrbrace{\lrvert{A}, \lrvert{B}}^2}}.
\end{align}
Then, with $\rho^\star_{AB}:=\Trr{B_2^n}{\rho^\star_{AB_1^n}}$ and by Hölder's inequality,
\begin{align}
\begin{split}
    \Trr{}{G_{AB}\rho_{AB}^\star}-\Trr{}{G_{AB}\lrbracket{\rho^\star_A\otimes\rho^\star_B}}&= \Trr{}{G_{AB}\lrbracket{\rho_{AB}^\star-\rho^\star_A\otimes\rho^\star_B}}\\
    &\leq \left\lVert G_{AB} \right\rVert_{\infty} \left\lVert \rho^\star_{AB}-\rho^\star_A\otimes\rho^\star_B\right\rVert_1 \\
    &\leq 2\cdot \left\lVert G_{AB} \right\rVert_{\infty}\lrbracket{1-\frac{1}{\min\lrbrace{\lrvert{A}, \lrvert{B}}^2}}.
\end{split}
\end{align}
If $n\geq 2$ and $k\in\lrbrace{1, 2, \ldots, n-1}$ the candidate is obtained from measurement-based rounding. Fix an informationally complete measurement $\mathcal{M}_{B}$ on $B$ with finite outcome alphabet $Z$. Let
\begin{align}
M_{B_1^k}^{z_1^k}=M_{B_1}^{z_1}\otimes\cdots\otimes M_{B_k}^{z_k}, \quad z_1^k\in Z^k,
\end{align}
be the corresponding POVM element on $B_1^k$. For every $z_1^k\in Z^k$, define the non-normalized post-measurement state on $A B_{k+1}$ by 
\begin{align}
\widetilde\rho^\star_{AB_{k+1}\vert z_1^k}:=\Trr{B_1^k B_{k+2}^n}{\lrbracket{\mathbb{1}_A\otimes M_{B_1^k}^{z_1^k}\otimes \mathbb{1}_{B_{k+1}^n}} \rho_{AB_1^n}^\star}.
\end{align}
Let $p\lrbracket{z_1^k}:=\Trr{}{\widetilde\rho^\star_{AB_{k+1}\vert z_1^k}}$. For outcomes $p\lrbracket{z_1^k}>0$, set 
\begin{align}
    \rho^\star_{AB_{k+1}\vert z_1^k} := \frac{\widetilde\rho^\star_{AB_{k+1}\vert z_1^k}}{p\lrbracket{z_1^k}},\quad \rho^\star_{A\vert z_1^k} :=\Trr{B_{k+1}}{ \rho^\star_{AB_{k+1}\vert z_1^k}},\quad \rho^\star_{B_{k+1}\vert z_1^k} := \Trr{A}{\rho^\star_{AB_{k+1}\vert z_1^k}}.
\end{align}
Outcomes with $p\lrbracket{z_1^k}=0$ may be ignored. By the same arguments as in the proof of the constrained de Finetti representation theorem in \autoref{lem:approximate_quantum_de_finetti}, we have for every $z_1^k\in Z^k$
\begin{align}\label{eqn:constraints_rounded_candidate_components}
            \begin{split}
                \begin{array}{cc}
                     \Theta_{A_L\rightarrow C_{A_L}}\lrbracket{\rho^\star_{A\vert z_1^k}} = W_{C_{A_L}}\otimes\rho^\star_{A_R\vert z_1^k},\, &
            \Omega_{A\rightarrow A}\lrbracket{\rho^\star_{A\vert z_1^k}}=\rho^\star_{A\vert z_1^k},\,\\
            \Upsilon_{\lrbracket{B_L}_{k+1}\rightarrow C_{B_L}}\lrbracket{\rho^\star_{B_{k+1}\vert z_1^k}} = K_{C_{B_L}}\otimes\rho^\star_{\lrbracket{B_R}_{k+1}\vert z_1^k},\, & \Xi_{B_{k+1}\rightarrow B_{k+1}}\lrbracket{\rho^\star_{B_{k+1}\vert z_1^k}}=\rho^\star_{B_{k+1}\vert z_1^k}.\,\\
                \end{array}
            \end{split}
        \end{align}
Hence, $\rho^\star_{A\vert z_1^k}\otimes \rho^\star_{B_{k+1}\vert z_1^k}$ is a feasible point of $\mathrm{cSEP}(G)$ for every $z_1^k\in Z^k$. Moreover, the convex mixture 
\begin{align}
    \sigma_{AB_{k+1}}^{(n,k)}:= \sum_{z_1^k} p\lrbracket{z_1^k}\,
\rho^\star_{A\vert z_1^k}\otimes\rho^\star_{B_{k+1}\vert z_1^k}
\end{align}
is also feasible for $\mathrm{cSEP}(G)$. For a fixed $n$, let $\lrbrace{\sigma_{AB_{k+1}}^{(n,k)}}_{k=0,\ldots, n-1}$ be the set of candidate states. By \autoref{lem:approximate_quantum_de_finetti}, there exists a $k_n\in\lrbrace{0,\ldots, n-1}$ such that 
\begin{align}\label{eqn:de_Finetti_bound_general_inner_sequence_state}
    \left\lVert \rho_{AB_{k_n+1}}^\star -  \sigma_{AB_{k_n}+1}^{(n,k_n)}\right\rVert_1 \leq \min\lbrace f(A, B), f(B\vert \cdot)\rbrace \sqrt{\frac{2\ln 2 \log \lrvert{A}}{n}}\,.
\end{align}
Since $\rho^\star_{AB_1^n}$ is permutation-invariant on $B_1,\ldots,B_n$, we may identify $B_{k+1}$ with $B_1=:B$ for every $k$. As $k_n$ is not known in advance, construct all candidates $\lrbrace{\sigma_{AB_{k+1}}^{(n,k)}}_{k=0,\ldots, n-1}$ and select the one closest to $\rho_{AB}^\star$ in $\lVert\cdot\rVert_1$-distance; denote its index again by $k$. Hence
\begin{align}
\mathrm{SDP}_n(G)=\Trr{}{G_{AB}\,\rho^\star_{AB}}=\Trr{}{G_{AB}\,\rho^\star_{AB_{k+1}}}.    
\end{align}
Then, by Hölder's inequality,
\begin{align}
\begin{split}
    0\leq \mathrm{SDP}_n(G) - \Trr{}{G_{AB}\, \sigma_{AB_{k+1}}^{(n,k)}} &= \Trr{}{G_{AB}\,\lrbracket{\rho_{AB_{k+1}}^\star-\sigma_{AB_{k+1}}^{(n,k)}}}\\
    &\leq \left\lVert G_{AB}\right\rVert_{\infty}  \left\lVert \rho_{AB_{k+1}}^\star -  \sigma_{AB_{k+1}}^{(n,k)}\right\rVert_1 \\
    &\leq \left\lVert G_{AB}\right\rVert_{\infty}\min\lbrace f(A, B), f(B\vert \cdot)\rbrace \sqrt{\frac{2\ln 2 \log \lrvert{A}}{n}},
\end{split}
\end{align}
where the last inequality uses the minimality of the selected $k$ together with \autoref{eqn:de_Finetti_bound_general_inner_sequence_state}. Moreover, since $\sigma_{AB_{k+1}}^{(n,k)}$ is feasible for $\mathrm{cSEP}(G)$, 
\begin{align}
     \Trr{}{G_{AB}\, \sigma_{AB_{k+1}}^{(n,k)}} \leq   \mathrm{cSEP}(G),
\end{align}
and thus, since $\mathrm{cSEP}(G)\leq\mathrm{SDP}_n(G)$ by \autoref{prop:hierarchy_outer_approx},
\begin{align}
\begin{split}
     0 \leq  \mathrm{cSEP}(G) - \Trr{}{G_{AB}\, \sigma_{AB_{k+1}}^{(n,k)}}&\leq \mathrm{SDP}_n(G)- \Trr{}{G_{AB}\, \sigma_{AB_{k+1}}^{(n,k)}}\\
     &\leq \left\lVert G_{AB}\right\rVert_{\infty}\min\lbrace f(A, B), f(B\vert \cdot)\rbrace \sqrt{\frac{2\ln 2 \log \lrvert{A}}{n}}.
\end{split}
\end{align}
Taking, among $\sigma^{(n,0)}_{AB}$ and this selected candidate, the one with the larger objective value yields \autoref{eqn:bound_inner_sequence} with $\epsilon$ as in \autoref{eqn:bound_epsilon_inner_sequence}, i.e., the minimum of the two derived bounds. Finally, from the set of feasible product states $\lrbrace{\rho^\star_{A\vert z_1^k}\otimes \rho^\star_{B_{k+1}\vert z_1^k}}_{z_1^k\in Z^k}$, whose convex mixture with weights $p\lrbracket{z_1^k}$ is $\sigma_{AB_{k+1}}^{(n,k)}$, choose 
\begin{align}
    z_\star\in\operatorname{argmax}_{z_1^k} \Trr{}{G_{AB}\,\lrbracket{\rho^\star_{A|z_1^k}\otimes\rho^\star_{B_{k+1}|z_1^k}}}.
\end{align}
Then,
\begin{align}
    \Trr{}{G_{AB}\, \lrbracket{\rho^\star_{A|z_\star}\otimes\rho^\star_{B_{k+1}|z_\star}}}\geq  \Trr{}{G_{AB}\, \sigma_{AB_{k+1}}^{(n,k)}}.
\end{align}
Therefore,
\begin{align}
    0\leq \mathrm{cSEP}(G) - \Trr{}{G_{AB}\, \lrbracket{\rho^\star_{A|z_\star}\otimes\rho^\star_{B_{k+1}|z_\star}}} \leq \mathrm{cSEP}(G) - \Trr{}{G_{AB}\, \sigma_{AB_{k+1}}^{(n,k)}}
\end{align}
which proves the claimed bound for the best product component. It remains to upgrade convergence of the values to convergence, along a subsequence, to an optimizer. We first note that the feasible set of $\mathrm{cSEP}(G)$ is compact. Indeed, define the local constrained state sets
\begin{align}
    \mathcal{F}_A:=
    \left\{\tau_A\in\mathcal{S}(A)\,\middle|\,\Theta_{A_L\to C_{A_L}}(\tau_A)=W_{C_{A_L}}\otimes \tau_{A_R},\;\Omega_{A\to A}(\tau_A)=\tau_A\right\},
\end{align}
and
\begin{align}
    \mathcal{F}_B:=\left\{\tau_B\in\mathcal{S}(B)\,\middle|\,\Upsilon_{B_L\to C_{B_L}}(\tau_B)=K_{C_{B_L}}\otimes\tau_{B_R},\;\Xi_{B\to B}(\tau_B)=\tau_B\right\}.
\end{align}
Since $\mathcal{S}(A)$ and $\mathcal{S}(B)$ are compact in finite dimensions and the above constraints are closed linear conditions, $\mathcal{F}_A$ and $\mathcal{F}_B$ are compact. Hence
\begin{align}
    \mathcal{F}_A\otimes \mathcal{F}_B:=\left\{\tau_A\otimes \tau_B\,\middle|\,\tau_A\in\mathcal{F}_A,\;\tau_B\in\mathcal{F}_B\right\}
\end{align}
is compact as the continuous image of $\mathcal{F}_A\times\mathcal{F}_B$. Therefore, by Carathéodory's theorem $\operatorname{conv}\left(\mathcal{F}_A\otimes\mathcal{F}_B\right)$ is compact. Consequently, the linear functional
\begin{align}
    \rho_{AB}\longmapsto \operatorname{tr}[G_{AB}\rho_{AB}]
\end{align}
attains its maximum on $\operatorname{conv}\left(\mathcal{F}_A\otimes\mathcal{F}_B\right)$. Thus an optimizer of $\mathrm{cSEP}(G)$ exists.

Now choose, for each $n\in\mathbb{N}$, an index $ k_n\in\{0,1,\ldots,n-1\}$ such that the rounded state $\sigma_{AB}^{(n,k_n)}$ satisfies
\begin{align}
    0\leq\mathrm{cSEP}(G)-\operatorname{tr}\left[G_{AB}\sigma_{AB}^{(n,k_n)}\right]\leq\epsilon(\lrvert{A},\lrvert{B},n),
\end{align}
with $\epsilon(\lrvert{A},\lrvert{B},n)$ bounded as in \autoref{eqn:bound_epsilon_inner_sequence}. We have $\epsilon(\lrvert{A},\lrvert{B},n)\to 0$ as $n\to\infty$. Since every $\sigma_{AB}^{(n,k_n)}$ is feasible for $\mathrm{cSEP}(G)$ and since the feasible set is compact, by the Bolzano--Weierstrass theorem there exists a subsequence $(n_j)_{j\in\mathbb{N}}$ and a state $\sigma_{AB}^{\star}\in\operatorname{conv}\left(\mathcal{F}_A\otimes\mathcal{F}_B\right)$ such that
\begin{align}
    \sigma_{AB}^{(n_j,k_{n_j})}\longrightarrow\sigma_{AB}^{\star}
\end{align}
in trace norm. By continuity of the objective functional,
\begin{align}
    \operatorname{tr}[G_{AB}\sigma_{AB}^{\star}]=\lim_{j\to\infty}\operatorname{tr}\left[G_{AB}\sigma_{AB}^{(n_j,k_{n_j})}\right]=\mathrm{cSEP}(G).
\end{align}
Hence $\sigma_{AB}^{\star}$ is an optimal solution of $\mathrm{cSEP}(G)$. This proves that the rounded feasible states admit a subsequence converging to an optimizer.
\end{proof}

Note, while in \autoref{prop:hierarchy_outer_approx} we were able to formulate a hierarchy, i.e.\ the feasible set to $\mathrm{SDP}_{n+1}(G)$ is a subset to the feasible set of $\mathrm{SDP}_{n}(G)$ for any $n\in\mathbb{N}$ and the feasible set to $\mathrm{cSEP}(G)$ lies densely in the limit set to $\mathrm{SDP}_{\infty}(G)$ (see \autoref{sec:cbo}), an analogous set structure cannot be given in the case of the inner sequence. Moreover, while informationally complete measurements yield a canonical choice for the rounding map and guarantee the stated convergence rate of the inner approximation, this choice is not required for feasibility. Any measurement scheme applied to the $B$-systems produces, after rounding, feasible points of $\mathrm{cSEP}(G)$, and therefore lower bounds on the associated objective value. In numerical implementations, alternative measurements with fewer outcomes may therefore be advantageous, either by improving the finite-level performance of the inner bounds or by simplifying the rounding procedure. The elements of the sequence $\lrbracket{\sigma^{(n, 0)}_{AB}}_{n\in\mathbb{N}}$ are particularly easy to construct. However, the sequence of values
\begin{align}
    \lrbracket{\Trr{}{G_{AB}\, \sigma^{(n, 0)}_{AB}}}_{n\in\mathbb{N}}
\end{align}
does not in general converge to $\mathrm{cSEP}(G)$ from below. Concretely, every element of $\lrbracket{\sigma^{(n, 0)}_{AB}}_{n\in\mathbb{N}}$ lies in the set of constrained separable states, which is compact\footnote{It equals $\operatorname{conv}\lrbracket{\mathcal{F}_A\otimes\mathcal{F}_B}$, whose compactness was established in the proof of \autoref{lem:inner_sequence_general}.}. Passing to a subsequence along which the marginals $\rho_A^{\star}$ and $\rho_{B_1}^{\star}$ of the optimizers converge, say to $\tau_A\in\mathcal{F}_A$ and $\tau_B\in\mathcal{F}_B$, we obtain $\sigma^{(n, 0)}_{AB}\rightarrow\tau_A\otimes\tau_B$ and hence
\begin{align}
    \Trr{}{G_{AB}\, \sigma^{(n, 0)}_{AB}}\longrightarrow\Trr{}{G_{AB}\,\lrbracket{\tau_A\otimes\tau_B}}\leq \mathrm{cSEP}(G).
\end{align}
This lower bound can be arbitrarily bad. As an example, for some finite $d\in\mathbb{N}$, consider
\begin{align}
    G_{AB}=\sum_{i=1}^d\ket{ii}\bra{ii}, \quad \rho_{AB}^\star= \frac{1}{d}\sum_{i=1}^d\ket{ii}\bra{ii},
\end{align}
where the optimal state is a classical probabilistic mixture. Then, 
\begin{align}
    \Trr{}{G_{AB}\,   \rho_{AB}^\star}=1,\quad \rho^\star_A=\rho^\star_B=\frac{\mathbb{1}}{d},\quad  \Trr{}{G_{AB}\,\rho^\star_A\otimes\rho^\star_B}=\frac{1}{d},
\end{align}
which tends to $0$ as $d\rightarrow\infty$. 

Bauer's maximum principle implies that a linear functional over a compact convex set attains its maximum at an extreme point. However, unlike for the unconstrained separable state space, the extreme points of the feasible set of $\mathrm{cSEP}(G)$ need not be pure product states. Thus, pure product optimality cannot be inferred from Bauer's principle alone. Nevertheless, compactness together with the identity $\mathrm{cSEP}(G)=\mathrm{cPROD}(G)$, cf.\ \autoref{sec:cbo}, implies that our rounding scheme admits a convergent subsequence whose limit is an optimizer of $\mathrm{cPROD}(G)$. 

The inner sequence and the outer hierarchy together yield a finite-level certificate of optimality. Concretely, we have the following corollary. 
\begin{corollary}[$\epsilon$-certificate of finite convergence]\label{cor:eps-cert}
Adopt the setting and notation of \autoref{lem:inner_sequence_general}. For each $n\in\mathbb{N}$ let $\mathcal{C}_n$ denote the finite set of product states feasible for $\mathrm{cSEP}(G)$ obtained from an optimizer $\rho_{AB_1^n}^\star$ of $\mathrm{SDP}_n(G)$ via informationally complete measurements, including the $k=0$ product of marginals. Define the upper and lower bounds
\begin{align}
     U_n := \operatorname{SDP}_n(G),\qquad L_n := \max_{\tau\in\mathcal{C}_n}\Trr{}{G_{AB}\,\tau}.
\end{align}
Then:
\begin{enumerate}
  \item For every $n\in\mathbb{N}$, $\;L_n \leq \operatorname{cSEP}(G) \leq U_n$.
  \item Fix $\epsilon\ge 0$. If at some finite level $n^\star$ the computable gap satisfies $U_{n^\star}-L_{n^\star}\leq\epsilon$, then $\operatorname{cSEP}(G)\in[L_{n^\star},U_{n^\star}]$, an interval of width $\leq\epsilon$, and any maximizer $\tau^\star\in\mathcal{C}_{n^\star}$ attaining $L_{n^\star}$ is a feasible $\epsilon$-optimal point. For $\epsilon=0$ this certifies $\operatorname{cSEP}(G)=U_{n^\star}=L_{n^\star}$ with $\tau^\star$ an exact optimizer.
  \item One always has $U_n-L_n\le\epsilon(\lrvert{A},\lrvert{B},n)$ with the rate of \autoref{eqn:bound_epsilon_inner_sequence}; hence $U_n-L_n\to 0$, and for every $\epsilon>0$ the gap obeys $U_n-L_n\leq\epsilon$ for all
  \begin{align}
       n \geq n_0(\epsilon):= \left\lceil\frac{2\ln 2\,\log\lrvert{A}\,\lVert G_{AB}\rVert_\infty^{2}\,\min\{f(A,B),f(B\mid\cdot)\}^{2}}{\epsilon^{2}}\right\rceil .
  \end{align}
\end{enumerate}
\end{corollary}

\begin{proof}
\textbf{(1)} Each $\tau\in\mathcal{C}_n$ is a product state whose marginals satisfy the constraints of $\operatorname{cSEP}(G)$ by \autoref{eqn:constraints_rounded_candidate_components} (and \autoref{eqn:constaints_k_zero_rounded_candidate} for $k=0$), so $\tau$ is feasible and $\operatorname{tr}[G_{AB}\tau]\leq\operatorname{cSEP}(G)$; taking the maximum over $\mathcal{C}_n$ gives $L_n\le\operatorname{cSEP}(G)$. By \autoref{prop:hierarchy_outer_approx}, $\operatorname{SDP}_n(G)$ is an outer approximation of $\operatorname{cSEP}(G)$: every constrained separable state admits an $n$-extension feasible for $\operatorname{SDP}_n(G)$, so maximizing over the larger feasible set gives $\operatorname{cSEP}(G)\leq U_n$.
\textbf{(2)} Assume $U_{n^\star}-L_{n^\star}\leq\epsilon$. By part (1) both differences below are non-negative and
\begin{align}
    \operatorname{cSEP}(G)-L_{n^\star}\leq U_{n^\star}-L_{n^\star}\leq\epsilon,\qquad U_{n^\star}-\operatorname{cSEP}(G)\leq U_{n^\star}-L_{n^\star}\leq\epsilon .
\end{align}
Thus $\operatorname{cSEP}(G)\in[L_{n^\star},U_{n^\star}]$, of width at most $\epsilon$. The maximizer $\tau^\star\in\mathcal{C}_{n^\star}$ is feasible with $0\le\operatorname{cSEP}(G)-\Trr{}{G_{AB}\tau^\star}=\operatorname{cSEP}(G)-L_{n^\star}\le\epsilon$, i.e.\ it is $\epsilon$-optimal. For $\epsilon=0$ the interval collapses, forcing $\operatorname{cSEP}(G)=U_{n^\star}=L_{n^\star}$ and exact optimality of $\tau^\star$.

\textbf{(3)} Let $k_n\in\{0,\dots,n-1\}$ be chosen so that the rounded convex mixture $\sigma_{AB}^{(n,k_n)}$ (cf.\ \autoref{eqn:candidate_rounded_convex_mixture}) satisfies 
\begin{align}
     \Trr{}{G_{AB}\,\sigma_{AB}^{(n,k_n)}}\ \geq\ U_n-\epsilon(\lrvert{A},\lrvert{B},n)
\end{align}
with $\epsilon(\lrvert{A},\lrvert{B},n)$ bounded by \autoref{eqn:bound_epsilon_inner_sequence}. Such a $k_n$ exists by the proof of \autoref{lem:inner_sequence_general}. A convex mixture can never exceed its best component: since $\sigma_{AB}^{(n,k_n)}=\sum_{z_1^{k_n}} p(z_1^{k_n})\,\rho^\star_{A\mid z_1^{k_n}}\otimes\rho^\star_{B\mid z_1^{k_n}}$ has weights $p(z_1^{k_n})\geq 0$ summing to $1$, and any convex combination of real numbers is at most their maximum,
\begin{align}
     \Trr{}{G_{AB}\,\sigma_{AB}^{(n,k_n)}}= \sum_{z_1^{k_n}} p(z_1^{k_n})\,\Trr{}{G_{AB}\big(\rho^\star_{A\mid z_1^{k_n}}\otimes\rho^\star_{B\mid z_1^{k_n}}\big)}\ \leq\ \max_{z_1^{k_n}}\Trr{}{G_{AB}\big(\rho^\star_{A\mid z_1^{k_n}}\otimes\rho^\star_{B\mid z_1^{k_n}}\big)}.
\end{align}
Let $z^\star$ attain this maximum. The maximizing component $\rho^\star_{A\mid z^\star}\otimes\rho^\star_{B\mid z^\star}$ lies in $\mathcal{C}_n$, so by definition of $L_n$,
\begin{align}
     L_n\ \geq\ \Trr{}{G_{AB}\big(\rho^\star_{A\mid z^\star}\otimes\rho^\star_{B\mid z^\star}\big)}
  \ \geq\ \Trr{}{G_{AB}\,\sigma_{AB}^{(n,k_n)}}\ \geq\ U_n-\epsilon(\lrvert{A},\lrvert{B},n).
\end{align}
Together with part (1), which gives $L_n\le\operatorname{cSEP}(G)\le U_n$, this yields
\begin{align}\label{eq:cor-gap}
  0\ \le\ U_n-L_n\ \le\ \epsilon(|A|,|B|,n).
\end{align}
Fix $\epsilon>0$. For every
\begin{align}
     n\ \geq\ n_0(\epsilon)\ := \left\lceil \frac{2\ln 2\,\log\lrvert{A}\, \lVert G_{AB}\rVert_{\infty}^{2}\,\min\{f(A,B),f(B|\cdot)\}^{2}}{\epsilon^{2}}\right\rceil   
\end{align}
one has $\epsilon(|A|,|B|,n)\le\epsilon$, and therefore $U_n-L_n\le\epsilon$, i.e.\ the hypothesis of part (2) is met.
\end{proof}

\begin{figure}
    \centering
    \begin{tikzpicture}[>=Stealth, line join=round, font=\small]


\fill[sdp1] (2.38,0) ellipse [x radius=3.5,  y radius=2.6];
\draw[blin!70,thick] (2.38,0) ellipse [x radius=3.5, y radius=2.6];
\fill[sdp2] (1.63,0) ellipse [x radius=2.75, y radius=2.02];
\draw[blin!70,thick] (1.63,0) ellipse [x radius=2.75, y radius=2.02];
\fill[sdpn] (0.88,0) ellipse [x radius=2.0,  y radius=1.46];
\draw[blin,thick,dashed] (0.88,0) ellipse [x radius=2.0, y radius=1.46];
\fill[csep] (0,0) ellipse [x radius=1.12, y radius=0.86];
\draw[csin,very thick] (0,0) ellipse [x radius=1.12, y radius=0.86];

\draw[black,densely dotted,thin,shorten >=1.8pt,-{Stealth[length=3.5pt]}] (2.85,0) .. controls (1.40,0.92) and (0.00,0.72) .. (-0.62,0.30);  
\draw[black,densely dotted,thin,shorten >=1.8pt,-{Stealth[length=3.5pt]}] (2.85,0) .. controls (1.60,1.00) and (0.55,0.90) .. (-0.02,0.52);  
\draw[black,densely dotted,thin,shorten >=1.8pt,-{Stealth[length=3.5pt]}] (2.85,0) .. controls (1.40,-0.30) and (-0.10,-0.35) .. (-0.68,-0.22); 
\draw[black,densely dotted,thin,shorten >=1.8pt,-{Stealth[length=3.5pt]}] (2.85,0) to[bend right=16] (-0.18,-0.50); 
\draw[black,densely dotted,thin,shorten >=1.8pt,-{Stealth[length=3.5pt]}] (2.85,0) to[bend right=8]  (0.34,-0.36);  

\foreach \p in {(-0.62,0.30),(-0.02,0.52),(-0.68,-0.22),%
                (-0.18,-0.50),(0.34,-0.36)}{\fill[pt] \p circle (1.5pt);}
\fill[star] (0.58,0.10) circle (2.3pt);
\draw[star] (0.58,0.10) circle (3.9pt);

\node[diamond,draw=csin,fill=csin,inner sep=1.7pt] (optc) at (1.12,0) {}; 
\node[anchor=west,text=csin,font=\scriptsize] at (1.22,0.30) {$\rho^\star_{\mathrm{cSEP}}$};

\fill[white] (2.88,0) circle (3.2pt);
\fill[blin] (2.88,0) circle (2.3pt);
\node[anchor=west,text=blin!80!black,font=\scriptsize] at (3.06,0.16)
   {$\rho^\star_{AB_1^n}$ -- optimizer of $\mathrm{SDP}_n(G)$};

\draw[->,black,very thick]
   (2.72,-0.10) .. controls (2.0,-1.05) and (1.15,-0.60) .. (0.74,0.02);
\node[text=black,font=\scriptsize,align=center] at (1.82,-1.11)
   {measurement-based rounding};

\node[blin!80!black] at (2.55,2.30) {$\mathrm{SDP}_1(G)$};
\node[blin!80!black] at (1.75,1.74) {$\mathrm{SDP}_2(G)$};
\node[blin!85!black] at (0.95,1.16) {$\mathrm{SDP}_n(G)$};

\draw[->,blin!70,thick] (4.15,2.13) to[bend right=16] (2.23,1.10);
\node[blin!60!black,font=\footnotesize] at (4.35,2.31) {$n\!\uparrow$};

\draw[gd,dashed] (-1.12,-2.75) -- (-1.12,2.7);
\fill[black] (-1.12,0) circle (2.1pt);
\node[font=\footnotesize] at (-1.32,0.22) {$P$};

\begin{scope}[shift={(-1.05,4.35)}, font=\footnotesize]
  \fill[csep] (0,0) rectangle (0.40,0.28); \draw[csin,very thick] (0,0) rectangle (0.40,0.28);
  \node[anchor=west,text=csin] at (0.52,0.14)
     {$\mathrm{cSEP}(G)$ \ -- target set};
  \fill[pt] (0.20,-0.31) circle (1.5pt);
  \node[anchor=west,text=pt] at (0.52,-0.31)
     {$\mathcal C_n$ \ -- feasible \emph{product} points $(\subseteq\mathrm{cSEP}(G))$};
  \fill[star] (0.20,-0.76) circle (2.3pt); \draw[star] (0.20,-0.76) circle (3.9pt);
  \node[anchor=west,text=star] at (0.52,-0.76)
     {$\tau^\star$ \ -- best inner point, attains $L_n=\max_{\tau\in\mathcal C_n}\mathrm{tr}[G_{AB}\tau]$};
\end{scope}

\def\yax{-3.55}
\draw[gd!70,dashed] (0.58,0.10) -- (0.58,\yax);
\draw[gd!70,dashed] (1.12,0)    -- (1.12,\yax);
\draw[gd!70,dashed] (2.88,0)    -- (2.88,\yax);

\draw[->,thick] (-1.4,\yax) -- (6.4,\yax);
\node[anchor=west,font=\footnotesize] at (6.02,{\yax-0.30})
   {increasing $\mathrm{tr}[G_{AB}\,\cdot\,]$};
\node[font=\footnotesize] at (-1.12,{\yax-0.30}) {$\leftarrow$ tight};

\foreach \x/\c in {0.58/star,1.12/csin,2.88/blin,4.38/gd,5.88/gd}
   \draw[\c,thick] (\x,{\yax+0.12}) -- (\x,{\yax-0.12});
\node[anchor=north east,text=star,font=\footnotesize]  at (0.66,{\yax-0.16}) {$L_n$};
\node[anchor=north west,text=csin,font=\footnotesize]  at (1.16,{\yax-0.16}) {$\mathrm{cSEP}(G)$};
\node[anchor=north,text=blin!80!black,font=\footnotesize] at (2.88,{\yax-0.16}) {$U_n$};
\node[anchor=north,text=gd,font=\scriptsize] at (4.38,{\yax-0.16}) {$U_2$};
\node[anchor=north,text=gd,font=\scriptsize] at (5.88,{\yax-0.16}) {$U_1$};

\draw[->,blin,thick] (2.70,{\yax+0.34}) -- (1.30,{\yax+0.34});
\draw[->,star,thick] (0.72,{\yax+0.34}) -- (1.00,{\yax+0.34});

\draw[decorate,decoration={brace,amplitude=5pt,mirror},thick]
   (0.58,{\yax-0.98}) -- (2.88,{\yax-0.98});
\node[anchor=north,font=\footnotesize] at (1.73,{\yax-1.10})
   {$U_n-L_n\le\epsilon(|A|,|B|,n)\xrightarrow{\,n\to\infty\,}0$};

\end{tikzpicture}
    \caption{\textbf{Geometry of the finite-level certificate (\autoref{cor:eps-cert}),}
projected onto the objective $G_{AB}$ (horizontal), so horizontal position equals $\mathrm{tr}[G_{AB}\,\cdot\,]$. The nested relaxations $\mathrm{SDP}_1(G)\supseteq\mathrm{SDP}_2(G)\supseteq\cdots$ (blue) enclose the target
$\mathrm{cSEP}(G)$ (amber), giving upper bounds
$U_1\ge U_2\ge\cdots\downarrow\mathrm{cSEP}(G)$; since any contact point lies on
$\partial\,\mathrm{SDP}_k(G)$ for every $k$, all levels touch $\mathrm{cSEP}(G)$ tangentially
at the common point $P$ (a constrained product state, tight in that direction), and the gap
opens only along $G_{AB}$. The optimizer $\rho^\star_{AB_1^n}$ of $\mathrm{SDP}_n(G)$ lies on
$\partial\,\mathrm{SDP}_n(G)$ and realizes $U_n$; measurement-based rounding (teal) maps it to a
feasible product point in $\mathcal{C}_n\subseteq\mathrm{cSEP}(G)$, the best of which,
$\tau^\star$, attains $L_n$. As $\mathrm{cSEP}(G)=\mathrm{cPROD}(G)$, its extreme points are
product states (Milman), so the optimizer $\rho^\star_{\mathrm{cSEP}}$ is one (Bauer)---not
necessarily pure. Hence $L_n\le\mathrm{cSEP}(G)\le U_n$ with
$U_n-L_n\le\epsilon(|A|,|B|,n)\to0$.}
    \label{fig:graphic_cor_certificate_finite_convergence}
\end{figure}

A schematic depiction of \autoref{cor:eps-cert} is given in \autoref{fig:graphic_cor_certificate_finite_convergence}. Both $U_n$ and $L_n$ are computed directly from $\mathrm{SDP}_n(G)$ and the rounded states, so the gap $U_n-L_n$ is an \emph{a posteriori}, fully computable optimality certificate, in contrast to the \emph{a priori} rate based on the finite de Finetti theorem. Because the inner sequence carries no nesting (cf.\ the discussion after \autoref{lem:inner_sequence_general}), $L_n$ need not be monotone; tracking the running best bounds $\underline{L}_n:=\max_{m\le n}L_m$ and $\overline{U}_n:=\min_{m\le n}U_m$ yields a monotone, non-increasing certified gap $\overline{U}_n-\underline{L}_n\le U_n-L_n$ that still brackets $\operatorname{cSEP}(G)$. 

Lastly, the proof of \autoref{thm:short_version_inner_sequence} follows directly from \autoref{lem:inner_sequence_general} with
\begin{align}
\begin{array}{cccccc}
     A_L= A_1Q_1,& C_{A_L}=Q_1,& A_R=T,& B_L= A_2Q_2\hat{T},& C_{B_L}=Q_2\hat{T},& B_R=\CC,\\
\end{array}
\end{align}
and
\begin{align}
\begin{split}
     \begin{array}{cc}
     \Theta_{A_L\rightarrow C_{A_L}}(\cdot)=\Trr{A_1}{\cdot}, & \Upsilon_{B_L\rightarrow C_{B_L}}(\cdot) = \Trr{A_2}{\cdot},\\
      W_{C_{A_L}}=\sum_{q_1\in Q_1}\pi_1(q_1)\ket{q_1}\bra{q_1}_{Q_1},   &  K_{C_{B_L}}= \sum_{q_2\in Q_2}\pi_2(q_2)\ket{q_2}\bra{q_2}_{Q_2}\otimes\frac{\mathbb{1}_{\hat{T}}}{\lrvert{T}}.
    \end{array}
\end{split}
\end{align}

Note, however, that due to the steering procedure in \autoref{sec:results_games_as_cbo}, in contrast to \cite{brandao_harrow_2014, kempe2011entangled}, the rounded strategy is not classical but quantum.

\section{Bose-symmetric hierarchy for constrained separability problems}
\label{sec:bose-symmetric}

\subsection{Introduction}

We introduce Bose-symmetric operators via the well-known notion of the symmetric subspace described in detail in, e.g., \cite{harrow2013church}. In the following, we first introduce in \autoref{sec:Bose_sym_easy} the minimal tools needed for our constrained Bose-symmetric de Finetti theorem (\autoref{sec:de-finetti-proof}), the corresponding SDP hierarchy (\autoref{sec:the-hierarchy}), its computational complexity (\autoref{sec:sdp-complexity}), and its application to non-local games (\autoref{sec:nonlocal-games}). We present concepts related to the symmetric subspace in greater mathematical detail in \autoref{sec:bose_via_schur_weyl} and provide the reader with general information about the relevant representation theory in \autoref{sec:representation_theory}. This is meant as a thorough discussion related to the literature \cite{harrow2013church, christandl2007one, fulton2013representation, etingof2009introduction, stevens2016schur}, which will be particularly important for \autoref{sec:symmetric_subspace_methods}. In particular, the structural consequence of Bose symmetry that drives our algorithm --- namely, that operators which are Bose-symmetric with respect to side systems form a \emph{full} matrix $*$-algebra of size polynomial in $n$ --- is stated and exploited at the beginning of \autoref{sec:symmetric_subspace_methods} (\autoref{prop:bose_full_matrix_algebra}); the present section only requires the associated dimension count (\autoref{prop:Bose_sym_space_char}).


\subsection{Bose-symmetry}\label{sec:Bose_sym_easy}

We collect a set of basic facts concerning Bose-symmetric operators that will be used throughout the paper.

\begin{definition}[Symmetric subspace]
Let $B$ be a finite-dimensional complex Hilbert space and $n\in\mathbb{N}$. The symmetric subspace of $B_1^n$ is defined as
    \begin{align}
        \lor^n\lrbracket{B} := \lrbrace{\ket{\psi}\in B_1^n \,:\, U_{B_1^n}(\pi)\ket{\psi}=\ket{\psi},\,\forall \pi\in S_n}=\text{span}_{\mathbb{C}}\lrbrace{\ket{\psi}^{\otimes n}\, :\, \ket{\psi}\in B}.
    \end{align}
    The orthogonal projector onto $\lor^n\lrbracket{B}$ is 
    \begin{align}
    	P_{\lor^n\lrbracket{B}} := \frac{1}{\lrvert{S_n}}\sum_{\pi\in S_n} U_{B_1^n}(\pi).\,
    \end{align}
\end{definition}
Since $U_{B_1^n}\,:\,S_n \rightarrow \operatorname{GL}\lrbracket{B_1^n}$ is a unitary representation (group homomorphism) with real-valued matrices, we have 
\begin{align}
    U^T_{B_1^n}(\pi)=U^{-1}_{B_1^n}(\pi)= U_{B_1^n}(\pi^{-1}),\quad \forall \pi\in S_n.
\end{align}
Moreover, $\pi\mapsto \pi^{-1}$ is a bijection of $S_n$. Thus,
\begin{align}
    P_{\lor^n\lrbracket{B}}^\dagger = P_{\lor^n\lrbracket{B}}^T = \frac{1}{\lrvert{S_n}}\sum_{\pi\in S_n} U^T_{B_1^n}(\pi)= \frac{1}{\lrvert{S_n}}\sum_{\pi\in S_n} U_{B_1^n}(\pi^{-1}) = P_{\lor^n\lrbracket{B}}.
\end{align}
The symmetric subspace induces a natural notion of symmetry for operators, namely Bose symmetry.
\begin{definition}\label{def:bose_symmetric} 
 An operator $X_{B_1^n}\in\mathcal{B}\lrbracket{B_1^n}$ is Bose-symmetric if 
    \begin{align}
        P_{\lor^n\lrbracket{B}}X_{B_1^n}P_{\lor^n\lrbracket{B}}=X_{B_1^n}.
    \end{align}
More generally, an operator $X_{AB_1^n}\in\mathcal{B}\lrbracket{AB_1^n}$ is Bose-symmetric w.r.t.\ $A$ if 
 \begin{align}
     \lrbracket{\mathbb{1}_A\otimes P_{\lor^n\lrbracket{B}}}X_{AB_1^n}\lrbracket{\mathbb{1}_A\otimes P_{\lor^n\lrbracket{B}}}=X_{AB_1^n}.
 \end{align}
  Moreover, the space of Bose-symmetric operators is $\End{}{\lor^n\lrbracket{B}}$ and the space of operators Bose-symmetric w.r.t.\ $A$ is $\End{}{A\otimes \lor^n\lrbracket{B}}$.
\end{definition}
The definition above implies that any Bose-symmetric operator annihilates $\lrbracket{\lor^n\lrbracket{B}}^{\perp}$ and has range contained in $\lor^n\lrbracket{B}$. Note that we canonically identify
\begin{align}
    \End{}{\vee^n(B)}=\End{\CSn}{\vee^n(B)}, \quad \End{}{A\otimes \vee^n(B)}=\End{\CSn}{A\otimes \vee^n(B)},
\end{align}
where we suppress the trivial group algebra action on $A$. We nevertheless use the notation $\End{\CSn}{\vee^n(B)}$ for the algebra of Bose-symmetric operators, in analogy with the notation $\End{\CSn}{B_1^n}$ for the algebra of exchange-symmetric operators. See \autoref{prop:Bose_sym_space_char} for further details.

\begin{remark}\label{rem:bose_simple_full_matrix}
    Since $\End{}{\lor^n\lrbracket{B}}$ is the full endomorphism algebra of a $\binom{\lrvert{B}+n-1}{n}$-dimensional space, the Bose-symmetric operators form a \emph{simple} unital associative matrix $*$-algebra, with unit given by the projector $P_{\lor^n\lrbracket{B}}$; the same holds in the presence of side systems. This stands in contrast to the algebra $\End{\CSn}{B_1^n}$ of merely symmetric operators, which decomposes into a direct sum of many full matrix blocks (cf.~\autoref{sec:sdp_symmetry_reduction}). We return to this observation in \autoref{prop:bose_full_matrix_algebra}, where it underpins the polynomial-time algorithm behind \autoref{thm:Bose_sym_complexity_results}.
\end{remark}

As in the case of exchange symmetry, invariance with respect to a quantum side system $A$ is a stronger requirement.
\begin{proposition}
    Let $X_{AB_1^n}\in\mathcal{B}\lrbracket{AB_1^n}$ be Bose-symmetric w.r.t.\ $A$. Then, $\Trr{A}{X_{AB_1^n}}$ is Bose-symmetric.
\end{proposition}
\begin{proof}
A direct calculation gives
    \begin{align}
        X_{B_1^n}= \Trr{A}{X_{AB_1^n}} = \Trr{A}{\lrbracket{\mathbb{1}_A\otimes P_{\lor^n\lrbracket{B}}}X_{AB_1^n}\lrbracket{\mathbb{1}_A\otimes P_{\lor^n\lrbracket{B}}}}=  P_{\lor^n\lrbracket{B}}X_{B_1^n}P_{\lor^n\lrbracket{B}},
    \end{align}
    which is the desired identity.
\end{proof}
For $n\geq 2$ and $\dim(B)\geq 2$, the identity operator on $B_1^n$ is symmetric but not Bose-symmetric. Accordingly, we henceforth restrict our arguments to the stronger condition of Bose symmetry with respect to a quantum side system.
\begin{proposition}\label{prop:Bose-symmetry_implies_symmetry}
    Let $X_{AB_1^n}\in\mathcal{B}\lrbracket{AB_1^n}$ be Bose-symmetric w.r.t.\ $A$. Then 
    \begin{align}
        \lrbracket{\mathbb{1}_A\otimes U_{B_1^n}(\pi)}X_{AB_1^n}\lrbracket{\mathbb{1}_A\otimes U^T_{B_1^n}(\pi)}=X_{AB_1^n},\quad \forall \pi\in S_n
    \end{align}
    holds. In other words, $X_{AB_1^n}$ is symmetric w.r.t.\ $A$. Moreover, let $X_{B_1^n}\in\mathcal{B}\lrbracket{B_1^n}$ be Bose-symmetric. Then $X_{B_1^n}$ is symmetric.
\end{proposition}
\begin{proof}
For any $\pi,\,\sigma \in S_n$ we have
\begin{align}
    U_{B_1^n}(\pi)U_{B_1^n}(\sigma)= U_{B_1^n}(\pi\sigma),\quad U^T_{B_1^n}(\pi)=U^{-1}_{B_1^n}(\pi)= U_{B_1^n}(\pi^{-1})
\end{align}
since $U_{B_1^n}\,:\,S_n \rightarrow \operatorname{GL}\lrbracket{B_1^n}$ is a unitary representation (group homomorphism). Thus, for a fixed $\pi\in S_n$ we have
\begin{align}
     U_{B_1^n}(\pi) P_{\lor^n\lrbracket{B}} = \frac{1}{\lrvert{S_n}}\sum_{\sigma\in S_n}  U_{B_1^n}(\pi)U_{B_1^n}(\sigma)= \frac{1}{\lrvert{S_n}}\sum_{\sigma\in S_n}U_{B_1^n}(\pi\sigma).
\end{align}
Since the map $\sigma\mapsto\pi\sigma$ is a bijection of $S_n$, the set $\lrbrace{\pi\sigma \,:\,\sigma\in S_n}$ is again all of $S_n$. Thus, $P_{\lor^n\lrbracket{B}}$ is invariant under the left action of $U_{B_1^n}(\pi)$ for any $\pi\in S_n$. Similarly,
\begin{align}
    P_{\lor^n\lrbracket{B}}U^T_{B_1^n}(\pi) = \frac{1}{\lrvert{S_n}}\sum_{\sigma\in S_n}U_{B_1^n}(\sigma\pi^{-1}).
\end{align}
Again, the map $\sigma\mapsto\sigma\pi^{-1}$ is a bijection of $S_n$. Thus, for any $\pi\in S_n$ we have $P_{\lor^n\lrbracket{B}}U^T_{B_1^n} (\pi) = P_{\lor^n\lrbracket{B}}$. Hence,
    \begin{align}
    \begin{split}
         \lrbracket{\mathbb{1}_A\otimes U_{B_1^n}(\pi)}X_{AB_1^n}\lrbracket{\mathbb{1}_A\otimes U^T_{B_1^n}(\pi)}&= \lrbracket{\mathbb{1}_A\otimes U_{B_1^n}(\pi)P_{\lor^n\lrbracket{B}}}X_{AB_1^n}\lrbracket{\mathbb{1}_A\otimes P_{\lor^n\lrbracket{B}}U^T_{B_1^n}(\pi)}\\
         &=\lrbracket{\mathbb{1}_A\otimes P_{\lor^n\lrbracket{B}}}X_{AB_1^n}\lrbracket{\mathbb{1}_A\otimes P_{\lor^n\lrbracket{B}}}\\
         &=X_{AB_1^n}.
    \end{split}
    \end{align}
    The second statement, without the side system $A$, follows by setting $A=\CC$.
\end{proof}
The following proposition is a key ingredient in the proof of the Bose-symmetric de Finetti representation theorem given in the next section.
\begin{proposition}\label{prop:Bose-symmetry_for_normal_operators_via_exchange_representation}
     Consider $X_{AB_1^n}\in\mathcal{B}\lrbracket{AB_1^n}$ such that  
     \begin{align}
         \ker\lrbracket{X_{AB_1^n}}=\ker\lrbracket{X_{AB_1^n}^\dagger}.
     \end{align}
     Then, $X_{AB_1^n}$ is Bose-symmetric w.r.t.\ $A$ if and only if
     \begin{align}\label{eqn:Bose-symmetry_w_A_normal}
         \lrbracket{\mathbb{1}_A\otimes U_{B_1^n}(\pi)}X_{AB_1^n} =  X_{AB_1^n},\,\forall \pi\in S_n.
     \end{align}
     Moreover, consider $X_{B_1^n}\in\mathcal{B}\lrbracket{B_1^n}$ such that
     \begin{align}
          \ker\lrbracket{X_{B_1^n}}=\ker\lrbracket{X_{B_1^n}^\dagger}.
     \end{align}
     Then, $X_{B_1^n}$ is Bose-symmetric if and only if
     \begin{align}
        U_{B_1^n}(\pi)X_{B_1^n} =  X_{B_1^n},\,\forall \pi\in S_n.
     \end{align}
\end{proposition}
\begin{proof}
First observe that
\begin{align}
     \lrbracket{\mathbb{1}_A\otimes U_{B_1^n}(\pi)}X_{AB_1^n}=  X_{AB_1^n},\,\forall \pi\in S_n
\end{align}
implies
\begin{align}
     \lrbracket{\mathbb{1}_A\otimes P_{\lor^n\lrbracket{B}}}X_{AB_1^n} = \frac{1}{\lrvert{S_n}}\sum_{\pi\in S_n}  \lrbracket{\mathbb{1}_A\otimes U_{B_1^n}(\pi)}X_{AB_1^n} = \frac{1}{\lrvert{S_n}}\sum_{\pi\in S_n} X_{AB_1^n} = X_{AB_1^n}.
\end{align}
Conversely, if $\lrbracket{\mathbb{1}_A\otimes P_{\lor^n\lrbracket{B}}}X_{AB_1^n}=X_{AB_1^n}$, then for every $\pi\in S_n$, 
\begin{align}
     \lrbracket{\mathbb{1}_A\otimes U_{B_1^n}(\pi)}X_{AB_1^n}=  \lrbracket{\mathbb{1}_A\otimes U_{B_1^n}(\pi)P_{\lor^n\lrbracket{B}}}X_{AB_1^n} = \lrbracket{\mathbb{1}_A\otimes P_{\lor^n\lrbracket{B}}}X_{AB_1^n}=X_{AB_1^n},
\end{align}
because 
\begin{align}\label{eqn:permutation_matrix_invariance_symmetric_subspace_projetor}
    \lrbracket{\mathbb{1}_A\otimes U_{B_1^n}(\pi)}\lrbracket{\mathbb{1}_A\otimes P_{\lor^n\lrbracket{B}}}=\lrbracket{\mathbb{1}_A\otimes P_{\lor^n\lrbracket{B}}}.    
\end{align}
Hence, 
\begin{align}
     \lrbracket{\mathbb{1}_A\otimes U_{B_1^n}(\pi)}X_{AB_1^n}=  X_{AB_1^n},\,\forall \pi\in S_n \quad\Leftrightarrow\quad \lrbracket{\mathbb{1}_A\otimes P_{\lor^n\lrbracket{B}}}X_{AB_1^n}=X_{AB_1^n}.
\end{align}
Since $\lrbracket{\mathbb{1}_A\otimes P_{\lor^n\lrbracket{B}}}X_{AB_1^n}=X_{AB_1^n}$ holds, taking adjoints gives 
\begin{align}
\begin{split}
     \lrrec{\lrbracket{\mathbb{1}_A\otimes P_{\lor^n\lrbracket{B}}}X_{AB_1^n}}^\dagger &= X_{AB_1^n}^\dagger\\
     \Leftrightarrow X_{AB_1^n}^\dagger \lrbracket{\mathbb{1}_A\otimes P_{\lor^n\lrbracket{B}}}^\dagger &=  X_{AB_1^n}^\dagger \lrbracket{\mathbb{1}_A\otimes P_{\lor^n\lrbracket{B}}}= X_{AB_1^n}^\dagger.
\end{split}
\end{align}
Therefore,
\begin{align}
     X_{AB_1^n}^\dagger\lrbracket{\mathbb{1}_{AB_1^n} - \lrbracket{\mathbb{1}_A\otimes P_{\lor^n\lrbracket{B}}}}=0.
\end{align}
Since $P_{\lor^n\lrbracket{B}}$ is an orthogonal projector, this implies that $X_{AB_1^n}^\dagger$ annihilates every vector in the orthogonal complement
\begin{align}
    \lrbracket{\mathbb{1}_{AB_1^n} - \lrbracket{\mathbb{1}_A\otimes P_{\lor^n\lrbracket{B}}}}\lrbracket{AB_1^n}
\end{align}
of the symmetric subspace. Thus, 
\begin{align}
    \lrbracket{\mathbb{1}_{AB_1^n} - \lrbracket{\mathbb{1}_A\otimes P_{\lor^n\lrbracket{B}}}}\lrbracket{AB_1^n} \subseteq \ker\lrbracket{X_{AB_1^n}^\dagger}.
\end{align}
With $\ker\lrbracket{X_{AB_1^n}}=\ker\lrbracket{X_{AB_1^n}^\dagger}$, it follows that 
\begin{align}
      \lrbracket{\mathbb{1}_{AB_1^n} - \lrbracket{\mathbb{1}_A\otimes P_{\lor^n\lrbracket{B}}}}\lrbracket{AB_1^n} \subseteq \ker\lrbracket{X_{AB_1^n}} \, \Leftrightarrow \, X_{AB_1^n}\lrbracket{\mathbb{1}_{AB_1^n} - \lrbracket{\mathbb{1}_A\otimes P_{\lor^n\lrbracket{B}}}}=0
\end{align}
and hence 
\begin{align}
    X_{AB_1^n}\lrbracket{\mathbb{1}_A\otimes P_{\lor^n\lrbracket{B}}} =  X_{AB_1^n}.
\end{align}
In other words, $X_{AB_1^n}$ also annihilates every vector orthogonal to the symmetric subspace. Combining this with $\lrbracket{\mathbb{1}_A\otimes P_{\lor^n\lrbracket{B}}}X_{AB_1^n}=X_{AB_1^n}$, we obtain
\begin{align}
    \lrbracket{\mathbb{1}_A\otimes P_{\lor^n\lrbracket{B}}}X_{AB_1^n}\lrbracket{\mathbb{1}_A\otimes P_{\lor^n\lrbracket{B}}}= X_{AB_1^n}\lrbracket{\mathbb{1}_A\otimes P_{\lor^n\lrbracket{B}}} =X_{AB_1^n}.
\end{align}
Conversely, assume 
\begin{align}
     \lrbracket{\mathbb{1}_A\otimes P_{\lor^n\lrbracket{B}}}X_{AB_1^n}\lrbracket{\mathbb{1}_A\otimes P_{\lor^n\lrbracket{B}}}=X_{AB_1^n}.
\end{align}
Then, for every $\pi\in S_n$,
\begin{align}
    \begin{split}
        \lrbracket{\mathbb{1}_A\otimes U_{B_1^n}(\pi)}X_{AB_1^n}& = \lrbracket{\mathbb{1}_A\otimes U_{B_1^n}(\pi)P_{\lor^n\lrbracket{B}}}X_{AB_1^n}\lrbracket{\mathbb{1}_A\otimes P_{\lor^n\lrbracket{B}}}\\
        &= \lrbracket{\mathbb{1}_A\otimes P_{\lor^n\lrbracket{B}}}X_{AB_1^n}\lrbracket{\mathbb{1}_A\otimes P_{\lor^n\lrbracket{B}}}=X_{AB_1^n},
    \end{split}
\end{align}
again using \autoref{eqn:permutation_matrix_invariance_symmetric_subspace_projetor}. This proves the equivalence. The second statement, without the side system $A$, follows by setting $A=\CC$.
\end{proof}
Note that without the condition on the kernel of  $X_{AB_1^n}$, an operator $X_{AB_1^n}$ satisfying \autoref{eqn:Bose-symmetry_w_A_normal} will only satisfy
\begin{align}
     \lrbracket{\mathbb{1}_A\otimes P_{\lor^n\lrbracket{B}}}X_{AB_1^n}=X_{AB_1^n}.
\end{align}
This left invariance implies that the range of $X_{AB_1^n}$ lies in $A\otimes\lor^n\lrbracket{B}$, but not that its support $\operatorname{supp}\lrbracket{X_{AB_1^n}}:=\ker\lrbracket{X_{AB_1^n}}^\perp$ does too. More formally, we give the following proposition for the reader's convenience.
\begin{proposition}
    In finite dimensions, let $X_{AB_1^n}\in\mathcal{B}\lrbracket{AB_1^n}$ satisfy $\ker\lrbracket{X_{AB_1^n}}=\ker\lrbracket{X_{AB_1^n}^\dagger}$. Then, 
     \begin{align}
    \operatorname{supp}(X_{AB_1^n})=\ker(X_{AB_1^n})^{\perp}=\operatorname{ran}(X_{AB_1^n}).
\end{align}
\end{proposition}
\begin{proof}
   A standard identity is 
   \begin{align}
   \operatorname{ran}\lrbracket{X_{AB_1^n}}^\perp=\ker\lrbracket{X_{AB_1^n}^\dagger}.  
   \end{align}
   Taking orthogonal complements and using finite-dimensionality yields 
   \begin{align}
   \operatorname{ran}\lrbracket{X_{AB_1^n}}=\ker\lrbracket{X_{AB_1^n}^\dagger}^\perp.    
   \end{align} 
   By the hypothesis and with the definition of the support the claim follows.
\end{proof}

We conclude this section with a proposition showing that the operators considered in this work satisfy the condition of \autoref{prop:Bose-symmetry_for_normal_operators_via_exchange_representation}.

\begin{proposition}
    Let $X_{AB_1^n}\in\mathcal{B}\lrbracket{AB_1^n}$ be a normal operator, i.e.
    \begin{align}
        X_{AB_1^n}^\dagger X_{AB_1^n}= X_{AB_1^n} X_{AB_1^n}^\dagger.
    \end{align}
    Then, $\ker(X_{AB_1^n})=\ker(X_{AB_1^n}^\dagger)$.
\end{proposition}
\begin{proof}
Let $\langle \cdot, \cdot\rangle$ be the inner product on the Hilbert space $AB_1^n$. Since $X_{AB_1^n}$ is normal, for every vector $v\in AB_1^n$,  
\begin{align}
\begin{split}
    \left\lVert X_{AB_1^n}\, v\right\rVert^2 &= \left\langle X_{AB_1^n}\, v, X_{AB_1^n}\, v  \right\rangle = \left\langle v, X_{AB_1^n}^\dagger X_{AB_1^n}\, v  \right\rangle\\
    &= \left\langle v,  X_{AB_1^n}X_{AB_1^n}^\dagger\, v  \right\rangle = \left\langle X_{AB_1^n}^\dagger\, v, X_{AB_1^n}^\dagger\, v  \right\rangle =  \left\lVert X_{AB_1^n}^\dagger\, v\right\rVert^2.
\end{split}
\end{align}
Hence
\begin{align}
    \left\lVert  X_{AB_1^n}\, v \right\rVert=0 \quad\Leftrightarrow \quad   \left\lVert  X_{AB_1^n}^\dagger\, v \right\rVert=0.
\end{align}
Since a vector has norm zero if and only if it is zero, we obtain $X_{AB_1^n}\,v=0$ if and only if $X_{AB_1^n}^\dagger\,v=0$. By the definition of the kernel, the claim follows.
\end{proof}
Throughout the remainder of this paper, all operators under consideration will be positive semidefinite. In particular, they are self-adjoint, and hence normal. Henceforth, we shall use the results of this section to identify Bose-symmetric states with states supported on the symmetric subspace.


\subsection{A de Finetti representation theorem for constrained Bose-symmetric states}
\label{sec:de-finetti-proof}

We now prove the following lemma approximating the set of constrained separable states via constrained Bose-symmetric states.

\begin{lemma}[Constrained de Finetti representation theorem for Bose-symmetric states (restated)]\label{lem:bose-symmetric_deFinetti}
Let $A=A_L\otimes A_R$, $B=B_L\otimes B_R$, $\bar{B}=\bar{B}_L\otimes \bar{B}_R$, $C_{A_L}$ and $C_{B_L}$ be Hilbert spaces with $B \cong\bar{B}$. Moreover, let $\rho_{A\lrbracket{B\Bar{B}}_1^n}$ be Bose-symmetric with respect to $A$. Furthermore, consider linear mappings $\Theta_{A_L\rightarrow C_{A_L}},\, \Omega_{A\rightarrow A},\, \Upsilon_{B_L\rightarrow C_{B_L}},\, \Xi_{B\rightarrow B}$ and operators $W_{C_{A_L}},\, K_{C_{B_L}}$. If
    \begin{align}\label{eqn:SDP_constraints_appendix}
    \begin{split}
        \begin{array}{cc}
            \Theta_{A_L\rightarrow C_{A_L}}\lrbracket{\rho_{A\lrbracket{B\Bar{B}}_1^n}} = W_{C_{A_L}}\otimes\rho_{A_R\lrbracket{B\Bar{B}}_1^n},\, &
            \Omega_{A\rightarrow A}\lrbracket{\rho_{A\lrbracket{B\Bar{B}}_1^n}}=\rho_{A\lrbracket{B\Bar{B}}_1^n},\,\\
            \Upsilon_{\lrbracket{B_L}_n\rightarrow C_{B_L}}\circ \Tr_{\Bar{B}_n}\lrbracket{\rho_{\lrbracket{B\Bar{B}}_1^n}} = K_{C_{B_L}}\otimes \rho_{\lrbracket{B\Bar{B}}_1^{n-1}\lrbracket{B_R}_n},\, & \Xi_{B_n\rightarrow B_n}\circ \Tr_{\Bar{B}_n}\lrbracket{\rho_{\lrbracket{B\Bar{B}}_1^n}}=\rho_{\lrbracket{B\Bar{B}}_1^{n-1}B_n},\,\\
        \end{array}
    \end{split}
    \end{align}
    then there exist a probability distribution $\lrbrace{p(x)}_{x\in\mathcal{X}}$ and states $\lrbrace{\sigma^x_A}_{x\in\mathcal{X}},\, \lrbrace{\sigma^x_{B\Bar{B}}}_{x\in\mathcal{X}}$ such that for any $i\in\lrrec{n}$
     \begin{align}\label{eqn:first_puri_bound}
        \left\lVert \rho_{A\lrbracket{B\Bar{B}}_i} - \sum_{x\in\mathcal{X}}p(x) \sigma_{A}^x\otimes\sigma^x_{B\Bar{B}} \right\rVert_1\leq \min\lrbrace{ f(A,  B\Bar{B}), f(B\Bar{B}\vert \cdot) }\sqrt{2\ln(2)}\sqrt{\frac{\log(\lrvert{A})}{n}},
    \end{align}
    where $f(A, B\Bar{B}), f(B\Bar{B}\vert \cdot)$ are minimal measurement distortions and for each $ x\in\mathcal{X}$
    \begin{align}
            \begin{split}
                \begin{array}{cc}
                    \Theta_{A_L\rightarrow C_{A_L}}\lrbracket{\sigma^x_{A}} = W_{C_{A_L}}\otimes\sigma^x_{A_R},\, &
            \Omega_{A\rightarrow A}\lrbracket{\sigma^x_{A}}=\sigma^x_{A},\,\\
            \Upsilon_{B_L\rightarrow C_{B_L}}\circ\Tr_{\Bar{B}}\lrbracket{\sigma^x_{B\Bar{B}}} = K_{C_{B_L}}\otimes\sigma^x_{B_R},\, & \Xi_{B\rightarrow B}\circ \Tr_{\Bar{B}}\lrbracket{\sigma^x_{B\Bar{B}}}=\sigma^x_{B}.\,\\
                \end{array}
            \end{split}
        \end{align}
Moreover, there exist a probability distribution $\lrbrace{\tilde{p}(x)}_{x\in\mathcal{X}}$ and states $\lrbrace{\omega^x_A}_{x\in\mathcal{X}},\, \lrbrace{\omega^x_{B}}_{x\in\mathcal{X}}$ such that for any $i\in\lrrec{n}$
 \begin{align}\label{eqn:bose-symmetric_de_Finetti_marginal_bound}
        \left\lVert \rho_{AB_i} - \sum_{x\in\mathcal{X}}\tilde{p}(x) \omega_{A}^x\otimes\omega^x_{B} \right\rVert_1\leq \min\lrbrace{ f(A,  B), f(B\vert \cdot) }\sqrt{2\ln(2)}\sqrt{\frac{\log(\lrvert{A})}{n}},
    \end{align}
    and for each $ x\in\mathcal{X}$
      \begin{align}\label{eqn:constraints_de_finetti_candidate_marginal_Bose_symmetric}
            \begin{split}
                \begin{array}{cc}
                    \Theta_{A_L\rightarrow C_{A_L}}\lrbracket{\omega^x_{A}} = W_{C_{A_L}}\otimes\omega^x_{A_R},\, &
            \Omega_{A\rightarrow A}\lrbracket{\omega^x_{A}}=\omega^x_{A},\,\\
            \Upsilon_{B_L\rightarrow C_{B_L}}\lrbracket{\omega^x_{B}} = K_{C_{B_L}}\otimes\omega^x_{B_R},\, & \Xi_{B\rightarrow B}\lrbracket{\omega^x_{B}}=\omega^x_{B}.\,\\
                \end{array}
            \end{split}
        \end{align}
\end{lemma}
\begin{proof}
The proof follows the general strategy of \autoref{lem:approximate_quantum_de_finetti}, supplemented by the additional arguments from \autoref{sec:Bose_sym_easy}; that lemma, in turn, extends \cite[Theorem 2.3]{berta2021semidefinite}. Fix an informationally complete measurement $\mathcal{M}_{\lrbracket{B\Bar{B}}}$ on $\lrbracket{B\bar{B}}$ with finite outcome alphabet $Z$. For any $k\in\lrbrace{0, 1, \ldots, n-1}$, let
\begin{align}
M_{\lrbracket{B\bar{B}}_1^k}^{z_1^k}=M_{\lrbracket{B\bar{B}}_1}^{z_1}\otimes\cdots\otimes M_{\lrbracket{B\bar{B}}_k}^{z_k}, \quad z_1^k\in Z^k,
\end{align}
be the corresponding POVM element of measurement $\mathcal{M}_{\lrbracket{B\bar{B}}_1^k}$ on $\lrbracket{B\bar{B}}_1^k$. As in the proof of the convergent rounding scheme in \autoref{lem:inner_sequence_general}, no measurement is performed when $k=0$. For every $z_1^k\in Z^k$, define the non-normalized post-measurement state on $A\lrbracket{B\bar{B}}_{k+1}$ by 
\begin{align}
\widetilde\rho_{A\lrbracket{B\bar{B}}_{k+1}\vert z_1^k}:=
\Trr{\lrbracket{B\bar{B}}_1^k \lrbracket{B\bar{B}}_{k+2}^n}{\lrbracket{\mathbb{1}_A\otimes M_{\lrbracket{B\bar{B}}_1^k}^{z_1^k}\otimes \mathbb{1}_{\lrbracket{B\bar{B}}_{k+1}^n}} \rho_{A\lrbracket{B\bar{B}}_1^n}}.
\end{align}
Let $p\lrbracket{z_1^k}:=\Trr{}{\widetilde\rho_{A\lrbracket{B\bar{B}}_{k+1}\vert z_1^k}}$. For outcomes $p\lrbracket{z_1^k}>0$, set 
\begin{align}
    \rho_{A\lrbracket{B\bar{B}}_{k+1}\vert z_1^k} := \frac{\widetilde\rho_{A\lrbracket{B\bar{B}}_{k+1}\vert z_1^k}}{p\lrbracket{z_1^k}},\quad \rho_{A\vert z_1^k} :=\Trr{\lrbracket{B\bar{B}}_{k+1}}{ \rho_{A\lrbracket{B\bar{B}}_{k+1}\vert z_1^k}},\quad \rho_{\lrbracket{B\bar{B}}_{k+1}\vert z_1^k} := \Trr{A}{\rho_{A\lrbracket{B\bar{B}}_{k+1}\vert z_1^k}}.
\end{align}
Outcomes with $p\lrbracket{z_1^k}=0$ may be ignored. Since $\Bose$ is Bose-symmetric w.r.t.\ $A$, by \autoref{prop:Bose-symmetry_implies_symmetry} it is also symmetric w.r.t.\ $A$. Thus, the de Finetti representation theorem in \autoref{lem:approximate_quantum_de_finetti} certifies the existence of a $k\in\lrbrace{0, 1, \ldots, n-1}$ such that
\begin{align}\label{eqn:Bose-de_finetti_candidate_bound}
\begin{split}
    \norm{\rho_{A(B\Bar{B})_{k+1}} - \sum_{z_1^k}p\lrbracket{z_1^k}\rho_{A\vert z_1^k}\otimes \rho_{\lrbracket{B\Bar{B}}_{k+1}\vert z_1^k}}_1\leq \min\lrbrace{ f(A, B\Bar{B}), f(B\Bar{B}\vert \cdot) }\sqrt{2\ln(2)}\sqrt{\frac{\log(\lrvert{A})}{n}}.\, \\
\end{split}
\end{align}
The symmetry of $\Bose$ implies that the bound holds for any marginal $\rho_{A(B\Bar{B})_{i}}$ with $i\in\lrrec{n}$. We now show that, for each $z_1^k\in Z^k$, the candidate product states satisfy the linear constraints. Due to the commutativity of Alice's maps with the measurement  $\mathcal{M}_{\lrbracket{B\bar{B}}_1^k}$ and the assumptions of the lemma in \autoref{eqn:SDP_constraints_appendix}, we obtain
\begin{align}
\begin{split}
     \Theta_{A_L\rightarrow C_{A_L}}\lrbracket{\rho_{A\vert z_1^k}}&=\frac{1}{p\lrbracket{z_1^k}}\, \Theta_{A_L\rightarrow C_{A_L}}\lrbracket{\Trr{\lrbracket{B\Bar{B}}_{1}^n}{\lrbracket{\mathbb{1}_{A}\otimes M_{\lrbracket{B\Bar{B}}_1^k}^{z_1^k}\otimes\mathbb{1}_{\lrbracket{B\Bar{B}}_{k+1}^n}}\rho_{A\lrbracket{B\Bar{B}}_1^n}}}\\
     &=\frac{1}{p\lrbracket{z_1^k}}\, \Trr{\lrbracket{B\Bar{B}}_{1}^n}{\lrbracket{\mathbb{1}_{C_{A_L}A_R}\otimes M_{\lrbracket{B\Bar{B}}_1^k}^{z_1^k}\otimes\mathbb{1}_{\lrbracket{B\Bar{B}}_{k+1}^n}}\Theta_{A_L\rightarrow C_{A_L}}\lrbracket{\rho_{A\lrbracket{B\Bar{B}}_1^n}}}\\
     &\stackrel{\autoref{eqn:SDP_constraints_appendix}}{=} \frac{1}{p\lrbracket{z_1^k}}\, \Trr{\lrbracket{B\Bar{B}}_{1}^n}{\lrbracket{\mathbb{1}_{C_{A_L}A_R}\otimes M_{\lrbracket{B\Bar{B}}_1^k}^{z_1^k}\otimes\mathbb{1}_{\lrbracket{B\Bar{B}}_{k+1}^n}}W_{C_{A_L}}\otimes\rho_{A_R\lrbracket{B\Bar{B}}_1^n}}\\
    &=W_{C_{A_L}}\otimes\rho_{A_R\vert z_1^k}
\end{split}
\end{align}
and
\begin{align}
\begin{split}
     \Omega_{A\rightarrow A}\lrbracket{\rho_{A\vert z_1^k}}&=\frac{1}{p\lrbracket{z_1^k}}\, \Omega_{A\rightarrow A}\lrbracket{\Trr{\lrbracket{B\Bar{B}}_{1}^n}{\lrbracket{\mathbb{1}_{A}\otimes M_{\lrbracket{B\Bar{B}}_1^k}^{z_1^k}\otimes\mathbb{1}_{\lrbracket{B\Bar{B}}_{k+1}^n}}\rho_{A\lrbracket{B\Bar{B}}_1^n}}}\\
     &=\frac{1}{p\lrbracket{z_1^k}}\, \Trr{\lrbracket{B\Bar{B}}_{1}^n}{\lrbracket{\mathbb{1}_{A}\otimes M_{\lrbracket{B\Bar{B}}_1^k}^{z_1^k}\otimes\mathbb{1}_{\lrbracket{B\Bar{B}}_{k+1}^n}}\Omega_{A\rightarrow A}\lrbracket{\rho_{A\lrbracket{B\Bar{B}}_1^n}}}\\
     &\stackrel{\autoref{eqn:SDP_constraints_appendix}}{=} \frac{1}{p\lrbracket{z_1^k}}\, \Trr{\lrbracket{B\Bar{B}}_{1}^n}{\lrbracket{\mathbb{1}_{A}\otimes M_{\lrbracket{B\Bar{B}}_1^k}^{z_1^k}\otimes\mathbb{1}_{\lrbracket{B\Bar{B}}_{k+1}^n}}\rho_{A\lrbracket{B\Bar{B}}_1^n}}=\rho_{A\vert z_1^k}\,.
\end{split}
\end{align}
Owing to the presence of a copy system on Bob’s side, establishing this implication for his constraints requires additional arguments beyond those used in \autoref{lem:approximate_quantum_de_finetti}. Concretely, due to the symmetry of $\rho_{A\lrbracket{B\Bar{B}}_1^n}$, \autoref{prop:symmetric_states_have_symmetric_reduced_states} certifies the symmetry of the marginals $\rho_{AB_1^n}$ and $\rho_{A\Bar{B}_1^n}$. Thus, $\Trr{\lrbracket{\Bar{B}}_{k+1}}{\cdot}$ acts as $\Trr{\lrbracket{\Bar{B}}_{n}}{\cdot}$ for any $k\in\lrbrace{0,1,\ldots, n-1}$. Consider $\Upsilon_{\lrbracket{B_L}_{k+1}\rightarrow C_{B_L}}(\cdot)$. Again, because $\rho_{B_1^n}$ is symmetric, \autoref{prop:symmetric_states_have_symmetric_reduced_states} implies that the reduced state $\rho_{\lrbracket{B_L}_1^n}$ is symmetric. We conclude that $ \Upsilon_{\lrbracket{B_L}_{k+1}\rightarrow C_{B_L}}$ acts in the same way as $ \Upsilon_{\lrbracket{B_L}_{n}\rightarrow C_{B_L}}$. Thus, the commutativity of this map with $\mathcal{M}_{\lrbracket{B\Bar{B}}_1^k}$ yields
\begin{align}
\begin{split}
    \Upsilon_{\lrbracket{B_L}_{k+1}\rightarrow C_{B_L}}&\circ\Trr{\lrbracket{\Bar{B}}_{k+1}}{\rho_{\lrbracket{B\Bar{B}}_{k+1}\vert z_1^k}}\\
    &=\frac{1}{p\lrbracket{z_1^k}}\,\Trr{\lrbracket{B\Bar{B}}_1^k\lrbracket{B\Bar{B}}_{k+2}^n}{\lrbracket{M_{\lrbracket{B\Bar{B}}_1^k}^{z_1^k}\otimes\mathbb{1}_{C_{B_L}\lrbracket{B_R}_{k+1}\lrbracket{B\Bar{B}}_{k+2}^n}}\Upsilon_{\lrbracket{B_L}_{k+1}\rightarrow C_{B_L}}\circ\Trr{\lrbracket{\Bar{B}}_{k+1}}{\BoseB}}\\
    &\stackrel{\autoref{eqn:SDP_constraints_appendix}}{=}\frac{1}{p\lrbracket{z_1^k}}\,\Trr{\lrbracket{B\Bar{B}}_1^k\lrbracket{B\Bar{B}}_{k+2}^n}{\lrbracket{M_{\lrbracket{B\Bar{B}}_1^k}^{z_1^k}\otimes\mathbb{1}_{C_{B_L}\lrbracket{B_R}_{k+1}\lrbracket{B\Bar{B}}_{k+2}^n}}K_{C_{B_L}}\otimes\rho_{\lrbracket{B\Bar{B}}_1^k\lrbracket{B_R}_{k+1}\lrbracket{B\Bar{B}}_{k+2}^n}}\\
    &=K_{C_{B_L}}\otimes\rho_{\lrbracket{B_R}_{k+1}\vert z_1^k}\,.
\end{split}
\end{align}
Lastly, from \autoref{prop:symmetric_states_have_symmetric_reduced_states}, $\Xi_{B_{k+1}\rightarrow B_{k+1}}\circ \Tr_{\Bar{B}_{k+1}}\lrbracket{\cdot}$ acts as $\Xi_{B_n\rightarrow B_n}\circ \Tr_{\Bar{B}_n}\lrbracket{\cdot}$ for any $k\in\lrbrace{0,1,\ldots,n-1}$. Thus, due to the commutativity with $\mathcal{M}_{\lrbracket{B\Bar{B}}_1^k}$, we obtain
\begin{align}
    \begin{split}
        \Xi_{B_{k+1}\rightarrow B_{k+1}}&\circ \Tr_{\Bar{B}_{k+1}}\lrbracket{\rho_{\lrbracket{B\Bar{B}}_{k+1}\vert z_1^k}}\\
        &=\frac{1}{p\lrbracket{z_1^k}}\,\Trr{\lrbracket{B\Bar{B}}_1^k\lrbracket{B\Bar{B}}_{k+2}^n}{\lrbracket{M_{\lrbracket{B\Bar{B}}_1^k}^{z_1^k}\otimes\mathbb{1}_{B_{k+1}\lrbracket{B\Bar{B}}_{k+2}^n}}\Xi_{B_{k+1}\rightarrow B_{k+1}}\circ \Trr{\Bar{B}_{k+1}}{\BoseB}}\\
         &\stackrel{\autoref{eqn:SDP_constraints_appendix}}{=}\frac{1}{p\lrbracket{z_1^k}}\,\Trr{\lrbracket{B\Bar{B}}_1^k\lrbracket{B\Bar{B}}_{k+2}^n}{\lrbracket{M_{\lrbracket{B\Bar{B}}_1^k}^{z_1^k}\otimes\mathbb{1}_{B_{k+1}\lrbracket{B\Bar{B}}_{k+2}^n}}\rho_{\lrbracket{B\Bar{B}}_1^k B_{k+1}\lrbracket{B\Bar{B}}_{k+2}^n}}=\rho_{B_{k+1}\vert z_1^k}\,.
    \end{split}
\end{align}
Since the foregoing arguments are independent of the particular choice of $z_1^k$, they hold for all $z_1^k\in Z^k$. Hence every product state appearing in the convex mixture defining our de Finetti candidate state satisfies the constraints. Lastly, since $\Bose$ is Bose-symmetric w.r.t.\ $A$, the marginal $\rho_{AB_1^n}$ is symmetric w.r.t.\ $A$. Moreover, since $\Bose$ satisfies the constraints in \autoref{eqn:SDP_constraints_appendix},  $\rho_{AB_1^n}$ satisfies
\begin{align}
        \begin{split}
        \begin{array}{cc}
            \Theta_{A_L\rightarrow C_{A_L}}\lrbracket{\rho_{AB_1^n}} = W_{C_{A_L}}\otimes\rho_{A_RB_1^n},\, &
            \Omega_{A\rightarrow A}\lrbracket{\rho_{AB_1^n}}=\rho_{AB_1^n},\,\\
            \Upsilon_{\lrbracket{B_L}_n\rightarrow C_{B_L}}\lrbracket{\rho_{B_1^n}} = K_{C_{B_L}}\otimes \rho_{B_1^{n-1}\lrbracket{B_R}_n},\, & \Xi_{B_n\rightarrow B_n}\lrbracket{\rho_{B_1^n}}=\rho_{B_1^{n-1}B_n}.\,\\
        \end{array}
    \end{split}
    \end{align}
Fix an informationally complete measurement $\mathcal{N}_B$ on $B$ with finite outcome alphabet $X$. For $m\in\lrbrace{0, 1, \ldots, n-1}$, let
\begin{align}
N_{B_1^m}^{x_1^m}=N_{B_1}^{x_1}\otimes\cdots\otimes N_{B_m}^{x_m}, \quad x_1^m\in X^m,
\end{align}
be the corresponding POVM element on $B_1^m$. For every $x_1^m\in X^m$, define the non-normalized post-measurement state on $A B_{m+1}$ by 
\begin{align}
\widetilde\omega_{AB_{m+1}\vert x_1^m}:=
\Trr{B_1^m B_{m+2}^n}{\lrbracket{\mathbb{1}_A\otimes N_{B_1^m}^{x_1^m}\otimes \mathbb{1}_{B_{m+1}^n}} \rho_{AB_1^n}}.
\end{align}
Let $p\lrbracket{x_1^m}:=\Trr{}{\widetilde\omega_{AB_{m+1}\vert x_1^m}}$. For outcomes $p\lrbracket{x_1^m}>0$, set 
\begin{align}
    \omega_{AB_{m+1}\vert x_1^m} := \frac{\widetilde\omega_{AB_{m+1}\vert x_1^m}}{p\lrbracket{x_1^m}},\quad \omega_{A\vert x_1^m} :=\Trr{B_{m+1}}{ \omega_{AB_{m+1}\vert x_1^m}},\quad \omega_{B_{m+1}\vert x_1^m} := \Trr{A}{\omega_{AB_{m+1}\vert x_1^m}}.
\end{align}
Outcomes with $p\lrbracket{x_1^m}=0$ may be ignored. By the symmetric de Finetti representation theorem in \autoref{lem:approximate_quantum_de_finetti}, there exists an $m\in\lrbrace{0,1,\ldots, n-1}$ such that
    \begin{align}
         \norm{\rho_{AB_{m+1}} - \sum_{x_1^m}p\lrbracket{x_1^m}\omega_{A\vert x_1^m}\otimes \omega_{B_{m+1}\vert x_1^m}}_1\leq \min\lrbrace{f(A, B), f(B\vert \cdot) }\sqrt{2\ln(2)}\sqrt{\frac{\log(\lrvert{A})}{n}},
    \end{align}
and for each $x_1^m\in X^m$, the product state $\omega_{A\vert x_1^m}\otimes \omega_{B_{m+1}\vert x_1^m}$ satisfies the constraints in \autoref{eqn:constraints_de_finetti_candidate_marginal_Bose_symmetric}, by arguments analogous to those given above. Symmetry then yields the required bound on $\rho_{AB_i}$ for any $i\in\lrrec{n}$.
\end{proof}

Note that tracing out $\bar{B}_{k+1}$ from the de Finetti candidate state used in \autoref{eqn:Bose-de_finetti_candidate_bound} produces a feasible candidate for the marginal problem in \autoref{eqn:constraints_de_finetti_candidate_marginal_Bose_symmetric}. However, this feasible candidate need not coincide with the candidate attaining the bound in \autoref{eqn:bose-symmetric_de_Finetti_marginal_bound}. Moreover, the direct proof of the marginal bound can be sharper. Indeed, by monotonicity of the trace norm under the completely positive and trace-preserving (CPTP) map $\Trr{\bar{B}_{k+1}}{\cdot}$, we have
\begin{align}
\begin{split}
    \norm{\rho_{AB_{k+1}} - \sum_{z_1^k}p\lrbracket{z_1^k}\rho_{A\vert z_1^k}\otimes \rho_{B_{k+1}\vert z_1^k}}_1 &= \norm{\Trr{\bar{B}_{k+1}}{\rho_{A(B\Bar{B})_{k+1}}} - \sum_{z_1^k}p\lrbracket{z_1^k}\rho_{A\vert z_1^k}\otimes \Trr{\bar{B}_{k+1}}{\rho_{\lrbracket{B\Bar{B}}_{k+1}\vert z_1^k}}}_1\\
    &\stackrel{\operatorname{DPI}}{\leq} \norm{\rho_{A(B\Bar{B})_{k+1}} - \sum_{z_1^k}p\lrbracket{z_1^k}\rho_{A\vert z_1^k}\otimes \rho_{\lrbracket{B\Bar{B}}_{k+1}\vert z_1^k}}_1\\
    &\leq \min\lrbrace{ f(A, B\Bar{B}), f(B\Bar{B}\vert \cdot) }\sqrt{2\ln(2)}\sqrt{\frac{\log(\lrvert{A})}{n}}.
\end{split}
\end{align}


\subsection{The hierarchy}\label{sec:the-hierarchy}

In analogy to \cite{berta2021semidefinite}, we can formulate a Bose-symmetric SDP hierarchy that approximates a constrained separability problem $\mathrm{cSEP}(G)$ (see \autoref{sec:cbo} or \autoref{eqn:cSEP_Main}) from above. 

\begin{lemma}[Bose-symmetric hierarchy]\label{lem:bose-symmetric_hierarchy}
Consider finite-dimensional Hilbert spaces $A=A_L\otimes A_R$, $B=B_L\otimes B_R$,
$\Bar{B}=\Bar{B}_L\otimes\Bar{B}_R$, $C_{A_L}$ and $C_{B_L}$, with $B\cong\Bar{B}$. Let
$G_{AB}\in\mathcal{B}\lrbracket{AB}$ be Hermitian with $\norm{G_{AB}}_{\infty}\leq 1$, and let $\mathrm{cSEP}(G)$ be the constrained separability problem defined in \autoref{eqn:cSEP_Main}. For every $n\in\mathbb{N}$, define
  \begin{align}\label{eqn:Bose_symmetric_hierarchy}
        \begin{split}
                \displaystyle &\mathrm{SDP}_n^{\mathrm{Bose}}(G) := \max_{\Bose\in\mathcal{B}\lrbracket{A\lrbracket{B\bar{B}}_1^n}}\Trr{}{G_{AB}\,\rho_{AB}}\\
                 \text{s.t. }\quad  &\Bose\succeq 0,\quad \Tr\left[\Bose\right] = 1,\\
                 &\rho_{AB}:=\Trr{\Bar{B}_1\lrbracket{B\Bar{B}}_2^n}{\Bose},\quad \Bose \text{ is Bose-symmetric w.r.t. } A,
                 \\
                 &\begin{array}{ll}
                \Theta_{A_L\rightarrow C_{A_L}}\lrbracket{\rho_{A\lrbracket{B\Bar{B}}_1^n}} = W_{C_{A_L}}\otimes\rho_{A_R\lrbracket{B\Bar{B}}_1^n},\, &  \Omega_{A\rightarrow A}\lrbracket{\rho_{A\lrbracket{B\Bar{B}}_1^n}}=\rho_{A\lrbracket{B\Bar{B}}_1^n},\,\\
                 \Upsilon_{\lrbracket{B_L}_n\rightarrow C_{B_L}}\circ \Tr_{\Bar{B}_n}\lrbracket{\rho_{\lrbracket{B\Bar{B}}_1^n}} = K_{C_{B_L}}\otimes \rho_{\lrbracket{B\Bar{B}}_1^{n-1}\lrbracket{B_R}_n},\, &\Xi_{B_n\rightarrow B_n}\circ \Tr_{\Bar{B}_n}\lrbracket{\rho_{\lrbracket{B\Bar{B}}_1^n}}=\rho_{\lrbracket{B\Bar{B}}_1^{n-1}B_n},\,\\
            \end{array}
        \end{split}
    \end{align}    
Then
\begin{align}\label{eqn:Bose_symmetric_hierarchy_error}
    0
    \leq
    \mathrm{SDP}_n^{\mathrm{Bose}}(G)-\mathrm{cSEP}(G)
    \leq
    \min\lrbrace{f(A,B),f(B\vert\cdot)}
    \sqrt{2\ln(2)}
    \sqrt{\frac{\log\lrvert{A}}{n}}.
\end{align}
\end{lemma}

\begin{proof}
We first prove that the Bose-symmetric relaxation is an outer approximation. Let
\begin{align}
    \sigma_{AB}=\sum_{x\in\mathcal{X}}p(x)\, \sigma_A^x\otimes\sigma_B^x
\end{align}
be feasible for $\mathrm{cSEP}(G)$. Thus, for every $x\in\mathcal{X}$,
\begin{align}\label{eqn:cSEP_constraints_product_states_Bose_proof}
\begin{array}{cc}
    \Theta_{A_L\rightarrow C_{A_L}}\lrbracket{\sigma_A^x}=W_{C_{A_L}}\otimes\sigma_{A_R}^x, & \Omega_{A\rightarrow A}\lrbracket{\sigma_A^x}=\sigma_A^x,
    \\[0.4em]
    \Upsilon_{B_L\rightarrow C_{B_L}}\lrbracket{\sigma_B^x}=K_{C_{B_L}}\otimes\sigma_{B_R}^x, &
    \Xi_{B\rightarrow B}\lrbracket{\sigma_B^x}=\sigma_B^x.
\end{array}
\end{align}
Let
\begin{align}
    \ket{\Psi}_{B\Bar{B}}:=\sum_{i=1}^{\lrvert{B}}\ket{i}_B\otimes\ket{i}_{\Bar{B}}
\end{align}
be the non-normalized maximally entangled pure state across the $B:\Bar{B}$ bipartition. For each $x\in\mathcal{X}$, define the pretty-good purification
\begin{align}
    \ket{\psi^x}_{B\Bar{B}}:=\lrbracket{\sqrt{\sigma_B^x}\otimes\mathbb{1}_{\Bar{B}}}
\ket{\Psi}_{B\Bar{B}}.
\end{align}
Then
\begin{align}
    \Trr{\Bar{B}}{\ket{\psi^x}\bra{\psi^x}_{B\Bar{B}}}=\sigma_B^x.
\end{align}
Now set
\begin{align}\label{eqn:Bose_extension_from_cSEP_state}
    \Bose:=\sum_{x\in\mathcal{X}}p(x)\,\sigma_A^x\otimes
    \lrbracket{\ket{\psi^x}\bra{\psi^x}_{B\Bar{B}}}^{\otimes n}.
\end{align}
This is a normalized positive semidefinite operator. Moreover, $\ket{\psi^x}^{\otimes n}\in\lor^n\lrbracket{B\Bar{B}}$ for every $x\in\mathcal{X}$, and hence $\Bose$ is Bose-symmetric with respect to $A$. Concretely, with
\begin{align} 
\sqrt{\left(\sigma_B^x\right)^{\otimes n}} = \sqrt{\sigma_B^x}^{\otimes n}, \forall x\in \mathcal{X}
\end{align} 
we have
\begin{align}
\begin{split}
    &\lrbracket{\mathbb{1}_A\otimes U_{B_1^n}(\pi)\otimes U_{\Bar{B}_1^n}(\pi)} \Bose \\
    &= \sum_{x\in \mathcal{X}}p(x)\sigma_A^x\otimes \left[\lrbracket{U_{B_1^n}(\pi)\otimes U_{\Bar{B}_1^n}(\pi)} \lrbracket{\sqrt{\sigma_B^x}^{\otimes n}\otimes\mathbb{1}_{\Bar{B}_1^n}} \lrbracket{\ket{\Psi}\bra{\Psi}_{B\Bar{B}}}^{\otimes n}\lrbracket{\sqrt{\sigma_B^x}^{\otimes n}\otimes\mathbb{1}_{\Bar{B}_1^n}}^\dagger\right]\\
    &=\sum_{x\in \mathcal{X}}p(x)\sigma_A^x\otimes\lrrec{\lrbracket{U_{B_1^n}(\pi)\sqrt{\sigma_B^x}^{\otimes n} \otimes U_{\Bar{B}_1^n}(\pi)}\lrbracket{\ket{\Psi}\bra{\Psi}_{B\Bar{B}}}^{\otimes n}\lrbracket{\sqrt{\sigma_B^x}^{\otimes n}\otimes\mathbb{1}_{\Bar{B}_1^n}}^\dagger}\\
    &= \sum_{x\in\mathcal{X}}p(x)\sigma_A^x\otimes\Biggl[\lrbracket{U_{B_1^n}(\pi)\sqrt{\sigma_B^x}^{\otimes n} \otimes U_{\Bar{B}_1^n}(\pi)} \\
     &\qquad\lrbracket{U^T_{B_1^n}(\pi)U_{B_1^n}(\pi)\otimes U^T_{\Bar{B}_1^n}(\pi)U_{\Bar{B}_1^n}(\pi)}\lrbracket{\ket{\Psi}\bra{\Psi}_{B\Bar{B}}}^{\otimes n}\lrbracket{\sqrt{\sigma_B^x}^{\otimes n}\otimes\mathbb{1}_{\Bar{B}_1^n}}^\dagger\Biggl]\\
    &=\sum_{x\in\mathcal{X}}p(x)\sigma_A^x\otimes\Biggl[\lrbracket{U_{B_1^n}(\pi)\sqrt{\sigma_B^x}^{\otimes n}U^T_{B_1^n}(\pi) \otimes U_{\Bar{B}_1^n}(\pi)U^T_{\Bar{B}_1^n}(\pi)}\\
    &\qquad\lrbracket{U_{B_1^n}(\pi)\otimes U_{\Bar{B}_1^n}(\pi)}\lrbracket{\ket{\Psi}\bra{\Psi}_{B\Bar{B}}}^{\otimes n}\lrbracket{\sqrt{\sigma_B^x}^{\otimes n}\otimes\mathbb{1}_{\Bar{B}_1^n}}^\dagger\Biggl]\\
    &=\sum_{x\in\mathcal{X}}p(x)\sigma_A^x\otimes\Biggl[\lrbracket{\underbrace{U_{B_1^n}(\pi)\sqrt{\sigma_B^x}^{\otimes n}U^T_{B_1^n}(\pi)}_{= \sqrt{\sigma_B^x}^{\otimes n}} \otimes \mathbb{1}_{\Bar{B}_1^n}}\\
    &\qquad \underbrace{\lrbracket{U_{B_1^n}(\pi)\otimes U_{\Bar{B}_1^n}(\pi)}\lrbracket{\ket{\Psi}\bra{\Psi}_{B\Bar{B}}}^{\otimes n}}_{= \lrbracket{\ket{\Psi}\bra{\Psi}_{B\Bar{B}}}^{\otimes n}}\lrbracket{\sqrt{\sigma_B^x}^{\otimes n}\otimes\mathbb{1}_{\Bar{B}_1^n}}^\dagger\Biggl]=\Bose.
\end{split}
\end{align}

The marginal appearing in the objective satisfies
\begin{align}
    \rho_{AB_1}=\Trr{\Bar{B}_1\lrbracket{B\Bar{B}}_2^n}{\Bose}=
    \sum_{x\in\mathcal{X}}p(x)\,\sigma_A^x\otimes\sigma_{B_1}^x=\sigma_{AB_1}.
\end{align}
The Alice-side constraints follow directly from
\autoref{eqn:cSEP_constraints_product_states_Bose_proof}:
\begin{align}
\begin{split}
    \Theta_{A_L\rightarrow C_{A_L}}\lrbracket{\Bose}&=\sum_{x\in\mathcal{X}}p(x)\,\Theta_{A_L\rightarrow C_{A_L}}\lrbracket{\sigma_A^x}\otimes\lrbracket{\ket{\psi^x}\bra{\psi^x}_{B\Bar{B}}}^{\otimes n}\\
    &=W_{C_{A_L}}\otimes\sum_{x\in\mathcal{X}}p(x)\,\sigma_{A_R}^x\otimes\lrbracket{\ket{\psi^x}
    \bra{\psi^x}_{B\Bar{B}}}^{\otimes n}\\
    &=W_{C_{A_L}}\otimes\rho_{A_R\lrbracket{B\Bar{B}}_1^n},
\end{split}
\end{align}
and similarly
\begin{align}
    \Omega_{A\rightarrow A}\lrbracket{\Bose}=\Bose.
\end{align}
For Bob's local constraint, using $\Trr{\Bar{B}}{\ket{\psi^x}\bra{\psi^x}}=\sigma_B^x$, we obtain
\begin{align}
\begin{split}
    & \Upsilon_{\lrbracket{B_L}_n\rightarrow C_{B_L}}\circ\Trr{\Bar{B}_n}{\rho_{\lrbracket{B\Bar{B}}_1^n}}\\
    &\qquad =\sum_{x\in\mathcal{X}}p(x)\,\lrbracket{\ket{\psi^x}\bra{\psi^x}_{B\Bar{B}}}^{\otimes(n-1)}\otimes\Upsilon_{\lrbracket{B_L}_n\rightarrow C_{B_L}}\lrbracket{\sigma_{B_n}^x}\\
    &\qquad = K_{C_{B_L}}\otimes\sum_{x\in\mathcal{X}}p(x)\,\lrbracket{\ket{\psi^x}\bra{\psi^x}_{B\Bar{B}}}^{\otimes(n-1)}\otimes\sigma_{\lrbracket{B_R}_n}^x\\
    &\qquad = K_{C_{B_L}}\otimes \rho_{\lrbracket{B\Bar{B}}_1^{n-1}\lrbracket{B_R}_n}.
\end{split}
\end{align}
The Bob-side fixed point constraint is analogous:
\begin{align}
\begin{split}
    \Xi_{B_n\rightarrow B_n}\circ\Tr_{\Bar{B}_n}\lrbracket{\rho_{\lrbracket{B\Bar{B}}_1^n}}
    &=\sum_{x\in\mathcal{X}}p(x)\,\lrbracket{\ket{\psi^x}\bra{\psi^x}_{B\Bar{B}}}^{\otimes(n-1)}\otimes\Xi_{B_n\rightarrow B_n}\lrbracket{\sigma_{B_n}^x}\\
    &=\rho_{\lrbracket{B\Bar{B}}_1^{n-1}B_n}.
\end{split}
\end{align}
Thus every feasible point of $\mathrm{cSEP}(G)$ has a feasible Bose-symmetric extension with the same objective value. Hence
\begin{align}\label{eqn:cSEP_below_Bose_SDP}
    \mathrm{cSEP}(G)\leq \mathrm{SDP}_n^{\mathrm{Bose}}(G).
\end{align}

Conversely, let $\Bose$ be an optimal feasible point of $\mathrm{SDP}_n^{\mathrm{Bose}}(G)$, and let
\begin{align}
    \rho_{AB}=\Trr{\Bar{B}_1\lrbracket{B\Bar{B}}_2^n}{\Bose}
\end{align}
be the marginal appearing in the objective. By \autoref{lem:bose-symmetric_deFinetti}, applied to the feasible state $\Bose$, there exist a probability distribution $\lrbrace{p(x)}_{x\in\mathcal{X}}$ and states $\lrbrace{\omega_A^x}_{x\in\mathcal{X}}$, $\lrbrace{\omega_B^x}_{x\in\mathcal{X}}$ such that
\begin{align}\label{eqn:Bose_hierarchy_rounded_state}
    \tau_{AB}:=\sum_{x\in\mathcal{X}}p(x)\,\omega_A^x\otimes\omega_B^x
\end{align}
is feasible for $\mathrm{cSEP}(G)$ and
\begin{align}\label{eqn:Bose_hierarchy_trace_norm_rounding}
    \norm{\rho_{AB}-\tau_{AB}}_1\leq\min\lrbrace{f(A,B),f(B\vert\cdot)}\sqrt{2\ln(2)}\sqrt{\frac{\log\lrvert{A}}{n}}.
\end{align}
Since $\tau_{AB}$ is feasible for $\mathrm{cSEP}(G)$, we have
\begin{align}
    \Trr{}{G_{AB}\tau_{AB}}\leq\mathrm{cSEP}(G).
\end{align}
Therefore, using Hölder's inequality and $\norm{G_{AB}}_{\infty}\leq 1$,
\begin{align}
\begin{split}
    \mathrm{SDP}_n^{\mathrm{Bose}}(G)&=\Trr{}{G_{AB}\rho_{AB}}\\
    &=\Trr{}{G_{AB}\tau_{AB}}+\Trr{}{G_{AB}\lrbracket{\rho_{AB}-\tau_{AB}}}\\
    &\leq\mathrm{cSEP}(G)+\norm{G_{AB}}_{\infty}\norm{\rho_{AB}-\tau_{AB}}_1\\
    &\leq\mathrm{cSEP}(G) +\min\lrbrace{f(A,B),f(B\vert\cdot)}\sqrt{2\ln(2)}\sqrt{\frac{\log\lrvert{A}}{n}}.
\end{split}
\end{align}
Together with \autoref{eqn:cSEP_below_Bose_SDP}, this proves \autoref{eqn:Bose_symmetric_hierarchy_error}.
\end{proof}

Note that in our cSEP applications, the bound in \autoref{eqn:first_puri_bound} from the Bose-symmetric de Finetti theorem in \autoref{lem:bose-symmetric_deFinetti} yields a worse approximation bound
\begin{align}
\begin{split}
     \norm{\rho_{AB} - \sum_{x\in\mathcal{X}}p(x)\sigma_A^x\otimes \sigma_B^x}_1 &\stackrel{\text{Data processing Ineq.}}{\leq} \norm{\rho_{A(B\Bar{B})} - \sum_{x\in\mathcal{X}}p(x)\sigma_A^x\otimes \sigma_{B\Bar{B}}^x}_1\\
    &\stackrel{\autoref{lem:bose-symmetric_deFinetti}}{\leq} \min\lrbrace{ f(A, B\bar{B}), f(B\bar{B}\vert \cdot) }\sqrt{2\ln(2)}\sqrt{\frac{ \log(\lrvert{A})}{n}},
\end{split}
\end{align}
where $\sigma_B^x:=\Trr{\Bar{B}}{\sigma^x_{B\Bar{B}}}$ and $\lrbrace{p(x)}_{x\in\mathcal{X}}$, $\lrbrace{\sigma^x_A}_{x\in\mathcal{X}}$, $\lrbrace{\sigma^x_{B\Bar{B}}}_{x\in\mathcal{X}}$ denote the ensemble from \autoref{eqn:first_puri_bound}.


\subsection{SDP complexity}\label{sec:sdp-complexity}

Our motivation for introducing the Bose-symmetric hierarchy is that the Bose-symmetry condition can be enforced implicitly by restricting the optimization space in the SDP of \autoref{lem:bose-symmetric_hierarchy} to operators supported on $A \otimes \lor^n\lrbracket{B\Bar{B}}$. Since the operators that are Bose-symmetric w.r.t.\ $A$ are precisely $\End{}{A \otimes \lor^n\lrbracket{B\Bar{B}}}$ (\autoref{def:bose_symmetric} and \autoref{prop:Bose_sym_space_char}), they form a full matrix $*$-algebra whose order, for fixed local dimension, grows only polynomially in $n$; this allows us to express the positive semidefiniteness and normalization constraints in terms of a matrix representation of polynomial size. The structural statement is recorded in \autoref{prop:bose_full_matrix_algebra}, where it opens the algorithmic part of this work: in \autoref{sec:explicit_star_isomorphism} we construct an explicit unital $*$-isomorphism onto this matrix algebra, and in \autoref{sec:symmetry_adapted_reformulation} and \autoref{sec:efficiency_of_Bose_trafo} we show that the resulting reformulation of $\mathrm{SDP}_n^{\mathrm{Bose}}(G)$ can be constructed in time polynomial in $n$ for fixed local dimension, without ever representing operators on $A\lrbracket{B\Bar{B}}_1^n$ in the computational basis.
We now show that an analogous efficient representation is available for the remaining constraints appearing in the SDP of \autoref{lem:bose-symmetric_hierarchy}.

\begin{lemma}\label{lem:bose-symmetric_hierarchy_general_complexity}
Assume the setting of \autoref{lem:bose-symmetric_hierarchy}, and suppose that
\begin{align}
    \lrvert{C_{A_L}}\leq \lrvert{A_L}, \qquad\lrvert{C_{B_L}}\leq \lrvert{B_L}.
\end{align}
Then, for every $\epsilon>0$, the value of $\mathrm{cSEP}(G)$ from \autoref{eqn:cSEP_Main} is approximated from above up to additive error $\epsilon$ by $\mathrm{SDP}_n^{\mathrm{Bose}}(G)$ whenever
\begin{align}
    n\geq \frac{2\ln(2)\,\lrbracket{\min\lrbrace{f(A,B),f(B\vert\cdot)}}^2\log\lrvert{A}}{\epsilon^2}.
\end{align}
Moreover, the positive semidefinite variable is a square matrix of order
\begin{align}
    r_n:=\lrvert{A}\binom{\lrvert{B}\lrvert{\Bar{B}}+n-1}{n}=\lrvert{A}\binom{\lrvert{B}^2+n-1}{n},
\end{align}
and therefore has at most
\begin{align}
    r_n^2=\lrvert{A}^2\binom{\lrvert{B}^2+n-1}{n}^2\leq\lrvert{A}^2 (n+1)^{2\lrvert{B}^2}
\end{align}
real degrees of freedom. In particular, for fixed $\lrvert{A}$ and $\lrvert{B}$,
the number of real degrees of freedom is polynomial in $n$. Thus, $\mathrm{cSEP}(G)$ can be approximated up to an additive error $\epsilon>0$ by solving an SDP with at most $\operatorname{poly}(\epsilon^{-1})$ real degrees of freedom involving square matrices of order at most
\begin{align}
   \max\lrbrace{\lrvert{A}\binom{\lrvert{B}^2+n-1}{n},\ \lrvert{B}\binom{\lrvert{B}^2+n-2}{n-1}}.
\end{align}
\end{lemma}
\begin{proof}
Set $D:=B\Bar{B}$ and $n\geq 1$. Since $B\cong \Bar{B}$, we have
\begin{align}
     \lrvert{D}=\lrvert{B}\lrvert{\Bar{B}}=\lrvert{B}^2.
\end{align}
A Bose-symmetric feasible point satisfies
\begin{align}
      \rho_{AD_1^n}=\lrbracket{\mathbb{1}_A\otimes P_{\lor^n(D)}}\rho_{AD_1^n}\lrbracket{\mathbb{1}_A\otimes P_{\lor^n(D)}}.
\end{align}
Thus, after choosing an orthonormal basis of $A\otimes \lor^n(D)$, the optimization variable is simply a positive semidefinite operator on a Hilbert space of dimension (cf.\ \autoref{prop:Bose_sym_space_char})
\begin{align}
    r_n=\dim\lrbracket{A\otimes \lor^n(D)}=\lrvert{A}\binom{\lrvert{D}+n-1}{n}=\lrvert{A}\binom{\lrvert{B}^2+n-1}{n}.
\end{align}
The real vector space of Hermitian operators on an $r_n$-dimensional complex Hilbert space has real dimension $r_n^2$. Hence the positive semidefinite variable has
\begin{align}
   r_n^2=\lrvert{A}^2\binom{\lrvert{B}^2+n-1}{n}^2  
\end{align}
real degrees of freedom. The coarse bound
\begin{align}
      \binom{\lrvert{B}^2+n-1}{n}\leq (n+1)^{\lrvert{B}^2}
\end{align}
gives
\begin{align}
     r_n^2\leq\lrvert{A}^2(n+1)^{2\lrvert{B}^2}.
\end{align}
It remains to check that, when $\lrvert{A}$ and $\lrvert{B}$ are fixed, the constraints can be expressed in terms of vector spaces of dimension polynomial in $n$. The maps
$\Theta_{A_L\rightarrow C_{A_L}}$ and $\Omega_{A\rightarrow A}$ act only on
Alice's registers, and therefore commute with the projection onto
$\lor^n(D)$ on Bob's extended registers. Consequently,
\begin{align}
    \Theta_{A_L\rightarrow C_{A_L}}
    \lrbracket{\rho_{A D_1^n}},
    \,
    W_{C_{A_L}}\otimes \rho_{A_R D_1^n}
    \in
    \mathcal{B}\lrbracket{C_{A_L}\otimes A_R\otimes \lor^n(D)}
\end{align}
and
\begin{align}
    \Omega_{A\rightarrow A}\lrbracket{\rho_{A D_1^n}},
    \,
    \rho_{A D_1^n}
    \in
    \mathcal{B}\lrbracket{A\otimes \lor^n(D)}.
\end{align}
The assumption $\lrvert{C_{A_L}}\leq \lrvert{A_L}$ ensures that
\begin{align}
     \dim\lrbracket{C_{A_L}\otimes A_R\otimes \lor^n(D)}\leq\dim\lrbracket{A\otimes \lor^n(D)}=r_n.    
\end{align}
For Bob's constraints, we use the vector space inclusion
\begin{align}
    \lor^n(D)\subseteq \lor^{n-1}(D)\otimes D_n,
\end{align}
which follows from $S_{n-1}\subset S_n$. Tracing out $\Bar{B}_n$ and applying maps only to the remaining $B_n$ register
does not affect the first $n-1$ copies of $D$. Hence
\begin{align}
    \Xi_{B_n\rightarrow B_n}\circ \Tr_{\Bar{B}_n}
    \lrbracket{\rho_{D_1^n}},
    \,
    \rho_{D_1^{n-1}B_n}
    \in
    \mathcal{B}\lrbracket{\lor^{n-1}(D)\otimes B_n}
\end{align}
and
\begin{align}
    \Upsilon_{\lrbracket{B_L}_n\rightarrow C_{B_L}}
    \circ \Tr_{\Bar{B}_n}\lrbracket{\rho_{D_1^n}},
    \,
    K_{C_{B_L}}\otimes \rho_{D_1^{n-1}\lrbracket{B_R}_n}
    \in
    \mathcal{B}\lrbracket{
        C_{B_L}\otimes \lor^{n-1}(D)\otimes \lrbracket{B_R}_n
    }.
\end{align}
The assumption $\lrvert{C_{B_L}}\leq \lrvert{B_L}$ implies
\begin{align}
    \dim\lrbracket{C_{B_L}\otimes \lor^{n-1}(D)\otimes \lrbracket{B_R}_n}\leq \dim\lrbracket{\lor^{n-1}(D)\otimes B_n}\leq\lrvert{B}\binom{\lrvert{B}^2+n-2}{n-1},
\end{align}
so these constraint spaces are also of dimension polynomial in $n$ for fixed $\lrvert{B}$. Although these constraints act on $\lor^{n-1}(D)\otimes B_n$ rather than directly on $\lor^n(D)$, this does not
change the complexity estimate, since
\begin{align}
    \dim\lrbracket{\lor^{n-1}(D)\otimes B_n}=\lrvert{B}\binom{\lrvert{B}^2+n-2}{n-1}=\mathcal{O}\lrbracket{n^{\lrvert{B}^2-1}}
\end{align}
for fixed $\lrvert{B}$. Finally, by \autoref{lem:bose-symmetric_hierarchy},
\begin{align}
    0\leq\mathrm{SDP}_n^{\mathrm{Bose}}(G)-\mathrm{cSEP}(G)\leq\min\lrbrace{f(A,B),f(B\vert\cdot)}\sqrt{2\ln(2)}\sqrt{\frac{\log\lrvert{A}}{n}}.
\end{align}
Thus, if
\begin{align}
     n\geq 
    \left\lceil
        \frac{2\ln(2)\,\lrbracket{\min\lrbrace{f(A,B),f(B\vert\cdot)}}^2\log\lrvert{A}}{\epsilon^2}
    \right\rceil,    
\end{align}
then
\begin{align}
    0\leq\mathrm{SDP}_n^{\mathrm{Bose}}(G)-\mathrm{cSEP}(G)\leq \epsilon.
\end{align}
This proves the claim.
\end{proof}


\subsection{Application to quantum non-local games}\label{sec:nonlocal-games}

In the free non-local game setting, the value of the game can be upper bounded via the following Bose-symmetric SDP hierarchy: 
\begin{align}\label{eqn:Bose_SDP_non-local_games}
    \begin{split}
        \begin{array}{ll}
             & \displaystyle \mathrm{SDP}_n^{\mathrm{Bose}}\lrbracket{T, V, \pi}= \lrvert{T}\max_{\rho} \Tr\lrrec{\lrbracket{V_{A_1A_2Q_1Q_2}\otimes S_{T\tilde{T}}}\rho_{\lrbracket{A_1Q_1T}\lrbracket{A_2Q_2\tilde{T}}}} \\
             &\\
             \text{s.t.}&  \rho_{\lrbracket{A_1Q_1T}\lrbracket{A_2Q_2\tilde{T}}} =\Trr{\lrbracket{\overline{A_2Q_2\tilde{T}}}_1\lrbracket{\lrbracket{A_2Q_2\tilde{T}}\lrbracket{\overline{A_2Q_2\tilde{T}}}}_2^n}{\rho_{(A_1Q_1T)\left((A_2Q_2\tilde{T})(\overline{A_2Q_2\tilde{T}})\right)_1^n}}\\
             &\\
             & \rho_{(A_1Q_1T)\left((A_2Q_2\tilde{T})(\overline{A_2Q_2\tilde{T}})\right)_1^n} \text{ is a Bose-symmetric state w.r.t.\ } \lrbracket{A_1Q_1T}\\
             & \Trr{A_1}{\rho_{(A_1Q_1T)\left((A_2Q_2\tilde{T})(\overline{A_2Q_2\tilde{T}})\right)_1^n}} = \sum_{q_1\in Q_1}\pi_1(q_1)\ket{q_1}\bra{q_1}_{Q_1} \otimes \rho_{T\left((A_2Q_2\tilde{T})(\overline{A_2Q_2\tilde{T}})\right)_1^n}\\
             &\\
             & \Trr{\lrbracket{A_2}_n(\overline{A_2Q_2\tilde{T}})_n}{\rho_{\left((A_2Q_2\tilde{T})(\overline{A_2Q_2\tilde{T}})\right)_1^n}} = \rho_{\left((A_2Q_2\tilde{T})(\overline{A_2Q_2\tilde{T}})\right)_1^{n-1}}\otimes \lrbracket{\sum_{q_2\in Q_2}\pi_2(q_2)\ket{q_2}\bra{q_2}_{Q_2}\otimes\frac{\mathbb{1}_{\tilde{T}}}{\lrvert{T}}}.\,\\
        \end{array}
    \end{split}
\end{align}

As in $\mathrm{SDP}{n}\lrbracket{T,V,\pi}$ from \autoref{sec:non_lacal_games_as_cbos}, the states considered in $\mathrm{SDP}{n}^{\mathrm{Bose}}\lrbracket{T,V,\pi}$ are not required to be classical--quantum. In the Bose-symmetric setting, however, the presence of the purifying systems leads to a substantially different situation and yields the following analogue of  \autoref{prop:restriction_to_cq_states}. We abbreviate $B:= A_2Q_2\tilde{T}$ and $\Bar{B}:=\overline{A_2Q_2\tilde{T}}$, as in \autoref{lem:bose-symmetric_hierarchy}, and write $\rho\equiv\rho_{(A_1Q_1T)(B\Bar{B})_1^n}$ for the optimization variable of \autoref{eqn:Bose_SDP_non-local_games}. For phase vectors $\alpha\in\R^{\lrvert{A_1}}$, $\beta\in\R^{\lrvert{Q_1}}$, $\varphi\in\R^{\lrvert{A_2}}$ and $\psi\in\R^{\lrvert{Q_2}}$, define the diagonal unitaries
\begin{align}\label{eqn:bose_dephasing_group}
    \begin{split}
    &u_{A_1}(\alpha):=\sum_{a_1}e^{i\alpha_{a_1}}\ket{a_1}\bra{a_1},\hspace{0.5cm} v_{Q_1}(\beta):=\sum_{q_1}e^{i\beta_{q_1}}\ket{q_1}\bra{q_1},\\
    &u_{A_2}(\varphi):=\sum_{a_2}e^{i\varphi_{a_2}}\ket{a_2}\bra{a_2},\hspace{0.5cm} v_{Q_2}(\psi):=\sum_{q_2}e^{i\psi_{q_2}}\ket{q_2}\bra{q_2},
    \end{split}
\end{align}
together with $g(\varphi,\psi):=u_{A_2}(\varphi)\otimes v_{Q_2}(\psi)\otimes\mathbb{1}_{\tilde{T}}\otimes\mathbb{1}_{\Bar{B}}$ acting on a single block $B\Bar{B}$, and
\begin{align}\label{eqn:bose_dephasing_unitary}
    W(\alpha,\beta,\varphi,\psi):=\lrbracket{u_{A_1}(\alpha)\otimes v_{Q_1}(\beta)\otimes\mathbb{1}_{T}}\otimes g(\varphi,\psi)^{\otimes n}.
\end{align}
The associated dephasing channel is the average
\begin{align}\label{eqn:bose_dephasing_average}
    \mathcal{D}\lrbracket{X}:=\mathbb{E}_{\alpha,\beta,\varphi,\psi}\lrbrace{W(\alpha,\beta,\varphi,\psi)\, X\, W(\alpha,\beta,\varphi,\psi)^{\dagger}},
\end{align}
where all phases are drawn independently and uniformly from $[0,2\pi)$.
 
\begin{proposition}[Bose-compatible classicality]\label{prop:bose_cq_restriction}
Let $n\in\mathbb{N}_{\geq 1}$. If $\rho$ is feasible for $\mathrm{SDP}_n^{\mathrm{Bose}}\lrbracket{T,V,\pi}$ in \autoref{eqn:Bose_SDP_non-local_games}, then $\mathcal{D}\lrbracket{\rho}$ is feasible and attains the same objective value. In particular, the optimum is attained at a state $\rho^{\star}$ with the following properties:
\begin{enumerate}[label=(\roman*)]
    \item $\rho^{\star}$ is classical-quantum with respect to $A_1Q_1$, i.e.\ $\rho^{\star}=\sum_{a_1,q_1}\ket{a_1q_1}\bra{a_1q_1}_{A_1Q_1}\otimes\rho^{\star\,\lrbracket{a_1,q_1}}_{T(B\Bar{B})_1^n}$;
    \item $\rho^{\star}$ commutes with the answer- and question-occupation observables of the extended blocks,
    \begin{align}\label{eqn:occupation_observables}
        N_{a_2}:=\sum_{i=1}^{n}\ket{a_2}\bra{a_2}_{(A_2)_i},\hspace{0.5cm} N_{q_2}:=\sum_{i=1}^{n}\ket{q_2}\bra{q_2}_{(Q_2)_i},\hspace{0.5cm} a_2\in\lrrec{\lrvert{A_2}},\ q_2\in\lrrec{\lrvert{Q_2}},
    \end{align}
    i.e.\ it is block-diagonal with respect to their joint eigenspaces;
    \item the marginal entering the objective is classical-quantum with respect to both $A_1Q_1$ and $\lrbracket{A_2Q_2}_1$,
    \begin{align}\label{eqn:bose_cq_marginal_form}
        \rho^{\star}_{\lrbracket{A_1Q_1T}\lrbracket{A_2Q_2\tilde{T}}}=\sum_{\substack{a_1,\,q_1,\\ a_2,\, q_2}}\ket{a_1q_1}\bra{a_1q_1}_{A_1Q_1}\otimes\ket{a_2q_2}\bra{a_2q_2}_{A_2Q_2}\otimes\rho^{\star\,\lrbracket{a_1,q_1,a_2,q_2}}_{T\tilde{T}}.
    \end{align}
\end{enumerate}
\end{proposition}
\begin{proof}
We first record two properties of the unitaries in \autoref{eqn:bose_dephasing_unitary}. Since $g(\varphi,\psi)^{\otimes n}$ is an $n$-fold tensor power of a single-block unitary, it commutes with every permutation unitary $U_{(B\Bar{B})_1^n}(\pi)$ and hence with $P_{\lor^n\lrbracket{B\Bar{B}}}$; consequently
\begin{align}\label{eqn:W_commutes_with_projector}
    \lrrec{W(\alpha,\beta,\varphi,\psi),\ \mathbb{1}_{A_1Q_1T}\otimes P_{\lor^n\lrbracket{B\Bar{B}}}}=0.
\end{align}
Moreover, $W$ is diagonal in the distinguished bases of all classical registers and acts as the identity on $T$, $\tilde{T}_1^n$ and $\Bar{B}_1^n$.
 
Fix phases $(\alpha,\beta,\varphi,\psi)$ and write $W\equiv W(\alpha,\beta,\varphi,\psi)$, $g\equiv g(\varphi,\psi)$. We claim that $W\rho\, W^{\dagger}$ is feasible whenever $\rho$ is. Positivity and normalization are clear. Bose symmetry with respect to $A_1Q_1T$ follows from \autoref{eqn:W_commutes_with_projector} and \autoref{def:bose_symmetric}. For the constraint on Alice's side in \autoref{eqn:Bose_SDP_non-local_games}, note that $u_{A_1}(\alpha)$ is traced out, so that
\begin{align}
    \Trr{A_1}{W\rho\, W^{\dagger}}=\lrbracket{v_{Q_1}(\beta)\otimes\mathbb{1}_T\otimes g^{\otimes n}}\Trr{A_1}{\rho}\lrbracket{v_{Q_1}(\beta)\otimes\mathbb{1}_T\otimes g^{\otimes n}}^{\dagger};
\end{align}
conjugating both sides of the constraint by the same unitary and using that $v_{Q_1}(\beta)$ commutes with $\sum_{q_1}\pi_1(q_1)\ket{q_1}\bra{q_1}_{Q_1}$, while $\Trr{A_1Q_1}{W\rho\,W^{\dagger}}=g^{\otimes n}\,\rho_{T(B\Bar{B})_1^n}\lrbracket{g^{\otimes n}}^{\dagger}$ is the corresponding marginal of $W\rho\,W^{\dagger}$, shows that the constraint holds for $W\rho\,W^{\dagger}$. For the constraint on the extended side, note first that $\Trr{A_1Q_1T}{W\rho\,W^{\dagger}}=g^{\otimes n}\rho_{(B\Bar{B})_1^n}\lrbracket{g^{\otimes n}}^{\dagger}$, as the Alice factors of $W$ are traced out. We then use that $u_{A_2}(\varphi)$ acts on $\lrbracket{A_2}_n$ alone and is traced out, while $g$ acts trivially on $\Bar{B}_n$; hence
\begin{align}
    \Trr{\lrbracket{A_2}_n\Bar{B}_n}{g^{\otimes n}\rho_{(B\Bar{B})_1^n}\lrbracket{g^{\otimes n}}^{\dagger}}
    =W'\,\Trr{\lrbracket{A_2}_n\Bar{B}_n}{\rho_{(B\Bar{B})_1^n}}\,W'^{\dagger},\hspace{0.5cm} W':=g^{\otimes (n-1)}\otimes\lrbracket{v_{Q_2}(\psi)\otimes\mathbb{1}_{\tilde{T}}}_{n}.
\end{align}
Conjugating the right-hand side of the constraint by $W'$ yields $g^{\otimes(n-1)}\rho_{(B\Bar{B})_1^{n-1}}\lrbracket{g^{\otimes(n-1)}}^{\dagger}\otimes K$ with $K:=\sum_{q_2}\pi_2(q_2)\ket{q_2}\bra{q_2}_{Q_2}\otimes\mathbb{1}_{\tilde{T}}/\lrvert{T}$, since $v_{Q_2}(\psi)$ commutes with $K$; as $g^{\otimes(n-1)}\rho_{(B\Bar{B})_1^{n-1}}\lrbracket{g^{\otimes(n-1)}}^{\dagger}$ is precisely the $\lrbracket{B\Bar{B}}_1^{n-1}$-marginal of $W\rho\,W^{\dagger}$, the constraint holds for $W\rho\,W^{\dagger}$. Lastly, tracing $\Bar{B}_1\lrbracket{B\Bar{B}}_2^n$ shows that the two-party marginal transforms as
\begin{align}\label{eqn:bose_marginal_transform}
    \rho_{\lrbracket{A_1Q_1T}\lrbracket{A_2Q_2\tilde{T}}}\;\longmapsto\;
    \lrbracket{u_{A_1}(\alpha)\otimes v_{Q_1}(\beta)\otimes\mathbb{1}_{T}\otimes u_{A_2}(\varphi)\otimes v_{Q_2}(\psi)\otimes\mathbb{1}_{\tilde{T}}}\,\rho_{\lrbracket{A_1Q_1T}\lrbracket{A_2Q_2\tilde{T}}}\,\lrbracket{\cdot}^{\dagger},
\end{align}
and since $V_{A_1A_2Q_1Q_2}$ is diagonal in the distinguished bases while $S_{T\tilde{T}}$ acts trivially on them, the objective value is unchanged.
 
Averaging over the phases preserves all of the above: the constraints are affine and the positive semidefinite cone is closed under integration, so $\mathcal{D}\lrbracket{\rho}$ is feasible, and by linearity it attains the same objective value. The feasible set is compact and the objective continuous, so the maximum is attained; for any optimizer $\rho_{\mathrm{opt}}$, the state $\rho^{\star}:=\mathcal{D}\lrbracket{\rho_{\mathrm{opt}}}$ is an optimizer in the fixed-point set of $\mathcal{D}$. Since the phases enter \autoref{eqn:bose_dephasing_average} independently, $\rho^{\star}$ commutes with $u_{A_1}(\alpha)\otimes\mathbb{1}$, $v_{Q_1}(\beta)\otimes\mathbb{1}$ and with $\mathbb{1}\otimes e^{i\sum_{a_2}\varphi_{a_2}N_{a_2}}$, $\mathbb{1}\otimes e^{i\sum_{q_2}\psi_{q_2}N_{q_2}}$ for all phase vectors. The first two conditions are equivalent to (i), and the latter two, upon differentiating with respect to the phases, are equivalent to (ii). For (iii), apply \autoref{eqn:bose_marginal_transform} to $\rho^{\star}=\mathcal{D}\lrbracket{\rho^{\star}}$: the two-party marginal of $\rho^{\star}$ equals its own dephasing with respect to the distinguished bases of $A_1$, $Q_1$, $\lrbracket{A_2}_1$ and $\lrbracket{Q_2}_1$, which is the form \autoref{eqn:bose_cq_marginal_form}.
\end{proof}
\begin{remark}\label{rem:bose_cq_restriction}
Full classicality of the extended registers $\lrbracket{A_2Q_2}_i$ is \emph{not} available in the Bose-symmetric setting: dephasing the individual blocks in their distinguished bases does not preserve Bose symmetry. For instance, for $n=2$ and a single classical register per block, the Bose-symmetric state $\frac{1}{2}\lrbracket{\ket{01}+\ket{10}}\lrbracket{\bra{01}+\bra{10}}$ dephases to $\frac{1}{2}\lrbracket{\ket{01}\bra{01}+\ket{10}\bra{10}}$, which has weight $\frac{1}{2}$ outside $\lor^2$. This is no accident of the proof: the purification construction in \autoref{lem:bose-symmetric_hierarchy} stores the classical randomness of Bob's strategy as entanglement between $A_2Q_2$ and the mirror registers $\overline{A_2Q_2}$, so even the feasible points arising from honest strategies carry coherences on $\lrbracket{A_2Q_2}_i$. The Bose-compatible remnant of classicality is precisely the occupation data of \autoref{prop:bose_cq_restriction}(ii): in the type basis of \autoref{prop:type_basis_symmetric_subspace}, condition (ii) states that $\rho^{\star}$ is block-diagonal with respect to the answer- and question-weights induced by the types, an additional block structure inside the full matrix algebra of \autoref{prop:bose_full_matrix_algebra} that can be exploited in the symmetry-adapted reformulation of \autoref{sec:symmetry_adapted_reformulation} to reduce the number of optimization variables. This parallels the situation for $\mathrm{SDP}_n\lrbracket{T,V,\pi}$ in \autoref{prop:restriction_to_cq_states}, where the absence of a support constraint permits the stronger, fully classical-quantum restriction.
\end{remark}

The following result characterizes the computational complexity of the problem and is a restatement of \autoref{lem:Bose_symmetry_first_step}.

\begin{lemma}[Bose-symmetric hierarchy for fixed-size free games]\label{lem:bose_symmetric_fixed_size_games_complexity}
Let $(V,\pi)$ be a two-player free non-local game with $\lrvert{A_1}=\lrvert{A_2}=\lrvert{A}$ answers, $\lrvert{Q_1}=\lrvert{Q_2}=\lrvert{Q}$ questions, and local quantum dimension $\lrvert{T}$. Set
\begin{align}
    d:=\lrvert{A}\lrvert{Q}\lrvert{T}=\lrvert{A_1Q_1T}=\lrvert{A_2Q_2\tilde{T}}.
\end{align}
For every $\epsilon>0$, define
\begin{align}\label{eqn:Bose_game_level_epsilon}
    n_{\epsilon}:=\max\lrbrace{1,\left\lceil\frac{8\ln(2)\,\lrvert{T}^4\log d}{\epsilon^2}\right\rceil}.
\end{align}
Then the Bose-symmetric relaxation in \autoref{eqn:Bose_SDP_non-local_games} satisfies
\begin{align}\label{eqn:Bose_game_error_epsilon}
    0\leq\mathrm{SDP}_{n_{\epsilon}}^{\mathrm{Bose}}\lrbracket{T,V,\pi}-w_{Q(T)}(V,\pi)\leq \epsilon.
\end{align}
Moreover, the positive semidefinite variable can be chosen to act on
\begin{align}
    \lrbracket{A_1Q_1T}\otimes
    \lor^{n_{\epsilon}}
    \lrbracket{
        \lrbracket{A_2Q_2\tilde{T}}
        \lrbracket{\overline{A_2Q_2\tilde{T}}}
    },
\end{align}
and is therefore a Hermitian matrix of order
\begin{align}\label{eqn:Bose_game_variable_order}
    R_{\epsilon}
    :=d\binom{d^2+n_{\epsilon}-1}{n_{\epsilon}}.
\end{align}
Consequently the SDP has at most
\begin{align}\label{eqn:Bose_game_dof_bound}
    R_{\epsilon}^2
    \leq
    d^2\lrbracket{n_{\epsilon}+1}^{2d^2}
    \leq
    d^2
    \lrbracket{
        2+\frac{8\ln(2)\,\lrvert{T}^4\log d}{\epsilon^2}
    }^{2d^2}
\end{align}
real degrees of freedom. In particular, for fixed $\lrvert{A}$, $\lrvert{Q}$, and $\lrvert{T}$, this number is polynomial in $\epsilon^{-1}$; more explicitly, it scales as $\mathcal{O}\lrbracket{\epsilon^{-4d^2}}$ with constants depending only on $\lrvert{A}$, $\lrvert{Q}$, and $\lrvert{T}$.
\end{lemma}

\begin{proof}
The Bose-symmetric hierarchy is an outer approximation: the purification construction in the proof of \autoref{lem:bose-symmetric_hierarchy} embeds every feasible strategy for $w_{Q(T)}(V,\pi)$ into a feasible point of $\mathrm{SDP}_{n}^{\mathrm{Bose}}\lrbracket{T,V,\pi}$ with the same objective value. Hence
\begin{align}\label{eqn:Bose_game_outer_bound}
    w_{Q(T)}(V,\pi)\leq \mathrm{SDP}_{n}^{\mathrm{Bose}}\lrbracket{T,V,\pi}\qquad \forall n\in\mathbb{N}.
\end{align}

Conversely, fix $n\in\mathbb{N}$ and let
\begin{align}
    \rho_{\lrbracket{A_1Q_1T}\lrbracket{\lrbracket{A_2Q_2\tilde{T}}\lrbracket{\overline{A_2Q_2\tilde{T}}}}_1^n}
\end{align}
be feasible for \autoref{eqn:Bose_SDP_non-local_games}. We apply the constrained
Bose-symmetric de Finetti theorem \autoref{lem:bose-symmetric_deFinetti} with
\begin{align}
\begin{array}{ccccccc}
     A_L=A_1Q_1,
     & C_{A_L}=Q_1,
     & A_R=T,
     & B_L=A_2Q_2\tilde{T},
     & C_{B_L}=Q_2\tilde{T},
     & B_R=\CC,
     & \bar{B}=\overline{A_2Q_2\tilde{T}},
\end{array}
\end{align}
and with the constraint maps and operators
\begin{align}
\begin{array}{cc}
    \Theta_{A_L\rightarrow C_{A_L}}(\cdot)=\Trr{A_1}{\cdot}, &\Upsilon_{B_L\rightarrow C_{B_L}}(\cdot)=\Trr{A_2}{\cdot},\\
    W_{C_{A_L}}=\displaystyle\sum_{q_1\in Q_1}\pi_1(q_1)\ket{q_1}\bra{q_1}_{Q_1},&K_{C_{B_L}}=\displaystyle\sum_{q_2\in Q_2}\pi_2(q_2)\ket{q_2}\bra{q_2}_{Q_2}\otimes\frac{\mathbb{1}_{\tilde{T}}}{\lrvert{T}}.
\end{array}
\end{align}
The fixed-point maps are taken to be the identity maps on $A_1Q_1T$ and $A_2Q_2\tilde{T}$. The assumptions $\lrvert{C_{A_L}}\leq \lrvert{A_L}$ and $\lrvert{C_{B_L}}\leq \lrvert{B_L}$ of \autoref{lem:bose-symmetric_hierarchy_general_complexity} are immediate. Moreover, by \autoref{prop:bose_cq_restriction}, w.l.o.g. the marginal $\rho_{\lrbracket{A_1Q_1T}\lrbracket{A_2Q_2\tilde{T}}}$ is a classical-quantum state. As the classical registers do not contribute to the measurement distortion we have
\begin{align}
    \min\lrbrace{f(A_1Q_1T, A_2Q_2\tilde{T}),f(A_2Q_2\tilde{T} \vert\cdot)} \leq 2\lrvert{T}.
\end{align}
Therefore the Bose-symmetric de Finetti theorem \autoref{lem:bose-symmetric_deFinetti} gives a feasible constrained separable state
\begin{align}
    \sigma_{\lrbracket{A_1Q_1T}\lrbracket{A_2Q_2\tilde{T}}}:=\sum_{x\in\mathcal{X}}p(x)\,
    \rho^x_{A_1Q_1T}\otimes\rho^x_{A_2Q_2\tilde{T}}
\end{align}
for the game formulation of \autoref{lem:non_local_games_as_cbo_sym} such that
\begin{align}\label{eqn:Bose_game_trace_distance}
    \left\lVert\rho_{\lrbracket{A_1Q_1T}\lrbracket{A_2Q_2\tilde{T}}}-\sigma_{\lrbracket{A_1Q_1T}\lrbracket{A_2Q_2\tilde{T}}}\right\rVert_1\leq2\lrvert{T}\sqrt{2\ln(2)}\sqrt{\frac{\log d}{n}}.
\end{align}
Since $\left\lVert V_{A_1A_2Q_1Q_2}\otimes S_{T\tilde{T}}\right\rVert_{\infty}\leq 1$, Hölder's inequality yields
\begin{align}\label{eqn:Bose_game_value_error_n}
\begin{split}
    \mathrm{SDP}_{n}^{\mathrm{Bose}}\lrbracket{T,V,\pi} & \leq w_{Q(T)}(V,\pi) +\lrvert{T}\,
    \left\lVert\rho_{\lrbracket{A_1Q_1T}\lrbracket{A_2Q_2\tilde{T}}}-\sigma_{\lrbracket{A_1Q_1T}\lrbracket{A_2Q_2\tilde{T}}}\right\rVert_1\\
    &\leq w_{Q(T)}(V,\pi)+2\lrvert{T}^2\sqrt{2\ln(2)}\sqrt{\frac{\log d}{n}}.
\end{split}
\end{align}
Combining \autoref{eqn:Bose_game_outer_bound} and \autoref{eqn:Bose_game_value_error_n}, and then taking $n=n_{\epsilon}$, proves \autoref{eqn:Bose_game_error_epsilon}. It remains only to count variables. Applying \autoref{lem:bose-symmetric_hierarchy_general_complexity} with $A=A_1Q_1T$ and $B=A_2Q_2\tilde{T}$ gives $\lrvert{A}=\lrvert{B}=d$ and $\lrvert{B\Bar{B}}=d^2$. Hence the Bose-symmetric positive semidefinite variable acts on a space of dimension
\begin{align}
    R_{\epsilon}=d\binom{d^2+n_{\epsilon}-1}{n_{\epsilon}}.
\end{align}
The real vector space of Hermitian operators on an $R_{\epsilon}$-dimensional complex Hilbert space has real dimension $R_{\epsilon}^2$. Finally,
\begin{align}
    \binom{d^2+n_{\epsilon}-1}{n_{\epsilon}}\leq\lrbracket{n_{\epsilon}+1}^{d^2}
\end{align}
and the definition of $n_{\epsilon}$ gives
\begin{align}
    n_{\epsilon}+1\leq2+\frac{8\ln(2)\,\lrvert{T}^4\log d}{\epsilon^2}.
\end{align}
This proves \autoref{eqn:Bose_game_dof_bound} and the stated polynomial scaling for fixed game size.
\end{proof}

\section{Approximating quantum non-local games via symmetric subspace methods}
\label{sec:symmetric_subspace_methods}

Semidefinite programs with inherent symmetry arise in numerous areas of research, motivating substantial work on methods to reduce their complexity \cite{schrijver2005new, laurent2005strengthened, gijswijt2006new, Klerk2007ReductionOS, polak2020new}. For an overview, see \cite[Sec.\ 1]{vallentin2009symmetry}. Related works have also examined the algorithmic complexity of these procedures in specialized cases where the symmetry is governed by the symmetric group, such as those involving the Terwilliger algebra \cite{gijswijt2009block}. More recently, \cite{fawzi2022hierarchy, chee2023efficient} explored applications in approximate quantum error correction. Building on this line of research, we establish connections to concepts from \cite{christandl2007one} using Schur-Weyl duality. 

The goal of this section is to develop an algorithm that expresses $\mathrm{SDP}_n^{\mathrm{Bose}}(G)$ in terms of a symmetry-adapted basis. To clarify, while $\mathrm{SDP}_n^{\mathrm{Bose}}(G)$ is defined as an SDP restricted to operators with support and range in $A\otimes\vee^n\lrbracket{B\bar{B}}$, it is not, a priori, expressed in a basis that yields a representation whose size scales only polynomially with $n$. We proceed in three steps of increasing concreteness. First, in \autoref{sec:bose_full_matrix_algebra}, we record the structural fact that the operators which are Bose-symmetric with respect to side systems form a full matrix $*$-algebra of dimension polynomial in $n$; this certifies the existence of the desired reformulation and uses no representation theory. Second, in \autoref{sec:explicit_star_isomorphism}, representation theory enters --- not to prove existence, but to provide an orthonormal basis of $\vee^n\lrbracket{B\bar{B}}$ whose coefficients over the computational basis are available in closed form, turning the abstract $*$-isomorphism into an explicit map. Third, in \autoref{sec:symmetry_adapted_reformulation} we use this map to reformulate $\mathrm{SDP}_n^{\mathrm{Bose}}(G)$ as a program $\Psi\lrbracket{\mathrm{SDP}_n^{\mathrm{Bose}}(G)}$ in which, by construction, all objects are of size polynomial in $n$, and in \autoref{sec:efficiency_of_Bose_trafo} we show that $\Psi\lrbracket{\mathrm{SDP}_n^{\mathrm{Bose}}(G)}$ can be \emph{constructed} in time polynomial in $n$. The latter step is crucial: first constructing $\mathrm{SDP}_n^{\mathrm{Bose}}(G)$ and subsequently changing basis would require handling exponentially large objects in $n$; instead, $\Psi\lrbracket{\mathrm{SDP}_n^{\mathrm{Bose}}(G)}$ is assembled directly.

The situation in this section is markedly simpler than for the exchange-symmetric algebra $\End{\CSn}{\lrbracket{B\Bar{B}}_1^n}$ treated in \autoref{sec:sdp_symmetry_reduction}. The latter decomposes into many blocks with nontrivial multiplicities, for which we will settle for a weaker, linear cone-preserving reduction in the sense of \cite{polak2020new}. The Bose-symmetric algebra, by contrast, consists of a single block, and the strongest notion of equivalence --- a unital $*$-isomorphism --- is available essentially for free (see \autoref{rem:symmetry_reduction_taxonomy} and \autoref{rem:rep_theory_reading}, as well as \autoref{fig:symmetry_vs_bose_symmetry_block_decomp}).


 \paragraph{\textbf{Notation and preliminaries}}

For $k\in\mathbb{N}$ write $\lrrec{k}:=\{1,\ldots,k\}$. A \emph{weak composition} of $n\in\mathbb{N}$ into $k$ parts is a vector $\alpha=\lrbracket{\alpha_1,\ldots,\alpha_k}\in\mathbb{N}_0^k$ with $\sum_i\alpha_i=n$ and $\alpha_i\geq 0$; we write $\alpha\vDash_k n$ and let $\operatorname{WComp}(k,n):=\{\alpha:\alpha\vDash_k n\}$ denote the set of all of them. A \emph{partition} of $n$ into at most $k$ parts is a weakly decreasing weak composition, i.e.\ $\lambda\vDash_k n$ with $\lambda_1\ge\cdots\ge\lambda_k$; in this case we write $\lambda\vdash_k n$ and set $\operatorname{Par}(k,n):=\{\lambda:\lambda\vdash_k n\}\subseteq\operatorname{WComp}(k,n)$. If $\lambda$ is a partition of $n$ into exactly $k$ nonzero parts, then we say that $\lambda$ has \emph{height} $k$. The \emph{Young diagram} $Y(\lambda)$ of $\lambda\vdash_k n$ is the left-justified array of $n$ boxes with $\lambda_i$ boxes in row $i$. A \emph{semistandard Young tableau} (SSYT) of shape $\lambda$ is a filling of the boxes of $Y(\lambda)$ with entries from $\lrrec{k}$ that is weakly increasing along each row and strictly increasing down each column; we write $\operatorname{SSYT}(\lambda,k)$ for the set of all such tableaux. The \emph{weight} (or \emph{type}) of $T\in\operatorname{SSYT}(\lambda,k)$ is the weak composition $\operatorname{wt}(T)=\lrbracket{m_1,\ldots,m_k}\in\operatorname{WComp}(k,n)$, where $m_i$ is the number of boxes of $T$ containing the entry $i$. Note that $\operatorname{SSYT}(\lambda,k)\neq \emptyset$ if and only if $\lambda$ has height at most $k$. In general, the weight map
\begin{align}
    \operatorname{wt}\colon \operatorname{SSYT}(\lambda,k)\longrightarrow \operatorname{WComp}(k,n)
\end{align}
is neither injective nor surjective. For $\alpha\in \operatorname{WComp}(k,n)$, the fibre of $\operatorname{wt}$ over $\alpha$ is
\begin{align}
    \operatorname{wt}^{-1}(\alpha)=\lrbrace{T\in \operatorname{SSYT}(\lambda,k) : \operatorname{wt}(T)=\alpha}.
\end{align}
Its cardinality is the \emph{Kostka number}
\begin{align}
    \lrvert{\operatorname{wt}^{-1}(\alpha)} = K_{\lambda,\alpha}\in \mathbb{N}_{0},
\end{align}
that is, $K_{\lambda,\alpha}$ is the number of semistandard Young tableaux of shape $\lambda$ and weight $\alpha$. We have
\begin{align}
    \lrvert{\mathrm{Par}(k,n)} \leq \lrvert{\operatorname{WComp}(k,n)} = \binom{n+k-1}{n}\leq (n+1)^{k-1}
\end{align}
which for fixed $k$ is $\mathcal{O}\lrbracket{\operatorname{poly}(n)}$ as $n\rightarrow\infty$. 
The hook-content formula  \cite[Ex.\ 6.4]{fulton2013representation} gives 
\begin{align}
    \lrvert{\operatorname{SSYT}(\lambda,k)} = \sum_{\alpha \vDash_k n} K_{\lambda, \alpha} = \prod_{(i,j)\in \lambda}\frac{k+j-i}{h(i,j)},\quad \text{with } h(i,j):= \lambda_i-j+\lambda'_j-i+1
\end{align}
with conjugate partition $\lambda'$. Equivalently, Weyl's dimension formula \cite[Thm.\ 6.3\ (2)]{fulton2013representation} yields
\begin{align}
    \lrvert{\operatorname{SSYT}(\lambda,k)} = \prod_{1\leq i<j\leq k} \frac{\lambda_i-\lambda_j+j-i}{j-i} \leq (n+1)^{k(k-1)/2}.
\end{align}
Thus, for fixed $k$, $\lrvert{\operatorname{SSYT}(\lambda,k)}$ is $\mathcal{O}\lrbracket{\operatorname{poly}(n)}$ as $n\rightarrow\infty$. The shape relevant for Bose symmetry is the single row $(n)$, i.e.\ the partition of $n$ of height $1$, which inside $\mathbb{N}_0^k$ we write $(n) = (n,0,\ldots,0)$. By Schur--Weyl duality it indexes the symmetric subspace $\vee^n\lrbracket{\HS}$ (see \autoref{sec:bose_via_schur_weyl}). For this shape the weight map specializes to a bijection
\begin{align}
    \operatorname{wt}\,:\,\operatorname{SSYT}((n),k)\;\xrightarrow{\ \sim\ }\;\operatorname{WComp}(k,n),
\end{align}
since a weakly increasing row of length $n$ is determined by its content; equivalently $K_{(n),\alpha}=1$ for every $\alpha\in\operatorname{WComp}(k,n)$. Consequently,
\begin{align}\label{eqn:ssyt_one_row_count}
    \lrvert{\operatorname{SSYT}((n),k)} = \lrvert{\operatorname{WComp}(k,n)} = \binom{n+k-1}{n} \leq (n+1)^{k-1},
\end{align}
which for fixed $k$ is $\mathcal{O}\lrbracket{\operatorname{poly}(n)}$ as $n\to\infty$.

\subsection{The algebra of Bose-symmetric operators is a full matrix $*$-algebra}\label{sec:bose_full_matrix_algebra}

Throughout this subsection, $\HS$ denotes a finite-dimensional complex Hilbert space with $d_{\HS}:=\dim_{\CC}\lrbracket{\HS}$, and $W_1, W_2$ denote finite-dimensional side systems; in the application to $\mathrm{SDP}_n^{\mathrm{Bose}}(G)$ we take $\HS=B\otimes\Bar{B}$, and $\lrbracket{W_1, W_2}$ ranges over the choices listed in \autoref{eqn:D_set} below. Extending \autoref{def:bose_symmetric}, we call $X\in\Op{W_1\otimes\HSn\otimes W_2}$ \emph{Bose-symmetric w.r.t.\ $\lrbracket{W_1,W_2}$} if
\begin{align}
    \Pi\, X\, \Pi = X, \qquad \Pi:=\mathbb{1}_{W_1}\otimes P_{\SymH}\otimes\mathbb{1}_{W_2}.
\end{align}
The following proposition is basis-free and requires no representation theory.

\begin{proposition}[Bose-symmetric operators form a full matrix $*$-algebra]\label{thm:bose-simple}\label{prop:bose_full_matrix_algebra}
Let $n\in\mathbb{N}$, let $m_{(n)}:=\dim_{\CC}\lrbracket{\SymH}=\binom{d_{\HS}+n-1}{n}$, and set $r:=d_{W_1}\, m_{(n)}\, d_{W_2}$ with $d_{W_i}:=\dim_{\CC}\lrbracket{W_i}$. Denote by $\mathcal{B}\subseteq\Op{W_1\otimes\HSn\otimes W_2}$ the set of operators that are Bose-symmetric w.r.t.\ $\lrbracket{W_1, W_2}$. Then:
\begin{enumerate}[label=(\roman*)]
    \item $\mathcal{B}=\End{}{W_1\otimes\SymH\otimes W_2}$ is a simple unital matrix $*$-algebra with unit $\Pi$;
    \item for every isometry $V:\CC^r\rightarrow W_1\otimes\HSn\otimes W_2$ with $VV^\dagger=\Pi$, the map
    \begin{align}
        \psi_V\,:\,\mathcal{B}\longrightarrow\CC^{r\times r}, \qquad X\longmapsto V^\dagger X V,
    \end{align}
    is a unital $*$-isomorphism; in particular, $\psi_V$ is a linear bijection satisfying $\psi_V(XY)=\psi_V(X)\psi_V(Y)$, $\psi_V(X^\dagger)=\psi_V(X)^\dagger$, $\Tr\lrrec{\psi_V(X)}=\Tr\lrrec{X}$, and $X\succeq 0$ if and only if $\psi_V(X)\succeq 0$. Moreover, the eigenvalues of $\psi_V(X)$ are those of the restriction of $X$ to $W_1\otimes\SymH\otimes W_2$;
    \item $r\leq d_{W_1}d_{W_2}\lrbracket{n+1}^{d_{\HS}-1}$; in particular, for fixed $d_{\HS}, d_{W_1}, d_{W_2}$, the algebra $\mathcal{B}$ admits a matrix representation of size polynomial in $n$.
\end{enumerate}
\end{proposition}

\begin{proof}
(i) If $\Pi X\Pi=X$, then the range of $X$ is contained in $\operatorname{ran}\Pi=W_1\otimes\SymH\otimes W_2$ and $X$ annihilates $\lrbracket{\operatorname{ran}\Pi}^{\perp}$; hence $X$ restricts to an endomorphism of $\operatorname{ran}\Pi$. Conversely, extending any endomorphism of $\operatorname{ran}\Pi$ by zero yields an operator satisfying $\Pi X\Pi=X$. This proves the set equality. For $X,Y\in\mathcal{B}$ we have $XY=\lrbracket{\Pi X\Pi}\lrbracket{\Pi Y\Pi}=\Pi\lrbracket{X\Pi Y}\Pi\in\mathcal{B}$ and $X^\dagger=\Pi X^\dagger\Pi\in\mathcal{B}$, and $\Pi X=X\Pi=X$; hence $\mathcal{B}$ is a unital matrix $*$-algebra with unit $\Pi$. Simplicity follows from (ii), since $\CC^{r\times r}$ is simple and simplicity is invariant under algebra isomorphisms.

(ii) Isometries with $VV^\dagger=\Pi$ exist: fix an orthonormal basis $\lrbrace{v_1,\ldots,v_r}$ of $\operatorname{ran}\Pi$ and set $Ve_k:=v_k$ with $\lrbrace{e_k}_{k=1}^r$ the canonical ONB of $\CC^r$; then $V^\dagger V=\mathbb{1}_r$ and $VV^\dagger=\Pi$. Linearity and $\psi_V(X^\dagger)=\psi_V(X)^\dagger$ are immediate. For $X,Y\in\mathcal{B}$,
\begin{align}
    \psi_V(X)\,\psi_V(Y)=V^\dagger X\,\lrbracket{VV^\dagger}\,YV=V^\dagger X\,\Pi\, YV=V^\dagger XY\,V=\psi_V(XY),
\end{align}
using $\Pi Y=Y$. Unitality: $\psi_V(\Pi)=V^\dagger VV^\dagger V=\mathbb{1}_r$. If $\psi_V(X)=0$, then $X=\Pi X\Pi=V\lrbracket{V^\dagger XV}V^\dagger=0$, so $\psi_V$ is injective; for $M\in\CC^{r\times r}$, the operator $X:=VMV^\dagger$ satisfies $\Pi X\Pi=X$ and $\psi_V(X)=M$, so $\psi_V$ is surjective. If $X\succeq 0$, then $\psi_V(X)=V^\dagger XV\succeq 0$; conversely, $X=V\psi_V(X)V^\dagger\succeq 0$ whenever $\psi_V(X)\succeq 0$. The trace identity follows from $\Tr\lrrec{V^\dagger XV}=\Tr\lrrec{XVV^\dagger}=\Tr\lrrec{X\Pi}=\Tr\lrrec{X}$, and the spectral statement from the fact that $V$ implements a unitary equivalence between $\CC^r$ and $\operatorname{ran}\Pi$ intertwining $\psi_V(X)$ with the restriction of $X$ to $\operatorname{ran}\Pi$.

(iii) By \autoref{eqn:ssyt_one_row_count}, $m_{(n)}=\binom{d_{\HS}+n-1}{n}\leq\lrbracket{n+1}^{d_{\HS}-1}$.
\end{proof}

For $\lrbracket{W_1, W_2}=\lrbracket{A,\CC}$ and $\HS=B\otimes\Bar{B}$, \autoref{prop:bose_full_matrix_algebra} specializes to
\begin{align}
    \EndBose{n}\simeq\CC^{\,r_n\times r_n}
\end{align}
with $r_n=\lrvert{A}\binom{\lrvert{B}^2+n-1}{n}$ as in \autoref{lem:bose-symmetric_hierarchy_general_complexity}; this is the structural result invoked in \autoref{sec:sdp-complexity}. The notation $m_{(n)}$ is chosen to match $m_{\lambda}=\dim_{\CC}\lrbracket{\Schurf{\lambda}\HS}$ of \autoref{sec:sdp_symmetry_reduction} at the one-row partition $\lambda=(n)$.

\begin{remark}[Relation to symmetry reduction of semidefinite programs]\label{rem:symmetry_reduction_taxonomy}
    \autoref{prop:bose_full_matrix_algebra} is the degenerate instance of a general phenomenon: every finite-dimensional unital matrix $*$-algebra over $\CC$ is $*$-isomorphic to a direct sum $\bigoplus_i\CC^{m_i\times m_i}$ of full matrix algebras (Artin--Wedderburn; see, e.g., \autoref{sec:representation_theory} or \cite{schrijver2005new, vallentin2009symmetry}), and computing such a block diagonalization is the basis of symmetry reduction of SDPs \cite{schrijver2005new, laurent2005strengthened, gijswijt2006new, Klerk2007ReductionOS, vallentin2009symmetry, polak2020new}. Three notions of reduction, in decreasing order of strength, appear in this literature.
    \begin{enumerate}[label=(\roman*)]
        \item An explicit unital $*$-isomorphism onto the Artin--Wedderburn decomposition preserves products, adjoints, traces and spectra --- hence every constraint expressible in these terms --- but generally requires a symmetry-adapted basis, i.e., knowledge of the irreducible representations involved.
        \item The \emph{regular $*$-representation} of \cite{Klerk2007ReductionOS} is a faithful $*$-representation constructed solely from the multiplication table of the canonical (orbit) basis with respect to the normalized trace inner product; no representation-theoretic data is needed, at the price of representing the algebra on itself, i.e., by matrices of order $r^2$ instead of $r$. In our setting this is still polynomial in $n$, so it would already yield a polynomial-time variant of \autoref{thm:Bose_sym_complexity_results}; moreover, the multiplication table is available in closed form (cf.\  \autoref{sec:efficient_acces_bose_operators}).
        \item A \emph{linear cone-preserving bijection} (an order isomorphism between the respective positive semidefinite cones) \cite{polak2020new} preserves linear structure and the signs of eigenvalues, but neither products nor spectra. This weakest notion suffices whenever the SDP is linear, i.e., whenever only linear functionals of the variable are constrained.
    \end{enumerate}
    Because the algebra of \autoref{prop:bose_full_matrix_algebra} is simple, its Artin--Wedderburn decomposition consists of a single block (cf.\ \autoref{fig:symmetry_vs_bose_symmetry_block_decomp}), and, as shown in \autoref{sec:explicit_star_isomorphism}, a symmetry-adapted orthonormal basis is available in closed form; we therefore obtain notion~(i) essentially for free. Since $\mathrm{SDP}_n^{\mathrm{Bose}}(G)$ is a linear SDP, notion~(iii) would in fact suffice for \autoref{thm:Bose_sym_complexity_results}. Adopting notion~(i) has the additional benefits that the trace normalization and SDP duality transfer verbatim, that the physical marginal $\rho_{AB}$ of the optimizer --- the input to the rounding procedure of \autoref{sec:inner_sequence} --- is recovered directly via $\psi_V^{-1}(M)=VMV^\dagger$, and that potential extensions involving operator products or spectral data (e.g., NPA-type localizing constraints) remain available. In the exchange-symmetric setting of \autoref{sec:sdp_symmetry_reduction}, in contrast, the algebra $\End{\CSn}{\lrbracket{B_1^n}}$ has many blocks with nontrivial multiplicities and no canonical orthogonal representative vectors; there, following \cite{polak2020new}, a linear cone-preserving bijection suffices and we will exploit this fact.
\end{remark}

\subsection{An explicit $*$-isomorphism from the type basis}\label{sec:explicit_star_isomorphism}

By \autoref{prop:bose_full_matrix_algebra}(ii), every orthonormal basis of $\SymH$ induces a unital $*$-isomorphism $\psi_V$. The algorithmic content of this subsection lies in choosing a basis whose coefficients over the computational basis are known in closed form; this is what representation theory provides.

\begin{proposition}[Type basis of the symmetric subspace]\label{prop:type_basis_symmetric_subspace}
    Let $\lrbrace{e_i}_{i=1}^{d_{\HS}}$ be the computational basis of $\HS$. For each type $\vec{t}\in\mathcal{T}_{(n),\,d_{\HS}}$ --- equivalently, for each semistandard labeling $\tau$ of the one-row shape $Y((n))$ with entries in $\lrrec{d_{\HS}}$ and weight $\vec{t}$ --- define the normalized polytabloid
    \begin{align}
        \ket{\tilde{s}_{\vec{t}}} = \tilde{u}_{\tau} := \frac{1}{\sqrt{\binom{n}{\vec{t}}}}\sum_{\vec{i}\,:\,T(\vec{i})=\vec{t}}\ket{i_1,\ldots,i_n}\;\in\;\HSn, \qquad \binom{n}{\vec{t}}:=\frac{n!}{t_1!\cdots t_{d_{\HS}}!},
    \end{align}
    where $T(\vec{i})$ denotes the type of $\vec{i}=(i_1,\ldots,i_n)$. Then $\lrbrace{\ket{\tilde{s}_{\vec{t}}}\,:\,\vec{t}\in\mathcal{T}_{(n),\,d_{\HS}}}$ is an orthonormal basis of $\SymH$. Consequently, the matrix $U_{(n)}\in\R^{d_{\HS}^n\times m_{(n)}}$ with these vectors as columns satisfies
    \begin{align}
        U_{(n)}^\dagger U_{(n)}=\mathbb{1}_{m_{(n)}}, \qquad U_{(n)}U_{(n)}^\dagger=P_{\SymH},
    \end{align}
    and each entry of $U_{(n)}$ equals either $0$ or $\binom{n}{\vec{t}}^{-1/2}$, computable in time polynomial in $n$ for fixed $d_{\HS}$.
\end{proposition}

\begin{proof}
    The vectors $\ket{s_{\vec{t}}}=\sum_{\vec{i}:T(\vec{i})=\vec{t}}\ket{i_1,\ldots,i_n}$ span $\SymH$ \cite{harrow2013church} (cf.\ \autoref{sec:bose_via_schur_weyl} for the derivation via highest-weight theory). Vectors of distinct type are supported on disjoint sets of computational basis vectors and are hence orthogonal, and $\braket{s_{\vec{t}}|s_{\vec{t}}}=\lrvert{\lrbrace{\vec{i}\,:\,T(\vec{i})=\vec{t}}}=\binom{n}{\vec{t}}$, so the normalized family is orthonormal. Its cardinality is $\lrvert{\mathcal{T}_{(n),\,d_{\HS}}}=\lrvert{\operatorname{WComp}(d_{\HS},n)}=\binom{d_{\HS}+n-1}{n}=m_{(n)}$ by \autoref{eqn:ssyt_one_row_count} and \autoref{prop:Bose_sym_space_char}, which matches $\dim_{\CC}\lrbracket{\SymH}$; hence it is a basis. Orthonormality of the columns gives $U_{(n)}^\dagger U_{(n)}=\mathbb{1}_{m_{(n)}}$, and $U_{(n)}U_{(n)}^\dagger=\sum_{\vec{t}}\ket{\tilde{s}_{\vec{t}}}\bra{\tilde{s}_{\vec{t}}}$ is the orthogonal projector onto their span, i.e., $P_{\SymH}$. Finally, for the matrix entries we have
    \begin{align}
        \lrbracket{U_{(n)}}_{\vec{i},\vec{t}}=\binom{n}{\vec{t}}^{-1/2}\mathbf{1}\lrrec{T(\vec{i})=\vec{t}},
    \end{align}
    where $\mathbf{1}[\cdot]$ is the indicator funtion and the multinomial coefficient is an integer with $\mathcal{O}\lrbracket{n\log d_{\HS}}$ bits, computable exactly in time polynomial in $n$ for any entry indexed by the pair of type vectors $\vec{i}, \vec{t}$.
\end{proof}

\begin{corollary}[Explicit unital $*$-isomorphism]\label{prop:isomorphism_main_part}\label{cor:explicit_star_isomorphism}
    Let $W_1, W_2$ be finite-dimensional Hilbert spaces and $r=d_{W_1}\,m_{(n)}\,d_{W_2}$. The map
    \begin{align}
        \begin{split}
            \psi\,:\,\End{}{W_1\otimes\SymH\otimes W_2}&\longrightarrow\CC^{\,r\times r}\\
            X&\longmapsto\lrbracket{\mathbb{1}_{W_1}\otimes U_{(n)}\otimes\mathbb{1}_{W_2}}^\dagger X\lrbracket{\mathbb{1}_{W_1}\otimes U_{(n)}\otimes\mathbb{1}_{W_2}}
        \end{split}
    \end{align}
    is a unital $*$-isomorphism; in particular, it is a bijection preserving the operator product, positive semidefiniteness, and the trace, with inverse $\psi^{-1}(M)=\lrbracket{\mathbb{1}_{W_1}\otimes U_{(n)}\otimes\mathbb{1}_{W_2}}\,M\,\lrbracket{\mathbb{1}_{W_1}\otimes U_{(n)}\otimes\mathbb{1}_{W_2}}^\dagger$.
\end{corollary}

\begin{proof}
    Set $V:=\mathbb{1}_{W_1}\otimes U_{(n)}\otimes\mathbb{1}_{W_2}$. By \autoref{prop:type_basis_symmetric_subspace}, $V^\dagger V=\mathbb{1}_r$ and $VV^\dagger=\mathbb{1}_{W_1}\otimes P_{\SymH}\otimes\mathbb{1}_{W_2}=\Pi$; the claim is \autoref{prop:bose_full_matrix_algebra}(ii).
\end{proof}

\begin{remark}\label{rem:rep_theory_reading}
    The basis of \autoref{prop:type_basis_symmetric_subspace} reflects the decomposition of $\SymH$ as a $\CSn$-module into $m_{(n)}$ copies of the trivial Specht module: the weight spaces of $\SymH$ are one-dimensional and are spanned by the polytabloids of the one-row shape, and, after normalization, the map $A\mapsto\lrbracket{\braket{Au_\tau, u_\gamma}}_{\tau,\gamma}$ of \cite{gijswijt2009block} coincides with $\psi$; see \autoref{sec:bose_via_schur_weyl} for the details, and \cite{christandl2007one} for the construction of Bose-symmetric states from a highest-weight vector. It is precisely the orthogonality of distinct weight spaces that renders the normalization trivial here. For a matrix $\tilde{U}$ whose columns merely span $\SymH$, with Gram matrix $G_{\lambda}:=\tilde{U}^\dagger\tilde{U}\neq\mathbb{1}$, the map $\psi_{G_\lambda}(X)=\tilde{U}^\dagger X\tilde{U}$ is still a linear bijection preserving positive semidefiniteness in both directions --- a cone-preserving bijection in the sense of \autoref{rem:symmetry_reduction_taxonomy}(iii), using $\Pi=\tilde{U}G_{\lambda}^{-1}\tilde{U}^\dagger$ --- but multiplicativity fails and is replaced by
    \begin{align}
        \psi_{G_\lambda}(XY)=\psi_{G_\lambda}(X)\,G_{\lambda}^{-1}\,\psi_{G_\lambda}(Y).
    \end{align}
    For a general shape $\lambda$, representative vectors within an isotypic component are no longer orthogonal, and one either orthogonalizes --- losing the closed form of the coefficients --- or works with the cone-preserving reduction; the latter is the route taken in \autoref{sec:sdp_symmetry_reduction}.
\end{remark}

\subsection{The symmetry-adapted reformulation of $\mathrm{SDP}_n^{\mathrm{Bose}}$}\label{sec:symmetry_adapted_reformulation}\label{sec:schur_sdp_bose}

With the explicit $*$-isomorphism of \autoref{cor:explicit_star_isomorphism} at hand, we can formulate an optimization problem approximating $w_{Q(T)}(V,\pi)$ from above, where the number of variables and number of constraints grow only polynomially in $n$. We work directly with the systems and constraint data of the Bose-symmetric hierarchy in \autoref{lem:bose-symmetric_hierarchy} --- equally, those of the constrained Bose-symmetric de Finetti theorem in \autoref{lem:bose-symmetric_deFinetti}: finite-dimensional Hilbert spaces $A=A_L\otimes A_R$, $B=B_L\otimes B_R$ and $\Bar{B}=\Bar{B}_L\otimes\Bar{B}_R$ with $B\cong\Bar{B}$, spaces $C_{A_L}$ and $C_{B_L}$, linear maps $\Theta_{A_L\rightarrow C_{A_L}}$ and $\Upsilon_{B_L\rightarrow C_{B_L}}$, and operators $W_{C_{A_L}}$ and $K_{C_{B_L}}$. The fixed-point maps $\Omega_{A\rightarrow A}$ and $\Xi_{B_n\rightarrow B_n}$ of \autoref{eqn:Bose_symmetric_hierarchy} are taken to be the identities throughout, as they are in the game application.\footnote{For general fixed-point maps the treatment is identical: $\Omega\lrbracket{\rho_{A\lrbracket{B\Bar{B}}_1^n}}=\rho_{A\lrbracket{B\Bar{B}}_1^n}$ relates operators in $\End{\CSn}{A\otimes\vee^n\lrbracket{B\Bar{B}}}$, and $\Xi_{B_n\rightarrow B_n}\circ\Tr_{\Bar{B}_n}\lrrec{\rho_{\lrbracket{B\Bar{B}}_1^n}}=\rho_{\lrbracket{B\Bar{B}}_1^{n-1}B_n}$ relates operators in $\End{\CC\lrrec{S_{n-1}}}{\vee^{n-1}\lrbracket{B\Bar{B}}\otimes B_n}$ --- a fourth space of the form covered by \autoref{cor:explicit_star_isomorphism}. Since $\Omega$ and $\Xi$ act on side registers of fixed dimension, both constraints transform exactly as the two treated below.} These data are precisely such that all three spaces collected in \autoref{eqn:D_set} are of the form $W_1\otimes\vee^{t}\lrbracket{B\Bar{B}}\otimes W_2$ for some integer $t$ covered by \autoref{cor:explicit_star_isomorphism}. Since $\Theta$ and $\Upsilon$ act on side registers only, while the $S_n$-action permutes only the $\lrbracket{B\Bar{B}}$-blocks, images of Bose-symmetric operators under these maps remain Bose-symmetric with respect to the surviving side systems. Concretely, since $\rho_{A\lrbracket{B\Bar{B}}_1^n}\in \End{\CSn}{A\otimes\vee^n\lrbracket{B\Bar{B}}}$, we deduce
    \begin{align}\label{eq:first_constraints_bose_par_trace}
        \Theta_{A_L\rightarrow C_{A_L}}\lrbracket{\rho_{A\lrbracket{B\Bar{B}}_1^n}},\, W_{C_{A_L}}\otimes\rho_{A_R\lrbracket{B\Bar{B}}_1^n}\,\in\, \End{\CSn}{C_{A_L}\otimes A_R\otimes \vee^n\lrbracket{B\Bar{B}}}
    \end{align}
    and
    \begin{align}\label{eq:second_constraints_bose_par_trace}
        \Upsilon_{\lrbracket{B_L}_n\rightarrow C_{B_L}}\circ \Tr_{\Bar{B}_n}\lrrec{\rho_{\lrbracket{B\Bar{B}}_1^n}},\, K_{C_{B_L}}\otimes \rho_{\lrbracket{B\Bar{B}}_1^{n-1}\lrbracket{B_R}_n}\,\in\, \End{\CC\lrrec{S_{n-1}}}{\vee^{n-1}\lrbracket{B\Bar{B}}\otimes C_{B_L}\otimes\lrbracket{B_R}_n},\,
    \end{align}
    by analogous arguments as in the proof of \autoref{lem:bose-symmetric_hierarchy_general_complexity}. Thus, 
    \begin{align}\label{eqn:abstract_game_sdp_bose_implicit}
    \begin{split}
            \begin{array}{cc}
                &\displaystyle \mathrm{SDP}_n^{\mathrm{Bose}}(G) = \max_{\rho_{A\lrbracket{B\Bar{B}}_1^n}\in \End{\CSn}{A\otimes\vee^n\lrbracket{B\Bar{B}}}} \Tr\left[G_{AB}\,\rho_{AB}\right]  \\
                \text{s.t.} & \rho_{A\lrbracket{B\Bar{B}}_1^n}\succeq 0,\, \Tr\left[\rho_{A\lrbracket{B\Bar{B}}_1^n}\right] = 1,\,\\
                 &\Theta_{A_L\rightarrow C_{A_L}}\lrbracket{\rho_{A\lrbracket{B\Bar{B}}_1^n}}=W_{C_{A_L}}\otimes\rho_{A_R\lrbracket{B\Bar{B}}_1^n},\,\\
                 &\Upsilon_{\lrbracket{B_L}_n\rightarrow C_{B_L}}\circ \Tr_{\Bar{B}_n}\lrrec{\rho_{\lrbracket{B\Bar{B}}_1^n}}=K_{C_{B_L}}\otimes \rho_{\lrbracket{B\Bar{B}}_1^{n-1}\lrbracket{B_R}_n}.\,
            \end{array}
        \end{split}
    \end{align}
    This is $\mathrm{SDP}_n^{\mathrm{Bose}}(G)$ from \autoref{eqn:Bose_symmetric_hierarchy} with identity fixed-point maps and with the Bose-symmetry constraint absorbed into the domain, cf.\ \autoref{prop:bose_full_matrix_algebra}(i). With
    \begin{align}
        \begin{array}{ccccccc}
             A_L=A_1Q_1, & C_{A_L}=Q_1, & A_R=T, & B_L=A_2Q_2\tilde{T}, & C_{B_L}=Q_2\tilde{T}, & B_R=\CC, & \Bar{B}=\overline{A_2Q_2\tilde{T}},
        \end{array}
    \end{align}
    and
    \begin{align}
        \begin{array}{cc}
            \Theta_{A_L\rightarrow C_{A_L}}(\cdot)=\Trr{A_1}{\cdot}, & \Upsilon_{B_L\rightarrow C_{B_L}}(\cdot)=\Trr{A_2}{\cdot},\\
            W_{C_{A_L}}=\displaystyle\sum_{q_1\in Q_1}\pi_1(q_1)\ket{q_1}\bra{q_1}_{Q_1}, & K_{C_{B_L}}=\displaystyle\sum_{q_2\in Q_2}\pi_2(q_2)\ket{q_2}\bra{q_2}_{Q_2}\otimes\frac{\mathbb{1}_{\tilde{T}}}{\lrvert{T}},
        \end{array}
    \end{align}
    --- the same instantiation as in the proof of \autoref{lem:bose_symmetric_fixed_size_games_complexity} --- it is equivalent as an optimization problem to $\mathrm{SDP}_n^{\mathrm{Bose}}\lrbracket{T, V, \pi}$ in \autoref{eqn:Bose_SDP_non-local_games}. For ease of notation, let
\begin{align}\label{eqn:D_set} 
	\mathcal{D} := \lrbrace{A\otimes\vee^n\lrbracket{B\Bar{B}},\, C_{A_L}\otimes A_R\otimes\vee^n\lrbracket{B\Bar{B}},\,\vee^{n-1}\lrbracket{B\Bar{B}}\otimes C_{B_L}\otimes\lrbracket{B_R}_n}
\end{align}
and
\begin{align}
	\begin{split}
		m \,:\, \mathcal{D} &\rightarrow \mathbb{N}\\
		 \mathcal{D}_i &\mapsto \dim_{\CC}\lrbracket{\mathcal{D}_i}.\,
	\end{split}
\end{align}
Furthermore, let 
\begin{align}
	t\lrbracket{\mathcal{D}_i}=\begin{cases}
		\begin{array}{cc}
			n & \text{if } i=1, 2,\,\\
			n-1 & \text{otherwise.}
		\end{array}
	\end{cases}
\end{align} 
Then,
\begin{align}\label{eqn:bijection_bose_sym_specifc_D}
\begin{split}
	\Psi_{(\cdot)} : \End{\CC\lrrec{S_{t\lrbracket{\D_i}}}}{\D_i}&\rightarrow \CC^{m\lrbracket{D_i}\times m\lrbracket{D_i}}\\
	Z &\mapsto \block{\Psi_{ \End{\CC\lrrec{S_{t\lrbracket{\D_i}}}}{\D_i}}\lrbracket{Z}}_{\lambda=\lrbracket{t\lrbracket{\D_i}}}
\end{split} 
\end{align} 
with $\Psi_{(\cdot)}$ acting as the isomorphism from \autoref{cor:explicit_star_isomorphism} on the corresponding spaces, is a bijection preserving positive semi-definiteness. In particular, $\Psi_{(\cdot)}$ is a unital, trace-preserving $*$-isomorphism on each of these spaces.
\begin{lemma}[Efficient Bose-symmetric hierarchy with partial trace constraints]\label{lem:Efficient_Bose-symmetric_hierarchy_with_partial_trace_constraints}
Let $ \rho_{A\lrbracket{B\Bar{B}}_1^n}\in\End{\CSn}{\D_1}$ and let $C_i\lrbracket{\D_1}$ denote the $i$-th canonical basis element of this space such that
\begin{align}
\rho_{A\lrbracket{B\Bar{B}}_1^n} = \sum_{i=1}^{m\lrbracket{\mathcal{D}_1}^2}x_iC_i\lrbracket{\mathcal{D}_1}.\,
\end{align} 
The following problem
	\begin{align}\label{eqn:reduced_Bose_SDP}
		\begin{split}
			\begin{array}{cc}
				&\displaystyle \displaystyle \Psi\lrbracket{\mathrm{SDP}_n^{\mathrm{Bose}}(G)} = \max_{\lrbrace{x_i}_{i=1}^{m\lrbracket{\mathcal{D}_1}^2}}\sum_{i=1}^{m\lrbracket{\mathcal{D}_1}^2}x_i\cdot\Tr\lrrec{G_{AB}\,C_i\lrbracket{\mathcal{D}_1}_{AB}}\\
				\\
				\text{s.th.} & \sum_{i=1}^{m\lrbracket{\mathcal{D}_1}^2}x_i C_i\lrbracket{\mathcal{D}_1}_{AB} := \sum_{i=1}^{m\lrbracket{\mathcal{D}_1}^2}x_i\cdot \Tr_{\Bar{B}_1\lrbracket{B\Bar{B}}_2^n}\lrrec{C_i\lrbracket{\mathcal{D}_1}},\,\\
				\\
				&\sum_{i=1}^{m\lrbracket{\mathcal{D}_1}^2} x_i\cdot \block{\Psi_{ \End{\CC\lrrec{S_{n}}}{\D_1}}\lrbracket{C_i\lrbracket{\D_1}}}_{(n)}\succeq 0,\,\\
				\\
				&\sum_{i=1}^{m\lrbracket{\mathcal{D}_1}^2} x_i\cdot \Tr\lrrec{C_i\lrbracket{\mathcal{D}_1}}=1,\,\\
				\\
				& \sum_{i=1}^{m\lrbracket{\mathcal{D}_1}^2} x_i\cdot \block{\Psi_{\End{\CC\lrrec{S_{t\lrbracket{\D_2}}}}{\D_2}}\lrbracket{\Theta_{A_L\rightarrow C_{A_L}}\lrbracket{C_i\lrbracket{\D_1}}}}_{(n)} \\
				\\
				&\hspace{1.5cm}=\sum_{i=1}^{m\lrbracket{\mathcal{D}_1}^2} x_i\cdot \block{\Psi_{\End{\CC\lrrec{S_{t\lrbracket{\D_2}}}}{\D_2}}\lrbracket{W_{C_{A_L}}\otimes \Tr_{A_{L}}\lrrec{C_i\lrbracket{\D_1}} }}_{(n)},\,\\
				\\
				& \sum_{i=1}^{m\lrbracket{\mathcal{D}_1}^2} x_i\cdot \block{\Psi_{\End{\CC\lrrec{S_{t\lrbracket{\D_3}}}}{\D_3}}\lrbracket{\Upsilon_{\lrbracket{B_L}_n\rightarrow C_{B_L}}\circ\Tr_{A\lrbracket{\Bar{B}}_n}\lrrec{C_i\lrbracket{\D_1}}}}_{(n-1)} \\
				\\
				&\hspace{1.5cm} = \sum_{i=1}^{m\lrbracket{\mathcal{D}_1}^2} x_i\cdot \block{\Psi_{\End{\CC\lrrec{S_{t\lrbracket{\D_3}}}}{\D_3}}\lrbracket{K_{C_{B_L}}\otimes\Tr_{A\lrbracket{B_L}_n\lrbracket{\Bar{B}}_n}\lrrec{C_i\lrbracket{\D_1}}}}_{(n-1)}	,\,
		\end{array}
		\end{split}
	\end{align}
	 is equivalent as an optimization problem to $\mathrm{SDP}_n^{\mathrm{Bose}}(G)$ in \autoref{eqn:abstract_game_sdp_bose_implicit}, where we optimize over polynomial many variables subject to polynomial many constraints of polynomial size.
\end{lemma}
 
\begin{proof}
    By \autoref{eq:first_constraints_bose_par_trace} and \autoref{eq:second_constraints_bose_par_trace}, all operators appearing in the constraints of \autoref{eqn:abstract_game_sdp_bose_implicit} are Bose-symmetric with respect to the appropriate side systems, so $\Psi_{(\cdot)}$ may be applied to both sides of every constraint. By \autoref{cor:explicit_star_isomorphism}, $\Psi_{(\cdot)}$ is linear and bijective on each of the three spaces in \autoref{eqn:D_set} and preserves positive semidefiniteness and the trace. Since the objective and all constraints of \autoref{eqn:abstract_game_sdp_bose_implicit} are linear in $\rho_{A\lrbracket{B\Bar{B}}_1^n}$, expanding $\rho_{A\lrbracket{B\Bar{B}}_1^n}=\sum_i x_i\, C_i\lrbracket{\mathcal{D}_1}$ in the canonical basis and applying $\Psi_{(\cdot)}$ constraint by constraint yields \autoref{eqn:reduced_Bose_SDP}; conversely, every feasible point of \autoref{eqn:reduced_Bose_SDP} defines, via $\Psi_{(\cdot)}^{-1}$, a feasible point of \autoref{eqn:abstract_game_sdp_bose_implicit} with the same objective value. The number of variables and the orders of all matrices involved are polynomial in $n$ by \autoref{prop:bose_full_matrix_algebra}(iii) (cf.\ \autoref{lem:bose-symmetric_hierarchy_general_complexity}).
\end{proof}


\subsection{Efficiency of the transformation}\label{sec:efficiency_of_Bose_trafo}

The program $\Psi\lrbracket{\mathrm{SDP}_n^{\mathrm{Bose}}(G)}$ of \autoref{lem:Efficient_Bose-symmetric_hierarchy_with_partial_trace_constraints} is of polynomial size, yet it is defined through operators acting on spaces whose dimension is exponential in $n$; a direct application of $\Psi$ from \autoref{eqn:bijection_bose_sym_specifc_D}, mapping $\mathrm{SDP}_n^{\mathrm{Bose}}(G)$ in \autoref{eqn:abstract_game_sdp_bose_implicit} to $\Psi\lrbracket{\mathrm{SDP}_n^{\mathrm{Bose}}(G)}$ in \autoref{eqn:reduced_Bose_SDP}, would therefore require exponentially many computations in $n$. Concretely, the quantities in $\mathrm{SDP}_n^{\mathrm{Bose}}(G)$ are of exponential size in $n$ when expressed in the computational basis.

To prove the polynomial-time algorithm stated in \autoref{thm:Bose_sym_complexity_results}, we provide an efficient procedure for constructing $\Psi\lrbracket{\mathrm{SDP}_n^{\mathrm{Bose}}(G)}$. The key insight is that, while the objects in $\mathrm{SDP}_n^{\mathrm{Bose}}(G)$ are exponentially large, they exhibit substantial redundancy. That is, accessing only a polynomial number of entries in $n$ suffices to carry out the transformation efficiently.

The structure of this section is as follows. First, in \autoref{sec:efficient_acces_bose_operators}, we establish an efficient method for accessing the unique coefficients of any $Z \in \End{\CSn}{\SymH}$ with respect to the canonical basis of $\End{\CSn}{\SymH}$. Then, in \autoref{sec:efficient_trafo_bose_operators}, we present an efficient procedure for computing $\psi(Z)$ for any such $Z$. Finally, in \autoref{sec:efficient_construction_of_bose_sdp}, we conclude with a proof of the overall efficiency of the SDP transformation.

\subsubsection{Efficiently accessing Bose-symmetric operators}\label{sec:efficient_acces_bose_operators} 

At a high level, the task of efficiently determining the unique coefficients of any Bose-symmetric operator amounts to efficiently identifying a representative element of each Bose-symmetric orbit. Recall from \autoref{prop:type_basis_symmetric_subspace} (cf.\ \cite{harrow2013church}) that $\lrbrace{\ket{s_{\vec{t}}}\, :\, \vec{t}\in \mathcal{T}_{(n),\, d_{\HS}}}$ with
    \begin{align}
        \text{$\ket{s_{\vec{t}}} :=\sum_{\vec{i} : T(\vec{i})=\vec{t}} \ket{i_1,\ldots, i_n}$ or the normalized $\ket{\tilde{s}_{\vec{t}}} := \frac{1}{\sqrt{\binom{n}{\vec{t}}}} \ket{s_{\vec{t}}}$,}
    \end{align}
where $T(\vec{i})$ denotes the type of $\vec{i} = (i_1, \ldots, i_n)$, is a basis of the symmetric subspace; throughout this subsection we work with the unnormalized vectors $\ket{s_{\vec{t}}}$, whose entries lie in $\lrbrace{0,1}$.

\begin{proposition}\label{prop:first_basis_Bose}
	The set
		\begin{align}
			\lrbrace{C_{\vec{t},\vec{t'}}^{\vee^n}:=\ket{s_{\vec{t}}}\bra{s_{\vec{t'}}} \,:\, \vec{t}, \vec{t'} \in \mathcal{T}_{(n),\, d_{\HS}}}
		\end{align}
		is a basis of $\End{\CSn}{\SymH}$.
	\end{proposition}
\begin{proof}
    By \autoref{prop:type_basis_symmetric_subspace}, the vectors $\ket{\tilde{s}_{\vec{t}}}$ form an orthonormal basis of $\SymH$; hence the operators $\CBose{t}{t'}{n}=\sqrt{\binom{n}{\vec{t}}\binom{n}{\vec{t'}}}\,\ket{\tilde{s}_{\vec{t}}}\bra{\tilde{s}_{\vec{t'}}}$ form a basis of $\End{}{\SymH}$, which coincides with $\End{\CSn}{\SymH}$ by \autoref{prop:Bose_sym_space_char}. Alternatively, the $m_{(n)}^2$ matrices $\CBose{t}{t'}{n}$ are nonzero with pairwise disjoint supports in the computational basis, hence linearly independent, and their number matches $\dim_{\CC}\End{\CSn}{\SymH}=m_{(n)}^2$.
\end{proof}
See \autoref{sec:proof_of_efficient_trafo_sdp_sym} for an example. Since any Bose-symmetric operator is symmetric, we will provide an alternative characterization of the basis of \autoref{prop:first_basis_Bose} as certain linear combinations of the canonical orbit basis elements of $\End{\CSn}{\HSn}$. We adopt the notation from \cite{gijswijt2009block}. For $a, b\in \lrrec{d_{\HS}}^n$, define $D(a,b)\in\mathbb{N}_0^{d_{\HS}\times d_{\HS}}$ by
\begin{align}
    \lrbracket{D(a,b)}_{i,j}=\lrvert{\lrbrace{k\,\vert\, a_k=i, b_k=j}}.\,
\end{align}
Note that for any $a,a',b,b'\in \lrrec{d_{\HS}}^n$, we have $D(a,b)=D(a',b')$ if and only if $a'=\pi(a)$ and $b'=\pi(b)$ for some $\pi\in S_n$. Thus, for any such $a,b$ the matrix $D(a,b)$ can be seen as a class function on the simultaneous $S_n$-orbits, i.e.\ 
\begin{align}
    D(\pi(a), \pi(b)) = D(a,b),\, \hspace{1cm}\forall \pi\in S_n.\,
\end{align} 
For the all ones column vector $\vec{1}$, let 
\begin{align}\label{eqn:P_n_d_H}
    P(n, d_{\HS}):= \lrbrace{D \in \mathbb{N}_0^{d_{\HS}\times d_{\HS}}\,\vert\, \vec{1}^TD\vec{1} =n}
\end{align}
denote the set of such $D$, where the sum over all entries equals $n$. In other words, $D(a,b)$ records the frequency of each pair $(i,j)\in\lrrec{d_{\HS}}\times\lrrec{d_{\HS}}$ among the pairs $\lrbracket{a_k, b_k}$, and $P(n, d_{\HS})$ is in bijection with the orbits of pairs $(a,b)\in\lrrec{d_{\HS}}^n\times\lrrec{d_{\HS}}^n$ under the simultaneous $S_n$-action; it hence indexes the orbits of the computational matrix units of $\Op{\HSn}$. Note that $D(a,b)\vec{1}$ yields the type of $a$ and $D(a,b)^T\vec{1}$ the type of $b$. From these orbit class functions, we can now construct the canonical basis to $\End{\CSn}{\HSn}$ by averaging over each orbit. Concretely, 
\begin{align}
    \lrbracket{A_{D}}_{a,b} = \begin{cases}
        \begin{array}{cc}
             1,\, & \text{if }  D(a,b)=D,\,\\
             0,\, & \text{otherwise.}
        \end{array}
    \end{cases}
\end{align}
By \cite[Proposition 1]{gijswijt2009block}, the set of $\lrbrace{A_D}_{D\in P(n, d_{\HS})}$ is a basis of $\End{\CSn}{\HSn}$. 
Thus, the following proposition follows naturally.
\begin{proposition}[Entries of the canonical basis]\label{prop:entries_canonical_bose}
Let $\vec{t}, \vec{t'}\in\mathcal{T}_{(n),\, d_{\HS}}$. Then, for all $a, b\in\lrrec{d_{\HS}}^n$,
    \begin{align}
    \lrrec{\CBose{t}{t'}{n}}_{a,b} = \begin{cases}
        \begin{array}{cc}
             1,\, & \text{if } D(a,b)\vec{1}=\vec{t}\,\text{ and }\,D(a,b)^T\vec{1}=\vec{t'},\,\\
             0,\, & \text{otherwise.}
        \end{array}
    \end{cases}
\end{align}
In particular, the entry $\lrrec{\CBose{t}{t'}{n}}_{a,b}$ depends on the pair $(a,b)$ only through the matrix $D(a,b)$.
\end{proposition}
\begin{proof}
    By definition, $\CBose{t}{t'}{n}=\ket{s_{\vec{t}}}\bra{s_{\vec{t'}}}$ has entry $1$ at $(a,b)$ if and only if $T(a)=\vec{t}$ and $T(b)=\vec{t'}$, and entry $0$ otherwise; the claim follows since $D(a,b)\vec{1}=T(a)$ and $D(a,b)^T\vec{1}=T(b)$.
\end{proof}

It is a well-established fact that $P(n, d_{\HS})$ can be computed efficiently (see e.g. \cite{chee2023efficient} or \autoref{eqn:opti_solving_cano_basis_E}). We obtain an analogous result in the Bose-symmetric variant. For $D\in P(n, d_{\HS})$, the vectors of row and column sums $D\vec{1}$ and $D^T\vec{1}$ are called the \emph{margins} of $D$. The terminology stems from contingency tables --- $D(a,b)$ being the two-way table that cross-tabulates the $n$ tensor factors of the pair $(a,b)$ by their value pairs $\lrbracket{a_k, b_k}$ --- where these totals are traditionally recorded in the margins of the table. For the pair-frequency matrices this recovers the observation below \autoref{eqn:P_n_d_H}: the margins of $D(a,b)$ are the types, $D(a,b)\vec{1}=T(a)$ and $D(a,b)^T\vec{1}=T(b)$; equivalently, $D(a,b)/n$ is the empirical distribution of the pairs $\lrbracket{a_k,b_k}_{k=1}^n$, with marginal distributions $T(a)/n$ and $T(b)/n$. Collecting the two margins defines the \emph{margin map}
\begin{align}\label{eqn:margin_map}
    \mu\,:\, P(n, d_{\HS}) \rightarrow \mathcal{T}_{(n),\, d_{\HS}}\times\mathcal{T}_{(n),\, d_{\HS}},\, \hspace{1cm} D\mapsto\lrbracket{D\vec{1},\, D^T\vec{1}},\,
\end{align}
which is surjective, since $\mu\lrbracket{D(a,b)}=\lrbracket{T(a), T(b)}$ for any $a$ of type $\vec{t}$ and $b$ of type $\vec{t'}$. Next, we partition $P(n, d_{\HS})$ into the fibres $\mu^{-1}\lrbracket{\vec{t},\vec{t'}}$ of the margin map: by \autoref{prop:entries_canonical_bose}, the support of $\CBose{t}{t'}{n}$ is the union of the $S_n$-orbits labelled by the fibre over $\lrbracket{\vec{t},\vec{t'}}$, so each fibre collects precisely those $S_n$-orbits that are merged by the additional symmetry of Bose-symmetric operators, in contrast to merely symmetric ones. The following lemma identifies the fibres explicitly and bounds their cardinality.

\begin{lemma}\label{lem:bound_bose_lin_comb}
Let $a\in\lrrec{d_{\HS}}^n$ be of type $\vec{t}$ with $h_a\leq d_{\HS}$ nonzero entries, and let $b\in\lrrec{d_{\HS}}^n$ be of type $\vec{t'}$ with $h_b\leq d_{\HS}$ nonzero entries, so that the sorted types are partitions of $n$ of heights $h_a$ and $h_b$, respectively. Then
    \begin{align}
    \mathcal{D}_{Bose}(a,b):= \lrbrace{D\lrbracket{\pi(a), \sigma(b)}\,:\,\pi, \sigma\in S_n} = \lrbrace{D\in P(n, d_{\HS})\,\vert\, D\vec{1}=\vec{t},\; D^T\vec{1}=\vec{t'}}
\end{align}
is a set of cardinality at most $\lrbracket{n+1}^{\lrbracket{h_a-1}\lrbracket{h_b-1}}$, i.e.\ polynomial in $n$ for fixed $h_a, h_b$, which, given $P(n, d_{\HS})$, can be computed in time polynomial in $n$. Furthermore, if $h_a=1$ or $h_b=1$, then $\lrvert{\mathcal{D}_{Bose}(a,b)}=1$.
\end{lemma}
\begin{proof}
    For the set equality: every $D\lrbracket{\pi(a),\sigma(b)}$ has row margins $D\vec{1}=T\lrbracket{\pi(a)}=\vec{t}$ and column margins $D^T\vec{1}=T\lrbracket{\sigma(b)}=\vec{t'}$. Conversely, given $D\in P(n, d_{\HS})$ with these margins, construct strings $a', b'\in\lrrec{d_{\HS}}^n$ by listing, for each pair $(i,j)$, exactly $D_{i,j}$ positions $k$ with $a'_k=i$ and $b'_k=j$; then $D(a',b')=D$, and since $T(a')=\vec{t}=T(a)$ and $T(b')=\vec{t'}=T(b)$, there are $\pi,\sigma\in S_n$ with $a'=\pi(a)$ and $b'=\sigma(b)$.

    For the cardinality, the margin conditions determine all entries of $D$ outside the $h_a\times h_b$ submatrix indexed by the supports of $\vec{t}$ and $\vec{t'}$ (they vanish), so $\mathcal{D}_{Bose}(a,b)$ is in bijection with the set of $h_a\times h_b$ contingency tables with row margins the nonzero part of $\vec{t}$ and column margins the nonzero part of $\vec{t'}$ (see, e.g., \cite{de2013combinatorics}). Since the margins are fixed, such a table is uniquely determined by its upper-left $\lrbracket{h_a-1}\times\lrbracket{h_b-1}$ submatrix, whose entries lie in $\lrbrace{0,\ldots, n}$; injectivity of this restriction yields $\lrvert{\mathcal{D}_{Bose}(a,b)}\leq\lrbracket{n+1}^{\lrbracket{h_a-1}\lrbracket{h_b-1}}$. If $h_a=1$, the unique nonzero row equals $\vec{t'}$, so the table is unique, and analogously for $h_b=1$. For fixed margins, the number of such tables can alternatively be analyzed via Ehrhart theory \cite{ehrhart1962polyedres, beck2007computing}; the elementary bound above suffices for our purposes.

    Finally, $P(n, d_{\HS})$ has cardinality polynomial in $n$ and can be computed efficiently (see, e.g., \cite{chee2023efficient} or \autoref{eqn:opti_solving_cano_basis_E}), and $\mathcal{D}_{Bose}(a,b)$ is obtained from it by filtering for the two margin conditions, at cost $\mathcal{O}\lrbracket{d_{\HS}^2}$ per element.
\end{proof}

Thus, we obtain an alternative reformulation of the canonical basis for the space of Bose-symmetric operators.
\begin{proposition}\label{prop:C_sym_from_A_D_Bose}
Let $a\in\lrrec{d_{\HS}}^n$ be of type $\vec{t}$, $b\in\lrrec{d_{\HS}}^n$ be of type $\vec{t'}$ and $\mathcal{D}_{Bose}(a,b)$ be defined as in \autoref{lem:bound_bose_lin_comb}. Then, 
    \begin{align}
         \CBose{t}{t'}{n} = \sum_{D\in\mathcal{D}_{Bose}(a,b)} A_D.
    \end{align}
\end{proposition}
\begin{proof}
    By \autoref{prop:entries_canonical_bose}, the support of $\CBose{t}{t'}{n}$ is $\lrbrace{(a',b')\,:\,T(a')=\vec{t},\, T(b')=\vec{t'}}$, which is the disjoint union, over $D\in\mathcal{D}_{Bose}(a,b)$, of the simultaneous $S_n$-orbits $\lrbrace{(a',b')\,:\,D(a',b')=D}$; the operators $A_D$ are precisely the indicators of these orbits.
\end{proof}

\begin{remark}[Orbit bases: simultaneous vs.\ independent orbits]\label{rem:simultaneous_vs_independent_orbits}\label{rem:orbit_bases_Bose}
The two canonical bases of this subsection arise from one construction --- $0/1$ indicator operators of orbits of index pairs $(a,b)\in\lrrec{d_{\HS}}^n\times\lrrec{d_{\HS}}^n$ --- applied to two different symmetrizations. Averaging by conjugation, $Z\mapsto\frac{1}{n!}\sum_{\pi\in S_n}U(\pi)\,Z\,U(\pi)^{\top}$, moves row and column indices \emph{simultaneously}, $\lrbracket{a,b}\mapsto\lrbracket{\pi(a),\pi(b)}$; the orbits of this $S_n$-action are classified by the full matrices $D(a,b)$, and their indicators $\lrbrace{A_D}_{D\in P(n,\,d_{\HS})}$ span the image $\End{\CSn}{\HSn}$. Two-sided symmetrization, $Z\mapsto P_{\SymH}\,Z\,P_{\SymH}=\frac{1}{\lrbracket{n!}^2}\sum_{\pi,\sigma\in S_n}U(\pi)\,Z\,U(\sigma)^{\top}$, moves the two indices \emph{independently}, $\lrbracket{a,b}\mapsto\lrbracket{\pi(a),\sigma(b)}$; the orbits of this $S_n\times S_n$-action are classified by the margins $\lrbracket{T(a),\,T(b)}$ alone (\autoref{prop:entries_canonical_bose}), and their indicators $\lrbrace{\CBose{t}{t'}{n}}_{\vec{t},\vec{t'}\in\mathcal{T}_{(n),\,d_{\HS}}}$ span the image $\End{\CSn}{\SymH}$. Indeed, up to normalization, $A_{D(a,b)}$ and $\CBose{T(a)}{T(b)}{n}$ are the averages of any computational matrix unit $\ket{a}\bra{b}$ under the respective symmetrization. The margin map $\mu$ of \autoref{eqn:margin_map} mediates between the two labelings: each independent orbit is the disjoint union of the simultaneous orbits in one margin fibre $\mathcal{D}_{Bose}$ (\autoref{prop:C_sym_from_A_D_Bose}). Consequently, $\End{\CSn}{\SymH}$ sits inside $\End{\CSn}{\HSn}$ precisely as the span of those $\sum_{D}z_D\,A_D$ whose coefficients are constant on every margin fibre; an individual $A_D$ over a fibre of size at least two satisfies the commutant condition $U(\pi)\,X\,U(\pi)^{\top}=X$ but violates the support condition $P_{\SymH}\,X\,P_{\SymH}=X$. This is the combinatorial counterpart of \autoref{fig:symmetry_vs_bose_symmetry_block_decomp}: restricting to the sector $\lambda=(n)$ merges, within each margin fibre, the finer simultaneous orbits into a single independent one.

The unnormalized convention thus makes the basis of \autoref{prop:first_basis_Bose} the \emph{orbit basis} of $\End{\CSn}{\SymH}$ in the same sense in which $\lrbrace{A_D}_{D\in P(n,\,d_{\HS})}$ is the orbit basis of $\End{\CSn}{\HSn}$: the expansion in \autoref{prop:C_sym_from_A_D_Bose} has all coefficients equal to one, and the coefficient extraction in \autoref{lem:efficient_access_bose} below amounts to reading single matrix entries, whereas normalized basis elements would carry the square-root factors $\binom{n}{\vec{t}}^{-1/2}\binom{n}{\vec{t'}}^{-1/2}$ through both. Normalization is confined to where orthonormality is used, namely the columns $\ket{\tilde{s}_{\vec{t}}}$ of the isometry $U_{(n)}$ in \autoref{prop:type_basis_symmetric_subspace} (cf.\ also the trace-normalized matrix units in \autoref{rem:symmetry_reduction_taxonomy}).
\end{remark}

\begin{lemma}[Efficient access]\label{lem:efficient_access_bose}
    Let $Z\in\End{\CSn}{\SymH}$ be a Bose-symmetric operator such that 
    \begin{align}
        Z=\sum_{\vec{t},\vec{t'}\in\mathcal{T}_{(n),\, d_{\HS}}} z_{\vec{t},\vec{t'}} \CBose{t}{t'}{n}\,,
    \end{align}
    and suppose we are given oracle access to the entries of $Z$ in the computational basis, i.e., to
    \begin{align}
    \begin{split}
        f_{Z} \,:\, \lrrec{d_{\HS}}^n\times \lrrec{d_{\HS}}^n &\rightarrow \CC \\
        (a, b) &\mapsto Z_{a,b},
    \end{split}
    \end{align}
    where a string in $\lrrec{d_{\HS}}^n$ is identified with the index of the corresponding computational basis vector. Then $\lrbrace{z_{\vec{t},\vec{t'}}}_{\vec{t},\vec{t'}\in\mathcal{T}_{(n),\, d_{\HS}}}$ can be computed with $m_{(n)}^2$ oracle queries and in time polynomial in $n$ for fixed $d_{\HS}$.
\end{lemma}

\begin{proof}
    For $\vec{t}\in\mathcal{T}_{(n),\, d_{\HS}}$, let 
    \begin{align}
        a(\vec{t}):=\lrbracket{\underbrace{1,\ldots,1}_{t_1\text{ times}},\, \underbrace{2,\ldots,2}_{t_2\text{ times}},\,\ldots,\,\underbrace{d_{\HS},\ldots,d_{\HS}}_{t_{d_{\HS}}\text{ times}}}
    \end{align}
    denote the nondecreasing string of type $\vec{t}$, computable in time $\mathcal{O}(n)$. By \autoref{prop:entries_canonical_bose}, $\CBose{u}{u'}{n}$ has entry $1$ at the index pair $\lrbracket{a(\vec{t}),\, a(\vec{t'})}$ if and only if $\lrbracket{\vec{u},\vec{u'}}=\lrbracket{\vec{t},\vec{t'}}$; since the supports of the basis elements are pairwise disjoint, it follows that
    \begin{align}
        z_{\vec{t},\vec{t'}} = Z_{a(\vec{t}),\, a(\vec{t'})} = f_Z\lrbracket{a(\vec{t}),\, a(\vec{t'})}.
    \end{align}
    The procedure below therefore recovers all coefficients using $m_{(n)}^2\leq\lrbracket{n+1}^{2\lrbracket{d_{\HS}-1}}$ oracle queries, cf.\ \autoref{eqn:ssyt_one_row_count}.
\begin{algorithm}[H]
\caption{Efficient access to Bose-symmetric operators}
\begin{algorithmic}
  \State $\mathcal{L} \gets \emptyset$
  \ForAll{$\lrbracket{\vec{t},\, \vec{t'}}\in\mathcal{T}_{(n),\, d_{\HS}}\times\mathcal{T}_{(n),\, d_{\HS}}$}
      \State $z_{\vec{t},\vec{t'}} \gets f_Z\lrbracket{a(\vec{t}),\, a(\vec{t'})}$
      \State $\mathcal{L} \gets \mathcal{L}\,\cup\,\lrbrace{\lrbracket{\vec{t},\, \vec{t'},\, z_{\vec{t},\vec{t'}}}}$
  \EndFor
  \State \textbf{Return} $\mathcal{L}$
\end{algorithmic}
\end{algorithm}
    \noindent Alternatively, in the spirit of \autoref{prop:C_sym_from_A_D_Bose}, one may iterate over $P(n, d_{\HS})$ and query a single representative per margin fibre $\mathcal{D}_{Bose}(a,b)$; the direct enumeration above uses the same number of queries.
\end{proof}

\subsubsection{Efficiently transforming Bose-symmetric operators}\label{sec:efficient_trafo_bose_operators}

By linearity, transforming an arbitrary Bose-symmetric operator reduces to transforming the corresponding canonical basis. First we show that, for any $i\in\lrrec{3}$, given the complex coefficients $\lrbrace{z_j}_{j=1}^{m\lrbracket{\D_i}^2}$ of an operator $Z\in\End{\CC\lrrec{S_{t\lrbracket{\D_i}}}}{\D_i}$ with respect to the canonical basis of this space, the block
	\begin{align}
		\block{\Psi_{ \End{\CC\lrrec{S_{t\lrbracket{\D_i}}}}{\D_i}}\lrbracket{Z}}_{\lambda=\lrbracket{t\lrbracket{\D_i}}}
	\end{align}
	can be computed in poly$(n)$ time. The canonical basis of $\End{\CSn}{\D_1}$ can be written as
	\begin{align}\label{eqn:canonical_basis_D1}
		\mathcal{B}_{\D_1}:=\lrbrace{\ket{i}\bra{j}_{A}\otimes \CBose{t}{t'}{n}}_{\substack{i,j\in \lrrec{\lrvert{A}},\,\\ \vec{t}, \vec{t'}\in \mathcal{T}_{(n),\, \lrvert{B\Bar{B}}}}}
	\end{align}
    and similarly for $\D_2$ and $\D_3$, with side systems $C_{A_L}\otimes A_R$ and $C_{B_L}\otimes\lrbracket{B_R}_n$, respectively. Since the side systems are of fixed dimension and the map of \autoref{cor:explicit_star_isomorphism} acts on them as the identity, it suffices to treat $\End{\CSn}{\SymH}$. Any $Z\in\End{\CSn}{\SymH}$ can be written in the canonical basis of \autoref{prop:first_basis_Bose} as $Z=\sum_{\vec{t},\vec{t'}\in\mathcal{T}_{(n),\, d_{\HS}}} z_{\vec{t},\vec{t'}}\,\CBose{t}{t'}{n}$. By \autoref{cor:explicit_star_isomorphism} together with \autoref{prop:type_basis_symmetric_subspace}, $\block{\psi(Z)}_{(n)}=U_{(n)}^\dagger\, Z\, U_{(n)}$, where the columns of the isometry $U_{(n)}$ are the normalized vectors $\ket{\tilde{s}_{\vec{\tau}}}=\binom{n}{\vec{\tau}}^{-1/2}\ket{s_{\vec{\tau}}}$; entrywise,
\begin{align}
    \lrrec{\block{\psi(Z)}_{(n)}}_{\vec{\tau},\vec{\gamma}} = \bra{\tilde{s}_{\vec{\tau}}}\, Z\, \ket{\tilde{s}_{\vec{\gamma}}} = \sum_{\vec{t},\vec{t'}\in\mathcal{T}_{(n),\, d_{\HS}}} z_{\vec{t},\vec{t'}}\, \bra{\tilde{s}_{\vec{\tau}}}\,\CBose{t}{t'}{n}\,\ket{\tilde{s}_{\vec{\gamma}}}.
\end{align}
The transformation therefore reduces to the following closed form.
\begin{proposition}[{\cite{harrow2013church}}]\label{prop:dimension_bound_sym_subspace}
	For any $k, n\in\mathbb{N}$,
	\begin{align}
		\lrvert{\mathcal{T}_{(n),\, k}} = \lrvert{\operatorname{WComp}(k,n)} = \binom{k+n-1}{n}\leq \lrbracket{n+1}^{k-1},\,\quad \lrvert{T^{-1}(\vec{t})}=\binom{n}{\vec{t}}:=\frac{n!}{\prod_{i=1}^k(t_i!)},\,
	\end{align}
	where $T^{-1}(\vec{t})$ is the set of strings of type $\vec{t}$; cf.\ \autoref{eqn:ssyt_one_row_count}.
\end{proposition}
\begin{lemma}\label{lem:efficient_basis_trafo_Bose_sym}
	For all $\vec{\tau}, \vec{\gamma}, \vec{t}, \vec{t'}\in \mathcal{T}_{(n),\, d_{\HS}}$,
	\begin{align}\label{eqn:closed_form_transformed_basis}
		\bra{\tilde{s}_{\vec{\tau}}}\,\CBose{t}{t'}{n}\,\ket{\tilde{s}_{\vec{\gamma}}} = \sqrt{\binom{n}{\vec{t}}\binom{n}{\vec{t'}}}\;\delta_{\vec{\tau},\vec{t}}\,\delta_{\vec{\gamma},\vec{t'}}.
	\end{align}
	Consequently, given complex numbers $\lrbrace{z_{\vec{t},\vec{t'}}}_{\vec{t},\vec{t'}\in\mathcal{T}_{(n),\, d_{\HS}}}$, the matrix $\block{\psi(Z)}_{(n)}\in\CC^{m_{(n)}\times m_{(n)}}$ of $Z=\sum_{\vec{t},\vec{t'}} z_{\vec{t},\vec{t'}}\,\CBose{t}{t'}{n}$, with entries
	\begin{align}
		\lrrec{\block{\psi(Z)}_{(n)}}_{\vec{\tau},\vec{\gamma}} = \sqrt{\binom{n}{\vec{\tau}}\binom{n}{\vec{\gamma}}}\; z_{\vec{\tau},\vec{\gamma}},
	\end{align}
	can be computed exactly in time polynomial in $n$ for fixed $d_{\HS}$.
\end{lemma}
\begin{proof}
	Vectors of distinct type are supported on disjoint sets of computational basis vectors, and $\braket{s_{\vec{\tau}}|s_{\vec{t}}}=\binom{n}{\vec{t}}\,\delta_{\vec{\tau},\vec{t}}$, since for $\vec{\tau}=\vec{t}$ the two $0/1$-vectors share exactly $\binom{n}{\vec{t}}$ entries equal to one. Hence
	\begin{align}
		\bra{\tilde{s}_{\vec{\tau}}}\,\CBose{t}{t'}{n}\,\ket{\tilde{s}_{\vec{\gamma}}} = \binom{n}{\vec{\tau}}^{-1/2}\binom{n}{\vec{\gamma}}^{-1/2}\,\braket{s_{\vec{\tau}}|s_{\vec{t}}}\,\braket{s_{\vec{t'}}|s_{\vec{\gamma}}} = \sqrt{\binom{n}{\vec{t}}\binom{n}{\vec{t'}}}\;\delta_{\vec{\tau},\vec{t}}\,\delta_{\vec{\gamma},\vec{t'}},
	\end{align}
	and the entry formula follows by linearity. Each multinomial coefficient is a positive integer with $\mathcal{O}\lrbracket{n\log d_{\HS}}$ bits, computable exactly with $\mathcal{O}(n)$ integer multiplications and one division, and there are $m_{(n)}^2\leq\lrbracket{n+1}^{2\lrbracket{d_{\HS}-1}}$ entries.
\end{proof}

With the unnormalized polytabloids $u_{\tau}=\ket{s_{\vec{\tau}}}$ in place of $\ket{\tilde{s}_{\vec{\tau}}}$, the same computation yields $u_{\tau}^\dagger\,\CBose{t}{t'}{n}\,u_{\gamma}=\binom{n}{\vec{t}}\binom{n}{\vec{t'}}\,\delta_{\vec{\tau},\vec{t}}\,\delta_{\vec{\gamma},\vec{t'}}$ --- the Gram-twisted, cone-preserving variant of \autoref{rem:rep_theory_reading}; cf.\ \cite{polak2020new} for the coordinate-ring formalism that extends this computation to general shapes, and \autoref{sec:coordinate_ring_proof} for an alternative derivation of \autoref{eqn:closed_form_transformed_basis} in this formalism. See \autoref{ex:change_of_basis_bose} for an example. Now by \autoref{lem:efficient_basis_trafo_Bose_sym}, for each $\D_i\in\mathcal{D}$ and any $Z\in\End{\CC\lrrec{S_{t\lrbracket{\D_i}}}}{\D_i}$ with $i\in\lrrec{3}$, if the decomposition of $Z$ in the corresponding canonical basis $\mathcal{B}_{\D_i}$ is given then 
	\begin{align}
		 \block{\Psi_{ \End{\CC\lrrec{S_{t\lrbracket{\D_i}}}}{\D_i}}\lrbracket{Z}}_{\lambda=\lrbracket{t\lrbracket{\D_i}}}
	\end{align}
	can be computed in $\text{poly}(n)$ time. 

\subsubsection{Efficient construction of $\Psi\lrbracket{\mathrm{SDP}_n^{\mathrm{Bose}}(G)}$}\label{sec:efficient_construction_of_bose_sdp}  

In summary, we obtain the following results.
	
\begin{lemma}\label{lem:efficiency_Bose_trafo}
For $n\in\mathbb{N}_{\geq 1}$  the transformation $\Psi$ mapping $\mathrm{SDP}_n^{\mathrm{Bose}}(G)$ in \autoref{eqn:abstract_game_sdp_bose_implicit} to $\Psi\lrbracket{\mathrm{SDP}_n^{\mathrm{Bose}}(G)}$ in \autoref{eqn:reduced_Bose_SDP} can be done in poly$\lrbracket{n}$ time for systems $A,B$ of fixed dimension. 
\end{lemma}
As a direct corollary, we obtain the following result in the fixed non-local game setting.
\begin{corollary}\label{cor:efficiency_Bose_trafo_games}
    The SDP hierarchy of \autoref{thm:Bose_sym_complexity_results} (see \autoref{sec:nonlocal-games}) can be constructed in poly$(n)$ time.
\end{corollary}
\begin{remark}
    The Bose-symmetric de Finetti theorem in \autoref{lem:bose-symmetric_deFinetti} certifies the approximation error between $\mathrm{SDP}_n^{\mathrm{Bose}}(G)$ and $\mathrm{cSEP}(G)$ for any $n$. Since $\Psi$ is a $*$-isomorphism, $\Psi\lrbracket{\mathrm{SDP}_n^{\mathrm{Bose}}(G)}$ and $\mathrm{SDP}_n^{\mathrm{Bose}}(G)$ have the same optimal value, and thus the approximation guarantee translates to $\Psi\lrbracket{\mathrm{SDP}_n^{\mathrm{Bose}}(G)}$.
\end{remark}

To prepare for the proof of \autoref{lem:efficiency_Bose_trafo}, we introduce the following notion.

\begin{definition}[Reduced type]\label{def:reduced_type}
	For $p\in\lrrec{d_{\HS}}$, let $T(p)\in\lrbrace{0,1}^{d_{\HS}}$ denote the indicator vector of $p$, i.e.\ $T(p)_j=\delta_{p,j}$. For a type $\vec{t}\in\mathcal{T}_{(n),\,d_{\HS}}$ with $t_p\geq 1$, we write $\vec{t}-T(p)\in\mathcal{T}_{(n-1),\,d_{\HS}}$ for the type obtained from $\vec{t}$ by removing one occurrence of the letter $p$. In particular, if $\vec{w}\in\lrrec{d_{\HS}}^n$ is of type $\vec{t}$, then $\vec{t}-T(w_i)$ is the type of the string obtained from $\vec{w}$ by deleting its $i$-th letter.
\end{definition}

\begin{proof}[Proof of \autoref{lem:efficiency_Bose_trafo}]	
	In order to prove \autoref{lem:efficiency_Bose_trafo}, we show that the objective function and constraints in $\Psi\lrbracket{\mathrm{SDP}_n^{\mathrm{Bose}}(G)}$ can be computed explicitly in poly$(n)$ time. Concretely, due to the bounds on the sums appearing in \autoref{eqn:reduced_Bose_SDP}, we show that for any $i\in \lrrec{m(\D_1)^2}$ the expressions
	\begin{enumerate}
		\item \label{psd_constraint} $\block{\Psi_{ \End{\CC\lrrec{S_{n}}}{\D_1}}\lrbracket{C_i\lrbracket{\D_1}}}_{(n)}$
		\item \label{trace_constraint} $\Tr\lrrec{C_i\lrbracket{\mathcal{D}_1}}$
		\item \label{objective_function} $\Tr_{\Bar{B}_1\lrbracket{B \bar{B}}_2^n}\lrrec{C_i\lrbracket{\mathcal{D}_1}}$
		\item \label{alice_constraint_left} $\block{\Psi_{\End{\CC\lrrec{S_{t\lrbracket{\D_2}}}}{\D_2}}\lrbracket{\Theta_{A_L\rightarrow C_{A_L}}\lrbracket{C_i\lrbracket{\D_1}}}}_{(n)}$
		\item \label{alice_constraint_right} $\block{\Psi_{\End{\CC\lrrec{S_{t\lrbracket{\D_2}}}}{\D_2}}\lrbracket{W_{C_{A_L}}\otimes \Tr_{A_{L}}\lrrec{C_i\lrbracket{\D_1}} }}_{(n)}$
		\item \label{bob_constraint_left} $\block{\Psi_{\End{\CC\lrrec{S_{t\lrbracket{\D_3}}}}{\D_3}}\lrbracket{\Upsilon_{\lrbracket{B_L}_n\rightarrow C_{B_L}}\circ\Tr_{A\lrbracket{\Bar{B}}_n}\lrrec{C_i\lrbracket{\D_1}}}}_{(n-1)}$
		\item \label{bob_constraint_right} $\block{\Psi_{\End{\CC\lrrec{S_{t\lrbracket{\D_3}}}}{\D_3}}\lrbracket{K_{C_{B_L}\otimes\Tr_{A\lrbracket{B_L}_n\lrbracket{\Bar{B}}_n}\lrrec{C_i\lrbracket{\D_1}}}}}_{(n-1)}$
	\end{enumerate}
are computable in poly$(n)$ time. While we show that $\Tr\lrrec{C_i\lrbracket{\mathcal{D}_1}}$ and $\Tr_{\bar{B}_1\lrbracket{B \bar{B}}_2^n}\lrrec{C_i\lrbracket{\mathcal{D}_1}}$ can be explicitly calculated, based on \autoref{lem:efficient_basis_trafo_Bose_sym} for the other expressions it suffices\footnote{Since we only need the coefficients in the expansion to efficiently perform the transformation.} to show that for any $i\in \lrrec{m(\D_1)^2}$
		\begin{enumerate}[label=\alph*)]
		\item \label{alice_combined} $\Theta_{A_L\rightarrow C_{A_L}}\lrbracket{C_i\lrbracket{\D_1}}$ and $W_{C_{A_L}}\otimes \Tr_{A_L}\lrrec{C_i\lrbracket{\D_1}}$ can be efficiently expressed in terms of a basis to $\End{\CC\lrrec{S_{t\lrbracket{\D_2}}}}{\D_2}$,\,
		\item \label{bob_combined}  $\Upsilon_{\lrbracket{B_L}_n\rightarrow C_{B_L}}\circ\Tr_{A\Bar{B}_n}\lrrec{C_i\lrbracket{\D_1}}$ and $K_{C_{B_L}}\otimes\Tr_{A\lrbracket{B_L}_n\Bar{B}_n}\lrrec{C_i\lrbracket{\D_1}}$ can be efficiently expressed in terms of a basis to $\End{\CC\lrrec{S_{t\lrbracket{\D_3}}}}{\D_3}$.\,
	\end{enumerate}
	In other words, we show that we can efficiently obtain the coefficients in the basis decomposition such that \autoref{lem:efficient_basis_trafo_Bose_sym} applies. \autoref{lem:efficient_basis_trafo_Bose_sym} implies that $\Psi_{ \End{\CC\lrrec{S_{n}}}{\D_1}}\lrbracket{C_i\lrbracket{\D_1}}$ can be efficiently computed. Thus, the positive semidefinite constraint corresponding to (\ref{psd_constraint}) can be efficiently constructed.	Since
	\begin{align}
		\CBose{t}{t'}{n} =\sum_{\vec{i}\,:\,T\lrbracket{\vec{i}}=\vec{t}}\sum_{\vec{j}\,:\,T\lrbracket{\vec{j}}=\vec{t'}}\ket{\vec{i}}\bra{\vec{j}},\,
	\end{align}
	by construction, $\CBose{t}{t'}{n}$ has $\lrvert{T^{-1}\lrbracket{\vec{t}}}\cdot \lrvert{T^{-1}\lrbracket{\vec{t'}}}$ many non-zero entries and $\lrvert{T^{-1}\lrbracket{\vec{t}}\,\cap\,T^{-1}\lrbracket{\vec{t'}}}$ many non-zero diagonal entries. Clearly, if $\vec{t}\neq\vec{t'}$ then $T^{-1}\lrbracket{\vec{t}}\,\cap\,T^{-1}\lrbracket{\vec{t'}}=\emptyset$. Thus, $\CBose{t}{t'}{n}$ has non-zero diagonal entries if and only if $\vec{t}=\vec{t'}$.
		\begin{align}
			\lrvert{T^{-1}\lrbracket{\vec{t}}\,\cap\,T^{-1}\lrbracket{\vec{t'}}} =  \lrvert{T^{-1}(\vec{t})} = \binom{n}{\vec{t}}.\,
		\end{align}
		Due to the chosen normalization, the non-zero entries are ones, thus 
		
		\begin{align}\label{eqn:trace_of_C}
			\begin{split}
			\Tr\lrrec{\CBose{t}{t'}{n}}=
				\begin{cases}
					\binom{n}{\vec{t}} & \text{if }\vec{t} = \vec{t'},\,\\
					0 & \text{else.}
				\end{cases}
			\end{split}
		\end{align}
		Alternatively, due to linearity of the trace
		\begin{align}
		\begin{split}
			\Tr\lrrec{\CBose{t}{t'}{n}} &= \Tr\lrrec{\ket{s_{\vec{t}}}\bra{s_{\vec{t'}}}} = \sum_{\vec{a}:T\lrbracket{\vec{a}}=\vec{t}}\sum_{\vec{b}:T\lrbracket{\vec{b}}=\vec{t'}}\underbrace{\prod_{i=1}^n\Tr\lrrec{\ket{a_i}\bra{b_i}}}_{=\delta_{\vec{a}, \vec{b}}\,\Rightarrow \,\vec{t}=\vec{t'}}\\
			&= \underbrace{\sum_{\vec{a}:T\lrbracket{\vec{a}} = \vec{t}}\sum_{\vec{b}:T\lrbracket{\vec{b}}=\vec{t'}} \delta_{\vec{a}, \vec{b}}}_{\displaystyle\sum_{\vec{a}:T\lrbracket{\vec{a}}=\vec{t}} 1 = \binom{n}{\vec{t}}} = \begin{cases}
					\binom{n}{\vec{t}} & \text{if }\vec{t} = \vec{t'},\,\\
					0 & \text{else.}
				\end{cases}
		\end{split}
		\end{align}
		Thus, (\ref{trace_constraint}) is efficiently computable. Grouping the strings in $\CBose{t}{t'}{n}$ according to their first letter yields
\begin{align}
	\CBose{t}{t'}{n} = \sum_{\substack{p, q\in\lrrec{d_{\HS}}\,:\,\\ t_p\geq 1,\; t'_q\geq 1}} \ket{p}\bra{q}\otimes K^{\vee^{n-1}}_{\vec{t}-T(p),\,\vec{t'}-T(q)},
\end{align}
where $K^{\vee^{n-1}}_{\vec{l},\vec{l'}}:=\sum_{\vec{x}:T(\vec{x})=\vec{l}}\,\sum_{\vec{y}:T(\vec{y})=\vec{l'}}\ket{\vec{x}}\bra{\vec{y}}$ denotes the canonical basis of $\End{\CC\lrrec{S_{n-1}}}{\vee^{n-1}\lrbracket{\HS}}$ from \autoref{prop:first_basis_Bose}, applied at level $n-1$. Together with $\Tr\lrrec{K^{\vee^{n-1}}_{\vec{l},\vec{l'}}}=\binom{n-1}{\vec{l}}\,\delta_{\vec{l},\vec{l'}}$, which is \autoref{eqn:trace_of_C} at level $n-1$, this gives
\begin{align}\label{eqn:partial_trace_first_copy}
	\Tr_{\HS_2^n}\lrrec{\CBose{t}{t'}{n}} = \sum_{\substack{p, q\in\lrrec{d_{\HS}}\,:\, t_p\geq 1,\, t'_q\geq 1,\,\\ \vec{t}-T(p)=\vec{t'}-T(q)}} \binom{n-1}{\vec{t}-T(p)}\;\ket{p}\bra{q}.
\end{align}
The sum contains at most $d_{\HS}^2$ terms. It is empty unless either $\vec{t}=\vec{t'}$, in which case it equals the diagonal operator $\sum_{p\,:\,t_p\geq 1}\binom{n-1}{\vec{t}-T(p)}\ket{p}\bra{p}$, or $\vec{t}-\vec{t'}=T(p)-T(q)$ for a (then unique) pair $p\neq q$, in which case it consists of the single term $\binom{n-1}{\vec{t}-T(p)}\ket{p}\bra{q}$. Since the remaining trace over $\Bar{B}_1$ acts on the fixed-dimensional first copy only, (\ref{objective_function}) is efficiently computable. To prove \ref{alice_combined}, consider the following basis of $\End{\CC\lrrec{S_{t\lrbracket{\D_2}}}}{\D_2}$,
	\begin{align}
		\mathcal{B}_{\D_2}:=\lrbrace{\ket{k}\bra{l}_{C_{A_L}\otimes A_R}\otimes\CBose{t}{t'}{n}}_{\substack{k,l\in\lrrec{\lrvert{C_{A_L}A_R}},\,\\ \vec{t},\vec{t'}\in\mathcal{T}_{(n),\, \lrvert{B\Bar{B}}}}}.
	\end{align}
	Since $\Theta_{A_L\rightarrow C_{A_L}}$, $\Tr_{A_L}$ and $W_{C_{A_L}}$ act on side registers only, we have, for any $i,j\in\lrrec{\lrvert{A}}$,
	\begin{align}
	\begin{split}
		\Theta_{A_L\rightarrow C_{A_L}}\lrbracket{\ket{i}\bra{j}_{A}\otimes \CBose{t}{t'}{n}}&=\Theta_{A_L\rightarrow C_{A_L}}\lrbracket{\ket{i}\bra{j}_{A}}\otimes\CBose{t}{t'}{n},\,\\
		W_{C_{A_L}}\otimes\Tr_{A_L}\lrrec{\ket{i}\bra{j}_{A}\otimes \CBose{t}{t'}{n}}&=\lrbracket{W_{C_{A_L}}\otimes\Tr_{A_L}\lrrec{\ket{i}\bra{j}_{A}}}\otimes\CBose{t}{t'}{n},\,
	\end{split}
	\end{align}
	where $\Theta_{A_L\rightarrow C_{A_L}}\lrbracket{\ket{i}\bra{j}_{A}}$ and $W_{C_{A_L}}\otimes\Tr_{A_L}\lrrec{\ket{i}\bra{j}_{A}}$ are operators on $C_{A_L}\otimes A_R$ of fixed dimension; their coefficients over $\lrbrace{\ket{k}\bra{l}_{C_{A_L}\otimes A_R}}$ are matrix elements of the problem data, e.g.\ $\braket{k\vert\,\Theta_{A_L\rightarrow C_{A_L}}\lrbracket{\ket{i}\bra{j}_{A}}\,\vert l}$, and independent of $n$. Thus, for any $i\in \lrrec{m(\D_1)^2}$, $\Theta_{A_L\rightarrow C_{A_L}}\lrbracket{C_i\lrbracket{\D_1}}$ and $W_{C_{A_L}}\otimes \Tr_{A_L}\lrrec{C_i\lrbracket{\D_1}}$ can be efficiently written in terms of $\mathcal{B}_{\D_2}$, with at most $\lrvert{C_{A_L}A_R}^2$ nonzero coefficients each. This proves \ref{alice_combined}, which together with \autoref{lem:efficient_basis_trafo_Bose_sym} implies that (\ref{alice_constraint_left}) and (\ref{alice_constraint_right}) are computable in poly$(n)$ time. To prove \ref{bob_combined}, we use the branching identity obtained by grouping the strings in $\CBose{t}{t'}{n}$ according to their last letter (cf.\ \autoref{def:reduced_type} and the first-letter analogue above): with $\lrbrace{\KBose{l}{l'}{n-1}}_{\vec{l},\vec{l'}\in\mathcal{T}_{(n-1),\,\lrvert{B\Bar{B}}}}$ the canonical basis of $\End{\CC\lrrec{S_{n-1}}}{\vee^{n-1}\lrbracket{B\Bar{B}}}$,
	\begin{align}\label{eqn:last_letter_branching}
		\CBose{t}{t'}{n} = \sum_{\substack{w, z\in\lrrec{\lrvert{B\Bar{B}}}\,:\,\\ t_{w}\geq 1,\; t'_{z}\geq 1}} \KBose{l}{l'}{n-1}\otimes\ket{w}\bra{z}_{\lrbracket{B\Bar{B}}_n}, \qquad \vec{l}=\vec{t}-T(w),\quad \vec{l'}=\vec{t'}-T(z);
	\end{align}
	see branching rules \cite{sagan2013symmetric} for a mathematical discussion of this relation. The maps acting on the distinguished $n$-th block are of fixed dimension: with $\Lambda_1:=\lrbracket{\Upsilon_{B_L\rightarrow C_{B_L}}\otimes\mathrm{id}_{B_R}}\circ\Tr_{\Bar{B}}$ and $\Lambda_2:=\Tr_{B_L\Bar{B}}$, both acting on $\lrbracket{B\Bar{B}}_n$, the operators of \ref{bob_combined} evaluate on the canonical basis of $\End{\CSn}{\D_1}$ as
	\begin{align}
	\begin{split}
		\Upsilon_{\lrbracket{B_L}_n\rightarrow C_{B_L}}\circ\Tr_{A\Bar{B}_n}\lrrec{\ket{i}\bra{j}_{A}\otimes\CBose{t}{t'}{n}} &= \delta_{i,j}\sum_{\substack{w, z\,:\,\\ t_{w}\geq 1,\, t'_{z}\geq 1}} \KBose{l}{l'}{n-1}\otimes\Lambda_1\lrbracket{\ket{w}\bra{z}},\,\\
		K_{C_{B_L}}\otimes\Tr_{A\lrbracket{B_L}_n\Bar{B}_n}\lrrec{\ket{i}\bra{j}_{A}\otimes\CBose{t}{t'}{n}} &= \delta_{i,j}\sum_{\substack{w, z\,:\,\\ t_{w}\geq 1,\, t'_{z}\geq 1}} \KBose{l}{l'}{n-1}\otimes K_{C_{B_L}}\otimes\Lambda_2\lrbracket{\ket{w}\bra{z}}.\,
	\end{split}
	\end{align}
	Expanding $\Lambda_1\lrbracket{\ket{w}\bra{z}}$ and $K_{C_{B_L}}\otimes\Lambda_2\lrbracket{\ket{w}\bra{z}}$ over the product basis $\lrbrace{\ket{u}\bra{v}_{C_{B_L}\otimes\lrbracket{B_R}_n}}$ --- their coefficients are matrix elements of the problem data, independent of $n$ --- expresses both operators over the canonical basis of $\End{\CC\lrrec{S_{t\lrbracket{\D_3}}}}{\D_3}$,
	\begin{align}
		\mathcal{B}_{\D_3}:=\lrbrace{\KBose{l}{l'}{n-1} \otimes\ket{u}\bra{v}_{C_{B_L}\otimes\lrbracket{B_R}_n}}_{\substack{u,v\in \lrrec{\lrvert{C_{B_L}B_R}},\, \\\vec{l}, \vec{l'}\in\mathcal{T}_{(n-1),\lrvert{B\Bar{B}}}}},\,
	\end{align}
	as a sum over at most $\lrvert{B\Bar{B}}^2$ pairs $\lrbracket{w,z}$, each contributing at most $\lrvert{C_{B_L}B_R}^2$ basis elements, with the reduced types $\vec{l}, \vec{l'}$ computable per pair and coefficients of coinciding basis elements added. In particular, all coefficients can be listed in time polynomial in $n$. This proves \ref{bob_combined}, which together with \autoref{lem:efficient_basis_trafo_Bose_sym} implies that (\ref{bob_constraint_left}) and (\ref{bob_constraint_right}) are computable in poly$(n)$ time.
	\end{proof}
 
\begin{remark}[Closed-form multiplication table]\label{rem:multiplication_table_bose}
Since $\braket{s_{\vec{t'}}\vert s_{\vec{u}}}=\delta_{\vec{t'},\vec{u}}\,\binom{n}{\vec{t'}}$, the canonical basis of \autoref{sec:efficient_acces_bose_operators} multiplies according to
\begin{align}\label{eqn:multiplication_table_bose}
    \CBose{t}{t'}{n}\,\CBose{u}{u'}{n}=\delta_{\vec{t'},\vec{u}}\,\binom{n}{\vec{t'}}\,\CBose{t}{u'}{n},\,
\end{align}
so the structure constants of $\End{\CSn}{\SymH}$ are available in closed form, computable in time polynomial in $n$ for fixed $d_{\HS}$; this makes notion~(ii) of \autoref{rem:symmetry_reduction_taxonomy}, the regular $*$-representation of \cite{Klerk2007ReductionOS}, fully explicit in our setting. Moreover, \autoref{eqn:multiplication_table_bose} shows that the trace-normalized elements $\ket{\tilde{s}_{\vec{t}}}\bra{\tilde{s}_{\vec{t'}}}$ satisfy the matrix-unit relations of $\CC^{m_{(n)}\times m_{(n)}}$, recovering the identification $\End{\CSn}{\SymH}\simeq\CC^{m_{(n)}\times m_{(n)}}$ of \autoref{prop:bose_full_matrix_algebra} directly and exhibiting the regular $*$-representation as $m_{(n)}$ redundant copies of this single block.
\end{remark}
\section{Approximating quantum non-local games via SDP symmetry reductions}
\label{sec:sdp_symmetry_reduction}

Analogous to the previous section, we now employ the full Schur--Weyl decomposition of $\HSn$ into isotypic components indexed by partitions $\lambda\vdash_{d_{\HS}} n$ (see \autoref{prop:schur_weyl_duality}), under which the commutant $\End{\CSn}{\HSn}$ becomes a direct sum of full matrix algebras --- one block per partition --- rather than the single block. In the language of \autoref{rem:symmetry_reduction_taxonomy}, the reduction employed in this section is of type \, (iii): a linear cone-preserving bijection in the sense of \cite{polak2020new}. The algebra $\End{\CSn}{\lrbracket{A_2Q_2\hat{T}}_1^n}$ decomposes into many blocks, and its representative vectors are neither orthogonal nor equipped with orthonormal closed-form coefficients, so we do not construct a unital $*$-isomorphism as in \autoref{sec:explicit_star_isomorphism}; since $\mathrm{SDP}_n\lrbracket{T, V, \pi}$ is a linear SDP, preservation of the linear structure and of positive semidefiniteness suffices (cf.\ \autoref{rem:rep_theory_reading}). Note also that the decomposition below contains exactly one block per partition $\lambda$ --- the multiplicity $m_{\lambda}$ determines the block size, not a repetition count --- and the reduction map is a bijection onto the full direct sum, so no block can be omitted without losing information.

\begin{proposition}[{Block decomposition; see, e.g., \cite{vallentin2009symmetry, gijswijt2009block}}]\label{prop:block_decomposition_symmetric}
We have
\begin{align}
	\End{\CSn}{\HSn}\simeq \bigoplus_{\lambda \vdash_{d_{\HS}}\, n} \CC^{m_{\lambda} \times m_{\lambda}},
\end{align}
where the sum ranges over all partitions $\lambda$ of $n$ into at most $d_{\HS}$ parts and $m_{\lambda}=\dim_{\CC}\lrbracket{\Schurf{\lambda}\HS}$; cf.\ \autoref{prop:schur_weyl_duality}, and \autoref{prop:bose_full_matrix_algebra} for the single-block case $\lambda=(n)$.
\end{proposition}
For a canonical basis $\lrbrace{e_{i}}_{i=1}^{\dHS}$ of $\HS$ and a semistandard tableau $\tau$ of shape $\lambda$, define the polytabloid
\begin{align}
		u_{\tau}= \sum_{\tau'\sim\tau}\sum_{c\in C_{\tau}} \text{sgn}(c)\bigotimes_{y\in Y(\lambda)}e_{\tau'(c(y))},
	\end{align}
	where $C_{\tau}$ denotes the column stabilizer of $\tau$ and $\tau'\sim\tau$ ranges over the row-equivalence class of $\tau$ (see \autoref{sec:representation_theory}). For the one-row shape $\lambda=(n)$ the column stabilizer is trivial and $u_{\tau}$ reduces to the vector $\ket{s_{\vec{t}}}$ of \autoref{prop:type_basis_symmetric_subspace}.

\begin{proposition}[{cf.\ \cite{polak2020new, gijswijt2009block}}]\label{prop:iso_full_amtrix_symmetry}\label{prop:iso_full_matrix_symmetry}
Let $\lambda\vdash_{d_{\HS}} n$ and let $\mathcal{T}_{\lambda, d_{\HS}}$ denote the set of semistandard Young tableaux of shape $Y(\lambda)$ with entries from $\lrrec{d_{\HS}}$, so that $m_{\lambda}=\lrvert{\mathcal{T}_{\lambda, d_{\HS}}}$. Let
\begin{align}
	U_{\lambda} = \lrrec{u_{\tau} \,:\, \tau \in \mathcal{T}_{\lambda, d_{\HS}}} \,\in\, \mathbb{R}^{d_{\HS}^n\times m_{\lambda}}
\end{align}
be the matrix whose columns are the polytabloids $u_{\tau}$. These form a representative set in the sense of \cite{polak2020new}: each $u_{\tau}$ is the image of a primitive idempotent of $\CSn$ in a copy of the Specht module $\SpechtH$, and for any $\tau, \gamma\in \mathcal{T}_{\lambda, d_{\HS}}$ there exists a $\CSn$-module isomorphism between the respective copies mapping $u_{\tau}$ to $u_{\gamma}$. Then the map
	\begin{align}
		\begin{split}
			\psi \,:\,  \End{\CSn}{\HSn} &\rightarrow \bigoplus_{\lambda\vdash_{d_{\HS}}\, n} \CC^{m_{\lambda}\times m_{\lambda}}\\
			 Z &\mapsto \bigoplus_{\lambda\vdash_{d_{\HS}}\, n} \lrbracket{\left\langle Z u_{\tau}, u_{\gamma}\right\rangle}_{\tau, \gamma \in \mathcal{T}_{\lambda, d_{\HS}}} = \bigoplus_{\lambda\vdash_{d_{\HS}}\, n} U_{\lambda}^T ZU_{\lambda}
		\end{split}
	\end{align}
    is a bijection preserving positive semidefiniteness (notion~(iii) of \autoref{rem:symmetry_reduction_taxonomy}); see \autoref{sec:representation_theory} for details.
\end{proposition}

\begin{proposition}\label{prop:iso_sym_with_side_info}
Let $\Al$ be a Hilbert space of finite dimension $d_{\Al}$. Then, in accordance with \autoref{prop:iso_full_amtrix_symmetry} the map
    \begin{align}
    \begin{split}
			\psi \,:\,  \End{\CSn}{\Al\otimes\HSn} &\rightarrow \bigoplus_{\lambda\vdash_{d_{\HS}}\, n} \CC^{m_{\lambda}(\Al, \HS^n)\times m_{\lambda}(\Al, \HS^n)}\\
			 Z &\mapsto \bigoplus_{\lambda\vdash_{d_{\HS}}\, n} \lrbracket{\mathbb{1}_{\Al}\otimes U_{\lambda}}^T Z\lrbracket{\mathbb{1}_{\Al}\otimes U_{\lambda}}
		\end{split}
    \end{align}
    with $m_{\lambda}\lrbracket{\Al, \HS^n}:= d_{\Al}\, m_{\lambda}$, is a bijection preserving positive semidefiniteness.
\end{proposition}

See \autoref{sec:additional_side_info} for a proof with additional quantum side information. 
\begin{proposition}\label{prop:bounding_our_case}
 For $\ESN\lrbracket{\Al\otimes \HS^n}$ in \autoref{prop:iso_sym_with_side_info}, the following bounds hold for the bijection:
\begin{align}\label{eqn:first_bounds_2}
         & m_{\lambda}(\Al, \HS^n) \leq d_{\Al}(n+1)^{\frac{d_{\HS}\lrbracket{d_{\HS}-1}}{2}},\,\forall \lambda\in \text{Par}\,(d_{\HS}, n),\, 
    \end{align}
    such that 
    \begin{align}
    \begin{split}
        m\lrbracket{\Al, \HS^n} &:= \dim\left[\ESN\lrbracket{\Al\otimes \HS^n}\right]\leq d_{\Al}^2(n+1)^{d_{\HS}^2}
    \end{split}
    \end{align}
\end{proposition}
\begin{proof}
    By \autoref{prop:iso_sym_with_side_info}, $\ESN\lrbracket{\Al\otimes \HS^n}\simeq\bigoplus_{\lambda\vdash_{d_{\HS}} n}\CC^{m_{\lambda}(\Al, \HS^n)\times m_{\lambda}(\Al, \HS^n)}$ with
    \begin{align}
         m_{\lambda}(\Al, \HS^n) = d_{\Al}\lrvert{\mathcal{T}_{\lambda, d_{\HS}}} \leq d_{\Al}(n+1)^{\frac{d_{\HS}\lrbracket{d_{\HS}-1}}{2}}
    \end{align}
    by Weyl's dimension formula (cf.\ the notation paragraph of \autoref{sec:symmetric_subspace_methods}); this proves \autoref{eqn:first_bounds_2}. The number of summands is
    \begin{align}\label{eqn:first_bounds}
        \begin{split}
            &\lrvert{\text{Par}\,(d_{\HS}, n)}\leq (n+1)^{d_{\HS}}.
        \end{split}
    \end{align}
    Since the dimension of a direct sum of matrix algebras is the sum of the squared block sizes,
    \begin{align}
        \begin{split}
            m\lrbracket{\Al, \HS^n} &= \sum_{\lambda \vdash_{d_{\HS}} n} m_{\lambda}\lrbracket{\Al, \HS^n}^2 \leq \lrvert{\text{Par}\,(d_{\HS}, n)}\cdot\max_{\lambda\vdash_{d_{\HS}} n} m_{\lambda}\lrbracket{\Al, \HS^n}^2\\
            &\stackrel{\autoref{eqn:first_bounds}}{\leq} (n+1)^{d_{\HS}}\, d_{\Al}^2\lrrec{(n+1)^{\frac{d_{\HS}\lrbracket{d_{\HS}-1}}{2}}}^2=d_{\Al}^2(n+1)^{d_{\HS}^2}.\,
        \end{split}
    \end{align}
\end{proof}

Recall the SDP hierarchy in \autoref{sec:non_local_games_as_cbos} approximating fixed size free non-local games from above. We begin by demonstrating that the analysis can be restricted to the corresponding invariant algebras. Specifically, the following lemma establishes that the $S_n$-symmetry is preserved under the given constraints. Subsequently, we apply \autoref{prop:iso_sym_with_side_info} to obtain a direct sum decomposition into full matrix algebras.

\begin{lemma}[Sufficiency in restricting optimization of $\mathrm{SDP}_n(T, V,\pi)$  to a dimension bounded invariant algebra]\label{lem:restriction_to_invariant_subspace}
    Let $n\in\mathbb{N}_{\geq 1}$. If \begin{align}
        \rho_{(A_1Q_1T)(A_2Q_2\hat{T})_1^n}\in\ESN\lrbracket{(A_1Q_1T)\otimes (A_2Q_2\hat{T})_1^n},\,
    \end{align} then,
    \begin{align}\label{eqn:sym_subspace_1}
    \begin{split}
        \Tr_{A_1}\left[\rho_{(A_1Q_1T)(A_2Q_2\hat{T})_1^n}\right] \in  \ESN\lrbracket{(Q_1T)\otimes (A_2Q_2\hat{T})_1^n},\,\\
        \sum_{q_1}\pi_1(q_1)\ket{q_1}\bra{q_1}_{Q_1} \otimes \Tr_{(A_1Q_1)}\left[\rho_{(A_1Q_1T)(A_2Q_2\hat{T})_1^n}\right]&\in \ESN\lrbracket{(Q_1T)\otimes (A_2Q_2\hat{T})_1^n}
    \end{split}
    \end{align}
    together with 
    \begin{align}\label{eqn:sym_subspace_2}
    \begin{split}
        &\Tr_{(A_2)_1}\left[\rho_{(A_2Q_2\hat{T})_1^n}\right] \in  \text{End}_{\CC\lrrec{S_{n-1}}\,}\lrbracket{(Q_2\hat{T})_1\otimes(A_2Q_2\hat{T})_2^n},\,\\
        &\left(\sum_{q_2}\pi_2(q_2)\ket{q_2}\bra{q_2}_{Q_2} \otimes\frac{\mathbb{1}_{\hat{T}}}{\lvert T\rvert}\right)_{\!(Q_2\hat{T})_1}\otimes \Tr_{(A_2Q_2\hat{T})_1}\left[ \rho_{(A_2Q_2\hat{T})_1^n}\right] \in\text{End}_{\CC\lrrec{S_{n-1}}\,}\lrbracket{(Q_2\hat{T})_1\otimes(A_2Q_2\hat{T})_2^n},\,
    \end{split}
    \end{align}
    where $\rho_{(A_2Q_2\hat{T})_1^n}:=\Tr_{A_1Q_1T}\lrrec{\rho_{(A_1Q_1T)(A_2Q_2\hat{T})_1^n}}$. Furthermore,  
        \begin{align}
            \dim\lrbracket{\ESN\lrbracket{(A_1Q_1T)\otimes(A_2Q_2\hat{T})_1^n}}\leq \lrvert{AQT}^2(n+1)^{\lrvert{AQT}^2}.\,
        \end{align}
\end{lemma}
\begin{proof}
    See \autoref{sec:restriction_to_invariant_subspace}.
\end{proof}

\begin{remark}
   It is immediate that the constraints  \begin{align}\label{eqn:sym_trace_pos_con}
        \begin{array}{ccc}
            \Trr{}{\rhoge}=1,\, & & \rhoge \succcurlyeq 0,\,
        \end{array} 
\end{align} can be formulated in terms of any basis of $\ESN\lrbracket{(A_1Q_1T)\otimes (A_2Q_2\hat{T})_1^n}$ --- the trace is linear in the coefficients, and positive semidefiniteness is imposed block-wise via \autoref{prop:iso_sym_with_side_info}.
\end{remark}

\begin{remark}
    Note that it does not matter which of the $(A_2Q_2\hat{T})_1^n$ systems is traced out and that the ordering in the tensor product is also irrelevant. Thus, \autoref{prop:iso_sym_with_side_info} can be applied to the spaces in \autoref{lem:restriction_to_invariant_subspace}. 
\end{remark}


\subsection{The block decomposition}\label{sec:sym_block_decomposition} 

We prove \autoref{thm:SDP_sym_complexity_results} in two steps. First, we demonstrate that determining the optimal winning probability of a fixed-size free two-player non-local game up to an additive error can be formulated as a collection of semidefinite programs with sizes specified in \autoref{thm:SDP_sym_complexity_results}. Then, we establish that this formulation can be constructed efficiently.

\begin{lemma}\label{lem:help_lemma_sym_SDP}
    There exists an SDP hierarchy indexed by $n$ approximating the value of a two-player free non-local game with $\lrvert{A}$ answers, $\lrvert{Q}$ questions and quantum assistance of size $\lrvert{T}$ up to an additive error with $\text{poly}(n)$-many variables, $\text{poly}(n)$-many constraints and matrices of size at most $\text{poly}(n)$ as given in \autoref{thm:SDP_sym_complexity_results}. 
\end{lemma}

\begin{proof}\label{proof:efficient_sdp}
As in \cite{chee2023efficient}, we introduce notation that avoids redundancy between conceptually related objects. Firstly, let 
\begin{align}
    \mathcal{D}:=\lrbrace{A_1Q_1T\otimes\lrbracket{A_2Q_2\hat{T}}_1^n,\,\, Q_1T\otimes\lrbracket{A_2Q_2\hat{T}}_1^n,\,\, (Q_2\hat{T})_1\otimes\lrbracket{A_2Q_2\hat{T}}_2^n}
\end{align}
be the set collecting the Hilbert spaces of relevance to our setting. Next, define a function 
\begin{align}
\begin{split}
    t \, : \, \mathcal{D} &\,\rightarrow \, \mathbb{N}\\
    \mathcal{D}_k &\,\mapsto\, \begin{cases}
    \begin{array}{cc}
         n-1 &  \text{ for } k=3,\,\\
         n & \text{ otherwise. }
    \end{array} 
    \end{cases}
\end{split}
\end{align}
Furthermore, let
\begin{align}
    \begin{split}
    \mathbb{1}_k :=\begin{cases}
    \begin{array}{cc}
         \mathbb{1}_{A_1Q_1T} & \text{ for }  k=1,\,\\
         \mathbb{1}_{Q_1T} & \text{ for } k=2,\,\\
         \mathbb{1}_{\lrbracket{Q_2\hat{T}}_1} & \text{ for } k=3
    \end{array} 
    \end{cases}
\end{split}
\end{align}
denote the corresponding identity matrices. To avoid notational issues we will denote $\text{End}_{\CC\lrrec{S_{t\lrbracket{\mathcal{D}_k}}}}(\cdot)$ by $\text{End}^{S_{t\lrbracket{\mathcal{D}_k}}}(\cdot)$. In a similar spirit to \autoref{prop:iso_sym_with_side_info}, for $k\in\lrrec{3}$ let
\begin{align}\label{eqn:our_bijection_applied}
\begin{split}
    \Psi_{\lrbracket{\cdot}} \, : \; \text{End}^{S_{t\lrbracket{\mathcal{D}_k}}}\lrbracket{\mathcal{D}_k} &\,\rightarrow \, \bigoplus_{\lambda \in \text{Par}\lrbracket{\lrvert{A_2Q_2\hat{T}}, t\lrbracket{\mathcal{D}_k}}} \mathbb{C}^{m_{\lambda}\lrbracket{\mathcal{D}_k}\times m_{\lambda}\lrbracket{\mathcal{D}_k}}\\
    Z &\,\mapsto\, \bigoplus_{\lambda \in \text{Par}\lrbracket{\lrvert{A_2Q_2\hat{T}}, t\lrbracket{\mathcal{D}_k}}}\lrbracket{\mathbb{1}_k\otimes U^T_{\lambda}}Z\lrbracket{\mathbb{1}_k\otimes U_{\lambda}}
\end{split}
\end{align}
denote the bijection obtained by applying \autoref{prop:iso_sym_with_side_info} to $\mathcal{D}_k$; it preserves positive semidefiniteness, and $m_{\lambda}\lrbracket{\mathcal{D}_k}$ is the quantity $m_{\lambda}\lrbracket{\Al, \HS^{t\lrbracket{\mathcal{D}_k}}}$ of \autoref{prop:iso_sym_with_side_info} for the respective side system. We write $\left[\!\left[\Psi_{\text{End}^{S_{t\lrbracket{\mathcal{D}_k}}}(\mathcal{D}_k)}(Z)\right]\!\right]_{\lambda}$ for the block of $\Psi_{\lrbracket{\cdot}}(Z)$ labeled by $\lambda$, and set
\begin{align}
    \begin{split}
       m(\mathcal{D}_k):= \sum_{\lambda \in \text{Par}\lrbracket{\lrvert{A_2Q_2\hat{T}}, t\lrbracket{\mathcal{D}_k}}} m^2_{\lambda}\lrbracket{\mathcal{D}_k} = \dim_{\CC}\lrbracket{\text{End}^{S_{t\lrbracket{\mathcal{D}_k}}}\lrbracket{\mathcal{D}_k}} 
        \stackrel{\autoref{prop:bounding_our_case}}{\leq} \lrvert{A_1Q_1T}^2(t\lrbracket{\mathcal{D}_k}+1)^{\lrvert{A_2Q_2\hat{T}}^2}.\,
    \end{split}
\end{align}
Correspondingly, let $\mathcal{C}\lrbracket{\mathcal{D}_k}= \lrbrace{C_i\lrbracket{\D_k}}_{i=1}^{m(\mathcal{D}_k)}$ denote the canonical basis\footnote{See \autoref{sec:proof_of_efficient_trafo_sdp_sym}.} of $\text{End}^{S_{t\lrbracket{\mathcal{D}_k}}}(\mathcal{D}_k)$, obtained from a canonical basis of $\mathcal{D}_k$ by averaging over the $S_{t\lrbracket{\mathcal{D}_k}}$-orbits.

With \autoref{lem:restriction_to_invariant_subspace} we can write $\mathrm{SDP}_{(n)}\lrbracket{T, V, \pi}$ from \autoref{sec:non_local_games_as_cbos} in terms of a direct sum of full matrix algebras by applying the bijection above. We express  $\mathrm{SDP}_{(n)}\lrbracket{T, V, \pi}$ in terms of the canonical bases of $\text{End}^{S_{t\lrbracket{\mathcal{D}_k}}}\lrbracket{\mathcal{D}_k}$, i.e.\ for any $Z\in \text{End}^{S_{t\lrbracket{\mathcal{D}_k}}}\lrbracket{\mathcal{D}_k}$ there exist $m\lrbracket{\D_k}$ coefficients $x_i\in\mathbb{C}$ such that
\begin{align}
    Z = \sum_{i=1}^{m\lrbracket{\D_k}}x_i C_i\lrbracket{\D_k},\,
\end{align}
and then apply the bijection given in \autoref{eqn:our_bijection_applied}. Note that due to the linearity of $\Psi_{(\cdot)}$ we can write for any $k\in [3]$
\begin{align}
    \block{\Psi_{\text{End}^{S_{t\lrbracket{\mathcal{D}_k}}}(\mathcal{D}_k)}\lrbracket{\sum_{i=1}^{m\lrbracket{\D_k}}x_i\cdot C_i\lrbracket{\D_k}}}_{\lambda} = \sum_{i=1}^{m\lrbracket{\D_k}}x_i \block{\Psi_{\text{End}^{S_{t\lrbracket{\mathcal{D}_k}}}(\mathcal{D}_k)}\lrbracket{C_i\lrbracket{\D_k}}}_{\lambda}.\,
\end{align}
Thus, while additionally exploiting the invariance of the trace under a basis transformation we arrive at 
\begin{align}\label{eqn:efficient_sdp_free_game}
\begin{split}
    & \Psi\lrbracket{\mathrm{SDP}_{(n)}\lrbracket{T, V, \pi}} = \lrvert{T} \max_{\lrbrace{x_i}_{i=1}^{m\lrbracket{\D_1}}} \sum_{i=1}^{m(\mathcal{D}_1)} x_i\cdot \Tr\lrrec{\left(V_{A_1 A_2Q_1 Q_2}\otimes S_{T\hat{T}}\right) C_i\lrbracket{\D_1}_{(A_1Q_1T)(A_2Q_2\hat{T})_1}}\\
    \text{s.t. }\hspace{0.5cm}
    &C_i\lrbracket{\D_1}_{(A_1Q_1T)(A_2Q_2\hat{T})_1} := \Trr{(A_2Q_2\hat{T})_2^n}{C_i\lrbracket{\D_1}},\;\; i\in\lrrec{m\lrbracket{\D_1}},\,\\
    &\sum_{i=1}^{m\lrbracket{\D_1}} x_i\cdot \left[\!\left[\Psi_{\text{End}^{S_{t\lrbracket{\mathcal{D}_1}}}(\mathcal{D}_1)}\lrbracket{C_i\lrbracket{\D_1}}\right]\!\right]_{\lambda} \succcurlyeq 0,\,\hspace{0.5cm} \forall \lambda\in\text{Par}\lrbracket{\lrvert{A_2Q_2\hat{T}}, n},\,\\
    &\sum_{i=1}^{m\lrbracket{\D_1}} x_i \cdot\Tr\lrrec{C_i\lrbracket{\D_1}} = 1,\,\\
    & \sum_{i=1}^{m\lrbracket{\D_1}} x_i\cdot \left[\!\left[\Psi_{\text{End}^{S_{t\lrbracket{\mathcal{D}_2}}}(\mathcal{D}_2)}\lrbracket{\Trr{A_1}{C_i\lrbracket{\D_1}}}\right]\!\right]_{\lambda}\\ 
    &\hspace{1cm}=\sum_{i=1}^{m\lrbracket{\D_1}} x_i\cdot \left[\!\left[\Psi_{\text{End}^{S_{t\lrbracket{\mathcal{D}_2}}}(\mathcal{D}_2)}\lrbracket{\sum_{q_1}\pi_1(q_1)\ket{q_1}\bra{q_1}_{Q_1} \otimes\Trr{A_1Q_1}{C_i\lrbracket{\D_1}}}\right]\!\right]_{\lambda},\,\\
    &\hspace{2cm}\forall \lambda\in\text{Par}\lrbracket{\lrvert{A_2Q_2\hat{T}}, n},\,\\
    \\
    & \sum_{i=1}^{m\lrbracket{\D_1}} x_i\cdot \left[\!\left[\Psi_{\text{End}^{S_{t\lrbracket{\mathcal{D}_3}}}(\mathcal{D}_3)}\lrbracket{\Trr{(A_1Q_1T)(A_2)_1}{C_i\lrbracket{\D_1}}}\right]\!\right]_{\lambda}\\ 
    &\hspace{1cm}=\sum_{i=1}^{m\lrbracket{\D_1}} x_i\cdot \left[\!\left[\Psi_{\text{End}^{S_{t\lrbracket{\mathcal{D}_3}}}(\mathcal{D}_3)}\lrbracket{\sum_{q_2}\pi_2(q_2)\ket{q_2}\bra{q_2}_{Q_2}\otimes\frac{\Id_{\hat{T}}}{\lrvert{T}}\otimes\Trr{(A_1Q_1T)(A_2Q_2\hat{T})_1}{C_i\lrbracket{\D_1}}}\right]\!\right]_{\lambda},\,\\
    &\hspace{2cm}\forall \lambda\in\text{Par}\lrbracket{\lrvert{A_2Q_2\hat{T}}, n-1},\,\\
\end{split}
\end{align}
where we optimize over 
\begin{align}
    m\lrbracket{\D_1}\stackrel{\autoref{prop:bounding_our_case}}{\leq} \lrvert{A_1Q_1T}^2(n+1)^{\lrvert{A_2Q_2\hat{T}}^2}
\end{align}
many variables and have specified  
\begin{align}
    2\,\lrvert{\text{Par}\lrbracket{\lrvert{A_2Q_2\hat{T}}, n}} + \lrvert{\text{Par}\lrbracket{\lrvert{A_2Q_2\hat{T}}, n-1}} + 1
\end{align}
many constraints, of which there are 
\begin{align}
    \begin{split}
       \lrvert{\text{Par}\lrbracket{\lrvert{A_2Q_2\hat{T}}, n}} \stackrel{\autoref{eqn:first_bounds}}{\leq} (n+1)^{\lrvert{A_2Q_2\hat{T}}}
    \end{split}
\end{align}
many PSD constraints. Furthermore, the Hermitian\footnote{Positive semidefinite matrices are Hermitian; with the real representative vectors of \autoref{prop:iso_full_matrix_symmetry} the constraint matrices are moreover real symmetric.} matrices involved are at most of the size $ m_{\lambda}\lrbracket{\D_1} \times m_{\lambda}\lrbracket{\D_1}$, with
\begin{align}
    m_{\lambda}\lrbracket{\D_1} \stackrel{\autoref{prop:bounding_our_case}}{\leq} \lrbracket{\lrvert{A_1Q_1T}(n+1)^{\lrbracket{\frac{\lrvert{A_2Q_2\hat{T}}\lrbracket{\lrvert{A_2Q_2\hat{T}}-1}}{2}}}}.\,
\end{align}
Setting $\lrvert{A_1}=\lrvert{A_2}$, $\lrvert{Q_1}=\lrvert{Q_2}$ and $\lrvert{T}=\lrvert{\hat{T}}$ concludes the proof.   
\end{proof}

Lastly, directly following the arguments in \cite{chee2023efficient}, we prove that the computation of the reduced problem, that is, determining the change-of-basis matrices and coefficients, can be done efficiently w.r.t.\ $n$ when the dimensions of the systems involved are fixed. 

\begin{lemma}\label{lem:efficiency_of_trafo_sdp_sym}
The transformation in \autoref{lem:help_lemma_sym_SDP} yielding \autoref{thm:SDP_sym_complexity_results} can be computed efficiently with respect to $n$.
\end{lemma}
The proof is deferred to \autoref{sec:proof_of_efficient_trafo_sdp_sym}. The proof of \autoref{thm:SDP_sym_complexity_results} then follows from \autoref{lem:efficiency_of_trafo_sdp_sym} and \autoref{lem:help_lemma_sym_SDP}.

\begin{remark}[Classical-quantum structure and the size of the reduction]\label{rem:cq_complexity}
The reduction above does not exploit the classical-quantum restriction of \autoref{prop:restriction_to_cq_states}. Doing so is possible: the optimization in \autoref{eqn:efficient_sdp_free_game} may be restricted to the elements of $\text{End}^{S_{t\lrbracket{\mathcal{D}_1}}}\lrbracket{\mathcal{D}_1}$ that are in addition diagonal on the classical registers, and these form the $S_n$-invariant subspace of $\mathrm{Diag}\lrbracket{A_1Q_1}\otimes\mathcal{B}\lrbracket{\HS_T}\otimes W^{\otimes n}$ with $W:=\mathrm{Diag}\lrbracket{A_2Q_2}\otimes\mathcal{B}\lrbracket{\HS_{\hat{T}}}$, of dimension
$\lrvert{A_1}\lrvert{Q_1}\lrvert{T}^2\binom{n+\lrvert{A_2}\lrvert{Q_2}\lrvert{T}^2-1}{n}\leq\lrvert{A_1}\lrvert{Q_1}\lrvert{T}^2\lrbracket{n+1}^{\lrvert{A_2}\lrvert{Q_2}\lrvert{T}^2-1}$.
Compared to the count $m\lrbracket{\D_1}\leq\lrvert{A_1Q_1T}^2\lrbracket{n+1}^{\lrvert{A_2Q_2\hat{T}}^2}$ of \autoref{prop:bounding_our_case}, the degree of the polynomial growth in $n$ would thus improve from $\lrvert{A_2Q_2\hat{T}}^2$ to $\lrvert{A_2}\lrvert{Q_2}\lrvert{T}^2-1$, i.e.\ roughly by a factor of $\lrvert{A}\lrvert{Q}$ in the exponent. Realizing this reduction, however, requires the block decomposition of the smaller algebra --- a refinement of \autoref{prop:iso_sym_with_side_info} that organizes the $S_{t\lrbracket{\mathcal{D}_k}}$-orbits of classical configurations into types and applies Schur--Weyl duality within the resulting Young subgroups.
An analogous refinement is available for the Bose-symmetric hierarchy via the occupation block structure of \autoref{prop:bose_cq_restriction}(ii).
Since \autoref{thm:SDP_sym_complexity_results} and \autoref{lem:efficiency_of_trafo_sdp_sym} concern the scaling in the hierarchy level $n$ --- equivalently, in the inverse additive error $\epsilon^{-1}$ --- for fixed $\lrvert{A}$, $\lrvert{Q}$ and $\lrvert{T}$, and this scaling is polynomial for either parametrization, we do not pursue this refinement here; it improves the degree, not the polynomial character, of the complexity bounds, and is chiefly of interest for practical implementations.
\end{remark}

\printbibliography 

\newpage



\clearpage

\phantomsection
\addcontentsline{toc}{section}{Supplementary Material}

\addtocontents{toc}{\protect\hidesuppinmaintoc}

\appendix
\pagestyle{suppheader}

\startcontents[supp]

\renewcommand{\hidesuppinmaintoc}{}

\thispagestyle{empty}              

\begingroup
  \centering
  \vspace*{0.5cm}
  {\LARGE\bfseries Supplementary Material\par}
  \vspace{1.8em}
  {\large for the article\par}
  \vspace{0.7em}
  {\large\itshape Approximating fixed size quantum correlations in polynomial time\par}
  \vspace{2.6em}
  {\large Julius A. Zeiss \quad Gereon Kossmann \quad Omar Fawzi \quad Mario Berta\par}
  \par
\endgroup

\vspace{1.8em}
\centerline{\rule{0.5\textwidth}{0.4pt}}
\vspace{1.8em}

\noindent{\bfseries Contents of the supplementary material}\par
\printcontents[supp]{}{0}{\setcounter{tocdepth}{1}}

\clearpage


\section{Semidefinite programming}\label{sec:sdp}

We follow \cite{boyd2004convex} and refer the reader there for a more comprehensive introduction to semidefinite programs.  A conic program is a convex optimization problem whose inequalities are expressed with respect to a cone-induced partial order. Semidefinite programs are conic programs over the cone of positive semidefinite matrices. In an SDP in inequality form, one minimizes a linear function of a variable $x=\lrbracket{x_1, \ldots, x_m}\in \mathbb{R}^m$ subject to an affine linear matrix inequality,
\begin{align} \label{eqn:SDP_ineq_form}
    \begin{array}{ll}
        \min_{x\in \mathbb{R}^m} & c^\top x\\
        \text{s.t.} & F_0 + x_1F_1 + \cdots + x_mF_m \succeq 0.
    \end{array}
\end{align}
Here the symmetric matrices $F_0,\ldots, F_m \in \mathbb{R}^{n\times n}$ and the vector $c\in \mathbb{R}^m$ specify the problem. We say that this inequality-form SDP has $m$ scalar decision variables. Another standard form of an SDP is written in terms of a symmetric matrix variable,
\begin{align}\label{eqn:SDP_standard_form}
    \begin{array}{ll}
        \min_{X\in \mathbb{R}^{n\times n}} & \left\langle C, X\right\rangle \\
        \text{s.t.} & \left\langle A_i, X\right\rangle=b_i, \quad i=1,\ldots, k, \\
        & X\succeq 0,
    \end{array}
\end{align}
where $\langle \cdot, \cdot \rangle$ is the trace (Frobenius) inner product. The symmetric matrices $C, A_1,\ldots, A_k\in \mathbb{R}^{n\times n}$ and $b=(b_1,\ldots,b_k)\in\mathbb{R}^k$ specify the problem. A matrix $X\in \mathbb{R}^{n\times n}$ is a feasible solution to \autoref{eqn:SDP_standard_form} if it fulfills the conditions $\left\langle A_i, X\right\rangle=b_i$ and $X\succeq 0$. A feasible solution $X$ is called optimal if for all other feasible solutions $X'$ we have $ \left\langle C, X\right\rangle \leq  \left\langle C, X'\right\rangle$. We say that the SDP has a matrix variable of size $n \times n$ or that the SDP involves square matrices of order $n$. Alternatively, one can specify the number of independent entries of $X$. In general, a real symmetric matrix, and in particular a real positive semidefinite matrix, $X\in\mathbb{R}^{n\times n}$ has $\frac{n\lrbracket{n+1}}{2}$ independent real scalar entries.

Consider a complex Hermitian matrix $Y\in\mathbb{C}^{n\times n}$. We have
\begin{align}
    Y\succeq 0 \quad\Leftrightarrow\quad \Phi(Y):=\begin{bmatrix}
        \operatorname{Re}(Y) & -\operatorname{Im}(Y) \\
        \operatorname{Im}(Y) & \operatorname{Re}(Y)
    \end{bmatrix}\succeq 0.
\end{align}
Thus, an SDP whose variable is a complex Hermitian positive semidefinite matrix $Y\in\mathbb{C}^{n\times n}$ can be reformulated using a real symmetric positive semidefinite matrix $\Phi(Y)\in\mathbb{R}^{2n\times 2n}$ together with the linear constraints that enforce this block structure. However, the image of $\Phi$ is not the entire cone of real positive semidefinite matrices of size $2n\times 2n$. While a general real symmetric $2n\times 2n$ matrix has $2n^2+n$ independent real scalar entries, a complex Hermitian positive semidefinite matrix $Y\in\mathbb{C}^{n\times n}$ has only $n^2$ independent real scalar entries. By construction, $\operatorname{Re}(Y)$ is symmetric and has $\frac{n\lrbracket{n+1}}{2}$ independent real scalar entries, while $\operatorname{Im}(Y)$ is skew-symmetric and has $\frac{n\lrbracket{n-1}}{2}$ independent real scalar entries. Thus, the structured real matrix $\Phi(Y)$ has
\begin{align}
    \frac{n\lrbracket{n+1}}{2} + \frac{n\lrbracket{n-1}}{2} = n^2
\end{align}
independent real scalar entries. Note that, in \autoref{eqn:SDP_standard_form}, independent equality constraints further reduce the dimension of the feasible affine space by their rank. In summary, for SDPs whose optimization variable is a complex Hermitian positive semidefinite $n\times n$ matrix, the number of real scalar degrees of freedom is at most $n^2$ before any problem-specific equality constraints, and no larger after imposing them.

\subsection{Symmetry-invariant SDPs}\label{sec:symmetry_invariant_sdp}

We follow \cite{vallentin2009symmetry} and briefly recall how semidefinite programs can be invariant under the action of a symmetry group. Let $G$ be a finite group acting on a finite set $\mathcal{Z}$ by
\begin{align}
     (g,z)\mapsto gz, \qquad g\in G,\ z\in \mathcal{Z}.
\end{align}
This action extends naturally to an action on pairs $(y,z)\in\mathcal{Z}\times \mathcal{Z}$ by
\begin{align}
     (g,(y,z))\mapsto (gy,gz).
\end{align}
Consequently, $G$ acts on square matrices whose rows and columns are indexed by $\mathcal{Z}$. More explicitly, for a $\mathcal{Z}\times\mathcal{Z}$ matrix $M$, we write
\begin{align}
    (gM)(y,z) := M(g^{-1}y,g^{-1}z).
\end{align}
A matrix $M$ is called $G$-invariant if $gM=M$ for all $g\in G$. Consider the SDP
\begin{align}\label{eqn:SDP_standard_form_2}
    \begin{array}{ll}
        \min_{X\in \mathbb{C}^{n\times n}} & \left\langle C, X\right\rangle \\
        \text{s.t.} & \left\langle A_i, X\right\rangle=b_i, \quad i=1,\ldots, k, \\
        & X\succeq 0,
    \end{array}
\end{align}
and assume that an optimal solution exists. We say that the SDP in \autoref{eqn:SDP_standard_form_2} is invariant under $G$ if, for every feasible solution $X$ and every $g\in G$, the matrix $gX$ is again feasible and the objective value is preserved, that is,
\begin{align}
     \left\langle C,gX\right\rangle=\left\langle C,X\right\rangle.
\end{align}
It is important to distinguish invariance of the SDP from invariance of its individual feasible points: a $G$-invariant SDP need not have the property that every feasible solution is itself $G$-invariant.

Nevertheless, convexity implies that the optimization in \autoref{eqn:SDP_standard_form_2} may be restricted to the subspace $\mathcal{B}$ of all $G$-invariant $\mathcal{Z}\times\mathcal{Z}$ complex matrices. Indeed, for any feasible $X$, define its group average by
\begin{align}
    R(X):=\frac{1}{\lrvert{G}}\sum_{g\in G} gX .
\end{align}
Then $R(X)\in\mathcal{B}$, since $gR(X)=R(X)$ for every $g\in G$. Moreover, since $X$ is feasible and the SDP is $G$-invariant, every matrix $gX$ is feasible. By convexity of the feasible set, their average $R(X)$ is therefore feasible as well. Finally, again by $G$-invariance of the SDP, the objective value is preserved:
\begin{align}
    \left\langle C,R(X)\right\rangle=\left\langle C,X\right\rangle .
\end{align}
Thus every feasible solution can be replaced by a $G$-invariant feasible solution with the same objective value. In particular, if an optimal solution exists, then there also exists an optimal solution which is $G$-invariant.

The key point is that $G$-invariance of the SDP means that the optimization problem itself is symmetric; it does not mean that every feasible solution must be symmetric. This differs from the SDPs considered in this work, where we explicitly impose that every feasible solution is $G$-invariant for $G=S_n$. In that case the feasible set is fixed pointwise by the group action, and hence the SDP is automatically $G$-invariant in the above sense.

\section{Constrained separability problems}\label{sec:cbo}

Constrained separability problems are pervasive in quantum information theory and constitute a subclass of the constrained bilinear optimization problems. For instance, they arise in the context of jointly constrained semidefinite bilinear programming, as developed by \cite{huber2019jointly}, which utilizes non-commutative extensions of the classical branch-and-bound algorithm, and in the unbounded variants described in \cite{berta2016quantum}. In this work, we adopt and adapt the formulation presented in \cite{berta2021semidefinite, ohst2024characterising}.

\begin{definition}\label{def:constraint_bilinear_optimization}
Let $\mathcal{H}_A=\HS_{A_L}\otimes\HS_{A_R}$ and $\HS_B=\HS_{B_L}\otimes\HS_{B_R}$ be Hilbert spaces. The composite system is described by $\mathcal{H}_A \otimes \mathcal{H}_B$. For a Hermitian cost Hamiltonian $G_{AB} \in \mathcal{B}(\mathcal{H}_A \otimes \mathcal{H}_B)$ with $\left\lVert G_{AB}\right\rVert_{\infty}\leq 1$, consider the constrained optimization over product states
{\allowdisplaybreaks
\begin{align}\label{eq:optimization_task}
        \begin{array}{lllll}
    &&&\displaystyle \mathrm{cPROD}(G) := \max_{\rho_{AB}\in \mathcal{B}(\HS_A\otimes\HS_B)} \Tr[G_{AB}\,\rho_{AB}]\\
    &&&\\
    &\text{s.t.}&& \text{Product states:} & \\
     &&&\\
    &&&\rho_{AB}=\rho_A\otimes\rho_B,\,  & \\
    &&&\rho_A\succeq 0,\, & \rho_B\succeq 0,\,\\
    &&& \Tr\left[\rho_A\right]=1,\, & \Tr\left[\rho_B\right]=1,\,\\
    &&&&\\
    &&&\text{Alice's linear constraints:} &\\
    &&&&\\
    &&& \Theta_{A_L\rightarrow C_{A_L}}\lrbracket{\rho_A} = W_{C_{A_L}}\otimes \rho_{A_R},\, & \Omega_{A\rightarrow A}\lrbracket{\rho_{A}}=\rho_{A} ,\,\\
    \\
     &&&\text{Bob's linear constraints:} &\\
    &&&&\\
    &&& \Upsilon_{B_L\rightarrow C_{B_L}}\lrbracket{\rho_B} = K_{C_{B_L}}\otimes \rho_{B_R},\, & \Xi_{B\rightarrow B}\lrbracket{\rho_{B}}=\rho_{B}, \\
        \end{array}
    \end{align}}
where the dependence on $G$ collects the freedom in specifying the cost Hamiltonian $G_{AB}$, the linear maps $\Theta_{A_L\rightarrow C_{A_L}}$, $\Omega_{A\rightarrow A}$, $\Upsilon_{B_L\rightarrow C_{B_L}}$ and $\Xi_{B\rightarrow B}$ the fixed operators $W_{C_{A_L}}, K_{C_{B_L}}$, and the dimensions involved, for the specific problem under consideration.  
\end{definition}
Since the objective is linear in $\rho_{AB}$, we can equivalently optimize over the convex hull of the constrained product states along the $A\,{:}\,B$-partitioning, i.e.\ over separable states whose components satisfy the constraints.

\begin{proposition}[Constrained separability problem]\label{lem:equivalent_cbo_sep_prod}
Let $G_{AB} \in \mathcal{B}(\mathcal{H}_A \otimes \mathcal{H}_B)$ with $\left\lVert G_{AB}\right\rVert_{\infty}\leq 1$ be a Hermitian cost Hamiltonian. Consider the bilinear optimization problem
    {\allowdisplaybreaks
    \begin{align}\label{eq:bilinear_optimisation_problem_sep}
        \begin{array}{lllll}
    &&&\displaystyle \mathrm{cSEP}(G) := \max_{\rho_{AB}\in \mathcal{B}(\HS_A\otimes\HS_B)} \Tr[G_{AB}\,\rho_{AB}]\\
    &&&&\\
    &\text{s.t.}&& \text{Separable states:} & \\
     &&&\\
    &&&\rho_{AB}=\sum_{x\in\mathcal{X}}p(x)\rho^x_A\otimes\rho^x_B,\,  & \text{s.th. }\forall x\in \mathcal{X}:\\
    &&&\rho^x_A\succeq 0,\, \rho^x_B\succeq 0,\,& \Tr\left[\rho^x_A\right]=1,\, \Tr\left[\rho^x_B\right]=1\\
    &&&&\\
    &&&\text{Alice's linear constraints $\forall x\in\mathcal{X}$:} &\\
    &&&&\\
    &&& \Theta_{A_L\rightarrow C_{A_L}}\lrbracket{\rho^x_A} = W_{C_{A_L}}\otimes \rho^x_{A_R} & \Omega_{A\rightarrow A}\lrbracket{\rho^x_{A}}=\rho^x_{A} ,\,\\
    \\
     &&&\text{Bob's linear constraints $\forall x\in\mathcal{X}$:} &\\
    &&&&\\
    &&& \Upsilon_{B_L\rightarrow C_{B_L}}\lrbracket{\rho^x_B} = K_{C_{B_L}}\otimes \rho^x_{B_R} ,\, & \Xi_{B\rightarrow B}\lrbracket{\rho^x_{B}}=\rho^x_{B}.\,\\ 
        \end{array}
    \end{align}}
    Then, $\mathrm{cSEP}(G)$ and $\mathrm{cPROD}(G)$ are equivalent.
\end{proposition}
\begin{proof}
We omit a complete proof, as the reasoning mirrors that of \cite{berta2021semidefinite, borderi2022finetti}; the additional constraint does not interfere with the structure or validity of the original argument.
\end{proof}


While $\mathrm{cSEP}(G)$, as opposed to $\mathrm{cPROD}(G)$, constitutes an optimization over a convex set, characterizing this set efficiently and thus solving $\mathrm{cSEP}(G)$ is NP-hard. Hence, in \cite{berta2016quantum, berta2019quantum, berta2021semidefinite} the authors provide a hierarchy of outer approximations to $\mathrm{cSEP}(G)$ yielding upper bounds on $\mathrm{cSEP}(G)$. Concretely, as in the DPS hierarchy, the set of separable states can be approximated by its superset of $n$-extendable states w.r.t.\ system $A$, denoted $n\text{-Ext}(A\,{:}\,B)$. While quantum generalizations of classical de Finetti representation theorems such as \cite{christandl2007one, caves2002unknown} proved that this approximation is tight for $n\rightarrow\infty$, as well as quantifying the corresponding rate of convergence, they do not show that the resulting states satisfy the linear constraints in $\mathrm{cSEP}(G)$. More precisely, the constructed measure in the de Finetti theorem defined on state space need not vanish, even approximately, on the subset of states not fulfilling these linear constraints. Hence, in order to prove convergence via a quantum de Finetti representation theorem, \cite{berta2019quantum, berta2021semidefinite} provided an adapted finite de Finetti representation theorem, which ensures this condition via information-theoretic arguments. We extend their result to accommodate additional constraints.
\begin{remark}
By choosing $A_R$ and $B_R$ to be one-dimensional, we recover the special case considered in \cite{berta2021semidefinite, jee2020quasi, ohst2024characterising}.
\end{remark}

\begin{proposition}\label{prop:symmetric_states_have_symmetric_reduced_states}
Let $\rho_{A\lrbracket{B_LB_R}_1^n}$ be symmetric w.r.t.\ $A$. Then, the marginal $\rho_{A\lrbracket{B_L}_1^n}$ is also symmetric w.r.t.\ $A$. 
\end{proposition}
\begin{proof}
This follows from a straightforward computation:
     \begin{align}
\begin{split}
     &\lrbracket{\mathbb{1}_A\otimes U_{\lrbracket{B_L}_1^n}(\pi)}\,\rho_{A\lrbracket{B_L}_1^n}\, \lrbracket{\mathbb{1}_A\otimes U^T_{\lrbracket{B_L}_1^n}(\pi)} \\
     &\hspace{0.6cm}= \Trr{\lrbracket{B_R}_1^n}{\lrbracket{\mathbb{1}_A\otimes U_{\lrbracket{B_L}_1^n}(\pi)\otimes\mathbb{1}_{\lrbracket{B_R}_1^n}}\,\rho_{A\lrbracket{B_LB_R}_1^n}\,\lrbracket{\mathbb{1}_A\otimes U^T_{\lrbracket{B_L}_1^n}(\pi)\otimes\mathbb{1}_{\lrbracket{B_R}_1^n}}}\\
     &\hspace{0.6cm}= \Trr{\lrbracket{B_R}_1^n}{\lrbracket{\mathbb{1}_A\otimes U_{\lrbracket{B_L}_1^n}(\pi)\otimes U^T_{\lrbracket{B_R}_1^n}(\pi)\,U_{\lrbracket{B_R}_1^n}(\pi)}\,\rho_{A\lrbracket{B_LB_R}_1^n}\,\lrbracket{\mathbb{1}_A\otimes U^T_{\lrbracket{B_L}_1^n}(\pi)\otimes\mathbb{1}_{\lrbracket{B_R}_1^n}}}\\
     &\hspace{0.6cm}= \Trr{\lrbracket{B_R}_1^n}{\lrbracket{\mathbb{1}_A\otimes U_{\lrbracket{B_L}_1^n}(\pi)\otimes U_{\lrbracket{B_R}_1^n}(\pi)}\,\rho_{A\lrbracket{B_LB_R}_1^n}\,\lrbracket{\mathbb{1}_A\otimes U^T_{\lrbracket{B_L}_1^n}(\pi)\otimes U^T_{\lrbracket{B_R}_1^n}(\pi)}}\\
     &\hspace{0.6cm}=\Trr{\lrbracket{B_R}_1^n}{\lrbracket{\mathbb{1}_A\otimes U_{\lrbracket{B_LB_R}_1^n}(\pi)}\, \rho_{A\lrbracket{B_LB_R}_1^n}\,\lrbracket{\mathbb{1}_A\otimes U^T_{\lrbracket{B_LB_R}_1^n}(\pi)}}= \rho_{A\lrbracket{B_L}_1^n},\,
\end{split}
\end{align}
where the second equality inserts $\mathbb{1}=U^T_{\lrbracket{B_R}_1^n}(\pi)\,U_{\lrbracket{B_R}_1^n}(\pi)$, the third moves the left-inserted $U^T_{\lrbracket{B_R}_1^n}(\pi)$ to the right factor using the cyclicity of the partial trace on the traced subsystem, and the last step uses $U_{\lrbracket{B_L}_1^n}(\pi)\otimes U_{\lrbracket{B_R}_1^n}(\pi)=U_{\lrbracket{B_LB_R}_1^n}(\pi)$ together with the symmetry of $\rho_{A\lrbracket{B_LB_R}_1^n}$ w.r.t.\ $A$.
\end{proof}

\begin{lemma}[De Finetti representation theorem with linear constraints (restated)]\label{lem:approximate_quantum_de_finetti}
    Let $\rho_{AB_1^n} \in \mathcal{S}(\mathcal{H}_A \otimes \mathcal{H}_B^{\otimes n})$ with $\HS_A=\HS_{A_L}\otimes\HS_{A_R}$ and $\HS_B=\HS_{B_L}\otimes\HS_{B_R}$ be a quantum state symmetric on $B_1^n$ w.r.t.\ $A$. Furthermore, consider linear mappings $\Theta_{A_L\rightarrow C_{A_L}},\, \Omega_{A\rightarrow A},\, \Upsilon_{\lrbracket{B_L}_n\rightarrow C_{B_L}},\, \Xi_{B_n\rightarrow B_n}$ and operators $W_{C_{A_L}}\in\mathcal{B}\lrbracket{C_{A_L}},\, K_{C_{B_L}}\in\mathcal{B}\lrbracket{C_{B_L}}$. Assuming
    \begin{align}\label{eqn:deFinetti_constraints_appendix}
    \begin{split}
        \begin{array}{cc}
            \Theta_{A_L\rightarrow C_{A_L}}\lrbracket{\rho_{AB_1^n}} = W_{C_{A_L}}\otimes\rho_{A_RB_1^n},\, &  \Omega_{A\rightarrow A}\lrbracket{\rho_{AB_1^n}}=\rho_{AB_1^n}\\
            \Upsilon_{\lrbracket{B_L}_n\rightarrow C_{B_L}}\lrbracket{\rho_{B_1^n}} = K_{C_{B_L}}\otimes \rho_{B_1^{n-1}\lrbracket{B_R}_n},\, & \Xi_{B_n\rightarrow B_n}\lrbracket{\rho_{B_1^n}}=\rho_{B_1^n},\,\\
        \end{array}
    \end{split}
    \end{align}
    then there exists a probability distribution $\lrbrace{p(x)}_{x\in\mathcal{X}}$ and states $\lrbrace{\sigma_A^x}_{x\in\mathcal{X}},\, \lrbrace{\sigma_B^x}_{x\in\mathcal{X}}$ such that the following is true
    \begin{enumerate}
        \item for $\rho_{AB} := \text{tr}_{B_2^{n}} [\rho_{AB_1^n}]$ we have 
        \begin{align*}
        &\Vert \rho_{AB} - \sum_{x\in \mathcal{X}} p(x) \sigma_A^x \otimes \sigma_B^x \Vert_1 \leq \min\lbrace f(A, B), f(B\vert \cdot)\rbrace \sqrt{\frac{2\ln 2 \log_2 \lrvert{A}}{n}}
        \end{align*}
        where $f(A, B)$ and $f(B\vert\cdot)$ denote the distortion factors of the informationally complete measurements, cf.\ \autoref{eqn:measurement_distortion_AB} and \autoref{eqn:measurement_distortion_A};
        \item \label{thm:approximate_quantum_de_finetti_2} for all $x \in \mathcal{X}$ we have
        \begin{align}
            \begin{split}
                \begin{array}{cc}
                    \Theta_{A_L\rightarrow C_{A_L}}\lrbracket{\sigma^x_{A}} = W_{C_{A_L}}\otimes\sigma^x_{A_R},\, & \Omega_{A\rightarrow A}\lrbracket{\sigma^x_{A}}=\sigma^x_{A},\,\\ 
            \Upsilon_{B_L\rightarrow C_{B_L}}\lrbracket{\sigma^x_{B}} = K_{C_{B_L}}\otimes\sigma^x_{B_R},\, & \Xi_{B\rightarrow B}\lrbracket{\sigma^x_{B}}=\sigma^x_{B}.\,\\
                \end{array}
            \end{split}
        \end{align}
    \end{enumerate}
\end{lemma}
\begin{proof}
The proof technique of \cite[Thm.\ 2.3]{berta2021semidefinite} extends to the additional constraints. Let $\mathcal{M}_{B\rightarrow Z}$ be an informationally complete measurement and consider the classical-quantum state $\rho_{AZ_1^n}:=\lrbracket{\mathrm{id}_A\otimes\mathcal{M}_{B\rightarrow Z}^{\otimes n}}\lrbracket{\rho_{AB_1^n}}$; since $\rho_{AB_1^n}$ is symmetric w.r.t.\ $A$, so is $\rho_{AZ_1^n}$. The self-decoupling lemma \cite[Lemma 2.1]{brandao2013product} then certifies the existence of an $m\in\lrbrace{0, 1, \ldots, n-1}$ such that
\begin{align}
    \mathbb{E}_{z_1^m}\lrbrace{\left\lVert \rho_{AZ_{m+1}\vert z_1^m}-\rho_{A\vert z_1^m}\otimes\rho_{Z_{m+1}\vert z_1^m}  \right\rVert_1^2}\leq\frac{2\ln 2 \log_2 \lrvert{A}}{n},\,
\end{align}
with the convention that for $m=0$ the conditioning on $z_1^0$ is void and the expectation is trivial. In the remainder we keep the measurement on the first $m$ systems only, i.e., we work with the POVM elements $M^{z_1^m}_{B_1^m}$ of $\mathcal{M}^{\otimes m}$, and undo the measurement on the $(m+1)$-st system at the price of the corresponding distortion factor. By convexity, we have
\begin{align}
    \left\lVert \rho_{AB_{m+1}}-\sum_{z_1^m}p(z_1^m)\,\rho_{A\vert z_1^m}\otimes\rho_{B_{m+1}\vert z_1^m} \right\rVert_1\leq\min \lrbracket{f(A, B),\, f(B\vert\cdot)}\sqrt{\frac{2\ln 2 \log_2 \lrvert{A}}{n}},\,
\end{align}
where the measurement distortions can be bounded by \autoref{eqn:measurement_distortion_AB} or \autoref{eqn:measurement_distortion_A} for the informationally complete measurement mapping $\rho_{AB_{m+1}}$ to $\rho_{AZ_{m+1}}$. Next, we show that the candidate state $\sum_{z_1^m}p(z_1^m)\,\rho_{A\vert z_1^m}\otimes\rho_{B_{m+1}\vert z_1^m}$ satisfies the new constraints. In short, all maps in \autoref{eqn:deFinetti_constraints_appendix} act on systems disjoint from $B_1^m$ and hence commute with the measurement, and the constraint on Bob's $n$-th subsystem transfers to the $(m+1)$-st by permutation invariance. Subsequently, let $\frac{1}{N}$ denote the corresponding normalization factor of the post measurement state. Thus,
\begin{align}
\begin{split}
    \Theta_{A_L\rightarrow C_{A_L}}\lrbracket{\rho_{A\vert{z_1^m}}} &= \frac{1}{N} \Theta_{A_L\rightarrow C_{A_L}}\lrbracket{\Trr{Z_1^m B_{m+1}^n}{\lrbracket{\mathbb{1}_{AB_{m+1}^n}\otimes M_{B_1^m}^{z_1^m}}\rho_{AB_1^n}}}\\
    &= \frac{1}{N}\Trr{Z_1^m B_{m+1}^n}{\lrbracket{\mathbb{1}_{C_{A_L}A_RB_{m+1}^n}\otimes M_{B_1^m}^{z_1^m}}\Theta_{A_L\rightarrow C_{A_L}}\lrbracket{\rho_{AB_1^n}}}\\
    &\stackrel{\autoref{eqn:deFinetti_constraints_appendix}}{=}\frac{1}{N}\Trr{Z_1^m B_{m+1}^n}{\lrbracket{\mathbb{1}_{C_{A_L}A_RB_{m+1}^n}\otimes M_{B_1^m}^{z_1^m}}W_{C_{A_L}}\otimes\rho_{A_RB_1^n}}=W_{C_{A_L}}\otimes\rho_{A_R\vert z_1^m}
\end{split}
\end{align}
and
\begin{align}
\begin{split}
    \Omega_{A\rightarrow A}\lrbracket{\rho_{A\vert z_1^m}} &= \frac{1}{N} \Omega_{A\rightarrow A}\lrbracket{\Trr{Z_1^m B_{m+1}^n}{\lrbracket{\mathbb{1}_{AB_{m+1}^n}\otimes M_{B_1^m}^{z_1^m}}\rho_{AB_1^n}}}\\
    &=\frac{1}{N} \Trr{Z_1^m B_{m+1}^n}{\lrbracket{\mathbb{1}_{AB_{m+1}^n}\otimes M_{B_1^m}^{z_1^m}}\Omega_{A\rightarrow A}\lrbracket{\rho_{AB_1^n}}}\\
    &\stackrel{\autoref{eqn:deFinetti_constraints_appendix}}{=} \frac{1}{N}\Trr{Z_1^m B_{m+1}^n}{\lrbracket{\mathbb{1}_{AB_{m+1}^n}\otimes M_{B_1^m}^{z_1^m}}\rho_{AB_1^n}}=\rho_{A\vert z_1^m}\,.
\end{split}
\end{align}
Again, let $\frac{1}{R}$ denote the corresponding normalization factor. The permutation invariance of $\rho_{B_1^n}$ over the subsystems $B_1^n$ --- inherited from the symmetry of $\rho_{AB_1^n}$ w.r.t.\ $A$ by tracing out $A$ --- implies that the constraint in \autoref{eqn:deFinetti_constraints_appendix}, stated for the $n$-th subsystem, holds verbatim with $\lrbracket{B_L}_{m+1}$ in place of $\lrbracket{B_L}_{n}$ for every $m\in\lrbrace{0,\ldots,n-1}$ (swap the subsystems $m+1$ and $n$; cf.\ \autoref{prop:symmetric_states_have_symmetric_reduced_states}). Thus,
\begin{align}
    \begin{split}
        \Upsilon_{\lrbracket{B_L}_{m+1}\rightarrow C_{B_L}}\lrbracket{\rho_{B_{m+1}\vert z_1^m}} &= \frac{1}{R}  \Upsilon_{\lrbracket{B_L}_{m+1}\rightarrow C_{B_L}}\lrbracket{\Trr{Z_1^m B_{m+2}^n}{\lrbracket{\mathbb{1}_{B_{m+1}^n}\otimes M_{B_1^m}^{z_1^m}}\rho_{B_1^n}}}\\
        &=\frac{1}{R} \Trr{Z_1^m B_{m+2}^n}{\lrbracket{\mathbb{1}_{C_{B_L}\lrbracket{B_R}_{m+1}B_{m+2}^n}\otimes M_{B_1^m}^{z_1^m}}\Upsilon_{\lrbracket{B_L}_{m+1}\rightarrow C_{B_L}}\lrbracket{\rho_{B_1^n}}}\\
        &\stackrel{\autoref{eqn:deFinetti_constraints_appendix}}{=}\frac{1}{R} \Trr{Z_1^m B_{m+2}^n}{\lrbracket{\mathbb{1}_{C_{B_L}\lrbracket{B_R}_{m+1}B_{m+2}^n}\otimes M_{B_1^m}^{z_1^m}} K_{C_{B_L}}\otimes \rho_{B_1^m\lrbracket{B_R}_{m+1}
        B_{m+2}^n}} \\
        &= K_{C_{B_L}}\otimes \rho_{\lrbracket{B_R}_{m+1}\vert z_1^m}
    \end{split}
\end{align}
and
\begin{align}
    \begin{split}
        \Xi_{B_{m+1}\rightarrow B_{m+1}}\lrbracket{\rho_{B_{m+1}\vert z_1^m}} &= \frac{1}{R}  \Xi_{B_{m+1}\rightarrow B_{m+1}}\lrbracket{\Trr{Z_1^m B_{m+2}^n}{\lrbracket{\mathbb{1}_{B_{m+1}^n}\otimes M_{B_1^m}^{z_1^m}}\rho_{B_1^n}}}\\
        &=\frac{1}{R} \Trr{Z_1^m B_{m+2}^n}{\lrbracket{\mathbb{1}_{B_{m+1}^n}\otimes M_{B_1^m}^{z_1^m}}\Xi_{B_{m+1}\rightarrow B_{m+1}}\lrbracket{\rho_{B_1^n}}}\\
        &\stackrel{\autoref{eqn:deFinetti_constraints_appendix}}{=}\frac{1}{R} \Trr{Z_1^m B_{m+2}^n}{\lrbracket{\mathbb{1}_{B_{m+1}^n}\otimes M_{B_1^m}^{z_1^m}}\rho_{B_1^n}} =\rho_{B_{m+1}\vert z_1^m}\,.\\
    \end{split}
\end{align}
Setting $\mathcal{X}:=\lrbrace{z_1^m}$, $p(x):=p(z_1^m)$, $\sigma_A^x:=\rho_{A\vert z_1^m}$ and $\sigma_B^x:=\rho_{B_{m+1}\vert z_1^m}$ establishes items 1 and 2; by the permutation invariance of $\rho_{AB_1^n}$, the identification of $B$ with $B_{m+1}$ is immaterial.
\end{proof}

\begin{remark}[Application to Bose-symmetric states]\label{rem:deFinetti_bose_symmetric}
    If $\rho_{AB_1^n}$ is Bose-symmetric w.r.t.\ $A$ (\autoref{def:bose_symmetric}), then it is in particular symmetric w.r.t.\ $A$ (cf.\ \autoref{sec:Bose_sym_easy}), and the classical-quantum state $\rho_{AZ_1^n}$ obtained from the product measurement $\mathcal{M}^{\otimes n}$ remains symmetric w.r.t.\ $A$ --- the informationally complete measurement is applied after the fact and need not respect Bose symmetry. \autoref{lem:approximate_quantum_de_finetti} therefore applies verbatim to the Bose-symmetric hierarchy of \autoref{lem:bose-symmetric_hierarchy}.
\end{remark}

Roughly speaking, \autoref{lem:approximate_quantum_de_finetti} states that even under the additional assumption of linear constraints a hierarchy of outer approximations can be used to approximate \eqref{eq:optimization_task}. Even more so, it provides an upper bound on the distance of the $n$-extendable state to a separable state. Importantly, each state in the separable state decomposition also adheres to the linear constraints. Using this de Finetti theorem, the constrained bilinear optimization problem \eqref{eq:optimization_task} can be relaxed to an SDP corresponding to an optimization over an outer approximation of the set of separable states yielding an upper bound on \eqref{eq:optimization_task}. Concretely, one obtains a hierarchy of SDP relaxations indexed by the (symmetric) extension level $n$ of $\rho_{AB_1^n}$ with respect to $A$ (i.e.\ approximating $\rho_{AB}\in\text{Sep}(A\,{:}\,B)$ with $n$-extendable states $\rho_{AB_1^n}$). The following proposition specifies this hierarchy and quantifies its rate of convergence.
    
\begin{proposition}[Hierarchy of SDP relaxations]\label{prop:hierarchy_outer_approx}
    Identify $B_1:=B$, i.e. $\rho_{AB}=\rho_{AB_1}=\Tr_{B_2^n}\left(\rho_{AB_1^n}\right)$. 
    The level $n$ SDP relaxation of  $\mathrm{cSEP}(G)$ is given by
    \begin{align*}
        \mathrm{SDP}_n(G) := \max_{\rho_{AB_1}} \text{tr}\left[G_{AB}\rho_{AB_1}\right]
    \end{align*}
    subject to
    \begin{multicols}{2}
    \begin{enumerate}
        \item $\rho_{AB_1^n}\succeq 0$ and $\Tr[\rho_{AB_1^n}]=1$,
        \item $\rho_{AB_1^n} = U_{B_1^n}^T(\sigma)\left(\rho_{AB_1^n}\right) U_{B_1^n}(\sigma) \quad \forall\sigma\in S_n$,
        \item $\Theta_{A_L\rightarrow C_{A_L}}\lrbracket{\rho_{AB_1^n}} = W_{C_{A_L}}\otimes\rho_{A_RB_1^n}$,
        \item $\Omega_{A\rightarrow A}\lrbracket{\rho_{AB_1^n}}=\rho_{AB_1^n}$,
        \item $\Upsilon_{\lrbracket{B_L}_n\rightarrow C_{B_L}}\lrbracket{\rho_{B_1^n}} = K_{C_{B_L}}\otimes \rho_{B_1^{n-1}\lrbracket{B_R}_n}$,
        \item $\Xi_{B_n\rightarrow B_n}\lrbracket{\rho_{B_1^n}}=\rho_{B_1^n}$,
\end{enumerate}
\end{multicols}
yielding a sequence of upper bounds on $\mathrm{cSEP}(G)$ converging from above with $n\rightarrow\infty$, i.e.\ 
\begin{align}\label{eqn:old_sdp_convergence}
     0\leq \mathrm{SDP}_n(G) - \mathrm{cPROD}(G) \leq \min\lbrace f(A, B),\, f(B\vert \cdot)\rbrace \sqrt{\frac{2\ln 2 \log_2 \lrvert{A}}{n}}.\,
\end{align}%
\end{proposition}
\begin{proof}
    Let $\rho_{AB}=\sum_{x\in\mathcal{X}}p(x)\,\rho_A^x\otimes\rho_B^x$ be feasible for $\mathrm{cSEP}(G)$. A separable state is $n$-extendable for any $n$. Define
    \begin{align}
        \rho_{AB_1^n}:=\sum_{x\in\mathcal{X}}p(x)\,\rho_A^x\otimes\lrbracket{\rho_B^x}^{\otimes n}\,.
    \end{align} This extension is positive semidefinite, has unit trace, and is invariant under permutations of $B_1^n$. Using that each pair $\lrbracket{\rho_A^x, \rho_B^x}$ satisfies the constraints of \autoref{eq:bilinear_optimisation_problem_sep}, we compute
    \begin{align}
    \begin{split}
         \Theta_{A_L\rightarrow C_{A_L}}\lrbracket{\rho_{AB_1^n}}&=\sum_{x\in\mathcal{X}}p(x)\, \Theta_{A_L\rightarrow C_{A_L}}\lrbracket{\rho_A^x}\otimes\lrbracket{\rho_B^x}^{\otimes n}=\sum_{x\in\mathcal{X}}p(x)\,W_{C_{A_L}}\otimes\rho_{A_R}^x\otimes\lrbracket{\rho_B^x}^{\otimes n}\\
         &=W_{C_{A_L}}\otimes\rho_{A_RB_1^n}
    \end{split}
    \end{align}
    and
    \begin{align}
        \Omega_{A\rightarrow A}\lrbracket{\rho_{AB_1^n}}= \sum_{x\in\mathcal{X}}p(x)\, \Omega_{A\rightarrow A}\lrbracket{\rho_A^x}\otimes\lrbracket{\rho_B^x}^{\otimes n}= \rho_{AB_1^n}\,.
    \end{align}
    Moreover,
        \begin{align}
        \begin{split}
            \Upsilon_{\lrbracket{B_L}_n\rightarrow C_{B_L}}\lrbracket{\rho_{B_1^n}}&=\sum_{x\in\mathcal{X}}p(x)\, \lrbracket{\rho_B^x}^{\otimes n-1}\otimes\Upsilon_{\lrbracket{B_L}\rightarrow C_{B_L}}\lrbracket{\rho_B^x}=\sum_{x\in\mathcal{X}}p(x)\, \lrbracket{\rho_B^x}^{\otimes n-1}\otimes K_{C_{B_L}}\otimes\rho^x_{B_R}\\
            &=K_{C_{B_L}}\otimes\rho_{B_1^{n-1}\lrbracket{B_R}_n}
        \end{split}
    \end{align}
    and 
    \begin{align}
        \Xi_{B_n\rightarrow B_n}\lrbracket{\rho_{B_1^n}}=\sum_{x\in\mathcal{X}}p(x)\, \lrbracket{\rho_B^x}^{\otimes n-1}\otimes\Xi_{B\rightarrow B}\lrbracket{\rho_B^x}=\rho_{B_1^n}.
    \end{align}
    Conversely, the constrained de Finetti theorem in \autoref{lem:approximate_quantum_de_finetti} certifies the convergence rate of $\mathrm{SDP}_n(G)$ to $\mathrm{cSEP}(G)$. 
\end{proof}

We briefly elaborate on the limit problem of the proposed SDP hierarchy in \autoref{prop:hierarchy_outer_approx}. For each $n\in\mathbb{N}$, let $\mathcal{F}_n$ denote the feasible set of $\mathrm{SDP}_n(G)$, and define its one-copy marginal feasible set by 
\begin{align}
    \mathcal{K}_n:= \lrbrace{\rho_{AB_1}\,:\, \exists\,\rho_{AB_1^n}\in \mathcal{F}_n \text{ s.th. } \rho_{AB_1}
        = \Trr{B_2^n}{\rho_{AB_1^n}}}
\end{align}
The limiting feasible set is 
\begin{align}
    \mathcal{K}_\infty:= \bigcap_{n\geq 1}\mathcal{K}_n.
\end{align}
We then define
\begin{align}
    \mathrm{SDP}_\infty(G):=\max_{\rho_{AB}\in \mathcal{K}_\infty}\Trr{}{G_{AB}\rho_{AB}}.
\end{align}
Since the sets $\mathcal{K}_n$ are compact, convex, and nested $\mathcal{K}_{n+1}\subseteq \mathcal{K}_n$, and since the objective function $ \rho_{AB}\mapsto \Trr{}{G_{AB}\rho_{AB}}$ is continuous, we have
\begin{align}
    \mathrm{SDP}_\infty(G)=\lim_{n\to\infty}\mathrm{SDP}_n(G)=\inf_{n\geq 1}\mathrm{SDP}_n(G).
\end{align}
By the de Finetti theorem in \autoref{lem:approximate_quantum_de_finetti}, every $\rho_{AB}\in \mathcal{K}_\infty$ is a constrained separable state. Conversely, every constrained separable state belongs to $\mathcal{K}_\infty$. Hence 
\begin{align}
    \mathrm{SDP}_\infty(G)=\mathrm{cSEP}(G)= \mathrm{cPROD}(G).
\end{align}
Moreover, let $\rho^{\star}_{AB_1^n}$ be optimal for $\mathrm{SDP}_n(G)$, and let
\begin{align}
    \rho^{(n, \star)}_{AB}:=\Trr{B_2^n}{\rho^{\star}_{AB_1^n}}. 
\end{align}
Since $\mathcal{S}(AB)$ is compact, the sequence $\lrbracket{\rho^{(n, \star)}_{AB}}_{n\in\mathbb{N}}$ has an accumulation point. Let $\sigma^\star_{AB}$ be such an accumulation point along a subsequence. Then
\begin{align}
    \sigma^\star_{AB}\in \mathcal{K}_\infty,\quad\text{and}\quad \Trr{}{G_{AB}\sigma^\star_{AB}}=\mathrm{SDP}_\infty(G)= \mathrm{cPROD}(G).
\end{align}

\section{Proof of \autoref{lem:non_local_games_as_cbo_sym}}\label{sec:quantum_steering}

For completeness, we recall two relevant results from the quantum steering literature, which will be used in the subsequent proof. For an introduction to the concepts underlying quantum steering, we refer the reader to \cite{Uola_2020}.

\begin{proposition}[Transpose trick]\label{prop:transpose_trick}
Let $\lrbrace{\ket{i}_A}_i$ and $\lrbrace{\ket{i}_B}_i$ be fixed orthonormal bases of $\HS_A$ and $\HS_B$ with $\text{dim}(\HS_A)=\text{dim}(\HS_B)$, and let $\ket{\Omega}_{AB}:=\sum_i\ket{i}_A\otimes\ket{i}_B$ denote the (non-normalized) maximally entangled vector. For any operator $M\in\mathcal{B}\left(\HS_A\right)$ with transpose $M^T$ taken with respect to these bases,
    \begin{align}
        \begin{split}
            \left(M_A\otimes\mathbb{1}_B\right)\ket{\Omega}_{AB} = \left(\mathbb{1}_A\otimes M^T_B\right)\ket{\Omega}_{AB}.\,
        \end{split}
    \end{align}
\end{proposition}
\begin{proof}
    Both sides equal $\sum_{i,j} M_{ji}\ket{j}_A\otimes\ket{i}_B$: the left-hand side by applying $M$ to the first tensor factor, the right-hand side since $\lrbracket{M^T}_{ij}=M_{ji}$.
\end{proof}

\begin{proposition}[Swap trick]\label{prop:swap_trick_renner}
    Consider any $N,M\in\mathcal{B}(\HS_A)$ and let $S_{A\hat{A}}\in\mathcal{B}(\HS_A\otimes\HS_{\hat{A}})$ be the (non-normalised) swap operator defined as $S_{A\hat{A}}:=\sum_{i,j}\ket{i}\bra{j}_{A}\otimes\ket{j}\bra{i}_{\hat{A}}$. Then, $ \Tr\left[MN\right]=\Tr\left[S_{A\hat{A}}\left(M\otimes N\right)\right]$.
\end{proposition}

To prove \autoref{lem:non_local_games_as_cbo_sym}, we first prove the following intermediate result.
\begin{proposition}\label{prop:non_local_games_as_cSEP_intermediate}
The value of a two-player free non-local game with $\lrvert{A}$ many answers, $\lrvert{Q}$ many questions and quantum assistance of size $\lrvert{T}$ is given by the following bipartite constrained separability problem
        \begin{align}\label{eqn:non_local_games_as_cSEP}
            \begin{split}
                 w_{Q(T)}(V, \pi) = &\lvert T\rvert\max_{(\alpha, D)}\hspace{0.5cm}\Tr\left[\left(V_{A_1A_2Q_1Q_2}\otimes S_{\hat{T}\tilde{T}}\right)\left(\alpha_{A_1Q_1\tilde{T}}\otimes D_{A_2Q_2\hat{T}}\right)\right]\\
                 \\
                 \text{s.t. } &\alpha_{A_1Q_1\tilde{T}} = \sum_{a_1,q_1}\pi_1(q_1)\ket{a_1q_1}\bra{a_1q_1}_{A_1Q_1}\otimes \alpha_{\tilde{T}}(a_1\vert q_1)\succcurlyeq 0,\,\\
        & \alpha_{Q_1\tilde{T}} = \sum_{q_1}\pi_1(q_1)\ket{q_1}\bra{q_1}_{Q_1}\otimes \sum_{a_1}\alpha_{\tilde{T}}(a_1\vert q_1=1),\, \hspace{0.5cm} \Trr{}{\alpha_{A_1Q_1\tilde{T}}}=1,\,\\
        & D_{A_2Q_2\hat{T}} =\sum_{a_2,q_2}\pi_2(q_2)\ket{a_2q_2}\bra{a_2q_2}_{A_2Q_2}\otimes\frac{D_{\hat{T}}(a_2\vert q_2)}{\lvert T\rvert}\succcurlyeq 0,\,\\
              &D_{Q_2\hat{T}} =\sum_{q_2}\pi_2(q_2)\ket{q_2}\bra{q_2}_{Q_2}\otimes\frac{\mathbb{1}_{\hat{T}}}{\lvert T\rvert},\,\\
            \end{split}
        \end{align}
        where $S_{\hat{T}\tilde{T}}$ is the swap operator, $\lrvert{T}=\lrvert{\tilde{T}}=\lrvert{\hat{T}}$, and $q_1=1$ in the second constraint denotes an arbitrary fixed reference question.
\end{proposition}

\begin{proof}
Let \autoref{eqn:deFinetti_CBO_free_games} be denoted by $\lrbracket{\mathrm{I}}$ and \autoref{eqn:non_local_games_as_cSEP} by $\lrbracket{\mathrm{II}}$. Clearly, any $D_{A_2Q_2\hat{T}}$ belonging to the feasible set of $\lrbracket{\mathrm{I}}$ is also a feasible point of problem $\lrbracket{\mathrm{II}}$ and vice versa. 

Let $E_{A_1Q_1T}$ and $\rho_{T\hat{T}}$ be feasible points for $\lrbracket{\mathrm{I}}$. Define for each $a_1\in A_1$, $q_1\in Q_1$
\begin{align}\label{eqn:assemblage_tick_in_proof}
    \beta_{\hat{T}}(a_1\vert q_1) := \Trr{T}{\lrbracket{E_T(a_1\vert q_1)\otimes \mathbb{1}_{\hat{T}}}\rho_{T\hat{T}}}
\end{align}
and
\begin{align}
\begin{split}
    \beta_{A_1Q_1\hat{T}}&:= \sum_{a_1,q_1}\pi_1(q_1)\ket{a_1q_1}\bra{a_1q_1}_{A_1Q_1} \otimes \beta_{\hat{T}}(a_1\vert q_1)=\Tr_{T}\lrrec{\lrbracket{E_{A_1Q_1T}\otimes\mathbb{1}_{\hat{T}}}\lrbracket{\mathbb{1}_{A_1Q_1}\otimes\rho_{T\hat{T}}}}.\,
\end{split}
\end{align}
Since $\rho_{T\hat{T}},\, E_{A_1Q_1T}\succcurlyeq 0$ and with the cyclicity of $\Trr{T}{\cdot}$ we have
\begin{align}
\begin{split}
    \beta_{A_1Q_1\hat{T}} &= \Tr_{T}\lrrec{\lrbracket{\sum_{a_1,q_1}\pi_1(q_1)\ket{a_1q_1}\bra{a_1q_1}_{A_1Q_1}\otimes E_{T}\lrbracket{a_1\vert q_1}\otimes\mathbb{1}_{\hat{T}}}\lrbracket{\mathbb{1}_{A_1Q_1}\otimes\rho_{T\hat{T}}}}\\
    &= \sum_{a_1,q_1}\pi_1(q_1)\ket{a_1q_1}\bra{a_1q_1}_{A_1Q_1} \otimes \Trr{T}{\lrbracket{E_T(a_1\vert q_1)\otimes\mathbb{1}_{\hat{T}}}\rho_{T\hat{T}}}\\
    &=\sum_{a_1,q_1}\pi_1(q_1)\ket{a_1q_1}\bra{a_1q_1}_{A_1Q_1} \otimes \Trr{T}{\lrbracket{E_T(a_1\vert q_1)\otimes\mathbb{1}_{\hat{T}}}^{\frac{1}{2}}\rho_{T\hat{T}}\lrbracket{E_T(a_1\vert q_1)\otimes\mathbb{1}_{\hat{T}}}^{\frac{1}{2}}}\succcurlyeq 0.\,\\
\end{split}
\end{align}
Furthermore, using the condition $\sum_{a_1}E_{T}(a_1\vert q_1)=\mathbb{1}_T$ for all $q_1\in Q_1$, we obtain
\begin{align}
   \begin{split}
       \Trr{A_1}{\beta_{A_1Q_1\hat{T}}}&=  \Trr{A_1}{\Trr{T}{\lrbracket{E_{A_1Q_1T}\otimes\mathbb{1}_{\hat{T}}}\lrbracket{\mathbb{1}_{A_1Q_1}\otimes\rho_{T\hat{T}}}}}
       = \Trr{T}{\lrbracket{E_{Q_1T}\otimes\mathbb{1}_{\hat{T}}}\lrbracket{\mathbb{1}_{Q_1}\otimes\rho_{T\hat{T}}}} \\
       &= \Trr{T}{\lrbracket{\sum_{q_1}\pi_1(q_1)\ket{q_1}\bra{q_1}_{Q_1}\otimes \mathbb{1}_{T\hat{T}}}\lrbracket{\mathbb{1}_{Q_1}\otimes \rho_{T\hat{T}}}}
       = \sum_{q_1}\pi_1(q_1)\ket{q_1}\bra{q_1}_{Q_1}\otimes \rho_{\hat{T}}\\
       &=\sum_{q_1}\pi_1(q_1)\ket{q_1}\bra{q_1}_{Q_1}\otimes\sum_{a_1}\beta_{\hat{T}}(a_1\vert q_1)
   \end{split}
\end{align}
and hence $\Trr{}{\rho_{\hat{T}}}=1 $ implies  $\Trr{}{\beta_{A_1Q_1\hat{T}}}=1$. Since $\lrvert{\hat{T}}=\lrvert{\tilde{T}}$, and thus, $\HS_{\hat{T}}\simeq\HS_{\tilde{T}}$, $\beta_{A_1Q_1\tilde{T}}:=\beta_{A_1Q_1\hat{T}}$ is a feasible point to $\lrbracket{\mathrm{II}}$. Since 
\begin{align}
\begin{split}
    \Tr\left[\left(V_{A_1A_2Q_1Q_2}\otimes\rho_{T\hat{T}}\right)\left(E_{A_1Q_1T}\otimes D_{A_2Q_2\hat{T}}\right)\right] &\stackrel{\autoref{eqn:assemblage_tick_in_proof}}{=} \Tr\lrrec{\lrbracket{V_{A_1A_2Q_1Q_2}\otimes\mathbb{1}_{\hat{T}}}\lrbracket{\beta_{A_1Q_1\hat{T}}\otimes\mathbb{1}_{A_2Q_2}}\lrbracket{\mathbb{1}_{A_1Q_1}\otimes D_{A_2Q_2\hat{T}}}}\\
    &\stackrel{\autoref{prop:swap_trick_renner}}{=}\Tr\left[\left(V_{A_1A_2Q_1Q_2}\otimes S_{\hat{T}\tilde{T}}\right)\left(\beta_{A_1Q_1\tilde{T}}\otimes D_{A_2Q_2\hat{T}}\right)\right],
\end{split}
\end{align}
the construction in \autoref{eqn:assemblage_tick_in_proof} converts feasible points from $\lrbracket{\mathrm{I}}$ to feasible points of $\lrbracket{\mathrm{II}}$ with the same objective value; hence $\lrbracket{\mathrm{II}}\geq\lrbracket{\mathrm{I}}$.

Conversely, given a feasible $\alpha_{A_1Q_1\tilde{T}}$ for $\lrbracket{\mathrm{II}}$, we define, for each $q_1\in Q_1$, $\sigma_{\tilde{T}}^{q_1} := \sum_{a_1}\alpha_{\tilde{T}}(a_1\vert q_1)$. The constraint on $\alpha_{Q_1\tilde{T}}$ implies $\sum_{a_1}\alpha_{\tilde{T}}(a_1\vert q_1)=\sum_{a_1}\alpha_{\tilde{T}}(a_1\vert q'_1)$ for any $q'_1\in Q_1$. Thus, we have  $\sigma_{\tilde{T}}^{q_1}=\sigma_{\tilde{T}}^{q'_1}=:\sigma_{\tilde{T}}$. Furthermore, $\alpha_{A_1Q_1\tilde{T}}\succcurlyeq 0$ implies $\alpha_{\tilde{T}}(a_1\vert q_1)\succcurlyeq 0$ for all $a_1\in A_1$ and $q_1\in Q_1$. This in turn implies $\sigma_{\tilde{T}}\succcurlyeq 0$. Moreover, $1=\Trr{}{\alpha_{A_1Q_1\tilde{T}}}=\sum_{q_1}\pi_1(q_1)\,\Trr{}{\sigma_{\tilde{T}}}=\Trr{}{\sigma_{\tilde{T}}}$. Let
\begin{align}
     \sigma_{\tilde{T}} = \sum_{k=1}^{\text{rank}\lrbracket{ \sigma_{\tilde{T}}}} \lambda_k\ket{e_k}\bra{e_k}_{\tilde{T}}
\end{align}
be its eigenbasis decomposition and 
\begin{align}
    \ket{\Psi\lrbracket{\sigma_{\tilde{T}}}}_{T\tilde{T}} := \sum_{k=1}^{\text{rank}\lrbracket{ \sigma_{\tilde{T}}}} \sqrt{ \lambda_k}\ket{e_k}_T\otimes\ket{e_k}_{\tilde{T}}
\end{align}
a purification of $\sigma_{\tilde{T}}$ written in Schmidt form. By construction, $\rho_{T\tilde{T}}^{\sigma}:=\ket{\Psi\lrbracket{\sigma_{\tilde{T}}}}\bra{\Psi\lrbracket{\sigma_{\tilde{T}}}}_{T\tilde{T}}$ is a normalized quantum state and thus feasible in  $\lrbracket{\mathrm{I}}$. To determine Alice's POVMs feasible in $\lrbracket{\mathrm{I}}$, we ask which operators yield $\alpha_{\tilde{T}}(a_1\vert q_1)$ as Bob's (subnormalized) post-measurement states upon Alice measuring $\rho_{T\tilde{T}}^{\sigma}$: for each $a_1\in A_1$, $q_1\in Q_1$ we seek $F_{T}(a_1\vert q_1)$ with
\begin{align}\label{eqn:steering_F_property}
    \alpha_{\tilde{T}}(a_1\vert q_1) = \Trr{T}{\lrbracket{F_{T}(a_1\vert q_1)\otimes\mathbb{1}_{\tilde{T}}}\rho_{T\tilde{T}}^{\sigma}};
\end{align}
existence and an explicit form follow below. Since any two purifications originating from $\sigma_{\tilde{T}}$ are equivalent up to isometries on $\HS_{T}$, there exists a unitary $U_{T}\in\mathcal{B}\lrbracket{\HS_{T}}$ such that 
\begin{align}
    \rho_{T\tilde{T}}^{\sigma}=\lrbracket{U^{\dagger}_{T}\otimes \sqrt{\sigma_{\tilde{T}}}}\rho^{max}_{T\tilde{T}} \lrbracket{U_{T}\otimes \sqrt{\sigma_{\tilde{T}}}}
\end{align}
with $\rho^{max}_{T\tilde{T}}$ the (non-normalized) maximally entangled state. Define
$\ket{\mathcal{I}_{\sigma}}_{\tilde{T}}:=\sum_{k=1}^{\text{rank}(\sigma_{\tilde{T}})}\ket{e_k}_{\tilde{T}}$ with $\lrbrace{\ket{e_k}}_k$ the eigenbasis to $\sigma_{\tilde{T}}$. Then, 
\begin{align}\label{eqn:proof_assemblages_construction_1}
    \begin{split}
    \rho_{T\tilde{T}}^{\sigma}&=\lrbracket{\mathbb{1}_T\otimes\sqrt{\sigma_{\tilde{T}}}}\ket{\mathcal{I}_{\sigma}}_{T}\bra{\mathcal{I}_{\sigma}}_{\tilde{T}}\lrbracket{\mathbb{1}_T\otimes\sqrt{\sigma_{\tilde{T}}}}
    \end{split}
\end{align}
and 
\begin{align}\label{eqn:proof_assemblages_construction_2}
    \ket{\mathcal{I}_{\sigma}}_{T}\bra{\mathcal{I}_{\sigma}}_{\tilde{T}}=\lrbracket{U^{\dagger}_{T}\otimes \mathbb{1}_{\tilde{T}}}\rho_{T\tilde{T}}^{max}\lrbracket{U_{T}\otimes \mathbb{1}_{\tilde{T}}}.\,
\end{align}
Assuming \autoref{eqn:steering_F_property}, we obtain
\begin{align}
\begin{split}
    \alpha_{\tilde{T}}(a_1\vert q_1)&=\Trr{T}{\lrbracket{F_{T}(a_1\vert q_1)\otimes\mathbb{1}_{\tilde{T}}}\rho_{T\tilde{T}}^{\sigma}}\\
    &\stackrel{\autoref{eqn:proof_assemblages_construction_1}}{=}\Trr{T}{\lrbracket{F_{T}(a_1\vert q_1)\otimes\mathbb{1}_{\tilde{T}}}\lrbracket{\mathbb{1}_T\otimes\sqrt{\sigma_{\tilde{T}}}}\ket{\mathcal{I}_{\sigma}}_{T}\bra{\mathcal{I}_{\sigma}}_{\tilde{T}}\lrbracket{\mathbb{1}_T\otimes\sqrt{\sigma_{\tilde{T}}}}}\\
    &= \Trr{T}{\lrbracket{\mathbb{1}_T\otimes\sqrt{\sigma_{\tilde{T}}}}\lrbracket{F_{T}(a_1\vert q_1)\otimes\mathbb{1}_{\tilde{T}}}\ket{\mathcal{I}_{\sigma}}_{T}\bra{\mathcal{I}_{\sigma}}_{\tilde{T}}\lrbracket{\mathbb{1}_T\otimes\sqrt{\sigma_{\tilde{T}}}}}\\
    &\stackrel{\autoref{prop:transpose_trick}}{=}\Trr{T}{\lrbracket{\mathbb{1}_T\otimes\sqrt{\sigma_{\tilde{T}}}}\lrbracket{\mathbb{1}_T\otimes \lrbracket{F_{\tilde{T}}(a_1\vert q_1)}^T}\ket{\mathcal{I}_{\sigma}}_{T}\bra{\mathcal{I}_{\sigma}}_{\tilde{T}}\lrbracket{\mathbb{1}_T\otimes\sqrt{\sigma_{\tilde{T}}}}}\\
    &\stackrel{\autoref{eqn:proof_assemblages_construction_2}}{=} \Trr{T}{\lrbracket{\mathbb{1}_T\otimes\sqrt{\sigma_{\tilde{T}}}}\lrbracket{\mathbb{1}_T\otimes \lrbracket{F_{\tilde{T}}(a_1\vert q_1)}^T}\lrbracket{U^{\dagger}_{T}\otimes \mathbb{1}_{\tilde{T}}}\rho_{T\tilde{T}}^{max}\lrbracket{U_{T}\otimes \mathbb{1}_{\tilde{T}}}\lrbracket{\mathbb{1}_T\otimes\sqrt{\sigma_{\tilde{T}}}}}\\
    &= \Trr{T}{\lrbracket{U_T^{\dagger}\otimes\sqrt{\sigma_{\tilde{T}}}\lrbracket{F_{\tilde{T}}(a_1\vert q_1)}^T}\rho_{T\tilde{T}}^{max}\lrbracket{U_{T}\otimes \sqrt{\sigma_{\tilde{T}}}}}\\
    &\stackrel{\text{cycl.}}{=}\Trr{T}{\lrbracket{U_T^{\dagger}U_T\otimes\sqrt{\sigma_{\tilde{T}}}\lrbracket{F_{\tilde{T}}(a_1\vert q_1)}^T}\rho_{T\tilde{T}}^{max}\lrbracket{\mathbb{1}_{T}\otimes \sqrt{\sigma_{\tilde{T}}}}}\\
    &=\Trr{T}{\lrbracket{\mathbb{1}_T\otimes\sqrt{\sigma_{\tilde{T}}}\lrbracket{F_{\tilde{T}}(a_1\vert q_1)}^T}\rho_{T\tilde{T}}^{max}\lrbracket{\mathbb{1}_{T}\otimes \sqrt{\sigma_{\tilde{T}}}}}\\
    &=\sqrt{\sigma_{\tilde{T}}}\lrbracket{F_{\tilde{T}}(a_1\vert q_1)}^T\sqrt{\sigma_{\tilde{T}}}
\end{split}
\end{align}
where the transposition is w.r.t.\ the eigenbasis of $\sigma_{\tilde{T}}$. Simple algebraic manipulation yields the pretty-good measurements
\begin{align}
    F_{\tilde{T}}(a_1\vert q_1)= \lrbracket{\frac{1}{\sqrt{\sigma_{\tilde{T}}}}\alpha_{\tilde{T}}(a_1\vert q_1)\frac{1}{\sqrt{\sigma_{\tilde{T}}}}}^T,
\end{align}
where $1/\sqrt{\sigma}$ denotes the square root of the Moore--Penrose pseudo-inverse (note that $\operatorname{supp}\lrbracket{\alpha_{\tilde{T}}(a_1\vert q_1)}\subseteq\operatorname{supp}\lrbracket{\sigma_{\tilde{T}}}$, since $0\preccurlyeq\alpha_{\tilde{T}}(a_1\vert q_1)\preccurlyeq\sigma_{\tilde{T}}$) and, since $\lrvert{T}=\lrvert{\tilde{T}}$, we may trivially reinterpret the operators as acting on $T$ instead of $\tilde{T}$. Since the transposition is w.r.t.\ the eigenbasis of $\sigma_{T}$, we have $\lrbracket{\sigma_{T}}^T=\sigma_T$ and similarly $\lrbracket{\frac{1}{\sqrt{\sigma_{T}}}}^T=\frac{1}{\sqrt{\sigma_{T}}}$. Thus, 
\begin{align}
    \begin{split}
        F_{T}(a_1\vert q_1)= \lrbracket{\alpha_{T}(a_1\vert q_1)\frac{1}{\sqrt{\sigma_{T}}}}^T\lrbracket{\frac{1}{\sqrt{\sigma_{T}}}}^T=\lrbracket{\frac{1}{\sqrt{\sigma_{T}}}}\lrbracket{\alpha_T(a_1\vert q_1)}^T\lrbracket{\frac{1}{\sqrt{\sigma_{T}}}}. 
    \end{split}
\end{align}
Moreover, since $\alpha_{T}(a_1\vert q_1)\succcurlyeq 0$, transposition preserves positive semidefiniteness, and $X\mapsto\sigma_T^{-1/2}X\,\sigma_T^{-1/2}$ is a congruence, it follows that $F_T(a_1\vert q_1)\succcurlyeq 0$. Thus, for all $q_1\in Q_1$
\begin{align}
\begin{split}
    \sum_{a_1}F_T(a_1\vert q_1) &= \frac{1}{\sqrt{\sigma_{T}}}\sum_{a_1}\lrbracket{\alpha_{T}(a_1\vert q_1)}^T\frac{1}{\sqrt{\sigma_{T}}}= \frac{1}{\sqrt{\sigma_{T}}}\lrbracket{\sum_{a_1}\alpha_{T}(a_1\vert q_1)}^T\frac{1}{\sqrt{\sigma_{T}}}\\
    &=\frac{1}{\sqrt{\sigma_{T}}}\lrbracket{\sigma_{T}}^T\frac{1}{\sqrt{\sigma_{T}}} =\Pi_{\operatorname{supp}\lrbracket{\sigma_T}}.
\end{split}
\end{align}
If $\sigma_T$ is rank-deficient, we complete the POVM by replacing $F_T(a_1^0\vert q_1)$ with $F_T(a_1^0\vert q_1)+\mathbb{1}_T-\Pi_{\operatorname{supp}\lrbracket{\sigma_T}}$ for an arbitrary fixed $a_1^0\in A_1$; since $\Trr{\tilde{T}}{\rho^{\sigma}_{T\tilde{T}}}$ is supported on $\operatorname{supp}\lrbracket{\sigma_T}$, this leaves \autoref{eqn:steering_F_property} unchanged. In summary,
\begin{align}
    F_{A_1Q_1T}:=\sum_{a_1,q_1}\pi_1(q_1)\ket{a_1q_1}\bra{a_1q_1}\otimes F_T(a_1\vert q_1)
\end{align}
is a feasible point of $\lrbracket{\mathrm{I}}$. Since
\begin{align}
    \begin{split}
        \Tr\left[\left(V_{A_1A_2Q_1Q_2}\otimes S_{\hat{T}\tilde{T}}\right)\left(\alpha_{A_1Q_1\tilde{T}}\otimes D_{A_2Q_2\hat{T}}\right)\right] &\stackrel{\autoref{prop:swap_trick_renner}}{=}  \Tr\left[\left(V_{A_1A_2Q_1Q_2}\otimes \mathbb{1}_{\hat{T}}\right)\lrbracket{\alpha_{A_1Q_1\hat{T}}\otimes\mathbb{1}_{A_2Q_2}}\lrbracket{\mathbb{1}_{A_1Q_1}\otimes D_{A_2Q_2\hat{T}}}\right]\\
        &\stackrel{\autoref{eqn:steering_F_property}}{=} \Tr\left[\left(V_{A_1A_2Q_1Q_2}\otimes \rho_{T\hat{T}}^{\sigma}\right)\left( F_{A_1Q_1T}\otimes D_{A_2Q_2\hat{T}}\right)\right],
    \end{split}
\end{align}
where in the last line we relabeled $\tilde{T}\to\hat{T}$, the construction maps feasible points of $\lrbracket{\mathrm{II}}$ to feasible points of $\lrbracket{\mathrm{I}}$ while preserving the objective value, whence $\lrbracket{\mathrm{I}}\geq\lrbracket{\mathrm{II}}$; combining both directions yields equality of the optimal values.
\end{proof}

\begin{proof}[Proof of \autoref{lem:non_local_games_as_cbo_sym}]
Denote the problem in \autoref{eqn:cSEP_non_local_main} from \autoref{lem:non_local_games_as_cbo_sym} as $\lrbracket{\mathrm{I}}$ and the problem in \autoref{eqn:non_local_games_as_cSEP} by $\lrbracket{\mathrm{II}}$. We prove the equivalence between $\lrbracket{\mathrm{I}}$ and $\lrbracket{\mathrm{II}}$. Recall, $\lrvert{T}=\lrvert{\hat{T}}=\lrvert{\tilde{T}}$; we write $T$ for $\tilde{T}$ throughout this proof. It suffices to analyze the one constraint in which the two problems differ. Let $\beta_{A_1Q_1T}=\sum_{a_1,q_1}\pi_1(q_1)\ket{a_1q_1}\bra{a_1q_1}\otimes\beta_T(a_1\vert q_1)$ be a feasible point of $\lrbracket{\mathrm{II}}$. Hence, 
\begin{align}
    \sum_{a_1}\beta_T(a_1\vert q_1)=\sum_{a_1}\beta_T(a_1\vert q_1'),\,\hspace{0.5cm}\forall q_1'\in Q_1
\end{align}
implies 
\begin{align}
    \beta_T=\sum_{q_1}\pi_1(q_1)\sum_{a_1}\beta_T(a_1\vert q_1)=\sum_{a_1}\beta_T(a_1\vert q_1'),\,\hspace{0.5cm}\forall q_1'\in Q_1
\end{align}
since $\sum_{a_1}\beta_T(a_1\vert q_1)$ is independent of $q_1$ for all $q_1\in Q_1$. Thus, $\beta_{A_1Q_1T}$ is also a feasible point of $\lrbracket{\mathrm{I}}$ attaining the same objective value. As this holds for all feasible points, optimality is preserved. Conversely, let $\gamma_{A_1Q_1T}=\sum_{a_1,q_1}\pi_1(q_1)\ket{a_1q_1}\bra{a_1q_1}\otimes\gamma_T(a_1\vert q_1)$ be a feasible point of $\lrbracket{\mathrm{I}}$, and let $\gamma_T$ denote the operator appearing in its marginal constraint. Comparing the $Q_1$-blocks of
\begin{align}
    \Trr{A_1}{\gamma_{A_1Q_1T}}=\sum_{q_1}\pi_1(q_1)\ket{q_1}\bra{q_1}_{Q_1}\otimes\sum_{a_1}\gamma_T(a_1\vert q_1) \stackrel{!}{=} \sum_{q_1}\pi_1(q_1)\ket{q_1}\bra{q_1}_{Q_1}\otimes\gamma_T
\end{align}
yields $\sum_{a_1}\gamma_T(a_1\vert q_1)=\gamma_T$ for every $q_1\in Q_1$. In particular, $\sum_{a_1}\gamma_T(a_1\vert q_1)$ is independent of $q_1$ and equals $\sum_{a_1}\gamma_T(a_1\vert q_1=1)$, so $\gamma_{A_1Q_1T}$ is feasible for $\lrbracket{\mathrm{II}}$ with the same objective value. As this holds for all feasible points, the optimal values coincide.
\end{proof}
\section{Proof of \autoref{prop:restriction_to_cq_states}}\label{sec:proof_cq_states}

\begin{proof}
The channel $\mathcal{P}_{\mathrm{cl}}$ is itself a pinching,
\begin{align}\label{eqn:pinching_kraus}
    \mathcal{P}_{\mathrm{cl}}\lrbracket{X}=\sum_{\substack{a_1,\, q_1,\,\\ \vec{a}_2,\, \vec{q}_2}}\Pi_{a_1q_1\vec{a}_2\vec{q}_2}\,X\,\Pi_{a_1q_1\vec{a}_2\vec{q}_2}\,,\hspace{0.5cm}\Pi_{a_1q_1\vec{a}_2\vec{q}_2}:=\ket{a_1q_1}\bra{a_1q_1}_{A_1Q_1}\otimes\mathbb{1}_{T}\otimes\ket{\vec{a}_2\vec{q}_2}\bra{\vec{a}_2\vec{q}_2}\otimes\mathbb{1}_{\hat{T}_1^n}\,,
\end{align}
with mutually orthogonal projectors summing to the identity. We freely use the following elementary consequences of \autoref{eqn:pinching_kraus}:
(i) $\mathcal{P}_{\mathrm{cl}}$, and likewise any subcomposition of the factors in \autoref{eqn:pinching_channel_cl}, is completely positive, trace-preserving, unital and idempotent;
(ii) $\mathcal{P}_{\mathrm{cl}}$ is self-adjoint with respect to the Hilbert-Schmidt inner product, i.e.\ $\Trr{}{Y\,\mathcal{P}_{\mathrm{cl}}\lrbracket{X}}=\Trr{}{\mathcal{P}_{\mathrm{cl}}\lrbracket{Y}\,X}$, by cyclicity of the trace and $\Pi^2_{a_1q_1\vec{a}_2\vec{q}_2}=\Pi_{a_1q_1\vec{a}_2\vec{q}_2}$;
(iii) for any register $C$ appearing in \autoref{eqn:pinching_channel_cl} we have $\Tr_C\circ\,\mathcal{P}_C=\Tr_C$, and $\mathcal{P}_C$ commutes with partial traces over registers disjoint from $C$;
(iv) $X=\mathcal{P}_{\mathrm{cl}}\lrbracket{X}$ holds if and only if $\Pi_{a_1q_1\vec{a}_2\vec{q}_2}\, X\,\Pi_{a_1'q_1'\vec{a}_2'\vec{q}_2'}=0$ whenever $\lrbracket{a_1,q_1,\vec{a}_2,\vec{q}_2}\neq\lrbracket{a_1',q_1',\vec{a}_2',\vec{q}_2'}$, i.e.\ if and only if $X$ is of the form \autoref{eqn:cq_states_form}; for $X\succcurlyeq 0$ the positivity of the blocks follows since each $\rho^{\lrbracket{a_1,q_1,\vec{a}_2,\vec{q}_2}}_{T\hat{T}_1^n}$ is a compression of $X$.
 
We verify that $\mathcal{P}_{\mathrm{cl}}\lrbracket{\rhoge}$ satisfies all constraints of $\mathrm{SDP}_n\lrbracket{T,V,\pi}$. Positivity and normalization in \autoref{eqn:sdp_free_game_definetti_2} are preserved by (i). For the symmetry constraint, recall that the permutation representation factorizes as $U_{\An}\lrbracket{\pi}=U_{(A_2Q_2)_1^n}\lrbracket{\pi}\otimes U_{\hat{T}_1^n}\lrbracket{\pi}$. Hence, conjugation by $\mathbb{1}_{A_1Q_1T}\otimes U_{\An}\lrbracket{\pi}$ maps $\Pi_{a_1q_1\vec{a}_2\vec{q}_2}$ to $\Pi_{a_1q_1\pi\lrbracket{\vec{a}_2}\pi\lrbracket{\vec{q}_2}}$ with $\pi\lrbracket{\vec{a}_2}_i:=a_{2,\pi^{-1}(i)}$, i.e.\ it merely relabels the Kraus operators in \autoref{eqn:pinching_kraus}. Writing $\bar{U}_{\pi}:=\mathbb{1}_{A_1Q_1T}\otimes U_{\An}\lrbracket{\pi}$, we thus have
\begin{align}\label{eqn:pinching_covariance}
    \mathcal{P}_{\mathrm{cl}}\lrbracket{\bar{U}_{\pi}\,X\,\bar{U}_{\pi}^T}=\bar{U}_{\pi}\,\mathcal{P}_{\mathrm{cl}}\lrbracket{X}\,\bar{U}_{\pi}^T\hspace{0.5cm}\forall\pi\in S_n\,,
\end{align}
so the symmetry of $\rhoge$ in \autoref{eqn:sdp_free_game_definetti_2} implies that of $\mathcal{P}_{\mathrm{cl}}\lrbracket{\rhoge}$. For the first constraint in \autoref{eqn:sdp_constraints}, property (iii) yields, with $\mathcal{P}_{(A_2Q_2)_1^n}:=\mathcal{P}_{(A_2Q_2)_1}\circ\ldots\circ\mathcal{P}_{(A_2Q_2)_n}$ and $\mathcal{P}_{A_1Q_1}=\mathcal{P}_{A_1}\circ\mathcal{P}_{Q_1}$,
\begin{align}
    \begin{split}
    &\Trr{A_1}{\mathcal{P}_{\mathrm{cl}}\lrbracket{\rhoge}}=\lrbracket{\mathcal{P}_{Q_1}\circ\mathcal{P}_{(A_2Q_2)_1^n}}\lrbracket{\Trr{A_1}{\rhoge}}\\
    &\stackrel{\autoref{eqn:sdp_constraints}}{=}\lrbracket{\mathcal{P}_{Q_1}\circ\mathcal{P}_{(A_2Q_2)_1^n}}\lrbracket{\sum_{q_1}\pi_1(q_1)\ket{q_1}\bra{q_1}\otimes\rho_{T\An}}\\
    &=\sum_{q_1}\pi_1(q_1)\ket{q_1}\bra{q_1}\otimes\mathcal{P}_{(A_2Q_2)_1^n}\lrbracket{\rho_{T\An}}\\
    &=\sum_{q_1}\pi_1(q_1)\ket{q_1}\bra{q_1}\otimes\Trr{A_1Q_1}{\mathcal{P}_{\mathrm{cl}}\lrbracket{\rhoge}}\,,
    \end{split}
\end{align}
where the last equality holds again by (iii). For the second constraint in \autoref{eqn:sdp_constraints}, note first that $\Trr{A_1Q_1T}{\mathcal{P}_{\mathrm{cl}}\lrbracket{\rhoge}}=\mathcal{P}_{(A_2Q_2)_1^n}\lrbracket{\rho_{\An}}$ by (iii), and thus
\begin{align}
    \begin{split}
    &\Trr{(A_2)_1}{\mathcal{P}_{(A_2Q_2)_1^n}\lrbracket{\rho_{\An}}}=\lrbracket{\mathcal{P}_{(Q_2)_1}\circ\mathcal{P}_{(A_2Q_2)_2^n}}\lrbracket{\Trr{(A_2)_1}{\rho_{\An}}}\\
    &\stackrel{\autoref{eqn:sdp_constraints}}{=}\lrbracket{\mathcal{P}_{(Q_2)_1}\circ\mathcal{P}_{(A_2Q_2)_2^n}}\lrbracket{\lrbracket{\sum_{q_2}\pi_2(q_2)\ket{q_2}\bra{q_2}_{(Q_2)_1}\otimes\frac{\mathbb{1}_{\hat{T}_1}}{\lrvert{T}}}\otimes\rho_{(A_2Q_2\hat{T})_2^n}}\\
    &=\lrbracket{\sum_{q_2}\pi_2(q_2)\ket{q_2}\bra{q_2}_{(Q_2)_1}\otimes\frac{\mathbb{1}_{\hat{T}_1}}{\lrvert{T}}}\otimes\mathcal{P}_{(A_2Q_2)_2^n}\lrbracket{\rho_{(A_2Q_2\hat{T})_2^n}}\,,
    \end{split}
\end{align}
where $\mathcal{P}_{(A_2Q_2)_2^n}\lrbracket{\rho_{(A_2Q_2\hat{T})_2^n}}$ is precisely the $\lrbracket{A_2Q_2\hat{T}}_2^n$-marginal of $\mathcal{P}_{\mathrm{cl}}\lrbracket{\rhoge}$, again by (iii). Lastly, for the objective, (iii) gives $\Trr{(A_2Q_2\hat{T})_2^n}{\mathcal{P}_{\mathrm{cl}}\lrbracket{\rhoge}}=\lrbracket{\mathcal{P}_{A_1Q_1}\circ\mathcal{P}_{(A_2Q_2)_1}}\lrbracket{\rhog}$. Since $V_{A_1A_2Q_1Q_2}$ is diagonal in the distinguished bases of the classical registers and $S_{T\hat{T}}$ acts trivially on them, $V_{A_1A_2Q_1Q_2}\otimes S_{T\hat{T}}$ commutes with every Kraus operator of $\mathcal{P}_{A_1Q_1}\circ\mathcal{P}_{(A_2Q_2)_1}$, and (ii) implies
\begin{align}
    \begin{split}
    &\Trr{}{\lrbracket{V_{A_1A_2Q_1Q_2}\otimes S_{T\hat{T}}}\lrbracket{\mathcal{P}_{A_1Q_1}\circ\mathcal{P}_{(A_2Q_2)_1}}\lrbracket{\rhog}}\\
    &=\Trr{}{\lrbracket{V_{A_1A_2Q_1Q_2}\otimes S_{T\hat{T}}}\,\rhog}\,.
    \end{split}
\end{align}
Hence $\mathcal{P}_{\mathrm{cl}}\lrbracket{\rhoge}$ is feasible for $\mathrm{SDP}_n\lrbracket{T,V,\pi}$ with the same objective value; the identical computation applies to $\mathrm{SDP}^{(\Phi)}_n\lrbracket{T,V,\pi}$, as $\Phi_{T\vert\hat{T}}$ likewise acts trivially on the classical registers. Finally, the feasible set is compact and the objective continuous, so the maximum is attained; composing any optimizer with $\mathcal{P}_{\mathrm{cl}}$ and invoking (i), (iv) yields an optimizer of the form \autoref{eqn:cq_states_form}.
\end{proof}
\section{Bose-symmetry via Schur-Weyl duality}\label{sec:bose_via_schur_weyl}

Central to the complexity result in \autoref{sec:symmetric_subspace_methods} is the intricate interplay between two groups that act naturally on $\SymH$, viewed as a subspace of $\HSn$. Concretely, for a Hilbert space $\HS$, the corresponding tensor product space $\HSn$ naturally admits a representation of $S_n$ and $\GLH$\footnote{That is, there exist group homomorphisms from $S_n$ and from $\GLH$ into $\operatorname{GL}\lrbracket{\HSn}$; both are made explicit below.}. Explicitly, let $S_n$ act\footnote{Each $\pi\in S_n$ is represented by a unitary permutation matrix $U_{\HSn}(\pi)$; since $S_n$ acts from the right, the assignment $\pi\mapsto U_{\HSn}(\pi)$ is a group antihomomorphism.} on $\HSn$ from the right by permuting the tensor factors of any element $\bigotimes_{i=1}^n v_i \in \HSn$, i.e.\ for $\pi\in S_n$,
\begin{align}\label{eq:nat_S_n_action}
	\lrbracket{\bigotimes_{i=1}^n v_i}\pi= \bigotimes_{i=1}^n v_{\pi(i)},
\end{align}
and let the Lie group $\GLH$ act on each tensor factor simultaneously from the left\footnote{The direction of the action is interchangeable}, i.e.\ for any $g\in\GLH$,
\begin{align}
	g\lrbracket{\bigotimes_{i=1}^n v_i} = \bigotimes_{i=1}^n g\cdot v_i.\,
\end{align} Importantly, these two actions commute, and in fact the algebras of $S_n$- and $\GLH$-equivariant maps from $\HSn$ to $\HSn$ are full mutual centralizers in $\End{}{\HSn}$. Schur-Weyl duality formalizes the relationship between the commuting actions of the finite-dimensional group algebra $\CSn$  and the general linear group $\GLH$, demonstrating how these actions induce a decomposition of $\HSn$ into simple modules. This duality establishes a correspondence between the simple $\CSn$-modules and the simple $\GLH$-modules within $\HSn$, providing a powerful algebraic framework for analyzing the structure of the symmetric subspace. See \autoref{ex:Schur_weyl_duality} for an application of Schur-Weyl duality. Next, we introduce the $\GLH$-submodules in $\HSn$ via  Weyl's tensorial construction and connect them to the symmetric subspace.

 Consider the decomposition of the complex group algebra $\CSn$ as a $\CSn$-module (regular representation) given by the Artin--Wedderburn theorem
\begin{align}\label{eqn:decomp_Sn_group_alebra}
	\CSn \simeq_{\CSn} \bigoplus_{\lambda\vdash n} \HS_{\lambda}^{\oplus\dim_{\CC}\lrbracket{\HS_{\lambda}}},
\end{align}
where for each $\lambda$, $\SpechtH = \CSn\cdot c_{\lambda}$ with Young symmetrizer $c_{\lambda}$ is a simple $\CSn$-module called a Specht module. By the Artin--Wedderburn theorem, each isotypical component $\HS_{\lambda}^{\oplus\dim_{\CC}\lrbracket{\HS_{\lambda}}}$ is the image of a primitive idempotent from the center of the group algebra $Z\lrbracket{\CSn}$. Note, while in general Young symmetrizers are minimal (primitive) idempotents, they are usually not central. In general, only suitable linear combinations of Young symmetrizers are again central. The primary advantage of Young symmetrizers, however, lies in the ability to explicitly construct a Young symmetrizer using a standard labeling of Young shape $Y(\lambda)$.

Weyl's tensorial construction extends to the case where the representation space is not the group algebra but rather the $n$-fold tensor power of a vector space — or, in our context, a Hilbert space, i.e.\ $\HSn$. Since $c_{(n)}\in Z\lrbracket{\CSn}$ is proportional to $P_{\SymH}$, we have $\SymH=\HSn\cdot c_{(n)}$. In general we define $\Schurf{\lambda}\HS := \HSn\cdot c_{\lambda}$ together with a Schur functor $\Schurf{\lambda}\lrbracket{\cdot}$ on the category of finite dimensional vector spaces. The functor takes the object $\HS$ to $\Schurf{\lambda}\HS$, and a morphism ($\CC$-linear map) on $\HS$ to the induced morphism on $\Schurf{\lambda}\HS$ obtained by restricting its $n$-th tensor power.
 \begin{remark}
 On a technical level, the Schur functor of partition $\lambda$ 
	\begin{align}
		\Schurf{\lambda}\, : \, \text{FinVect} \rightarrow \text{FinVect}
	\end{align}
	maps the category of finite dimensional vector spaces and linear maps thereon to itself. If $f\,:\, V\rightarrow W$ is a linear map between finite dimensional vector spaces $V, W$, we define
	\begin{align}
	\begin{split}
		\Schurf{\lambda}(f) \, : \, &\Schurf{\lambda}V \rightarrow \Schurf{\lambda}W,\,\\
	\end{split}
	\end{align}
	with $\Schurf{\lambda}(f\circ g)= \Schurf{\lambda}(f)\circ \Schurf{\lambda}(g)$ for composable linear maps $f$ and $g$, and $\Schurf{\lambda}\lrbracket{\mathcal{I}_V}=\mathcal{I}_{\Schurf{\lambda}V}$. Strictly speaking the functor $\Schurf{\lambda}\lrbracket{\cdot}$ is distinct from its image space $\Schurf{\lambda}\HS$ which we refer to as a Weyl module of the reductive group $\GLH\simeq \text{GL}_{d_{\HS}}(\CC)$. See \cite[Sec. 6]{fulton2013representation}  and \cite[Appendix A]{macdonald1998symmetric} for a more in-depth discussion on polynomial functors.
 \end{remark}
 Importantly, the Schur functor carries the $\GLH$-module structure of $\HS$ over to $\Schurf{\lambda}\HS$. Next, we will see how the Weyl modules relate to simple $\CSn$-modules. 
\begin{proposition}
We have
	\begin{align}
		\Schurf{\lambda}\HS \simeq \text{Hom}_{\CSn}\lrbracket{\HS_{\lambda}, \HSn}.
	\end{align}
\end{proposition}
\begin{proof}
We have 	
\begin{align}
\begin{split}
	\Schurf{\lambda}\HS &= \HSn \cdot c_{\lambda}\simeq \CSn \otimes_{\CSn} \HSn c_{\lambda}= \CSn c_{\lambda} \otimes_{\CSn} \HSn\\
	&\stackrel{\HS_{\lambda}=\CSn\cdot c_{\lambda}}{=} \SpechtH \otimes_{\CSn} \HSn\stackrel{\SpechtH^*\simeq \SpechtH}{\simeq} \SpechtH^*\otimes_{\CSn} \HSn \simeq \Hom{\CSn}{\SpechtH}{\HSn}.\,
\end{split}
\end{align}
In the first line we used the definition of the $\CSn$-balanced product, i.e.\ the tensor product between right $\CSn$-module $\CSn$ and left $\CSn$-module $\HSn$ (see \autoref{rem:KG_balanced_product} and \autoref{prop:fulton_harris_1}). The second equality in the first line follows from the property of the $\CSn$-balanced product to move $c_{\lambda}$ from the second into the first $\CSn$-module tensor factor. The last line follows from the definition of the contragredient (dual) $\CSn$-module $\SpechtH^* = \Hom{\CC}{\SpechtH}{\CC}$ isomorphic to $\SpechtH$.
\end{proof}

 Thus, the symmetric subspace is the space of $\CSn$-invariant maps from the trivial Specht module $\HS_{(n)}$ to $\HSn$. By Schur's lemma, any such map is either an isomorphism between $\CSn$-modules or the zero map. Thus, every non-zero element of $\Schurf{(n)}\HS$ is an isomorphism onto a copy of the trivial Specht module inside $\HSn$ and $\dim_{\CC}\lrbracket{\Schurf{(n)}\HS}$ is the number of times (multiplicity) $\HS_{(n)}$ appears in the $\CSn$-module $\HSn$.
\begin{remark}
	While the Young symmetrizers are primitive idempotents in $\CSn$, i.e.\ they yield simple $\CSn$-modules, this property does not hold in $\HSn$.
\end{remark}

As the $\GLH$-action commutes with the $\CSn$ action on $\HSn$, the image of $\GLH$ in $\End{}{\HSn}$ is in $\End{\CSn}{\HSn}$. Furthermore, the image of any element of $\CSn$ is in $\End{\CSn}{\HSn}$ if and only if it is in $Z\lrbracket{\CSn}$. We present the following proposition, which establishes Weyl modules — particularly the symmetric subspace — as submodules of $\HSn$, viewed as a module over the algebra of $\CSn$-invariant endomorphisms.

\begin{proposition}[\cite{fulton2013representation}, Lemma 6.23]\label{prop:fulton_harris_2}
Let $\HS$ be a finite dimensional complex vector space. The algebra $\End{\CSn}{\HSn}$ is spanned as a linear subspace of $\End{}{\HSn}$ by the operators $g^{\otimes n}$ with $g\in\End{}{\HS}$. A subspace of $\HSn$ is a sub-$\End{\CSn}{\HSn}$-module if and only if it is invariant under $\GLH$, i.e.\ a $\GLH$-submodule.
\end{proposition}

The character of the $\GLH$-module $\Schurf{\lambda}\HS$, evaluated at $g\in\GLH$ with eigenvalues $x_1,\ldots,x_{d_{\HS}}$, is
\begin{align}
    \chi_{\Schurf{\lambda}\HS}(g)=S_{\lambda}\lrbracket{x_1,\ldots,x_{d_{\HS}}},
\end{align}
where $S_{\lambda}\lrbracket{x_1,\ldots,x_{d_{\HS}}}$ is the Schur polynomial (see e.g.\ \cite{macdonald1998symmetric}) of degree $n$ in $d_{\HS}$ variables. Concretely, for the symmetric subspace
\begin{align}
     \chi_{\Schurf{(n)}\HS}(g)=H_{n}\lrbracket{x_1,\ldots,x_{d_{\HS}}},
\end{align}
where $H_{n}\lrbracket{x_1,\ldots,x_{d_{\HS}}}$ is the complete homogeneous symmetric polynomial of degree $n$ (see e.g.\ \cite[Appendix A]{fulton2013representation} for a relation between these polynomials). The Schur polynomials can be given in terms of semistandard Young tableaux
\begin{align}
S_{\lambda}\lrbracket{x_1,\ldots ,x_{d_{\HS}}} = \sum_{T\in \mathcal{T}_{\lambda,\, d_{\HS}}}x_1^{t_1}\cdots x_{d_{\HS}}^{t_{d_{\HS}}}=\sum_{\mu}K_{\lambda\mu}m_{\mu}
\end{align}
where the first sum is indexed by semistandard labelings $T$ of Young shape $\lambda$ with weight $(t_1,\ldots, t_{d_{\HS}})$. Alternatively, Kostka numbers $K_{\lambda\mu}$ quantify how many semistandard labelings of Young shape $Y(\lambda)$ of weight $\mu$ exist and $m_{\mu}$ are monomial symmetric functions. While a general closed form formula for $K_{\lambda\mu}$ does not exist, for $\lambda=(n)$ and any composition $\mu$ of $n$ into at most $d_{\HS}$ parts we have $K_{\lambda\mu}=1$.
In general, 
\begin{align}
	\dim_{\CC}\lrbracket{\Schurf{\lambda}\HS}= S_{\lambda}(1,\ldots, 1)= \prod_{(i,j)\in Y(\lambda)}\frac{d_{\HS}-i+j}{h_{\lambda}(i,j)}
\end{align}
with hook length $h_{\lambda}(i,j)$, i.e.\ the number of cells in the hook $H_{\lambda}(i,j)$. Then, by \cite[Theorem 6.3]{fulton2013representation} we recover the $\CC$-vector space dimension of the symmetric subspace given by \cite{harrow2013church} in
\begin{align}
    \dim_{\CC}\lrbracket{\SymH}= H_{n}\lrbracket{1,\ldots,1} = \binom{d_{\HS}+n-1}{n}.\,
\end{align}

Furthermore, the characters certify that some $\Schurf{\lambda}\HS$ are zero, if the dimension of $\HS$ is small with respect to $n$. 
\begin{proposition}[\cite{fulton2013representation} Theorem 6.3 ]
	Let $k=\dim_{\CC}\lrbracket{\HS}$. Then, $\Schurf{\lambda}\HS$ is zero if $\lambda_{k+1}\neq 0$. 
\end{proposition}

The commutativity of the $\CSn$- and $\GLH$-actions also translates the simplicity of Specht modules to the Weyl modules. By \autoref{prop:fulton_harris_lemma_6_22} (ii), the simplicity of the $\CSn$-module $\CSn\cdot c_{\lambda}$ implies the simplicity of $\HSn \otimes_{\CSn}\CSn\cdot c_{\lambda}$ as an$\End{\CSn}{\HSn}$-module. Furthermore, by \autoref{prop:fulton_harris_lemma_6_22} (i)  \begin{align}
	\HSn \otimes_{\CSn}\CSn\cdot c_{\lambda} \simeq_{\End{\CSn}{\HSn}} \HSn\cdot c_{\lambda}
\end{align} and thus by \autoref{prop:fulton_harris_2} the spaces are also isomorphic as $\GLH$-modules. In summary, we have the following proposition.
	
\begin{proposition}[\cite{fulton2013representation}, Theorem 6.3]
	For each partition $\lambda\vdash_{d_{\HS}} n$, the Weyl module $\Schurf{\lambda}\HS$ is a simple $\GLH$-module. 
\end{proposition}

\begin{remark}
	While in this setting, the $\Schurf{\lambda}\HS$ are simple $\GLH$-modules for every partition $\lambda$ of height at most $d_{\HS}$, there exist other (e.g.\ dual representations) simple $\GLH$-modules which cannot be given in terms of $\Schurf{\lambda}\HS$.
\end{remark}

Schur-Weyl duality provides the following decompositions of $\HSn$. For the reader's convenience, we give a concise proof.
\begin{proposition}[Schur-Weyl Duality for $\CSn\times \GLH$]\label{prop:schur_weyl_duality} We have the decomposition of $\HSn$ as a $\CSn\times \GLH$-module (i.e.\ $\lrbracket{\CSn, \GLH}$-bimodule)
	\begin{align}\label{eqn:schur_weyl_S_n_GL_H}
	\begin{split}
	\HSn \simeq \bigoplus_{\lambda\vdash_{d_{\HS}} n} \HS_{\lambda} \otimes \Schurf{\lambda}\HS	\end{split}
\end{align}
or as a $\GLH$-module
\begin{align}\label{eqn:schur_weyl_GL_H}
	\HSn\simeq \bigoplus_{\lambda\vdash_{d_{\HS}} n}\lrbracket{\Schurf{\lambda}\HS}^{\oplus \dim_{\CC}\lrbracket{\SpechtH}}.
\end{align}
\end{proposition}
\begin{proof}
Follows directly from an application of the double centralizer theorem
	\begin{align}
		\begin{split}
			\HSn &\simeq \HSn \otimes_{\CSn} \CSn \stackrel{\autoref{eqn:decomp_Sn_group_alebra}}{\simeq} \HSn\otimes_{\CSn} \bigoplus_{\lambda} \SpechtH^{\oplus \dim_{\CC}\lrbracket{\SpechtH}}\\
			&\simeq \bigoplus_{\lambda} \lrbracket{\HSn \otimes_{\CSn} \SpechtH}^{\oplus \dim_{\CC}\lrbracket{\SpechtH}} \simeq \bigoplus_{\lambda} \lrbracket{\Schurf{\lambda}\HS}^{\oplus \dim_{\CC}\lrbracket{\SpechtH}}\\
			&\stackrel{\autoref{prop:V_decomp}}{\simeq} \bigoplus_{\lambda} \SpechtH \otimes \Hom{\CSn}{\SpechtH}{\HSn} \simeq  \bigoplus_{\lambda} \SpechtH \otimes \Schurf{\lambda} \HS.\,
		\end{split}
	\end{align}
	Note that we used the semisimple decomposition of the regular $\CSn$-module, in which each Specht module $\SpechtH$ appears with multiplicity equal to its dimension $\dim_{\CC}\lrbracket{\SpechtH}$.
\end{proof}

\begin{remark}\label{rem:at_most_d}
	Recall, $\Schurf{\lambda}V=0$ for $\lambda$ a partition of $n$ of height greater than $\dim V$. In other words, the direct sum in \autoref{eqn:schur_weyl_S_n_GL_H} and \autoref{eqn:schur_weyl_GL_H} runs through all partitions $\lambda$ of $n$ with height at most $\dim V$. 
\end{remark}
Thus, considered as a $\GLH$-module, $\HSn$ decomposes into a direct sum of Weyl modules with multiplicity given by the $\CC$-vector space dimension of the corresponding Specht module. Importantly, the symmetric subspace is one component in this decomposition. As an example consider \autoref{ex:complex_Weyl_decomp}. 

We are interested in operators with support and range on the symmetric subspace, i.e.\ elements from $\End{}{ \vee^{n}\lrbracket{\HS}}$.
\begin{proposition}\label{prop:Bose_sym_space_char}
We have
	\begin{align}
		\End{}{ \vee^{n}\lrbracket{\HS}}\simeq \End{\CSn}{ \vee^{n}\lrbracket{\HS}}
	\end{align}
	and 
	\begin{align}
		\dim_{\CC}\lrbracket{\End{}{ \vee^{n}\lrbracket{\HS}}}=\binom{d_{\HS}+n-1}{n}^2\leq \lrbracket{n+1}^{2d_{\HS}}.\,
	\end{align}
\end{proposition}
\begin{proof}
The isomorphism follows from the triviality of the $S_n$-action on the symmetric subspace. Trivially,
\begin{align}
    \End{\CSn}{ \vee^{n}\lrbracket{\HS}}\subseteq \End{}{ \vee^{n}\lrbracket{\HS}}.\,
\end{align}
Conversely, every $\pi\in S_n$ acts as the identity on $\vee^{n}\lrbracket{\HS}$, so the restriction of the $\CSn$-action to $\vee^{n}\lrbracket{\HS}$ is by scalar multiples of the identity. Hence every endomorphism of $\vee^{n}\lrbracket{\HS}$ commutes with this action, i.e.
\begin{align}
   \End{}{ \vee^{n}\lrbracket{\HS}} \subseteq \End{\CSn}{ \vee^{n}\lrbracket{\HS}}.\,
\end{align}
Furthermore, the $\CC$ vector space dimension follows from the character together with the definition of the symmetric subspace.
\end{proof}
\begin{remark}
    Note that, when considered as a vector space, the symmetric subspace of
 $\End{}{\HSn}\simeq\End{}{\HS}^{\otimes n}$ -- i.e.\ $\vee^n\lrbracket{\End{}{\HS}}$ -- is distinct from $\End{}{\SymH}$.
\end{remark}

\subsection{Full matrix $*$-algebra structure of Bose-symmetric operators}\label{sec:algebraic_iso}

We prove that the space of Bose-symmetric operators is isomorphic, as an algebra, to a full matrix $*$-algebra of polynomial size. We then demonstrate the explicit construction of this isomorphism. Although the subsequent concepts are well established in the mathematical literature, we include the following material for the convenience of readers with a background in quantum information.

We require the decomposition of $\HSn$ as $\CSn$, $\GLH$ and $\CSn\times \GLH$-modules. For the reader's convenience, relevant mathematical background is provided in \autoref{sec:representation_theory}. 

While the decomposition of $\HSn$ as a $\GLH$-module yields the symmetric subspace component, in order to consider $\End{\CSn}{\SymH}$ as a full matrix algebra, we need to consider its decomposition as a $\CSn$-module. Thus, we introduce an additional step in the decomposition of $\HSn$ to establish the computational complexity of $\Sdpbose(G)$.

It is a classical result that any unital matrix $*$-algebra is isomorphic to a direct sum of full matrix algebras. In what follows, we concretize this result in our specific setting and subsequently construct an explicit $*$-isomorphism.

\begin{proposition}
The space $\mathcal{B}:=\End{\CSn}{\SymH}$ is a unital matrix $*$-algebra.
\end{proposition}

\begin{proof}
Since $\HS$ is of finite dimension $d_{\HS}$, $\mathcal{B}\subset \CC^{d_{\HS}^n\times d_{\HS}^n}$. Clearly, by linearity $\mathcal{B}$ is closed under addition and scalar multiplication. 
For $T_1,T_2\in \mathcal{B}$, $v\in \SymH $ and any $\sigma\in S_n$ we have
\begin{align}
    \lrbracket{T_1 T_2}\lrbracket{\sigma\cdot v} = T_1\lrbracket{T_2\lrbracket{\sigma\cdot v}}= T_1\lrbracket{\sigma \cdot T_2(v)}= \sigma\cdot T_1\lrbracket{T_2(v)}=\sigma\cdot\lrbracket{T_1T_2}(v)
\end{align}
extending to $\CSn$ via linearity. Thus, $\mathcal{B}$ is closed under matrix multiplication. Since $S_n$ is finite, its group action is unitary and thus
\begin{align}
    \left\langle \sigma\cdot v, T^*w \right\rangle = \left\langle T\lrbracket{\sigma\cdot v}, w \right\rangle = \left\langle  T(v), \sigma^{-1}\cdot w \right\rangle = \left\langle  v, T^{*}\lrbracket{\sigma^{-1}\cdot w }\right\rangle
\end{align}
implying $T^{*}\lrbracket{\sigma \cdot v}= \sigma\cdot T^{*}(v)$. If $\mathcal{B}$ is viewed as a subalgebra of $\End{}{\HSn}$, the unit is given by the corresponding orthogonal projector onto $\SymH$.
\end{proof}

\begin{remark}
    For finite dimensional $\HS\simeq\CC^{d_{\HS}}$, the canonical inner product $\langle x, y\rangle =y^{\dagger}x$ is $S_n$-invariant. More generally, for a finite group $G$ acting on $\CC^Z$ where $Z$ is a set of finite cardinality, any permutation representation
	\begin{align}
	\begin{split}
		R\,:\, &G\rightarrow \operatorname{GL}\lrbracket{\CC^Z}\\
		&g \mapsto  R_g,\,
	\end{split}
	\end{align}
	 is unitary \cite[Theorem 3.11]{etingof2009introduction} and thus,
	 \begin{align}
	 	 \left\langle R_g v, R_gw \right\rangle = \lrbracket{R_g w}^{\dagger}R_g v = w^{\dagger}R_g^{\dagger}R_gv= w^{\dagger}v = \left\langle v, w \right\rangle
	 \end{align}
	 proves the $G$-invariance of the inner product. Importantly, the unitarity of the action implies that the orthogonal complement of any $\CSn$-submodule of $\SymH$ is itself a $\CSn$-submodule.
\end{remark}

\begin{proposition}
    The matrix $*$-algebra $\End{\CSn}{\SymH}$ is simple.
\end{proposition}

\begin{proof}
    By \autoref{prop:Bose_sym_space_char}, $\End{\CSn}{\SymH}\simeq\End{}{\SymH}$ is the full endomorphism algebra of a finite-dimensional complex vector space. Such an algebra contains no non-trivial (proper) two-sided ideals; equivalently, its center is given by the $\CC$-span of the unit of $\End{\CSn}{\SymH}$.
\end{proof}
Thus, it is isomorphic as an algebra to a single full matrix algebra.

Next, we exploit the fact that while the Weyl modules $\Schurf{\lambda}\HS$ are irreducible as $\GLH$-modules, they become reducible when viewed as $\CSn$-modules. By Maschke's theorem, $\CSn$ is semisimple and thus any $\CSn$-module decomposes into a direct sum of simple Specht modules\footnote{For fields of characteristic zero, i.e.\ $\CC$, Specht modules are simple.}. Importantly, the symmetric subspace $\Schurf{(n)}\HS$ is isomorphic, as a $\CSn$-module, to a direct sum of polynomially many (in $n$) isomorphic copies of the 1-dimensional trivial Specht module.

\begin{proposition}
The symmetric subspace $\SymH$ is isomorphic as a $\CSn$-module to a direct sum of $m_{(n)}=\binom{d_{\HS}+n-1}{n}$ many simple Specht modules isomorphic to the trivial Specht module $\HS_{(n)}$, i.e.\
	\begin{align}
		\SymH\simeq_{\CSn} \HS_{(n)}^{\oplus \dim_{\CC}\lrbracket{\SymH}}.\,
	\end{align}
\end{proposition}

\begin{proof}
	By definition, $\sigma \cdot v = v$ for any $v\in\SymH$ and every $\sigma\in S_n$. Thus, $v$ is an eigenvector with eigenvalue 1 under every group element of $S_n$. Hence, the restriction of the $S_n$-action to $\SymH$ is trivial: each $\sigma\in S_n$ acts as the identity on $\SymH$. Thus, $\text{span}_{\CC}\lrbracket{v}$ is a simple $\CSn$-module isomorphic to the trivial Specht module. The multiplicity corresponds to the $\CC$-vector space dimension of $\SymH$.
\end{proof}

In other words, the isotypic component of the trivial representation in $\HSn$ is precisely $\SymH$. In general, we have the following proposition.

\begin{proposition}\label{prop:Weyl_to_Specht}
The Weyl modules in the Weyl tensorial construction are isomorphic as $\CSn$-modules to a direct sum of isomorphic Specht modules
	\begin{align}
		\Schurf{\lambda}\HS^{\oplus\dim_{\CC}\lrbracket{\HS_{\lambda}}}\simeq \HS_{\lambda}^{\oplus \dim_{\CC}\lrbracket{\Schurf{\lambda}\HS}}.\,
	\end{align}
\end{proposition}

\begin{proof}
With \autoref{prop:V_decomp} we have
	\begin{align}
		\Schurf{\lambda}\HS^{\oplus\dim_{\CC}\lrbracket{\HS_{\lambda}}} \simeq \SpechtH\otimes\Hom{\CSn}{\SpechtH}{\HSn} \simeq \SpechtH^{\oplus \dim_{\CC}\lrbracket{\Schurf{\lambda}\HS}}.\,
	\end{align}
\end{proof}
By Schur's lemma, $\CSn$-invariant maps with support and range on $\SymH$ decompose accordingly.

\begin{proposition}
We have 
	\begin{align}
		\End{\CSn}{\SymH} \simeq \CC^{m_{(n)}\times m_{(n)}}.\,
	\end{align}
\end{proposition}

\begin{proof}
In short, due to Schur's lemma, $\operatorname{Hom}$-space additivity and Artin--Wedderburn,
\begin{align}
	\begin{split}
		\End{\CSn}{\Schurf{\lambda}\HS^{\oplus \dim_{\CC}\lrbracket{\SpechtH}}}&\simeq \Hom{\CSn}{\SpechtH\otimes\Hom{\CSn}{\SpechtH}{\HSn}}{\Schurf{\lambda}\HS^{\oplus \dim_{\CC}\lrbracket{\SpechtH}}}\\
		&\simeq \Hom{\CSn}{\SpechtH^{\oplus \dim_{\CC}\lrbracket{\Schurf{\lambda}\HS}}}{\Schurf{\lambda}\HS^{\oplus \dim_{\CC}\lrbracket{\SpechtH}}} \simeq \End{\CSn}{\SpechtH^{\oplus \dim_{\CC}\lrbracket{\Schurf{\lambda}\HS}}}\\
		&\simeq M_{\dim_{\CC}\lrbracket{\Schurf{\lambda}\HS}}\lrbracket{\End{\CSn}{\SpechtH}} \simeq M_{\dim_{\CC}\lrbracket{\Schurf{\lambda}\HS}}\lrbracket{\CC}.\,
	\end{split}
	\end{align}
	Specifically, 
	\begin{align}
		\End{\CSn}{\SymH}\simeq M_{\dim_{\CC}\lrbracket{\Schurf{(n)}\HS}}\lrbracket{\CC}.\,
	\end{align}
\end{proof}

Importantly, as $\End{\CSn}{\SymH}$ and $M_{\binom{n+d-1}{n}}\lrbracket{\CC}$ are $*$-isomorphic as unital algebras, any such mapping preserves eigenvalues and thus positive semidefiniteness. The explicit construction of the $*$-isomorphism requires a decomposition of $\SymH$ as a $\CSn$-module. Since $\SymH$ is invariant under the action of $\CSn$, any subspace of $\SymH$ is similarly invariant. Consequently, we can identify $m_{(n)}$ distinct subspaces, each of which is isomorphic as a $\CSn$-module to the trivial Specht module $\HS_{(n)}$. Let $\mathcal{U}_{(n)}\subset\SymH$ be isomorphic as a $\CSn$-module to the trivial Specht module. Let $e_{(n)}\in \mathcal{U}_{(n)}\backslash\lrbrace{0}$ be a non-zero vector in this module. Importantly, since $\mathcal{U}_{(n)}$ is simple, $\mathcal{U}_{(n)}\simeq \CSn\cdot e_{(n)}$ by cyclicity.

Furthermore, consider the image of $e_{(n)}$ under $\End{\CSn}{\SymH}$, i.e.\
\begin{align}
    W_{(n)}:=\lrbrace{Ae_{(n)} \,\vert\, A\in\End{\CSn}{\SymH}},\,
\end{align}
with an orthonormal basis $B_{(n)}$. We have the following corollary of \cite[Theorem 3]{gijswijt2009block}.

\begin{proposition}
The map
    \begin{align}
        \begin{split}
            \psi \,:\,  \End{\CSn}{\SymH} &\rightarrow \CC^{m_{(n)}\times m_{(n)}}\\
			 A &\mapsto \lrbracket{\left\langle Ab, b' \right\rangle}_{b, b'\in B_{(n)}}
        \end{split}
    \end{align}
    is an algebra $*$-isomorphism.
\end{proposition}

Each $A\in  \End{\CSn}{\SymH}$ is thus mapped to its restriction to $W_{(n)}$ which in turn is the image of $e_{(n)}$ under the action of $\End{\CSn}{\SymH}$. Thus, the $*$-isomorphism can be made explicit given a non-zero vector in $\mathcal{U}_{(n)}$ and an explicitly computable\footnote{Since $\End{\CSn}{\SymH}$ as an algebra is closed under multiplication, $AA'$ is again in $\End{\CSn}{\SymH}$.} $\left\langle AA'e_{(n)}, A''e_{(n)}\right\rangle$ for any $A, A', A''$ from a basis of $\End{\CSn}{\SymH}$. Since isomorphic $\CSn$-modules are connected by a $\CSn$-invariant map, i.e.\ an element from $\End{\CSn}{\SymH}$, the expressions $A'e_{(n)}$ and $A''e_{(n)}$ correspond to $\CSn$-modules isomorphic to $\mathcal{U}_{(n)}$.

\begin{proposition}
We have  $W_{(n)} \simeq \SymH$.
\end{proposition}
\begin{proof}
Recall,
    \begin{align}
        \SymH\simeq \bigoplus_{i=1}^{m_{(n)}} V_i,\,
    \end{align}
    where $U_{(n)}\simeq V_i\simeq V_j$ for $i,j=1,\ldots,m_{(n)}$ are isomorphic as $\CSn$-modules. For each $i=1,\ldots, m_{(n)}$, let $v_i$ be a non-zero basis vector spanning $V_i$. Thus, $\lrbrace{v_i}_{i=1}^{m_{(n)}}$ is a direct sum basis for $\SymH$. From \autoref{prop:Bose_sym_space_char} any linear operator on $\SymH$ is Bose-symmetric, and hence, the full matrix algebra $M_{m_{(n)}}(\CC)$ indexed in $\lrbrace{v_i}_{i=1}^{m_{(n)}}$ basis is Bose-symmetric. Let $w\in\SymH$ be any vector which can be uniquely written as
    \begin{align}
        w=\sum_{i=1}^{m_{(n)}} \alpha_iv_i.\,
    \end{align}
    Define $A$ as the $\CC$-linear map determined by
    \begin{align}
        A(v_1)=w,\, \hspace{1cm} A(v_j)=0 \hspace{0.5cm} \text{for } j=2,\ldots, m_{(n)}.\,
    \end{align}
    Let $\lrrec{A}$ denote the matrix representation of $A$ in the $\lrbrace{v_i}_{i=1}^{m_{(n)}}$ basis. Then, 
    \begin{align}
        \lrrec{A}_{i,1}=\alpha_i,\,  \hspace{1cm} \lrrec{A}_{i,j}=0 \hspace{0.5cm} \text{for } i=1,\ldots,m_{(n)},\, j=2, \ldots, m_{(n)}.\,
    \end{align}
    Alternatively, let $v_1^*$ be the dual vector to $v_1$, then $A=wv_1^*$. 
    Since $S_n$ acts trivially on any element of the symmetric subspace, $A$ is Bose-symmetric. Thus, $v_1$ can be mapped to any $w\in\SymH$ under some intertwiner $A$; identifying $e_{(n)}=v_1$, this yields $W_{(n)}=\SymH$. The same argument applies to any other choice of generating vector $v_j$.
\end{proof}

With the proposition above, we can provide an explicit $*$-isomorphism given a basis of the symmetric subspace. In \autoref{sec:symmetric_subspace_methods}, we provided such a basis following the approach of \cite{harrow2013church}. Here, we provide a more detailed derivation of this basis using highest weight theory and semistandard Young tableaux. We can decompose each Weyl module, a highest weight module, into its weight spaces, i.e.\ a collection of vectors with the same weight. The action of the Cartan subalgebra $\mathfrak{h}$ of $\mathfrak{gl}\lrbracket{\HS}$, i.e.\ the Lie algebra of the maximal torus of $\GLH$, makes this concrete. Consider the tensor power action of the Lie algebra $\mathfrak{gl}(\HS)$ on $\HSn$. Then, the weight module admits a decomposition
\begin{align}
    \HSn = \bigoplus_{\mu\in\mathfrak{h}^{*}}\underbrace{\lrbrace{w\in \HSn\, \vert\, H\cdot w = \mu(H)\, w,\, \forall H\in \mathfrak{h}}}_{=: \lrbracket{\HSn}_{\mu}}
\end{align}
where each $\mu$ arises as a sum over the tensor factors of the weights of $\mathfrak{h}$ and $H\in\mathfrak{h}$ with $H=\text{diag}\lrbracket{h_1,\ldots,h_{d_{\HS}}}$ acts on simple tensors in the canonical basis to $\HSn$ as
\begin{align}
	H\cdot \lrbracket{\bigotimes_{j=1}^n e_{i_j}} = \lrbracket{\sum_{j=1}^n h_{i_j}}\lrbracket{\bigotimes_{j=1}^ne_{i_j}},
\end{align}
and thus extends to $\HSn$ by linearity. Restricting the action of $\mathfrak{h}$ to a linear subspace of $\HSn$ we get a corresponding weight space decomposition of the symmetric subspace
\begin{align}
    \SymH \simeq \bigoplus_{\mu\in\mathfrak{h}^{*}} \SymH \cap \lrbracket{\HSn}_{\mu}.\,
\end{align}
It is easy to see that the intersection of the linear subspaces is again a linear subspace. While there exist potentially many elements of the canonical basis to $\HSn$ with the same weight, under the idempotent $c_{(n)}$ projecting onto $\SymH$, they become a single vector in $\SymH$. In other words, the collection of vectors with the same weight becomes the same symmetrized vector in $\SymH$ (up to normalization), i.e.\
\begin{align}
	\SymH\cap\lrbracket{\HSn}_{\mu} = \text{span}_{\CC}\lrbracket{\lrbracket{\HSn}_{\mu}\cdot c_{(n)}}.\,
\end{align} 
Since $\mathfrak{h}$ is abelian, the weight of a canonical basis element of $\HSn$ depends only on its type, i.e.\ on the tuple $(\alpha_1, \alpha_2, \ldots, \alpha_{d_{\HS}})$, where $\alpha_j$ counts how many times the canonical basis vector $e_j$ of $\HS$ appears in the corresponding simple tensor. Thus, the weight spaces in $\SymH$ can be identified with semistandard labelings of Young shape $Y((n))$ with entries from $\lrrec{d_{\HS}}$. The dimension of the weight spaces is given by the number of semistandard tableaux with the corresponding weight content.

\begin{proposition}
    For each weight $\mu$ occurring in $\SymH$, the $\CC$-vector space dimension of $\SymH \cap \lrbracket{\HSn}_{\mu}$ is one.
\end{proposition}

\begin{proof}
    For Young shape $Y((n))$ there exists exactly one semistandard labeling encoding a given weight.
\end{proof}

\begin{proposition}
Let $\tau$ be a semistandard labeling with weight $\mu$ of Young shape $Y((n))$. Consider the image of the corresponding Young symmetrizer in $\HSn$ (polytabloid), i.e.\
\begin{align}
	\mathbb{R}^{d_{\HS}^n}\ni u_{\tau}= \sum_{\tau'\sim\tau}\bigotimes_{y\in Y((n))}e_{\tau'(y)},
\end{align}
where the sum ranges over all tableaux $\tau'$ in the orbit of $\tau$ under the row stabilizer subgroup of $\tau$. Then,
\begin{align}\label{eqn:iso_weight_c_n_image}
	 \SymH\cap\lrbracket{\HSn}_{\mu} \simeq \text{span}_{\CC} \lrbracket{u_{\tau}}\simeq \CSn\cdot u_{\tau}.\,
\end{align}
Furthermore, the weight spaces in $\SymH$ are isomorphic as $\CSn$-modules to the trivial Specht module.
\end{proposition}

\begin{proof}
Since $u_{\tau}$ has weight $\mu$ it follows that $u_{\tau}\in \SymH\cap\lrbracket{\HSn}_{\mu}$. As  $\SymH\cap\lrbracket{\HSn}_{\mu}$ has $\CC$-vector space dimension one, the first vector space isomorphism in \autoref{eqn:iso_weight_c_n_image} follows. By construction, for any $v\in\SymH$, $S_n \cdot v=v$. Thus, $\CSn$ acts by $\CC$-scalar multiplication on $u_{\tau}$, i.e.\ $\CSn\cdot u_{\tau} = \text{span}_{\CC}\lrbracket{u_{\tau}}$ and the second isomorphism follows. As $\CSn$-modules, the spaces are isomorphic to the trivial Specht module.  
\end{proof}

\begin{proposition}\label{prop:basis_weight_sym_sub}
The set $\lrbrace{u_{\tau}\,:\, \tau\in \mathcal{T}_{(n),d_{\HS}}}$ is a basis of the symmetric subspace and a representative set for the action of $\CSn$ on $\SymH$.
\end{proposition}
\begin{proof}
By the arguments above, the set spans $\SymH$ and is of cardinality $\dim_{\CC}\lrbracket{\SymH}$. The set satisfies the conditions in  \autoref{prop:polak_certification_of_representative_set}, and thus, is a representative set for $\CSn$ on $\SymH$.
\end{proof}

\begin{remark}\label{rem:connection_semistandard_partitions}
    For partition $\lambda=(n)$, there exists a bijection between the set $\mathcal{T}_{(n),d_{\HS}}$ of semistandard labelings with entries from $\lrrec{d_{\HS}}$ of Young shape $Y\lrbracket{(n)}$ and the set $\text{Par}\lrbracket{d_{\HS}, n}$ of partitions of $n$ into at most $d_{\HS}$ parts. Thus, the connection to $\lrbrace{\ket{s_{\vec{t}}}\, :\, \vec{t}\in\text{Par}\lrbracket{d_{\HS},n}}$ from \cite{harrow2013church} is clear. For more general partitions this bijection no longer holds.
\end{remark}

We are now prepared to explicitly construct a $*$-isomorphism between the space of Bose-symmetric operators and a full matrix algebra over $\CC$, whose size scales polynomially in $n$.

\begin{proposition}\label{prop:change_of_basis_bose_sym}
Let $\tilde{u}_{\tau}:=u_{\tau}/\norm{u_{\tau}}$ denote the normalized polytabloids. The map
	\begin{align}
		\begin{split}
			\psi \,:\,  \End{\CSn}{\SymH} &\rightarrow \CC^{m_{(n)}\times m_{(n)}}\\
			 A &\mapsto \lrbracket{\left\langle A \tilde{u}_{\tau}, \tilde{u}_{\gamma}\right\rangle}_{\tau, \gamma \in \mathcal{T}_{(n), d_{\HS}}}
		\end{split}
	\end{align}
	is a bijection preserving semi-definiteness. Moreover, there exists an isometry $U_{(n)}\in\R^{d_{\HS}^n\times m_{(n)}}$ such that $ U_{(n)}^* A U_{(n)} = \psi(A)$, for all $A\in \End{\CSn}{\SymH}$.
\end{proposition}

\begin{proof}
Follows directly from the arguments above. The change of basis matrix $U_{(n)}$ is constructed from the normalized column vectors $\tilde{u}_{\tau}$ with $\tau\in \mathcal{T}_{(n),d_{\HS}}$ from \autoref{prop:basis_weight_sym_sub}; cf.\ \autoref{prop:type_basis_symmetric_subspace}.
\end{proof}

\begin{remark}
    Since $U_{(n)}\in\mathbb{R}^{d_{\HS}^n\times m_{(n)}}$, $U_{(n)}^*=U_{(n)}^T$. Since
    \begin{align}
        \SymH = \bigoplus_{\tau\in \mathcal{T}_{(n),d_{\HS}}}\CSn \cdot u_{\tau},
    \end{align}
    $U_{(n)}$ is also referred to as an element of a representative matrix set (see \cite[Definition 2.4.5]{polak2020new}) for the $\CSn$-action on $\SymH$. In other words, it contains a collection of non-zero vectors from $\SymH$ which generate $\SymH$ upon $\CSn$ action. While in the case of $\HSn$, the representative matrix set has a higher cardinality (corresponding to the number of $\CSn$-isotypic components in $\HSn$), for $\SymH$, it consists solely of $U_{(n)}$. 
\end{remark}

\begin{remark}
	Similarly, we can obtain a bijection preserving positive semidefiniteness for any Weyl module $\Schurf{\lambda}\HS$, i.e.\
	\begin{align}
		\begin{split}
			\psi \,:\,  \End{\CSn}{\Schurf{\lambda}\HS^{\oplus \dim_{\CC}\lrbracket{\SpechtH}}} &\rightarrow \CC^{m_{\lambda}\times m_{\lambda}}\\
			 A &\mapsto \lrbracket{\left\langle A u_{\tau}, u_{\gamma}\right\rangle}_{\tau, \gamma \in \mathcal{T}_{\lambda, d_{\HS}}}
		\end{split}
	\end{align}
with $m_{\lambda}=\dim_{\CC}\lrbracket{\Schurf{\lambda}\HS}$ and $\mathcal{T}_{\lambda, d_{\HS}}$ the set of semistandard labelings of Young shape $Y(\lambda)$ with entries from $\lrrec{d_{\HS}}$ (cf.\ \autoref{prop:iso_full_amtrix_symmetry}).
\end{remark}
\section{Proof of \autoref{lem:restriction_to_invariant_subspace}}\label{sec:restriction_to_invariant_subspace}

\begin{proof}
    	We explicitly express $\rho_{(A_1Q_1T)(A_2Q_2\hat{T})_1^n}$ in terms of the standard basis of $\mathcal{B}\lrbracket{A_1Q_1T}$, i.e.
\begin{align}
    \rhoge= \sum_{\substack{a_1, a_1'\in\left[\lrvert{A_1}\right]\,,\\ q_1, q_1'\in\left[\lrvert{Q_1}\right]}}\ket{a_1q_1}\bra{a_1'q_1'}_{A_1Q_1}\otimes\sum_{i,j\in \left[\lvert T \rvert\right]}\ket{i}\bra{j}_T  \otimes \rho^{a_1, q_1,a_1', q_1',i,j}_{(A_2Q_2\hat{T})_1^n}\,.
\end{align}
Thus,
\begin{align}\label{eqn:reduced_state}
    \Trr{A_1}{\rhoge}= \sum_{q_1, q_1'\in\left[\lrvert{Q_1}\right]}\ket{q_1}\bra{q_1'}_{Q_1}\otimes\sum_{i,j\in \left[\lvert T \rvert\right]}\ket{i}\bra{j}_T \otimes \sum_{a_1\in\lrrec{\lrvert{A_1}}}\rho^{a_1, q_1,a_1, q_1'i,j}_{(A_2Q_2\hat{T})_1^n}.
\end{align}
Since $\rhoge\in\ESN\lrbracket{(A_1Q_1T)\otimes\An}$, i.e.
\begin{align}\label{eqn:sym_ext_state}
    \rhoge =\sum_{\substack{a_1, a_1'\in\left[\lrvert{A_1}\right]\,,\\ q_1, q_1'\in\left[\lrvert{Q_1}\right]}}\ket{a_1q_1}\bra{a_1'q_1'}_{A_1Q_1}\otimes\sum_{i,j\in \left[\lvert T \rvert\right]}\ket{i}\bra{j}_T  \otimes U_{(A_2Q_2\hat{T})_1^n}\,(\pi)\rho^{a_1, q_1,a_1', q_1',i,j}_{(A_2Q_2\hat{T})_1^n}\, U^*_{(A_2Q_2\hat{T})_1^n}(\pi)
\end{align}
where $U_{(A_2Q_2\hat{T})_1^n}(\pi)$ corresponds to a permutation matrix representation for any $\pi\in S_n$, we deduce that for any $\pi\in S_n$
\begin{align}
    \begin{split}
    &\Trr{A_1}{\rhoge} \stackrel{\ref{eqn:sym_ext_state}}{=}\\  &\sum_{q_1, q_1'\in\left[\lrvert{Q_1}\right]}\ket{q_1}\bra{q_1'}_{Q_1}\otimes\sum_{i,j\in \left[\lvert T \rvert\right]}\ket{i}\bra{j}_T \otimes \sum_{a_1\in\lrrec{\lrvert{A_1}}}U_{(A_2Q_2\hat{T})_1^n}\,(\pi)\rho^{a_1, q_1,a_1, q_1'i,j}_{(A_2Q_2\hat{T})_1^n}\, U^*_{(A_2Q_2\hat{T})_1^n}(\pi)\\
    &=\lrbracket{\mathbb{1}_{Q_1T}\otimes U_{(A_2Q_2\hat{T})_1^n}}\lrbracket{\sum_{q_1, q_1'\in\left[\lrvert{Q_1}\right]}\ket{q_1}\bra{q_1'}\otimes\sum_{i,j\in \left[\lvert T \rvert\right]}\ket{i}\bra{j} \otimes   \sum_{a_1\in\lrrec{\lrvert{A_1}}}\rho^{a_1, q_1,a_1, q_1'i,j}_{(A_2Q_2\hat{T})_1^n}}\lrbracket{\mathbb{1}_{Q_1T}\otimes U^*_{(A_2Q_2\hat{T})_1^n}}\\
    &\stackrel{\ref{eqn:reduced_state}}{=}\lrbracket{\mathbb{1}_{Q_1T}\otimes U_{(A_2Q_2\hat{T})_1^n}}\lrbracket{ \Trr{A_1}{\rhoge}}\lrbracket{\mathbb{1}_{Q_1T}\otimes U^*_{(A_2Q_2\hat{T})_1^n}}\\
        \Rightarrow \hspace{0,5cm} &\Trr{A_1}{\rhoge}\in \ESN\lrbracket{(Q_1T)\otimes(A_2Q_2\hat{T})_1^n}\,.
    \end{split}
\end{align}
Now for the corresponding right-hand side of the constraint in $\mathrm{SDP}_n(T, V,\pi)$ we have
\begin{align}
   \sum_{q_1}\pi_1(q_1)\ket{q_1}\bra{q_1}_{Q_1}\otimes \Tr_{(A_1Q_1)}\left[\rho_{(A_1Q_1T)(A_2Q_2\hat{T})_1^n}\right] \in \ESN\lrbracket{(Q_1T)\otimes(A_2Q_2\hat{T})_1^n}
\end{align}
by a similar argument, concluding the proof of \autoref{eqn:sym_subspace_1}. For the proof of \autoref{eqn:sym_subspace_2}, let $\lrbrace{\ket{i}}_{i\in \lrrec{\lrvert{A_1Q_1T}}}$ be a basis of $A_1Q_1T$ and
$\lrbrace{\ket{a_2q_2u}}_{\substack{a_2\in\lrrec{\lrvert{A_2}}\,,\\ q_2\in\lrrec{\lrvert{Q_2}}\,,\\ u\in \lrrec{\lrvert{(\hat{T})_1}}}}$ a basis of $(A_2Q_2\hat{T})_1$ such that
\begin{align}
    \rhoge = \sum_{i,j\in\lrrec{\lrvert{A_1Q_1T}}}\ket{i}\bra{j}_{A_1Q_1T}\otimes \sum_{\substack{a_2 a_2'\in\lrrec{\lrvert{A_2}}\,,\\ q_2, q_2'\in\lrrec{\lrvert{Q_2}}\,,\\ u, v\in \lrrec{\lrvert{(\hat{T})_1}}}} \ket{a_2q_2u}\bra{a_2'q_2'v}_{(A_2Q_2\hat{T})_1}\otimes \rho^{\lrbracket{i,j,a_2,a_2',q_2,q_2', u, v}}_{(A_2Q_2\hat{T})_2^n}\,.
\end{align}
Then, we have
\begin{align}\label{eqn:reduced_state_in_standard_basis}
    \Trr{(A_2)_1}{\rhoge} = \sum_{i,j\in\lrrec{\lrvert{A_1Q_1T}}}\ket{i}\bra{j}_{A_1Q_1T}\otimes \sum_{\substack{ q_2, q_2'\in \lrrec{\lrvert{(Q_2)_1}}\,,\\u,v\in\lrrec{\lrvert{(\hat{T})_1}}}}\ket{q_2u}\bra{q_2'v}_{(Q_2\hat{T})_1}\otimes\sum_{a_2\in\lrrec{\lrvert{(A_2)_1}}}\rho^{\lrbracket{i,j,a_2,a_2,q_2,q_2',} u, v}_{(A_2Q_2\hat{T})_2^n}
\end{align}
in the standard basis of the respective spaces. From the subgroup structure $S_{n-1}\subset S_n$, it follows that
\begin{align}\label{eqn:different_way_to_permute}
\begin{split}
    \rhoge &= \lrbracket{\mathbb{1}_{A_1Q_1T} \otimes \mathbb{1}_{(A_2Q_2\hat{T})_1}\otimes U_{(A_2Q_2\hat{T})_2^n}(\pi)}\\
    &\lrbracket{\rhoge}\lrbracket{\mathbb{1}_{A_1Q_1T}\otimes  \mathbb{1}_{(A_2Q_2\hat{T})_1}\otimes U^*_{(A_2Q_2\hat{T})_2^n}(\pi)} \\
    &= \lrbracket{\mathbb{1}_{A_1Q_1T}\otimes U_{(A_2Q_2\hat{T})_1^n}^{\lrbracket{(A_2Q_2\hat{T})_2^n\uparrow (A_2Q_2\hat{T})_1^n}}(\pi)}\\
    &\lrbracket{\rhoge}\lrbracket{\mathbb{1}_{A_1Q_1T}\otimes \lrbracket{U_{(A_2Q_2\hat{T})_1^n}^{\lrbracket{(A_2Q_2\hat{T})_2^n\uparrow (A_2Q_2\hat{T})_1^n}}}^*(\pi)}
\end{split}
\end{align}
$\forall \pi\in S_{n-1}$. Hence, for any $\pi\in S_{n-1}$ 
\begin{align}
    \begin{split}
        &\Trr{(A_2)_1}{\rhoge}\\
    &\stackrel{\autoref{eqn:reduced_state_in_standard_basis}}{=}\sum_{i,j\in\lrrec{\lrvert{A_1Q_1T}}}\ket{i}\bra{j}_{A_1Q_1T}\otimes \sum_{\substack{ q_2, q_2'\in  \lrrec{\lrvert{(Q_2)_1}}\,,\\u,v\in\lrrec{\lrvert{(\hat{T})_1}}}}\ket{q_2u}\bra{q_2'v}_{(Q_2\hat{T})_1}\\
    &\hspace{1.5cm}\otimes\sum_{a_2\in\lrrec{\lrvert{(A_2)_1}}}U_{(A_2Q_2\hat{T})_2^n}(\pi)\rho^{\lrbracket{i,j,a_2,a_2,q_2,q_2', u, v}}_{(A_2Q_2\hat{T})_2^n}U^*_{(A_2Q_2\hat{T})_2^n}(\pi)\\
        &\stackrel{\autoref{eqn:different_way_to_permute}}{=}\lrbracket{\mathbb{1}_{A_1Q_1T}\otimes\mathbb{1}_{(Q_2\hat{T})_1}\otimes U_{(A_2Q_2\hat{T})_2^n}(\pi)}\\
        &\lrbracket{\Trr{(A_2)_1}{\rhoge}}\lrbracket{\mathbb{1}_{A_1Q_1T}\otimes\mathbb{1}_{(Q_2\hat{T})_1}\otimes U^*_{(A_2Q_2\hat{T})_2^n}(\pi)}\\
        \Rightarrow\, &\Trr{(A_2)_1}{\rhoge}\, \in \text{End}_{\CC\lrrec{S_{n-1}}\,}\lrbracket{(A_1Q_1T)\otimes (Q_2\hat{T})_1\otimes(A_2Q_2\hat{T})_2^n}\\
        \Rightarrow \, &\Trr{(A_1Q_1T)(A_2)_1}{\rhoge}\, \in \text{End}_{\CC\lrrec{S_{n-1}}\,}\lrbracket{(Q_2\hat{T})_1\otimes(A_2Q_2\hat{T})_2^n}\,.
    \end{split}
\end{align}
For the corresponding right-hand side, we have
\begin{align}
    \left(\sum_{q_2}\pi_2(q_2)\ket{q_2}\bra{q_2}_{Q_2} \otimes\frac{\mathbb{1}_{\hat{T}}}{\lvert T\rvert}\right)\otimes \rho_{(A_2Q_2\hat{T})_2^n}\in\text{End}_{\CC\lrrec{S_{n-1}}\,}\lrbracket{(Q_2\hat{T})_1\otimes(A_2Q_2\hat{T})_2^n}
\end{align}
by a similar argument, thus proving \autoref{eqn:sym_subspace_2}. Lastly, the bound on the dimension follows directly from \autoref{prop:bounding_our_case}.
\end{proof}

\section{Alternative proof of \autoref{lem:efficient_basis_trafo_Bose_sym}}\label{sec:coordinate_ring_proof}

    The proof of \autoref{lem:efficient_basis_trafo_Bose_sym} given in \autoref{sec:efficient_trafo_bose_operators} evaluates the pairing directly. For completeness, we record an alternative derivation of the closed form \autoref{eqn:closed_form_transformed_basis} in the coordinate-ring formalism of \cite{polak2020new} (see also \cite{gijswijt2009block, litjens2017semidefinite, chee2023efficient}), which makes explicit how the general-shape machinery underlying \autoref{prop:efficient_computation_explicit_part} specializes to the one-row shape: the column stabilizers are trivial, so no sign factors occur, and the evaluation collapses to a single monomial with a closed-form coefficient.

\begin{proof}[Alternative proof of \autoref{lem:efficient_basis_trafo_Bose_sym}]
    The proof is an adaptation of \cite{polak2020new} to Bose-symmetric operators, with explicit closed-form expressions. Consider $W_{\HS}:=\HS\otimes\HS$ with canonical basis $\mathcal{W}:=\lrbrace{a_{x,y}:=e_x\otimes e_y\,:\, x,y\in\lrrec{d_{\HS}}}$, where $\lrbrace{e_x}_{x=1}^{d_{\HS}}$ is the computational basis of $\HS$, and, with respect to its natural inner product, its dual space $W_{\HS}^*=\HS^*\otimes\HS^*$ with dual basis $\mathcal{W}^*:=\lrbrace{a^*_{x,y}\,:\, x,y\in\lrrec{d_{\HS}}}$, i.e.\
    \begin{align}
        a^*_{x,y}a_{w,z}=\delta_{x,w}\,\delta_{y,z}.\,
    \end{align}
    Furthermore, consider the coordinate ring of homogeneous polynomials of degree $n$ on $W_{\HS}$, denoted $\mathcal{O}_n\lrbracket{W_{\HS}}$, i.e.\ the homogeneous polynomials of degree $n$ in the commuting variables $\mathcal{W}^*$. Importantly, as opposed to a free algebra generated from the same alphabet, in the coordinate ring the ordering of letters in monomials is irrelevant. Let $\mu\,:\,\lrrec{d_{\HS}}^n\times\lrrec{d_{\HS}}^n\rightarrow\CC$ be a class function for types, i.e.\ for any $\pi, \sigma\in S_n$,
    \begin{align}
        \mu\lrbracket{\pi\lrbracket{\vec{w}},\sigma\lrbracket{\vec{z}}} = \mu\lrbracket{\vec{w}, \vec{z}},\,
    \end{align}
    so that $\mu\lrbracket{\vec{w},\vec{z}}$ depends only on the pair of types $\lrbracket{T\lrbracket{\vec{w}}, T\lrbracket{\vec{z}}}$: for any $\vec{w}, \vec{v}$ of type $\vec{t}$ and any $\vec{z}, \vec{m}$ of type $\vec{t'}$ we have $\mu\lrbracket{\vec{w}, \vec{z}}=\mu\lrbracket{\vec{v}, \vec{m}}=:\mu\lrbracket{\vec{t}, \vec{t'}}$. The class functions for types are precisely the entry functions of Bose-symmetric operators: by \autoref{prop:entries_canonical_bose}, $Z=\sum_{\vec{t},\vec{t'}\in\mathcal{T}_{(n),\, d_{\HS}}} z_{\vec{t},\vec{t'}}\,\CBose{t}{t'}{n}$ has entries $Z_{\vec{w},\vec{z}}=z_{T\lrbracket{\vec{w}},\,T\lrbracket{\vec{z}}}$, so its entry function $f_Z$ from \autoref{lem:efficient_access_bose} is a class function for types with $\mu\lrbracket{\vec{t},\vec{t'}}=z_{\vec{t},\vec{t'}}$.

    For the one-row shape, a semistandard labeling $\tau$ of $Y((n))$ with entries in $\lrrec{d_{\HS}}$ is determined by its weight $\vec{\tau}\in\mathcal{T}_{(n),\, d_{\HS}}$, and the associated polytabloid is
    \begin{align}
        u_{\tau}= \sum_{\tau'\sim\tau}\bigotimes_{y\in Y((n))}e_{\tau'(y)},\,
    \end{align}
    where $\tau'\sim\tau$ ranges over the orbit of $\tau$ under its row stabilizer --- reading labelings of $Y((n))$ as strings, over the $\binom{n}{\vec{\tau}}$ distinct strings in $\lrrec{d_{\HS}}^n$ of type $\vec{\tau}$. In particular, $u_{\tau}=\ket{s_{\vec{\tau}}}$ with entries in $\lrbrace{0,1}$, and $\ket{\tilde{s}_{\vec{\tau}}}=\binom{n}{\vec{\tau}}^{-1/2}u_{\tau}$ by \autoref{prop:type_basis_symmetric_subspace}. Identify $\Op{\HSn}\simeq W_{\HS}^{\otimes n}$ via the computational-basis identification $\HS^*\simeq\HS$ followed by a pairwise reordering of tensor factors, i.e.\ $\ket{\vec{w}}\bra{\vec{z}}\mapsto\bigotimes_{i=1}^n a_{w_i, z_i}$; by \autoref{prop:first_basis_Bose},
    \begin{align}\label{eqn:CBose_as_dual_tensor}
        \CBose{t}{t'}{n}=\ket{s_{\vec{t}}}\bra{s_{\vec{t'}}}=\sum_{\substack{\vec{w}\, :\, T\lrbracket{\vec{w}}=\vec{t},\,\\ \vec{z}\, :\, T\lrbracket{\vec{z}}=\vec{t'}}}\bigotimes_{i=1}^n a_{w_i, z_i}.\,
    \end{align}
    Dually, $u_{\tau}^T X u_{\gamma}=\lrbracket{u_{\tau}\otimes u_{\gamma}}\lrbracket{X}$ under $\lrbracket{\HSn}^*\otimes\lrbracket{\HSn}^*\simeq\lrbracket{\HSn\otimes\HSn}^*\simeq\lrbracket{W_{\HS}^{\otimes n}}^*$. Furthermore,
    \begin{align}\label{eqn:g_def_bose}
        g := u_{\tau}\otimes u_{\gamma} = \sum_{\substack{\tau'\sim \tau,\,\\ \gamma' \sim \gamma}}\bigotimes_{y\in Y\lrbracket{(n)}} F_{\tau'(y), \gamma'(y)},\,
    \end{align}
    where $F\in\lrbracket{\mathcal{W}^*}^{d_{\HS}\times d_{\HS}}$ is the matrix encoding of the dual basis, i.e.\ $\lrbracket{F}_{x,y}=a^*_{x,y}$. Since the rearrangements $\tau'$ and $\gamma'$ range independently, $g$ is invariant under the diagonal $S_n$-action on the tensor factors of $\lrbracket{W_{\HS}^*}^{\otimes n}$ and may hence be regarded as an element of $\mathcal{O}_n\lrbracket{W_{\HS}}$, namely the polynomial $\sum_{\tau'\sim\tau,\, \gamma'\sim\gamma}\prod_{y\in Y((n))}a^*_{\tau'(y),\gamma'(y)}$ --- the analogue of the polynomial $G_{\tau,\gamma}$ from the proof of \autoref{prop:efficient_computation_explicit_part}, with trivial column stabilizers. Then, identifying the cells of $Y((n))$ with the tensor positions $\lrrec{n}$,
    \begin{align}\label{eqn:coordinate_ring_evaluation}
        \begin{split}
            u_{\tau}^T\,\CBose{t}{t'}{n}\,u_{\gamma}
            &= \lrbracket{u_{\tau}\otimes u_{\gamma}}\lrbracket{\CBose{t}{t'}{n}}=g\lrbracket{\CBose{t}{t'}{n}}\\
            &= \sum_{\substack{\tau'\sim \tau,\,\\ \gamma' \sim \gamma}} \bigotimes_{y\in Y\lrbracket{(n)}} F_{\tau'(y), \gamma'(y)}\lrbracket{\sum_{\substack{\vec{w}\, :\, T\lrbracket{\vec{w}}=\vec{t},\,\\ \vec{z}\, :\, T\lrbracket{\vec{z}}=\vec{t'}}}\bigotimes_{i=1}^n a_{w_i, z_i}}\\
            &= \sum_{\substack{\vec{w}\, :\, T\lrbracket{\vec{w}}=\vec{t},\,\\ \vec{z}\, :\, T\lrbracket{\vec{z}}=\vec{t'}}}\; \sum_{\substack{\tau'\sim \tau,\,\\ \gamma' \sim \gamma}}\; \prod_{y\in Y\lrbracket{(n)}} \underbrace{a^*_{\tau'(y), \gamma'(y)}a_{w_y, z_y}}_{=\begin{cases}\begin{array}{cc}
                1, & \text{if } \tau'(y)=w_y\, \land\, \gamma'(y)=z_y,\,\\
                0, & \text{otherwise,}
            \end{array}
            \end{cases}}\\
            &= \sum_{\substack{\vec{w}\, :\, T\lrbracket{\vec{w}}=\vec{t},\,\\ \vec{z}\, :\, T\lrbracket{\vec{z}}=\vec{t'}}} \lrvert{\lrbrace{\lrbracket{\tau',\gamma'}\,:\,\tau'\sim\tau,\,\gamma'\sim\gamma,\,\tau'=\vec{w},\,\gamma'=\vec{z}}}\\
            &= \sum_{\substack{\vec{w}\, :\, T\lrbracket{\vec{w}}=\vec{t},\,\\ \vec{z}\, :\, T\lrbracket{\vec{z}}=\vec{t'}}} \delta_{\vec{\tau}, \vec{t}}\,\delta_{\vec{\gamma}, \vec{t'}} = \binom{n}{\vec{t}}\binom{n}{\vec{t'}}\, \delta_{\vec{\tau}, \vec{t}}\,\delta_{\vec{\gamma}, \vec{t'}}.\,
        \end{split}
    \end{align}
    Here the fourth equality counts, for fixed $\lrbracket{\vec{w},\vec{z}}$, the pairs of rearrangements matching $\lrbracket{\vec{w},\vec{z}}$ entrywise: since $\tau'$ ranges over the distinct strings of type $\vec{\tau}$, there is exactly one matching $\tau'$ if $\vec{\tau}=T\lrbracket{\vec{w}}=\vec{t}$ and none otherwise, and similarly for $\gamma'$; the final equality counts the $\binom{n}{\vec{t}}\binom{n}{\vec{t'}}$ pairs $\lrbracket{\vec{w},\vec{z}}$. Normalizing yields the closed form \autoref{eqn:closed_form_transformed_basis}:
    \begin{align}
        \bra{\tilde{s}_{\vec{\tau}}}\,\CBose{t}{t'}{n}\,\ket{\tilde{s}_{\vec{\gamma}}}=\binom{n}{\vec{\tau}}^{-1/2}\binom{n}{\vec{\gamma}}^{-1/2}\, u_{\tau}^T\,\CBose{t}{t'}{n}\,u_{\gamma} = \sqrt{\binom{n}{\vec{t}}\binom{n}{\vec{t'}}}\;\delta_{\vec{\tau},\vec{t}}\,\delta_{\vec{\gamma},\vec{t'}},\,
    \end{align}
    where the Kronecker deltas allow replacing $\binom{n}{\vec{\tau}}\binom{n}{\vec{\gamma}}$ by $\binom{n}{\vec{t}}\binom{n}{\vec{t'}}$. Thus, for any class function for types $\mu$ we have
    \begin{align}\label{eqn:class_function_pairing_bose}
        \sum_{\vec{t}, \vec{t'}\in \mathcal{T}_{(n),\, d_{\HS}}} \bra{\tilde{s}_{\vec{\tau}}}\,\CBose{t}{t'}{n}\,\ket{\tilde{s}_{\vec{\gamma}}}\; \mu\lrbracket{\vec{t}, \vec{t'}} = \sqrt{\binom{n}{\vec{\tau}}\binom{n}{\vec{\gamma}}}\; \mu\lrbracket{\vec{\tau}, \vec{\gamma}}.\,
    \end{align}
    Choosing for $\mu$ the class function of the given coefficients, $\mu\lrbracket{\vec{t},\vec{t'}}=z_{\vec{t},\vec{t'}}$ (equivalently, $\mu=f_Z$), i.e.\ replacing $\mu\lrbracket{\vec{t},\vec{t'}}$ with the corresponding $z_{\vec{t},\vec{t'}}$, the left-hand side of \autoref{eqn:class_function_pairing_bose} becomes $\bra{\tilde{s}_{\vec{\tau}}}Z\ket{\tilde{s}_{\vec{\gamma}}}=\lrrec{\block{\psi\lrbracket{Z}}_{(n)}}_{\vec{\tau},\vec{\gamma}}$ by \autoref{prop:type_basis_symmetric_subspace} and \autoref{cor:explicit_star_isomorphism}, and the right-hand side is the entry formula of \autoref{lem:efficient_basis_trafo_Bose_sym}. Since each multinomial coefficient can be computed with $\mathcal{O}\lrbracket{n}$ integer operations and there are $m_{(n)}^2\leq\lrbracket{n+1}^{2\lrbracket{d_{\HS}-1}}$ entries (\autoref{prop:dimension_bound_sym_subspace}), the matrix $\block{\psi\lrbracket{Z}}_{(n)}$ can be computed exactly in time polynomial in $n$ for fixed $d_{\HS}$.
\end{proof}
\section{Proof of \autoref{lem:efficiency_of_trafo_sdp_sym}}\label{sec:proof_of_efficient_trafo_sdp_sym}
Before we prove \autoref{lem:efficiency_of_trafo_sdp_sym}, we construct a canonical basis of $\End{\CSn}{\mathcal{A}\otimes\HS^n}$ from $S_n$-orbits, where $\Al$ is some Hilbert space. See \autoref{sec:symmetric_subspace_methods} for an alternative construction following \cite{gijswijt2009block}. See \autoref{ex:sdp_sym_reduction} for examples that elaborate on various concepts discussed in this section. 

\paragraph{\textbf{Notation}}
Since $\lrbracket{\CC^d}^{\otimes n}\simeq \CC^{d^n}$, we often switch between these descriptions and only explicitly denote the isomorphism by $\phi(\cdot)$ when needed.

Any $Z\in \End{\CSn}{\mathcal{A}\otimes\HS^n}$ is invariant under the group action. Averaging an arbitrary operator over the group yields such an invariant operator. In particular, averaging a standard basis element of $\mathcal{B}\lrbracket{\mathcal{A}\otimes\HS^n}$ over its $S_n$-orbit produces an invariant operator, and the distinct operators obtained in this way form a basis of $\ESN\lrbracket{\mathcal{A}\otimes\HS^n}$.

Concretely, we consider the orbits of the group action of $S_n$ on pairs\footnote{The standard basis of a matrix algebra is indexed by tuples (row index, column index).}
\begin{align}\label{eqn:orbit_index_set}
	\lrrec{\lrvert{\Al}\times\lrvert{\HS}^n}^2.\,
\end{align}
 Let $i, j\in \lrrec{\lrvert{\Al}\times\lrvert{\HS}^n}$ denote an index of the standard basis of $\Al\otimes\HS^n$. Then,
\begin{align}\label{eqn:orbit}
    O(i,j):= \lrbrace{\lrbracket{\pi(i), \pi(j)} \, :\, \pi\in S_n}
\end{align}
with $\pi(i) = \lrbracket{\mathbb{1}_{\Al}\otimes U_{(\HS)_1^n}(\pi)}\ket{i}$, i.e.\ for $\ket{i}\in \HS^n$
\begin{align}
    \begin{split}
        U_{(\HS)_1^n}(\pi)\bigotimes_{j=1}^n\ket{\phi_j(i)}_j = \bigotimes_{j=1}^n\ket{\phi_j(i)}_{\pi^{-1}(j)}.\,
    \end{split}
\end{align}
Naturally, the action of $S_n$ partitions the set in \autoref{eqn:orbit_index_set} into $m\lrbracket{\Al, \HS^n}$ orbits, where
\begin{align}
	m\lrbracket{\Al, \HSn} = \lrvert{\Al}^2\binom{\lrvert{\HS}^2+n-1}{n}.
\end{align}
Indeed, the $S_n$-action leaves the two $\Al$-indices untouched, and an orbit is uniquely determined by these together with the multiplicities $n_{l,k}$ of the single-system index pairs $(l,k)$, $l,k\in\lrrec{\lrvert{\HS}}$; counting the solutions of $\sum_{l,k=1}^{\lrvert{\HS}}n_{l,k}=n$ in $\mathbb{N}_0$ yields the binomial coefficient. We represent each orbit $O_r$ indexed by $r\in \lrrec{m\lrbracket{\Al, \HS^n}}$ as an adjacency matrix
\begin{align}\label{eqn:rep_elemet_canonical_basis}
    C_r\in \lrbrace{0, 1}^{\lrbracket{\lrvert{\Al}\times\lrvert{\HS}^n} \times \lrbracket{\lrvert{\Al}\times\lrvert{\HS}^n}},\,
\end{align} where the matrix entries are given by
\begin{align}\label{eqn:cano_basis}
    \lrbracket{C_r}_{i,j}=\begin{cases}
        1 \hspace{0.15cm}\text{ if } (i,j)\in O_r \\
        0 \hspace{0.15cm}\text{ otherwise. }
    \end{cases}
\end{align}

\begin{remark}
    Given an element $(a,b)$ we denote the corresponding orbit via $O_{(a,b)}$ and the corresponding adjacency matrix by $C_{(a,b)}$. 
\end{remark}

Then, 
\begin{align}
    \mathcal{C}\lrbracket{\Al, \HS^n} := \lrbrace{C_i}_{i=1}^{m\lrbracket{\Al, \HS^n}}
\end{align}
is a canonical basis of $\ESN\lrbracket{\Al\otimes \HS^n}$ and
\begin{align}
\begin{split}
    \lrvert{\mathcal{C}\lrbracket{\Al, \HS^n}} &= \dim\lrbracket{ \ESN\lrbracket{\Al\otimes \HS^n}} \stackrel{\autoref{prop:bounding_our_case}}{\leq} \lrvert{\Al}^2(n+1)^{\lrvert{\HS}^2}.\,
\end{split}
\end{align}

\begin{remark}
	 In each orbit, we can choose the representative element to be of the form
\begin{align}
   \bigotimes_{l,k=1}^{\lrvert{\HS}}\lrbracket{\ket{l}\bra{k}}^{\otimes n_{l,k}}
\end{align}
with multiplicities $n_{l,k}\in\mathbb{N}_0$ satisfying $\sum_{l,k=1}^{\lrvert{\HS}}n_{l,k}=n$, where the tensor factors are arranged in lexicographic order with respect to $(l,k)$, i.e.\ all systems carrying the same basis element $\ket{l}\bra{k}$ are grouped into a consecutive block. In particular, for an orbit consisting of diagonal pairs, the representative element takes the form
\begin{align}
    \bigotimes_{j=1}^{\lrvert{\HS}}\lrrec{\bigotimes_{i=\lambda_{j-1}+1}^{\lambda_j}\ket{j}_i}\lrrec{\bigotimes_{i=\lambda_{j-1}+1}^{\lambda_j}\bra{j}_{i}}
\end{align}
with $\lambda_0=0$ and $\lambda=(\lambda_1, \lambda_2,\ldots, \lambda_{\lrvert{\HS}})\vdash_{\lrvert{\HS}} n$. In general, in each system $i\in\lrrec{n}$ a basis element is given by $\ket{l}\bra{k}_i$ for $l,k\in\lrrec{\lrvert{\HS}}$, and as the $S_n$-action permutes the system indices, it does not alter the number of systems described by basis element $\ket{l}\bra{k}_i$ for $l,k\in\lrrec{\lrvert{\HS}}$. The $S_n$-action conserves the described set partitioning.
\end{remark}

\begin{remark}\label{rem:split_basis}
	Note that with a canonical basis $\lrbrace{\ket{i}}_{i=1}^{\lrvert{\Al}}$ of $\Al$ a canonical basis of $\End{\CSn}{\Al\otimes\HSn}$ is given by
	\begin{align}
		\lrbrace{\ket{i}\bra{j}_{\Al}\otimes C_t^{\HS}}_{i,\,j\in \lrrec{\lrvert{\Al}},\, t\in \lrrec{m(\HS^n)}},
	\end{align}
	where the $C_t^{\HS}$ are  canonical basis elements of $\End{\CSn}{\HSn}$.
\end{remark}

Thus, ignoring $\Al$, a common procedure (see e.g.\ \cite{chee2023efficient}) to efficiently determine the representative elements of the $S_n$-orbits exploits a bijection between said orbits and certain matrices 
\begin{align}
\lrbrace{E_i}_{i=1}^{m\lrbracket{\HS^n}}\subset\mathbb{Z}^{\lrvert{\HS}\times \lrvert{\HS}}_{\geq 0}.\,
\end{align}
Concretely, for an arbitrary representative element $(x,y)$ to orbit $O^{\HS^n}_i$ we construct the matrix
\begin{align}
    \lrbracket{E^{(x,y)}}_{a,b} = \lrvert{\lrbrace{k\in\lrrec{n}\,:\, \phi_k(x)= a, \phi_k(y)= b}},\,\text{ with } (a,b)\in \lrrec{\lrvert{\HS}}\times\lrrec{\lrvert{\HS}},\,
\end{align}
by determining the frequencies of basis elements $\ket{\phi_i(x)}\bra{\phi_i(y)}$ to $\mathcal{B}\lrbracket{\HS^n}$ appearing in $(x,y)\in O^{\HS^n}_i$. Conversely, by solving
\begin{align}\label{eqn:opti_solving_cano_basis_E}
\begin{split}
    \begin{array}{cc}
         \text{find} &  \mathcal{E}:=\lrbrace{E^{(x,y)}}_{(x,y) \in \lrrec{\lrvert{\HS}^n}\times\lrrec{\lrvert{\HS}^n}} \\ \\
         \text{s.t.} & \sum_{a,b} \lrbracket{E^{(x,y)}}_{a,b} = n\\
          & \lrbracket{E^{(x,y)}}_{a,b} \in \mathbb{N}_0
    \end{array}
\end{split}
\end{align}
we can efficiently (in $n$) determine a representative element for each orbit $O_i^{\HS^n}$.

\begin{remark}
	Computing the full orbit would require at worst $\lrvert{S_n}=n!$ many permutations.
\end{remark}

\begin{proof}[Proof of \autoref{lem:efficiency_of_trafo_sdp_sym}] 

The proof is based on \cite{chee2023efficient, fawzi2022hierarchy, polak2020new, litjens2017semidefinite} and is structured as follows: 

First, we show that any $Z\in \text{End}^{S_{t\lrbracket{\mathcal{D}_1}}}\lrbracket{\mathcal{D}_1}$, and in particular any canonical basis element thereof, can be efficiently transformed into its block-diagonal form allowing for an efficient formulation of the positive semidefinite constraint. Then, we demonstrate that the other SDP constraints and the objective function can be efficiently formulated in terms of the canonical basis of the respective spaces. See \autoref{sec:efficiency_of_Bose_trafo} for a similar proof strategy.

Since \autoref{prop:bounding_our_case} bounds the number of blocks to be growing polynomially in $n$, it suffices to show that each block in the decomposition given in \autoref{eqn:efficient_sdp_free_game} can be computed efficiently w.r.t. $n$. Thus, for any $Z \in \text{End}^{S_{t\lrbracket{\mathcal{D}_1}}}\lrbracket{\mathcal{D}_1}$, the only costly basis transformation is 
\begin{align}\label{eqn:efficient_cal_basis_trafo}
\begin{split}
    \sum_{i=1}^{m\lrbracket{\HS^n}}z_i\cdot U^T_{\lambda}C_i^{\HS}U_{\lambda},\,
\end{split}
\end{align}
with $z_i\in\mathbb{C}$. The $U_{\lambda}$ consist of column vectors $\lrbrace{u_{\tau}}_{\tau\in T_{\lambda, \lrvert{\HS}}}$ indexed in the number of semistandard Young tableaus of shape $\lambda$ indexed in $\lrrec{\lrvert{\HS}}$ and as these are bounded to only grow polynomially in $n$ we can deduce the efficient computation of \autoref{eqn:efficient_cal_basis_trafo} from an efficient computation of 
\begin{align}\label{eqn:efficient_cal_basis_trafo_2}
    \sum_{i=1}^{m\lrbracket{\HS^n}}z_i\cdot u_{\tau}^TC_i^{\HS}u_{\tau'}
\end{align}
with $\tau, \tau'\in T_{\lambda, \lrvert{\HS}}$. While clearly $u_{\tau}\in\mathbb{R}^{\lrvert{\HS}^n}$ and $C_i^{\HS}\in\mathbb{R}^{\lrvert{\HS}^n\times\lrvert{\HS}^n}$ grow exponentially in $n$, $\forall \tau\in T_{\lambda, \lrvert{\HS}}, i\in\lrrec{m\lrbracket{\HS^n}}$, we will argue that while we cannot efficiently compute these quantities, we can however efficiently compute $u_{\tau}^TC_i^{\HS}u_{\tau'}$, i.e. the entries in the transformed matrix $U^T_{\lambda}C_i^{\HS}U_{\lambda}$. Concretely, by \autoref{prop:efficient_computation_explicit_part} below, given $\lrbrace{z_i}_{i\in m(\HS^n)}$ \autoref{eqn:efficient_cal_basis_trafo_2} can be computed polynomially in $n$ for fixed $\lrvert{\HS}$. 

While we have now shown that we can efficiently compute the bases to the set of full matrix algebras in the image of $\Psi_{\text{End}^{S_{t\lrbracket{\mathcal{D}_1}}}(\mathcal{D}_1)}(\cdot)$ in \autoref{eqn:our_bijection_applied}, we have yet to show that each constraint in \autoref{eqn:efficient_sdp_free_game} can be efficiently formulated in their respective canonical basis. Let
\begin{align}
    \lrbrace{\ket{a_1q_1i}}_{\substack{a_1\in\lrrec{A_1},\,\\ q_1\in\lrrec{Q_1},\,\\ i\in\lrrec{\lrvert{T}}}}
\end{align}
be the standard basis to $A_1Q_1T$ and 
\begin{align}
    \lrbrace{C_j^{\lrbracket{A_2Q_2\hat{T}}_1^n}}_{j\in\lrrec{m\lrbracket{\lrbracket{A_2Q_2\hat{T}}_1^n}}}
\end{align}
be the canonical basis of $\End{\CSn}{\An}$. Since the game predicate $V_{A_1A_2Q_1Q_2}$ is diagonal in the computational bases of the classical answer and question registers, and all constraints of \autoref{eqn:efficient_sdp_free_game} commute with the corresponding dephasing channels, feasible points can without loss of generality be taken classical (diagonal) on $A_1Q_1$; it therefore suffices to expand $\rhoge$ in the classical-quantum sub-basis below (cf.\ \autoref{prop:restriction_to_cq_states}). Then, by \autoref{rem:split_basis} we can express $\rhoge\in\text{End}^{S_{t\lrbracket{\mathcal{D}_1}}}(\mathcal{D}_1)$ in the corresponding basis
\begin{align}
    \mathcal{B}_{\lrbracket{\D_1}} := \lrbrace{\ket{a_1q_1i}\bra{a_1q_1j}_{A_1Q_1T}\otimes C_t^{\An}}_{\substack{a_1\in\lrrec{A_1},\,\\ q_1\in\lrrec{Q_1},\,\\ i,j\in\lrrec{\lrvert{T}},\,\\ t\in\lrrec{m\lrbracket{\An}}}},\,
\end{align}
i.e.
\begin{align}
    \rhoge = \sum_{\substack{a_1\in\lrrec{A_1},\,\\ q_1\in\lrrec{Q_1},\,\\ i,j\in\lrrec{\lrvert{T}},\,\\ t\in\lrrec{m\lrbracket{\An}}}}v_{a_1, q_1,i, j,t}\cdot\ket{a_1q_1i}\bra{a_1q_1j}_{A_1Q_1T}\otimes C_t^{\An},\,
\end{align}
with coefficients $v_{a_1, q_1,i,j,t}\in\mathbb{C}$. For ease of notation, let us define for every $a_1\in\lrrec{\lrvert{A_1}}, q_1\in\lrrec{\lrvert{Q_1}}, i,j\in\lrrec{\lrvert{T}},\, t\in\lrrec{m\lrbracket{\An}}$
\begin{align}
    Z(a_1,q_1, i, j, t):= \ket{a_1q_1i}\bra{a_1q_1j}_{A_1Q_1T}\otimes C_t^{\An},\,
\end{align}
such that
\begin{align}
    \rhoge = \sum_{\substack{a_1\in\lrrec{A_1},\,\\ q_1\in\lrrec{Q_1},\,\\ i,j\in\lrrec{\lrvert{T}},\,\\ t\in\lrrec{m\lrbracket{\An}}}}v_{a_1, q_1,i, j,t}\cdot Z(a_1,q_1, i, j, t).\,
\end{align}
Then, we explicitly express \autoref{eqn:efficient_sdp_free_game} in terms of these $Z(a_1,q_1, i, j, t)$, i.e.

\begin{align}\label{eqn:basis_efficient_sdp}
        &\Psi\lrbracket{\mathrm{SDP}_{(n)}\lrbracket{T, V,\pi}} = \lrvert{T}\max_{v_{a_1, q_1,i,j,t}\in\mathbb{C}} \sum_{\substack{a_1\in\lrrec{A_1},\,\\ q_1\in\lrrec{Q_1},\,\notag\\ i,j\in\lrrec{\lrvert{T}},\,\\ t\in\lrrec{m\lrbracket{\An}}}} v_{a_1, q_1,i,j,t} \\
        &\hspace{4cm}\cdot\Tr\lrrec{\lrbracket{V_{A_1A_2Q_1Q_2}\otimes S_{T\hat{T}}}\Trr{(A_2Q_2\hat{T})_2^n}{Z(a_1,q_1, i, j, t)}}\notag\\
        \notag\\
        \text{s.t.}\hspace{0.5cm} &\sum_{\substack{a_1\in\lrrec{A_1},\,\notag\\ q_1\in\lrrec{Q_1},\,\\ i,j\in\lrrec{\lrvert{T}},\,\notag\\ t\in\lrrec{m\lrbracket{\An}}}} v_{a_1, q_1,i,j,t}\cdot \left[\!\left[\Psi_{\text{End}^{S_{t\lrbracket{\mathcal{D}_1}}}(\mathcal{D}_1)}\lrbracket{Z(a_1,q_1, i, j, t)}\right]\!\right]_{\lambda} \succcurlyeq 0,\,\hspace{0.5cm} \forall \lambda\in\text{Par}\lrbracket{\lrvert{A_2Q_2\hat{T}}, n},\,\notag\\ 
        &\sum_{\substack{a_1\in\lrrec{A_1},\,\\ q_1\in\lrrec{Q_1},\,\notag\\ i,j\in\lrrec{\lrvert{T}},\,\\ t\in\lrrec{m\lrbracket{\An}}}} v_{a_1, q_1,i,j,t}\cdot\Tr\left[Z(a_1,q_1, i, j, t) \right] = 1,\,\notag\\
        \notag\\
         & \sum_{\substack{a_1\in\lrrec{A_1},\,\\ q_1\in\lrrec{Q_1},\,\notag\\ i,j\in\lrrec{\lrvert{T}},\,\notag\\ t\in\lrrec{m\lrbracket{\An}}}} v_{a_1, q_1,i,j,t}\cdot \left[\!\left[\Psi_{\text{End}^{S_{t\lrbracket{\mathcal{D}_2}}}(\mathcal{D}_2)}\lrbracket{\Trr{A_1}{Z(a_1,q_1, i, j, t)}}\right]\!\right]_{\lambda}\notag\\ 
    &\hspace{0.50cm}=\sum_{\substack{a_1\in\lrrec{A_1},\,\notag\\ q_1\in\lrrec{Q_1},\,\\ i,j\in\lrrec{\lrvert{T}},\,\notag\\ t\in\lrrec{m\lrbracket{\An}}}} v_{a_1, q_1,i,j,t}\notag\\
    &\hspace{3cm}\cdot\left[\!\left[\Psi_{\text{End}^{S_{t\lrbracket{\mathcal{D}_2}}}(\mathcal{D}_2)}\lrbracket{\sum_{q_1'}\pi_1(q_1')\ket{q_1'}\bra{q_1'}_{Q_1}\otimes\Trr{A_1Q_1}{Z(a_1,q_1, i, j, t)}}\right]\!\right]_{\lambda},\,\notag\\
    &\hspace{2cm}\forall \lambda\in\text{Par}\lrbracket{\lrvert{A_2Q_2\hat{T}}, n},\,\notag\\
    \notag\\
    &\lrvert{T}\cdot\sum_{\substack{a_1\in\lrrec{A_1},\,\notag\\ q_1\in\lrrec{Q_1},\,\notag\\ i,j\in\lrrec{\lrvert{T}},\,\notag\\ t\in\lrrec{m\lrbracket{\An}}}} v_{a_1, q_1,i,j,t}\cdot \left[\!\left[\Psi_{\text{End}^{S_{t\lrbracket{\mathcal{D}_3}}}(\mathcal{D}_3)}\lrbracket{\Trr{(A_1Q_1T)(A_2)_1}{Z(a_1,q_1, i, j, t)}}\right]\!\right]_{\lambda}\notag\\ 
    &\hspace{0.50cm}=\sum_{\substack{a_1\in\lrrec{A_1},\,\notag\\ q_1\in\lrrec{Q_1},\,\notag\\ i,j\in\lrrec{\lrvert{T}},\,\notag\\ t\in\lrrec{m\lrbracket{\An}}}} v_{a_1, q_1,i,j,t}\notag\\
    &\hspace{3cm}\cdot \left[\!\left[\Psi_{\text{End}^{S_{t\lrbracket{\mathcal{D}_3}}}(\mathcal{D}_3)}\lrbracket{\sum_{q_2}\pi_2(q_2)\ket{q_2}\bra{q_2}_{Q_2}\otimes\mathbb{1}_{\hat{T}}\otimes\Trr{(A_1Q_1T)(A_2Q_2\hat{T})_1}{Z(a_1,q_1, i, j, t)}}\right]\!\right]_{\lambda},\,\notag\\
    &\hspace{2cm}\forall \lambda\in\text{Par}\lrbracket{\lrvert{A_2Q_2\hat{T}}, n-1}.\,\notag\\
\end{align}

We argued that any $Z\in \text{End}^{S_{t(\D_1})}\lrbracket{\D_1}$ can be efficiently transformed into block diagonal form. Thus, the positive semidefinite constraints are efficiently computable with respect to these bases. Subsequently, we demonstrate that the normalization constraint can also be computed efficiently. Next, we show that the operators (matrices) appearing in the constraints in \autoref{eqn:basis_efficient_sdp} can be efficiently expressed in the respective canonical basis. With the bound on $m\lrbracket{\An}$ from \autoref{prop:bounding_our_case}, it suffices to restrict the proofs to a respective summand in the sums appearing in \autoref{eqn:basis_efficient_sdp}, i.e. we show that for all indices
\begin{align}\label{eqn:constraints_in_cano_basis_1}
\begin{split}
   \Trr{A_1}{Z(a_1,q_1, i, j, t)} \,\text{ and }\, \sum_{q_1'}\pi_1(q_1')\ket{q_1'}\bra{q_1'}_{Q_1}\otimes\Trr{A_1Q_1}{Z(a_1,q_1, i, j, t)}
\end{split}
\end{align}
can be efficiently expressed in terms of the canonical basis of $\text{End}^{S_{t\lrbracket{\mathcal{D}_2}}}(\mathcal{D}_2)$
\begin{align}
    \mathcal{B}_{\D_2}:=\lrbrace{\ket{q_1}\bra{q_1}_{Q_1}\otimes\ket{i}\bra{j}_{T}\otimes C_t^{\An}}_{\substack{q_1\in\lrrec{\lrvert{Q_1}},\,\\ i,j\in\lrrec{\lrvert{T}},\, \\t\in m\lrbracket{\lrbracket{A_2Q_2\hat{T}}_1^n} }}
\end{align}and analogously, 
\begin{align}\label{eqn:constraints_in_cano_basis_2}
    \begin{split}
       \Trr{(A_1Q_1T)(A_2)_1}{Z(a_1,q_1, i, j, t)} \,\text{ and }\, \sum_{q_2}\pi_2(q_2)\ket{q_2}\bra{q_2}_{Q_2}\otimes\mathbb{1}_{\hat{T}}\otimes\Trr{(A_1Q_1T)(A_2Q_2\hat{T})_1}{Z(a_1,q_1, i, j, t)}
    \end{split}
\end{align}
in terms of the canonical basis of $\text{End}^{S_{t\lrbracket{\mathcal{D}_3}}}(\mathcal{D}_3)$
\begin{align}
    \begin{split}
         \mathcal{B}_{\D_3}:=\lrbrace{\ket{q_2}\bra{q_2}_{(Q_2)_1}\otimes\ket{u}\bra{v}_{\hat{T}_1}\otimes K_t^{\lrbracket{A_2Q_2\hat{T}}_2^n}}_{\substack{q_2\in \lrrec{\lrvert{(Q_2)_1}},\,\\ u,v\in \lrrec{\lrvert{(\hat{T})_1}}\\t\in m\lrbracket{\lrbracket{A_2Q_2\hat{T}}_2^n} }},\,
    \end{split}
\end{align}
where $\lrbrace{K_t^{\lrbracket{A_2Q_2\hat{T}}_2^n}}_{t\in \lrrec{m\lrbracket{\lrbracket{A_2Q_2\hat{T}}_2^n}}}$ with
\begin{align}
    m\lrbracket{\lrbracket{A_2Q_2\hat{T}}_2^n}= \dim\lrbracket{\text{End}^{\,S_{n-1}\,}\lrbracket{\lrbracket{A_2Q_2\hat{T}}_2^n}}\stackrel{\autoref{prop:bounding_our_case}}{\leq} \lrbracket{n+1}^{\lrvert{A_2Q_2\hat{T}}^2}
\end{align}
is the canonical basis of $\text{End}^{\,S_{n-1}\,}\lrbracket{\lrbracket{A_2Q_2\hat{T}}_2^n}$ which can be efficiently computed by solving the corresponding \autoref{eqn:opti_solving_cano_basis_E} followed by applying the inverse of the bijection $\phi$. Let $\lrbrace{O_t^{\An}}_t$ be the corresponding orbit partitioning. Recall, to each $t$ a representative element $(a,b)\in O_t^{\An}$ can be efficiently computed. Since 
\begin{align}
    \lrrec{\lrvert{A_2Q_2\hat{T}}^n}\simeq \bigtimes_{i=1}^n \lrrec{\lrvert{A_2Q_2\hat{T}}}
\end{align}
we define a mapping for all $j\in\lrrec{n}$
\begin{align}\label{eqn:bijection_psi}
    \begin{split}
        \phi_j \,:\, \lrrec{\lrvert{A_2Q_2\hat{T}}^n}  & \rightarrow \lrrec{\lrvert{A_2Q_2\hat{T}}},\,\\
        a &\mapsto \phi_j(a) =  \frac{a-\sum_{i=1}^{j-1}(\phi_i(a)-1)\lrvert{A_2Q_2\hat{T}}^{i-1}-1}{\lrvert{A_2Q_2\hat{T}}^{j-1}}\mod\lrbracket{\lrvert{A_2Q_2\hat{T}}} +1
    \end{split}
\end{align}
where $\phi_1(a)= (a-1)\mod\lrbracket{\lrvert{A_2Q_2\hat{T}}}+1$. In other words, we choose a representation of each $(a,b)\in  O_t^{\An}$ in terms of coefficients from $\lrrec{\lrvert{A_2Q_2\hat{T}}}$ indexed in $n$, i.e. for any $a\in \lrrec{\lrvert{A_2Q_2\hat{T}}^n}$
\begin{align}
    a = 1+\sum_{j=1}^n (a_j-1)\lrvert{A_2Q_2\hat{T}}^{j-1}
\end{align}
where the $a_j$ correspond to $\phi_j(a)$ for all $j\in\lrrec{n}$. Each $(a,b)\in O_t^{\An}$ corresponds to a basis element of $\mathcal{B}\lrbracket{\An}$, i.e. 
\begin{align}\label{eqn:orbit_rep_corresponding_basis_element}
    (a,b) \leftrightarrow \bigotimes_{i=1}^n\ket{\phi_i(a)}\bra{\phi_i(b)}.\,
\end{align}
Additionally, with
\begin{align}\label{eqn:subsplitting}
    \lrrec{\lrvert{A_2Q_2\hat{T}}}\simeq \lrrec{\lrvert{A_2}}\times\lrrec{\lrvert{Q_2\hat{T}}}
\end{align}
any $a\in \lrrec{\lrvert{A_2Q_2\hat{T}}}$ can be decomposed into $a=(a_{A_2}, a_{Q_2\hat{T}})\in \lrrec{\lrvert{A_2}}\times\lrrec{\lrvert{Q_2\hat{T}}}$.

By the linearity of the trace and the bound on $m\lrbracket{\An}$ from \autoref{prop:bounding_our_case} we deduce that an efficient computation of $\Tr\lrrec{\rhoge}$ follows from an efficient computation of 
\begin{align}\label{eqn:trace_equation_SDP_sym}
\begin{split}
    \Tr\lrrec{\ket{a_1q_1i}\bra{a_1q_1j}_{A_1Q_1T}\otimes C_t^{\An}} &= \Tr\lrrec{\ket{a_1q_1i}\bra{a_1q_1j}_{A_1Q_1T}}\Tr\lrrec{C_t^{\An}}\\
    &=\begin{cases}
        \begin{array}{cc}
             \lrvert{O_t^{\An}} & \text{if } i=j \,\land\, a=b \\
             0 & \text{otherwise,} 
        \end{array}
    \end{cases}
\end{split}
\end{align}
for every $a_1\in\lrrec{A_1},\, q_1\in\lrrec{Q_1},\, i,j\in\lrrec{\lrvert{T}},\, t\in\lrrec{m\lrbracket{\An}}$ and where $(a,b)\in O_t^{\An}$ such that requiring $a=b$ is equivalent to requiring the Hamming-distance
\begin{align}
    H_D(a,b)=\lrvert{\lrbrace{i\in [n]\,:\, \phi_i(a)\neq\phi_i(b)}}
\end{align} to vanish, i.e. $H_D(a,b)=0$. Note that for any basis element $(a,b)\in \lrrec{\lrvert{A_2Q_2\hat{T}}^n}\times \lrrec{\lrvert{A_2Q_2\hat{T}}^n}$, $H_D(a,b)$ is invariant under exchanging system indices, i.e. the action of $S_n$. Consequently, for any $(a,b)\in O_t^{\An}$
\begin{align}
    H_D(a,b)=H_D(c,d),\,\hspace{0.5cm}\forall (c,d)\in O_t^{\An}.\,
\end{align}
Since $\lrvert{O_t^{\An}}$ can be computed efficiently for orbits where elements are diagonal (via the closed form below), the normalization condition in \autoref{eqn:efficient_sdp_free_game} can be computed efficiently. Concretely, let $a$ be of type $\vec{t}$, i.e. $T(a)=\vec{t}$, then by \autoref{prop:dimension_bound_sym_subspace}
 \begin{align}
 	\lrvert{O_t^{\An}} = \binom{n}{\vec{t}},\, \hspace{0.5 cm}\text{if } (a,a)\in O_t^{\An}.\,
 \end{align}
Next, 
\begin{align}
    \begin{split}
        \Trr{A_1}{Z(a_1,q_1, i, j, t)}=\Trr{A_1}{\ket{a_1q_1i}\bra{a_1q_1j}_{A_1Q_1T}\otimes C_t^{\An}} = \ket{q_1i}\bra{q_1j}_{Q_1T}\otimes C_t^{\An}
    \end{split}
\end{align}
and
\begin{align}
    \begin{split}
        \Trr{A_1Q_1T}{Z(a_1,q_1, i, j, t)} = \Trr{A_1Q_1T}{\ket{a_1q_1i}\bra{a_1q_1j}_{A_1Q_1T}\otimes C_t^{\An}} = \begin{cases}
            \begin{array}{cc}
                 C_t^{\An} &  \text{ if } i=j,\,\\
                 0 & \text{ otherwise. }
            \end{array}
        \end{cases}
    \end{split}
\end{align}
Thus, for every $a_1\in\lrrec{\lrvert{A_1}}, q_1\in\lrrec{\lrvert{Q_1}}, i,j\in\lrrec{\lrvert{T}}$ and $t\in\lrrec{m\lrbracket{\lrbracket{A_2Q_2\hat{T}}_1^n}}$ the operators in \autoref{eqn:constraints_in_cano_basis_1} can be efficiently expressed in terms of $\mathcal{B}_{\D_2}$. In order to prove an analogous statement for the expressions in \autoref{eqn:constraints_in_cano_basis_2}, we need to relate the $K_t^{\lrbracket{A_2Q_2\hat{T}}_2^n}$'s from $\mathcal{B}_{\D_3}$ (canonical basis to $\D_3$) to the $C_t^{\An}$'s from $\mathcal{B}_{\D_2}$. We express the orbits of $S_n$ via orbits of $S_{n-1}$. Firstly, note that for any $t\in m\lrbracket{\An}$ the orbit $O_t^{\An}$ can be partitioned further under the action of $S_{n-1}$. To see this, take a representative element $(a,b)$ from $O_t^{\An}$ and collect in a set $O_{(\phi_1(a), \phi_1(b))}^t$ elements from $O_t^{\An}$ which can be reached via permutations on $\lrbracket{A_2Q_2\hat{T}}_2^n$ while keeping $\ket{\phi_1(a)}\bra{\phi_1(b)}_1$ representing $(A_2Q_2\hat{T})_1$ fixed. Repeat this procedure for every tuple from the set
\begin{align}
    \mathcal{S}= \lrrec{\lrvert{A_2Q_2\hat{T}}}\times \lrrec{\lrvert{A_2Q_2\hat{T}}}
\end{align}
until all elements from $O_t^{\An}$ are assigned to an $S_{n-1}$-orbit. Concretely, $S_{n-1}$ partitions $O_t^{\An}$ into $\lrvert{\mathcal{S}}$ sets 
\begin{align}
    O_{(c,d)}^t:=\lrbrace{(a,b)\in O_t^{\An} \, : \, \phi_1(a)=c,\ \phi_1(b)=d},\,
\end{align}
i.e. $O_{(c,d)}^t$ contains all elements of $O_t^{\An}$ where system one is described by $\ket{c}\bra{d}_1$ while the remaining systems are in an $S_{n-1}$-orbit relation. Note that $\forall t\in \lrrec{m\lrbracket{\An}}$ the orbit length $\lrvert{O_{(c,d)}^t}$ is determined by the multiplicities of the basis elements appearing in the $\lrbracket{A_2Q_2\hat{T}}_2^n$-systems, i.e.
\begin{align}\label{eqn:suborbit_length}
    \lrvert{O_{(c,d)}^t} = \binom{n-1}{\vec{m}}\quad\text{with}\quad \lrbracket{\vec{m}}_{(x,y)} := \lrvert{\lrbrace{i\in\lrbrace{2,\ldots,n} \,:\, (\phi_i(a), \phi_i(b)) = (x,y)}},\,
\end{align}
for a representative element $(a,b)$ of $O_{(c,d)}^t$, which can be computed efficiently with respect to $n$\footnote{By construction, elements within an orbit $O_{(c,d)}^t$ necessarily contain the same multiset of basis elements in the $n-1$-copy equivalent systems $\lrbracket{A_2Q_2\hat{T}}_2^n$, so the pair-type $\vec{m}$ does not depend on the chosen representative.}. Next, we write the canonical basis of $\ESN\lrbracket{\An}$ in terms of orbits of $S_{n-1}$, i.e.
\begin{align}\label{eqn:symmetric_efficiency_decomposing_C}
    \begin{split}
        C_t^{\An}&= \sum_{(a,b)\in O_t^{\An}}\bigotimes_{i=1}^n \ket{\phi_i(a)}\bra{\phi_i(b)} = \sum_{(c,d)\in \mathcal{S}} \sum_{(a,b)\in O_{(c,d)}^t} \ket{c}\bra{d}\otimes \lrbracket{\bigotimes_{i=2}^n \ket{\phi_i(a)}\bra{\phi_i(b)}}\\
        &\stackrel{\autoref{eqn:subsplitting}}{=} \sum_{(c,d)\in\mathcal{S}}\sum_{(a, b)\in O_{(c,d)}^t}\ket{c_{A_2}}\bra{d_{A_2}} \otimes \ket{c_{Q_2\hat{T}}}\bra{d_{Q_2\hat{T}}}\\
    &\hspace{2cm}\otimes\lrbracket{\bigotimes_{i=2}^n\ket{\phi_i(a)}\bra{\phi_i(b)}},\,
    \end{split}
\end{align}
where $c_{A_2}, c_{Q_2\hat{T}}$ (and likewise $d_{A_2}, d_{Q_2\hat{T}}$) denote the components of $c$ under \autoref{eqn:subsplitting}. Thus, 
\begin{align}
\begin{split}
     &\Trr{(A_1Q_1T)(A_2)_1}{Z(a_1,q_1, i, j, t)}=\Trr{(A_1Q_1T)(A_2)_1}{\ket{a_1q_1i}\bra{a_1q_1j}_{A_1Q_1T}\otimes C_t^{\An}}\\
     &=\delta_{i,j}\cdot\Trr{(A_2)_1}{C_t^{\An}}\\
    &\stackrel{\autoref{eqn:symmetric_efficiency_decomposing_C}}{=}\delta_{i,j}\cdot \sum_{\substack{(c,d)\in\mathcal{S},\,\\ c_{A_2}=d_{A_2}}}\sum_{(a, b)\in O_{(c,d)}^t}\ket{c_{Q_2\hat{T}}}\bra{d_{Q_2\hat{T}}}
    \otimes\lrbracket{\bigotimes_{k=2}^n\ket{\phi_k(a)}\bra{\phi_k(b)}}
\end{split}
\end{align}
which can be efficiently expressed in terms of $\mathcal{B}_{\D_3}$, since we can efficiently determine one\footnote{Since the indexing of orbits of $S_{n-1}$ in \autoref{eqn:opti_solving_cano_basis_E} does not necessarily coincide with the indexing of orbits of $S_n$.} $t'\in \lrrec{m\lrbracket{\lrbracket{A_2Q_2\hat{T}}_2^n}}$ such that
\begin{align}
    K_{t'}^{\lrbracket{A_2Q_2\hat{T}}_2^n}=\sum_{(a, b)\in O_{(c,d)}^t}
    \bigotimes_{k=2}^n\ket{\phi_k(a)}\bra{\phi_k(b)}.\,
\end{align}
Equivalently,
\begin{align}
    \sum_{q_2}\pi_2(q_2)\ket{q_2}\bra{q_2}_{Q_2}\otimes\mathbb{1}_{\hat{T}}\otimes\Trr{(A_1Q_1T)(A_2Q_2\hat{T})_1}{Z(a_1,q_1, i, j, t)}
\end{align}
can be efficiently expressed by $\mathcal{B}_{\D_3}$ as it constitutes a special case of the argument above without dissecting $\lrbracket{A_2Q_2\hat{T}}_1$ into components.

Lastly, an efficient expansion of the objective function follows from an efficient computation of  
\begin{align}
    &\Trr{(A_2Q_2\hat{T})_2^n}{Z(a_1,q_1, i, j, t)} \\
    &= \Trr{(A_2Q_2\hat{T})_2^n}{\ket{a_1q_1i}\bra{a_1q_1j}_{A_1Q_1T}\otimes C_t^{\An}}\\
    &=\ket{a_1q_1i}\bra{a_1q_1j}_{A_1Q_1T}\otimes \Trr{(A_2Q_2\hat{T})_2^n}{\sum_{(a,b)\in O_t^{\An}}\bigotimes_{k=1}^n \ket{\phi_{k}(a)}\bra{\phi_{k}(b)}_{\lrbracket{A_2Q_2\hat{T}}_{k}}}\\
    &=\begin{cases}
        \begin{array}{cc}
        	  \displaystyle\sum_{c\,:\,\vec{t}_c>0}\binom{n-1}{\vec{t}-\vec{e}_c}\lrrec{\ket{a_1q_1i}\bra{a_1q_1j}_{A_1Q_1T}\otimes\ket{c}\bra{c}_{\lrbracket{A_2Q_2\hat{T}}_1}} & \text{if } H_{D}(a,b)=0,\,\\
        	  \\
             C(t)\lrrec{\ket{a_1q_1i}\bra{a_1q_1j}_{A_1Q_1T}\otimes\ket{a_1}\bra{b_1}_{\lrbracket{A_2Q_2\hat{T}}_1}}& \text{ if } H_D(a,b)=1 \,\text{ and }\, a_1\neq b_1 ,\,\\
           	\\
             0 & \text{ otherwise, }
        \end{array}
    \end{cases}
\end{align}
where $a$ is of type $\vec{t}$ (so that the multinomial coefficients above are efficiently computable), the representative $(a,b)$ of an orbit with $H_D(a,b)=1$ is chosen such that the mismatched position is the first one, and we have exploited the Hamming-distance to efficiently formulate the condition
\begin{align}
    H_D(a,b)=1\hspace{0.5cm} \Leftrightarrow \hspace{0.5cm} \exists ! i\in [n] \,: a_i\neq b_i .\,
\end{align}
In other words, if $H_D(a,b)=1$ the partial trace yields a non-zero contribution whenever the off-diagonal term appears in the first tensor factor. Concretely, for $a$ of type $\vec{t}$ we have 
\begin{align}
    C(t) := \lrvert{\lrbrace{(c,d)\in O_t^{\An}\,:\, \phi_1(c)=a_1,\, \phi_1(d)=b_1}} = \binom{n-1}{\vec{t}_{\mathrm{red}}}
\end{align}
extracts those elements of $O_t^{\An}$ where system $\lrbracket{A_2Q_2\hat{T}}_1$ is in $\ket{a_1}\bra{b_1}_{\lrbracket{A_2Q_2\hat{T}}_1}$, where $\vec{t}_{\mathrm{red}}$ denotes the type of $a$ with the first (mismatched) letter removed. Clearly, for $a$ of type $\vec{t}$ and $b$ of type $\vec{t'}$, $H_D(a,b)=0$ implies $\vec{t}=\vec{t'}$ and similarly $H_D(a,b)=1$ implies $\vec{t}_{\mathrm{red}}=\vec{t'}_{\mathrm{red}}$. Since we can determine a representative element of $O_t^{\An}$ efficiently for every $t$, computing $C(t)$ for every $t\in m\lrbracket{\An}$ can also be done efficiently.
\end{proof}

For the reader's convenience we include the following proposition (see \cite[Lemma 4.5]{chee2023efficient}, \cite[p. 30]{polak2020new} or \cite{fawzi2022hierarchy}) together with its proof. 

\begin{proposition}\label{prop:efficient_computation_explicit_part}
For any $\tau, \tau'\in \mathcal{T}_{\lambda, \lrvert{\HS}}$ and any $\lambda\in\text{Par}(\lrvert{\HS},n)$, given $\lrbrace{z_i}_{i=1}^{m\lrbracket{\HS^n}}$ the expression
	\begin{align}
    \sum_{i=1}^{m\lrbracket{\HS^n}}z_i\cdot u_{\tau}^TC_i^{\HS}u_{\tau'}
\end{align}
can be computed efficiently. 
\end{proposition}

\begin{proof}
	Let $\HS:=A_2Q_2\hat{T}$. Throughout this proof we abbreviate $m\lrbracket{\HS}:=m\lrbracket{\HS^n}$ and write $O_r^{\HS}$ for the orbits of the $S_n$-action on $\lrrec{\lrvert{\HS}^n}^2$. The proof is based on a dual system consisting of $W_{\HS}:= \HS\otimes \HS$ and $W_{\HS}^*:=\HS^*\otimes \HS^*$. For each $p=(x,y)\in \lrrec{\lrvert{\HS}}\times \lrrec{\lrvert{\HS}}$, define $a_p := e_x\otimes e_y \in W_{\HS}$ for the canonical basis $\lrbrace{e_x}_{x\in\lrrec{\lrvert{\HS}}}$ of $\HS$, and let $\mathcal{W}_{\HS}$ denote the corresponding basis. Similarly, consider $a_p^*$ and the dual basis $\mathcal{W}_{\HS}^*$ of $W_{\HS}^*$. Note, that
	\begin{align}
    a^*_{p_i}a_{p_j}=\langle a_{p_i}, a_{p_j}\rangle \stackrel{\lrbrace{a_{p_i}}_{i} \text{ ONB}}{=}\delta_{i,j}.\,
\end{align}
	Consider the coordinate ring of homogeneous polynomials of degree $n$ on $W_\HS$ denoted $\mathcal{O}_n\lrbracket{W_{\HS}}$. 	Importantly, as opposed to a free algebra generated from the same alphabet, in the coordinate ring, the ordering of letters in monomials is irrelevant. 
	We bijectively identify a representative element $(a,b)$ of each orbit $O_{(a,b)}^{\HS}$ with a degree $n$ monomial in the coordinate ring. Any other element $(c,d)\in O_{(a,b)}^{\HS}$ is mapped to the same monomial. Concretely, the map\footnote{Equivalently, we can define $\mu$ on $\lrrec{\lrvert{\HS}^2}^n$.} 
\begin{align}\label{eqn:mu_map}
	\begin{split}
		\mu \,:\, &\lrbracket{\lrrec{\lrvert{\HS}}\times\lrrec{\lrvert{\HS}}}^{n} \rightarrow \mathcal{O}_n\lrbracket{W_{\HS}}\\
		&(p_1,\ldots, p_n) \mapsto \prod_{i=1}^n a_{p_i}^*
	\end{split}
\end{align}
is an $S_n$-orbit class function since
\begin{align}
    \mu\lrbracket{\sigma\lrbracket{p_1, \ldots, p_n}} = \prod_{i=1}^n a^*_{p_{\sigma^{-1}(i)}}= \prod_{i=1}^n a^*_{p_i} =   \mu\lrbracket{p_1, \ldots, p_n},\,
\end{align}
for any $\sigma\in S_n$.

For semistandard Young tableaux $\tau,\gamma\in\mathcal{T}_{\lambda, \lrvert{\HS}}$ of Young shape $Y(\lambda)$, consider the polynomial $G_{\tau,\gamma}$ in the polynomial ring with complex coefficients $\mathbb{C}\lrrec{x_{i,j}\,\vert\, i,j=1,\ldots,\lrvert{\HS}}$ given by
\begin{align}\label{eqn:big_G_def}
    G_{\tau,\gamma}(X) := \sum_{\substack{\tau'\sim \tau,\\ \gamma' \sim \gamma\\}} \sum_{c, c'\in C_{\lambda}} \text{sgn}(cc')\prod_{y\in Y(\lambda)} x_{\tau'\lrbracket{c(y)}, \gamma'\lrbracket{c'(y)}},\,
\end{align}
where $X=\lrbracket{x_{i,j}}_{i,j=1}^{\lrvert{\HS}}$ is a matrix indexed in the canonical basis of $\mathcal{B}\lrbracket{\HS}$. Now by \cite[Proposition 3]{gijswijt2009block} and \cite[Theorem 7]{litjens2017semidefinite}, for fixed $\lrvert{\HS}$, the polynomial $G_{\tau,\gamma}$ can be efficiently expressed w.r.t. $n$ as a linear combination of monomials. 

Let $(a,b)\in \lrrec{\lrvert{\HS}^n}\times \lrrec{\lrvert{\HS}^n}$ be mapped under $\phi$ to $(p_1, \ldots, p_n)\in \lrbracket{\lrrec{\lrvert{\HS}} \times \lrrec{\lrvert{\HS}}}^n$, i.e. $p_i = (\phi_i(a),\phi_i(b))$ so that we can write
\begin{align}\label{eqn:adjacency_matrix}
    C_{(a,b)}^{\HS} = \sum_{(c,d) \in O_{(a,b)}^{\HS}} \bigotimes_{i=1}^n a_{\phi_i(c),\phi_i(d)}=\sum_{(p'_1, \ldots, p'_n) \in O_{(p_1, \ldots, p_n)}^{\HS}} \bigotimes_{i=1}^n a_{p'_i}.\,
\end{align}
Let us additionally introduce a linear map translating $n$-fold tensor products to monomials
\begin{align}\label{eqn:zeta}
    \begin{split}
        \zeta \,:\, \lrbracket{W_{\HS}^*}^{\otimes n} &\rightarrow \mathcal{O}_n\lrbracket{W_{\HS}}\\
         \bigotimes_{i=1}^n w_i^* &\mapsto \prod_{i=1}^n w_i^*.\,
    \end{split}
\end{align}
With the canonical inner product on $\HS^{\otimes n}$ we can identify the polytabloids as dual vectors, i.e.  $u_{\tau}, u_{\gamma}\in \lrbracket{\HS^{\otimes n}}^*$ such that
\begin{align}\label{eqn:dual_polytabloids}
    u_{\tau}^TC_r^{\HS}u_{\gamma} = \lrbracket{u_{\tau} \otimes u_{\gamma}}\lrbracket{C_r^{\HS}}\,\in \CC
\end{align} 
where we have used 
\begin{align}
     \lrbracket{\HS^{\otimes n}}^*\otimes \lrbracket{\HS^{\otimes n}}^* \simeq \lrbracket{\HS^{\otimes n}\otimes \HS^{\otimes n}}^*.\,
\end{align}
Recall that the computation of polytabloid $u_{\tau}$ (or $u_{\gamma}$) constituted an application of the corresponding Young symmetrizer constructed from Young tableau $\tau$ (or $\gamma$) of shape $\lambda$ on $n$ copies of the canonical basis of $\HS$. Thus, for dual vector $u_{\tau}\otimes u_{\gamma}$ we apply the Young symmetrizers on the basis of the dual space $\lrbracket{\HS^{\otimes n}\otimes \HS^{\otimes n}}^*$. Concretely, by definition
\begin{align}\label{eqn:g_def}
     g := u_{\tau} \otimes u_{\gamma} =  \underbrace{\sum_{\substack{\tau'\sim \tau,\\ \gamma' \sim \gamma\\}} \sum_{c, c'\in C_{\lambda}} \text{sgn}(cc')}_{\text{Young sym.}}\underbrace{\bigotimes_{y\in Y(\lambda)} F_{\tau'\lrbracket{c(y)}, \gamma'\lrbracket{c'(y)}}}_{n-\text{fold tensor product}},\,
\end{align}
where $F\in \lrbracket{\mathcal{W}_{\HS}^*}^{\lrvert{\HS}\times \lrvert{\HS}}$ is an encoding of the canonical basis of the dual space, i.e. $\lrbracket{F}_{x,y}=a^*_{(x,y)}$. Then, 
\begin{align}
    \sum_{(p_1,\ldots, p_n)\in \lrbracket{\lrrec{\lrvert{\HS}}\times\lrrec{\lrvert{\HS}}}^n} g\lrbracket{\bigotimes_{i=1}^n a_{p_i}}\prod_{i=1}^n a^*_{p_i} \,\in \mathcal{O}_n\lrbracket{W_{\HS}}
\end{align}
is a linear combination of $n$-degree dual monomials with coefficients 
\begin{align}\label{eqn:concrete_coeff}
\begin{split}
    g\lrbracket{\bigotimes_{i=1}^n a_{p_i}} \stackrel{\ref{eqn:g_def}}{=} \sum_{\substack{\tau'\sim \tau,\\ \gamma' \sim \gamma\\}} \sum_{c, c'\in C_{\lambda}} \text{sgn}(cc') \prod_{y\in Y(\lambda)} a^*_{\tau'(c(y)), \gamma'(c'(y))}a_{p_y},\,
\end{split}
\end{align}
where orthonormality between basis elements yields
\begin{align}\label{eqn:ONB_concrete}
    \begin{split}
        a^*_{\tau'(c(y)), \gamma'(c'(y))}a_{p_y} = \begin{cases}
            \begin{array}{cc}
                1,\, &  \text{if }\, p_y = \lrbracket{\tau'(c(y)), \gamma'(c'(y))} ,\,\\
                0,\, & \text{else}.\,
            \end{array}
        \end{cases}
    \end{split}
\end{align}
Next, we show that the polynomial 
\begin{align}
    \sum_{r=1}^{m(\HS)} \lrbracket{u_{\tau} \otimes u_{\gamma}}\lrbracket{C_r^{\HS}} \mu(O_{r}^{\HS})
\end{align}
with $n$-degree dual monomials representing each orbit can be equated to $G_{\tau, \gamma}(F)$ from \autoref{eqn:big_G_def} for a specific $F$. Concretely, 
\begin{align}
\begin{split}
    \sum_{r=1}^{m(\HS)} \lrbracket{u_{\tau} \otimes u_{\gamma}}\lrbracket{C_r^{\HS}} \mu(O_{r}^{\HS})&\stackrel{\ref{eqn:g_def}}{=} \sum_{r=1}^{m(\HS)} g\lrbracket{C_r^{\HS}} \mu(O_{r}^{\HS})\\
    &\stackrel{\ref{eqn:adjacency_matrix}}{=} \sum_{r=1}^{m(\HS)} g\lrbracket{\sum_{(p_1, \ldots, p_n) \in O_{r}^{\HS}}\bigotimes_{i=1}^n a_{p_i}} \mu(O_{r}^{\HS})\\
    &\stackrel{g \text{ linear}}{=} \sum_{r=1}^{m(\HS)}\sum_{(p_1, \ldots, p_n) \in O_{r}^{\HS}} g\lrbracket{\bigotimes_{i=1}^n a_{p_i}} \mu(O_{r}^{\HS})\\
    \\
    &\stackrel{\substack{\ref{eqn:mu_map},\, \ref{eqn:concrete_coeff},\, \ref{eqn:ONB_concrete}\\ \lrbracket{\lrrec{\lrvert{\HS}}\times\lrrec{\lrvert{\HS}}}^n = \bigsqcup_{r=1}^{m(\HS)} O_r^{\HS}}}{=} \sum_{(p_1,\ldots, p_n)\in \lrbracket{\lrrec{\lrvert{\HS}}\times\lrrec{\lrvert{\HS}}}^n} g\lrbracket{\bigotimes_{i=1}^n a_{p_i}}\prod_{i=1}^n a^*_{p_i}\\
    &\stackrel{\ref{eqn:zeta}}{=} \sum_{(p_1,\ldots, p_n)\in \lrbracket{\lrrec{\lrvert{\HS}}\times\lrrec{\lrvert{\HS}}}^n} g\lrbracket{\bigotimes_{i=1}^n a_{p_i}}\zeta\lrbracket{\bigotimes_{i=1}^n a^*_{p_i}}\\
    &\stackrel{\ref{eqn:g_def}}{=}\zeta\lrbracket{u_{\tau}\otimes u_{\gamma}} \stackrel{\zeta \text{ lin.}}{=} \sum_{\substack{\tau'\sim \tau,\\ \gamma' \sim \gamma\\}} \sum_{c, c'\in C_{\lambda}} \text{sgn}(cc')\,\zeta\lrbracket{\bigotimes_{y\in Y(\lambda)} F_{\tau'\lrbracket{c(y)}, \gamma'\lrbracket{c'(y)}}}\\
    &= \sum_{\substack{\tau'\sim \tau,\\ \gamma' \sim \gamma\\}} \sum_{c, c'\in C_{\lambda}} \text{sgn}(cc')\,\prod_{y\in Y(\lambda)} F_{\tau'\lrbracket{c(y)}, \gamma'\lrbracket{c'(y)}} \\
    &\stackrel{\ref{eqn:big_G_def}}{=} G_{\tau, \gamma}\lrbracket{F}.\,
    \end{split}
\end{align}
 In other words, $G_{\tau, \gamma}\lrbracket{F}$ is a linear combination of poly$(n)$-many (i.e. $m(\HS)$) monomials $\mu(O_r^{\HS})$. Recall that any $Z\in\ESNH$ can be given in the canonical basis $\lrbrace{C_r^{\HS}}_{r=1}^{m(\HS)}$ of $\ESNH$, i.e.
 \begin{align}
     Z = \sum_{r=1}^{m(\HS)} z_rC_r^{\HS}.\,
 \end{align}
 Since the $\lrbrace{z_r}_{r\in\lrrec{m(\HS)}}$ are a collection of class functions on orbits $\lrbrace{O_r^{\HS}}_{r\in\lrrec{m(\HS)}}$ and the $\lrbrace{C_r^{\HS}}_{r\in\lrrec{m(\HS)}}$ are linearly independent as basis elements, determining each $z_r$ requires comparing $m(\HS)$ many entries of $Z$ with a corresponding element in each of the $C_r^{\HS}$. Concretely, since the $C_r^{\HS}$ contain exclusively zeroes and ones, for $(x,y)\in O_r^{\HS}$ efficiently computed via \autoref{eqn:opti_solving_cano_basis_E} we obtain $ z_r= (Z)_{x,y}$. In summary, we have shown that given $\lrbrace{z_r}_{r\in\lrrec{m(\HS)}}$
\begin{align}
     \sum_{r=1}^{m(\HS)} z_ru_{\tau}^TC_r^{\HS}u_{\gamma}
\end{align} 
can be computed efficiently.
\end{proof}

\section{Representation theory}\label{sec:representation_theory}

In this section, we primarily adopt a module-theoretic approach to representation theory. Nevertheless, group-theoretic language is occasionally employed where appropriate. The exposition is largely based on \cite{fulton2013representation, blyth2018module, etingof2009introduction, lassueurshort, lassueur2020modular}, and is intended as a review for the reader. We claim no novelty for the material presented.
	
We will briefly introduce the main mathematical objects relevant to this work. We assume familiarity with group theory. A ring $(R,+, \cdot)$ is an abelian group under addition, a semigroup under multiplication and the multiplication is distributive with respect to the addition. $R$ is unitary if it contains a multiplicative identity element, i.e.\ if it is a monoid under multiplication. A division ring is a unitary ring in which the non-zero elements form a group under multiplication, i.e. to each such element there exists a multiplicative inverse. In a commutative ring, the multiplication is commutative. An associative algebra $\Al$ over a commutative ring $R$ is a ring together with a ring homomorphism from $R$ into $Z(\Al)$, the center of $\Al$.

\begin{remark}
		A field is a commutative division ring. By Wedderburn's little theorem, any finite division ring is commutative and, therefore, a field.
\end{remark}
	
Importantly, $\Al$ is also an $R$-module.
	
\begin{definition}[$R$-module]
		Let $R$ be a unitary ring with identity $1_R$. A left $R$-module, or a module over $R$ is an additive abelian group $M$ together with a left action $R \times M \rightarrow M$, described by $(\lambda, x)\mapsto \lambda x$, such that
		\begin{enumerate}
			\item $\lambda(x+ y)= \lambda x + \lambda y$,\,\hspace{1cm} $\forall \lambda\in R,\, \forall x,\, y\in M$\,;
			\item $(\lambda + \mu)x = \lambda x+ \mu x$,\,\hspace{1cm} $\forall \lambda,\, \mu\in R,\, \forall x\in M$\,;
			\item $\lambda(\mu x)= (\lambda \mu)x$,\,\hspace{1cm} $\forall \lambda,\, \mu\in R,\, \forall x\in M$\,;
			\item $1_Rx=x$,\,\hspace{1cm} $\forall x\in M$.\,
		\end{enumerate}
\end{definition}
    
	\begin{remark}
		A right $R$-module is defined analogously with a right action $M\times R\rightarrow M$ described by $(x, \lambda)\mapsto x\lambda$. For $R$ a commutative ring, the definitions coincide.
	\end{remark}
	
\begin{remark}
		Note that for a ring $R$, the set of $n\times n$ matrices with entries from $R$ forms a ring under matrix addition and matrix multiplication, called the matrix ring and denoted $R^{n\times n}$ or $M_n(R)$. Moreover, for $R$ commutative, $R^{n\times n}$ is an associative algebra.
\end{remark}
	
A free $R$-module $M$ is an $R$-module with a basis, i.e. there  exists a set $B\subseteq M$ of linearly independent elements generating $M$.
    
\begin{remark}
		For any set $S$ and ring $R$, one can construct a free $R$-module with basis $S$, called the free module on $S$. Any vector space over a field is a free module by virtue of having a basis.
\end{remark}
	
We will now relate these concepts to groups. 
	
	\begin{definition}[Group algebra]
		For a finite group $G$ and ring $R$, the group ring of $G$ over $R$ denoted $R\lrrec{G}$ is a free module on $G$, i.e. all formal $R$-linear combinations of elements from $G$. Furthermore, if $R$ is a field, $R\lrrec{G}$ is an $R$-vector space with basis $G$, hence an $R$-algebra, which we refer to as the group algebra.
\end{definition}
    
\begin{remark}
		By construction, the $R$-vector space structure naturally makes $R\lrrec{G}$ an abelian group under addition. Furthermore, the group multiplication is extended by $R$-bilinearity to yield the necessary multiplicative monoid structure. Moreover, $R\lrrec{G}$ inherits the identity from $G$. By definition of a free module, the $R$-vector space dimension of $R\lrrec{G}$ is given by the order of $G$ and $R\lrrec{G}$ is commutative if and only if $G$ is abelian. Notably, the ring structure of $R\lrrec{G}$ allows us to apply general results on $R$-modules to $R\lrrec{G}$-modules. We will mainly work with 
		\begin{align}
			\CC\lrrec{G}=\lrbrace{\sum_{g\in G}\alpha_g g\,:\, \alpha_g\in\CC}.\,
		\end{align}
		Furthermore, we can think of the group algebra as the space of complex-valued functions on $G$, i.e.
		\begin{align}
			\CG \simeq \lrbrace{f \,:\, G\rightarrow \CC},\,
		\end{align}
		as $G$ forms a linearly independent basis of $\CG$.
\end{remark}
    
The coefficients $\alpha_g$ of an element of the center $Z\lrbracket{R\lrrec{G}}$ are class functions, i.e. constant on conjugacy classes of $G$. There exists a natural connection between group representations and group algebra modules. Let $\mathbb{K}$ be a field of characteristic zero.

	\begin{proposition}
		Any $\KK$-representation $\rho \,:\, G\rightarrow \GLV$ of $G$ yields a $\KG$-module structure on $V$,
		\begin{align}
		\begin{split}
			\cdot \,:\, \KG\times V &\rightarrow V\\
			\lrbracket{\sum_{g\in G}c_gg, v} &\mapsto \lrbracket{\sum_{g\in G}c_gg}\cdot v := \sum_{g\in G}c_g\rho(g)v.\,
		\end{split}
		\end{align}
		Conversely, every $\KG$-module $(V, +, \cdot)$ defines a $\KK$-representation
		\begin{align}
			\begin{split}
				\rho_V \,:\, G&\rightarrow \GLV\\
				g &\mapsto \rho_V(g)\,:\,V\rightarrow V, v\mapsto \rho_V(g)(v) := g\cdot v
			\end{split}
		\end{align}
		of $G$. 
\end{proposition}
    
\begin{remark}
	Note that $\text{Aut}(V)\subseteq\End{}{V}$ denotes the group of invertible endomorphisms, i.e. automorphisms on $V$. Note that $\GLV\simeq \text{Aut}(V)$ and $\GLV\simeq \text{GL}_n(\KK)$ as groups if $\dim(V)=n$.
\end{remark}

\begin{remark}
	Note that every unitary ring $R$ is an $R$-module where the $R\times R\rightarrow R$ action is the multiplication in $R$ and similarly, any field $\mathbb{K}$ is a $\mathbb{K}$-vector space. In particular, the field $\KK$ itself becomes a $\mathbb{K}\lrrec{G}$-module via the $G$-action
			\begin{align}
			\begin{split}
				\cdot \;:\; G\times \KK \rightarrow \KK\\
				(g, \lambda) \mapsto g\cdot \lambda := \lambda
			\end{split}
			\end{align}
			extended by $\KK$-linearity to $\KG$; this is called the trivial $\KG$-module. Accordingly, we refer to $\KG$, viewed as a $\KG$-module via left multiplication, as the regular $\KG$-module and denote it by $\KG^{\circ}$.
\end{remark}
    
\begin{remark}
		Since $G$ is a group, the anti-automorphism
		\begin{align}
		\begin{split}
			\KG &\rightarrow \KG\\
			 g &\mapsto g^{-1}
		\end{split}
		\end{align}
		naturally makes any left $\KG$-module into a right $\KG$-module.
\end{remark}

\begin{remark}\label{rem:KG_balanced_product}
	For $N, M$ two $\KG$-modules, the $\KK$-balanced tensor product $N\otimes_{\KK}M$ is also a $\KG$-module via diagonal action of $G$, i.e.
		\begin{align}
			\begin{split}
				\cdot \, :\, G\times (N\otimes_{\KK}M) &\rightarrow N\otimes_{\KK}M \\
				(g, n\otimes m) &\mapsto g\cdot (n \otimes m):= gn\otimes gm
			\end{split}
		\end{align}
	extended to $\KG$ via $\KK$-linearity.
\end{remark}

We are especially interested in substructures of $R$-modules.

\begin{definition}
	A submodule of an $R$-module $M$ is a subgroup $N$ of $M$ that is stable under the action of $R$ on $M$, in the sense that if $x\in N$ and $\lambda\in R$ then $\lambda x\in N$. Equivalently, $N$ is an $R$-submodule of $M$, if and only if
	\begin{align}
		\lambda x+ \mu y\in N,\, \hspace{1cm} \forall x,y \in N, \forall \lambda, \mu \in R.\,
	\end{align} 
	Accordingly, a non-empty subset $\BS$ of an $R$-algebra $\Al$ is a subalgebra of $\Al$ if
	\begin{align}
		x-y\in \BS, xy\in \BS, \lambda x\in \BS,\, \hspace{1cm} \forall x,y\in \BS, \forall\lambda\in R.\,
	\end{align}
\end{definition}

\begin{remark}
	For any $R$-module $M$, $M$ and $\lrbrace{0}$ are always trivial submodules.
\end{remark}

A module is simple if and only if it contains no proper nonzero submodules, while a semisimple module is decomposable into a direct sum of simple modules. We extend the notion to algebras and say that an algebra is semisimple if all its finite-dimensional modules are semisimple.

\begin{remark}
	Simple modules correspond to irreducible representations and semisimple to completely reducible representations. For finite groups and fields of characteristic zero, we can always assume there to be a semisimple decomposition. 
\end{remark}

Next, we introduce structure-preserving maps, i.e. homomorphisms between modules. 

\begin{definition}
	If $M$ and $N$ are $R$-modules then a mapping $f\,:\, M\rightarrow N$ is called an $R$-homomorphism if
	\begin{enumerate}
		\item $f(x + y) = f(x) + f(y)$, \hspace{1cm} $\forall x,y\in M$,
		\item $f(\lambda x) = \lambda f(x)$, \hspace{1cm} $\forall x\in M$, $\forall\lambda\in R$.
	\end{enumerate}
	$M$ and $N$ are equivalent as $R$-modules, i.e. $R$-isomorphic, if there exists a bijective $R$-homomorphism between them. The set of $R$-homomorphisms from $M$ to $N$ forms an abelian group under addition and is denoted by $\Hom{R}{M}{N}$. Furthermore, for $M=N$, the set $\End{R}{M}:=\Hom{R}{M}{M}$ admits the structure of an algebra under composition. Moreover, $\End{R}{M}$ is the centralizer algebra of the $R$-action on $M$, i.e. endomorphisms on $M$ commuting with the $R$-action.
\end{definition}

\begin{remark}
	In representation theory, $G$-equivariant maps constitute an analogous concept to $\KG$-homomorphisms. 
\end{remark}

While in general, $\Hom{R}{M}{N}$ does not form an $R$-module, restricting to $\KG$ allows for the following proposition.

\begin{proposition}
	Let $M, N$ be $\KG$-modules. Then, the abelian group $\Hom{\KK}{M}{N}$ is a $\KG$-module via the conjugate $G$-action
	\begin{align}
		\begin{split}
			\cdot \,:\, G \times \Hom{\KK}{M}{N}&\rightarrow \Hom{\KK}{M}{N}\\
			(g, f) &\mapsto g\cdot f: M \rightarrow N, m \mapsto (g\cdot f)(m):= g\cdot f(g^{-1}\cdot m) 
		\end{split}
	\end{align}
	extended to $\KG$ by $\KK$-linearity.
\end{proposition}

\begin{remark}
    Note that $\Hom{\KG}{M}{N}$ is in general not a $\KG$-module.
\end{remark}

\begin{remark}
	For a $\KG$-module $M$ we can impose a  $\KG$-module structure on its $\KK$-dual $M^*:=\Hom{\KK}{M}{\KK}$ via
	\begin{align}
		\begin{split}
			\cdot \,:\, G\times M^* &\rightarrow M^*\\
			(g, f) &\mapsto g\cdot f \,:\, M\rightarrow \KK,\, m\mapsto(g\cdot f)(m) := f(g^{-1} \cdot m)
		\end{split}
	\end{align}
	extended to $\KG$ by $\KK$-linearity. Furthermore, for another $\KG$-module $N$, $\Hom{\KK}{M}{N}\simeq M^*\otimes_{\KK} N$.
\end{remark}

Importantly, Schur's lemma characterizes $R$-homomorphisms. 

\begin{proposition}[Schur \cite{blyth2018module}]\label{prop:Schur}
	Let $M$, $N$ be simple modules over ring $R$. Then, any homomorphism $f: M\rightarrow N$ of $R$-modules is either invertible or zero. Furthermore, the ring $\End{R}{M}$ of $R$-morphisms is a division ring.
\end{proposition}	

\begin{remark}
	Since all division rings are simple, $\End{R}{M}$ is not decomposable. 
\end{remark}

We can further specify Schur's lemma for the group algebra over an adequate field. 

\begin{proposition}[Schur's theorem for the group algebra $\KG$]
	\begin{enumerate}
		\item Let $V, W$ be simple $\mathbb{K}\lrrec{G}$-modules. Then, 
		\begin{enumerate} 
		\item Any homomorphism of $\KG$-modules $\phi\,:\, V\rightarrow V$ is either zero or invertible. Equivalently, $\End{\KG}{V}$ is a division ring. 
		\item If $V\ncong W$, then $\Hom{\KG}{V}{W}=0$.
		\end{enumerate}
	\item If $\mathbb{K}$ is an algebraically closed field and $V$ is a simple $\KG$-module, then \begin{align}
		\End{\KG}{V}=\lrbrace{\lambda \text{Id}_V \,:\, \lambda\in\KK}\simeq \KK.\,
	\end{align}
	\end{enumerate}
\end{proposition}

\begin{remark}
	Significantly, by the fundamental theorem of algebra, $\CC$ is algebraically closed.
\end{remark}

For a ring $R$, we collect the isomorphism classes of simple $R$-modules in the set
\begin{align}
	\Irr{R}:=\lrbrace{\text{isomorphism classes of simple } R\text{-modules}}.\,
\end{align}

\begin{definition}
	For $M$ a semisimple $R$-module and $S$ a simple $R$-module, the $S$-homogeneous component of $M$, denoted $S(M)$, is the sum of all simple $R$-submodules of $M$ isomorphic to $S$. 
\end{definition}

\begin{proposition}[Wedderburn decomposition (cf.\ Artin--Wedderburn)]\label{prop:wedderburn_little}
	Let $R$ be a semisimple ring, then the following assertions hold. 
	\begin{enumerate}
		\item If $S\in\Irr{R}$, then $S\lrbracket{R^{\circ}}\neq 0$. Furthermore, $\lrvert{\Irr{R}}<\infty$.
		\item We have \begin{align}
			R^{\circ} = \bigoplus_{S\in\Irr{R}}S\lrbracket{R^{\circ}},\,
		\end{align}
		where each homogeneous component $S\lrbracket{R^{\circ}}$ is a two-sided ideal of $R$ and $S(R^{\circ})T(R^{\circ})=0$ if $S\neq T\in \Irr{R}$.
		\item Each $S\lrbracket{R^{\circ}}$ is a simple left Artinian ring, the identity element of which is an idempotent element of $R$ lying in $Z(R)$.
	\end{enumerate}
\end{proposition}

The center of $R$ is spanned by the primitive central idempotents, each projecting onto an isotypical component of $R^{\circ}$. Characters encode how minimal central idempotents are constructed as linear combinations of group elements.
\begin{remark}
Clearly, for $x\in\KG$, the image of $x$ in $\End{}{V}$ lies in $\End{\KG}{V}$ if $x\in Z\lrbracket{\KG}$.
\end{remark}

We will also be needing the Artin--Wedderburn theorem for algebras over a field.

\begin{proposition}[Corollary to the Wedderburn and Artin--Wedderburn theorems]\label{prop:wederburn_artin_wederburn}
Let $\Al$ be a semisimple, split\footnote{In our considerations, $\mathbb{K}$ is an algebraically closed field.} $\KK$-algebra and let $S\in\text{Irr}\lrbracket{\Al}$ be a simple $\Al$-module. Then,
\begin{enumerate}
	\item For the regular (and semisimple) $\Al$-module $\Al^{\circ}$ the $S$-homogeneous component of $\Al^{\circ}$ is isomorphic to a matrix algebra. Concretely,  $S\lrbracket{\Al^{\circ}}\simeq \KK^{n_S\times n_{S}}$ with $\dim_{\KK}\lrbracket{S\lrbracket{\Al^{\circ}}}=n_S^2$;
	\item $\dim_{\KK}(S)=n_S$;
	\item $\dim_{\KK}\lrbracket{\Al}=\sum_{S\in \Irr{\Al}}\dim_{\KK}(S)^2$;
	\item $\lrvert{\Irr{\Al}}=\dim_{\KK}\lrbracket{Z(\Al)}$.
\end{enumerate}
\end{proposition}

Importantly, we thus have the unique isotypic decomposition 
\begin{align}
	R^{\circ} = \bigoplus_{S\in\Irr{R}}S\lrbracket{R^{\circ}}\simeq \bigoplus_{S\in\Irr{R}} \KK^{n_S\times n_S}
\end{align}
into isotypical components $\lrbrace{S\lrbracket{R^{\circ}}}_{S\in\Irr{R}}$.

\begin{remark}
	Notably, while this decomposition into isotypical components is unique, i.e. each component is unique, the decomposition of each isotypical component into isomorphic simple $R$-modules is not unique. Furthermore, other semisimple decompositions are basis dependent. 
\end{remark}

Importantly, a submodule of $R^{\circ}$ is simple if and only if it is a minimal left ideal. Minimal ideals are generated from minimal (or primitive) idempotents, i.e. those that cannot be decomposed into a sum of two nonzero, orthogonal idempotents.

\begin{proposition}[\cite{blyth2018module}, Theorem 6.13]
	An endomorphism $f\in\End{R}{M}$ is an idempotent, if and only if it is a projection.
\end{proposition}

There is a strong connection between projections and submodules.

\begin{proposition}[\cite{blyth2018module}, Theorem 6.14]
\label{prop:connection_projections_and_submodules}
An $R$-module $M$ is the direct sum of submodules $M_1,\ldots, M_n$ if and only if there are non-zero $R$-morphisms $p_1,\ldots,p_n\in\End{R}{M}$ such that
	\begin{itemize}
		\item $\sum_{i=1}^np_i=\mathcal{I}_M,\,$
		\item $p_i\circ p_j = 0$ for $i\neq j$.
	\end{itemize}   
\end{proposition}

\begin{remark}
	While at the moment, the existence of primitive idempotents for each simple $R$-module follows from its minimal ideal structure, in the case of $R=\CSn$, we will later use bijections to Young tableaux to construct them explicitly.
\end{remark}

Character theory allows us to formalize some aspects of the decomposition of $\CG$ as a $\CG$-module. Concretely, by the Wedderburn decomposition (\autoref{prop:wedderburn_little}), the homogeneous components are the images of idempotents from $Z\lrbracket{\CG}$. We can identify the space of class functions on $G$ with $Z(\CG)$ as vector spaces. The characters form a basis of the space of class functions on $G$. 
Thus, we can construct elements from $Z(\CG)$ via simple (irreducible) characters. The characters tell us how group elements need to be combined to yield idempotents as these must yield class functions when evaluated at group elements. Furthermore, each minimal central idempotent is determined by a simple character.

\begin{proposition}
    Let $V^{(i)}$ be a simple $\CG$-module in $\CG$ and let $\chi^{(i)}$ denote the corresponding character. Then, 
         \begin{align}\label{eqn:idempotents_characters}
  	e_{i} = \frac{d_{i}}{\lrvert{G}}\sum_{g\in G} \chi^{(i)}\lrbracket{g^{-1}}g,
  \end{align}
  where $d_i=\dim_{\CC}\lrbracket{V^{(i)}}$, is a central minimal idempotent projecting onto the isotypical component of $V^{(i)}$.
\end{proposition}

\begin{proposition}[\cite{etingof2009introduction}]\label{prop:decomp_semisimple_algebra}
	Let $\Al$ be a finite-dimensional algebra. Then $\Al$ has finitely many simple modules $V_i$ up to isomorphism. These simple modules are finite-dimensional. Moreover, $\Al$ is semisimple if and only if as an algebra, 
	\begin{align}
		\Al \simeq \bigoplus_i \End{}{V_i}
	\end{align}
	where $V_i$ are simple $\Al$-modules.
\end{proposition}

\begin{proposition}\label{prop:direct_prod_of_simple_mods}
	Let $\Al$ be an algebra with simple $\Al$-module $V$ of finite dimension $n\in\mathbb{N}$. The space $\End{}{V}$ with left $\Al$-action is a semisimple left $\Al$-module such that
	\begin{align}
		\End{}{V}\simeq V^{\oplus n}.\,
	\end{align}
\end{proposition}

\begin{proof}
	Let $\lrbrace{v_i}_{i=1}^n$ be a basis of $V$, then 
	\begin{align}
	\begin{split}
		\psi \,:\, &\End{}{V} \rightarrow V^{\oplus n}\\
		& x \mapsto (xv_1, \ldots, xv_n)
	\end{split} 
	\end{align}
	is an isomorphism. 
\end{proof}

\begin{proposition}\label{prop:V_decomp}
	Let $\Al$ be a finite-dimensional semisimple subalgebra of $\End{}{V}$ for finite-dimensional vector space $V$. Then, 
	\begin{align}
		V\simeq \bigoplus_i V_i \otimes \Hom{\Al}{V_i}{V} \simeq \bigoplus_i \lrbracket{V_i}^{\oplus m_i}
	\end{align}
	with $\Hom{\Al}{V_i}{V}$ the multiplicity space of $V_i$ in $V$ of $\CC$-vector space dimension $m_i$ and for any $a\in\Al$, $v\otimes f \in V_i \otimes \Hom{\Al}{V_i}{V}$ we have
	\begin{align}
		a\lrbracket{v\otimes f} = av\otimes f.\,
	\end{align}
\end{proposition}

\begin{proof}
	By Schur's lemma, we have the canonical identification of the semisimple finite-dimensional representation $V$ of $\Al$ as $\bigoplus_i V_i \otimes \Hom{\Al}{V_i}{V}$ by considering the map
	\begin{align}
		\begin{split}
			\phi \,:\, & \bigoplus_i V_i \otimes \Hom{\Al}{V_i}{V} \rightarrow V \\
			&v\otimes f \mapsto f(v)
		\end{split}
	\end{align}
	which is an isomorphism. The second isomorphism follows from \autoref{prop:direct_prod_of_simple_mods}. 
\end{proof}

\begin{remark}
	This semisimple decomposition is non-unique, as the choice of the $V_i$'s is up to isomorphism. Concretely, while in the isotypic decomposition any isotypic component was a sum over all isomorphic simple modules, here we ``choose'' one simple module $V_i$ and take it to the $m_i$-th power. Clearly, there is freedom in the choice of $V_i$. However, the number of isomorphic simple modules in each  isotypical component and the number of isotypical components are unique.
\end{remark}

\begin{remark}
For $V$ the regular $\KG$-module,
	\begin{align}
		\dim_{\KK}\lrbracket{S}=\dim_{\KK}\lrbracket{\Hom{\KG}{S}{V}} \qquad\text{and}\qquad \dim_{\KK}\lrbracket{S}^2=\dim_{\KK}\lrbracket{S(V)}.
	\end{align}
\end{remark}

Maschke's theorem assures the semisimplicity of $\KG$ as an algebra.

\begin{proposition}[Maschke]
Let $G$ be a finite group and $\mathbb{K}$ a field whose characteristic\footnote{We restrict to fields of characteristic $0$.} does not divide $\lrvert{G}$. Then, $\KG$ is a semisimple $\KK$-algebra.
\end{proposition}

\begin{remark}
	Importantly, any $\KG$-module is semisimple. However, the concrete decomposition depends on the specific $\KG$-module under consideration. Moreover, not all simple $\KG$-modules appear in all $\KG$-modules. Furthermore, there might be multiple isomorphic copies of a simple $\KG$-module appearing.
\end{remark}

We will now use Schur's lemma to translate the decomposition of a semisimple $\KG$-module $V$ to a decomposition of $\End{\KG}{V}$.

\begin{proposition}\label{prop:decomp_of_End_KG_V}
Let $V$ be a semisimple $\KG$-module and $\lrbrace{V_i}_{i=1}^t$ the corresponding simple  $\KG$-modules. Then,
	\begin{align}
		\End{\KG}{V}\simeq \bigoplus_{i=1}^t \KK^{m_i \times m_i}
	\end{align}
	where $m_i=\dim_{\KK}\lrbracket{\Hom{\KG}{V_i}{V}}$.  
\end{proposition}

\begin{proof}
	Since $V$ is semisimple, it admits a decomposition into simple $\KG$-modules
	\begin{align}
		V\simeq \bigoplus_{i} V_i^{\oplus m_i}.\,
	\end{align}
	Furthermore, $\End{\KG}{V_i}\subseteq \End{}{V_i}$. By Schur's lemma, $\operatorname{Hom}$-space additivity and Artin--Wedderburn,
	\begin{align}
	\begin{split}
		\End{\KG}{V} &\simeq \Hom{\KG}{\bigoplus_{i=1}^t V_i^{\oplus m_i}}{V}\\
		&= \bigoplus_{i=1}^t\Hom{\KG}{V_i^{\oplus m_i}}{V} \simeq \bigoplus_{i=1}^t \bigoplus_{j=1}^t \Hom{\KG}{V_i^{\oplus m_i}}{V_j^{\oplus m_j}} = \bigoplus_{i=1}^t\Hom{\KG}{V_i^{\oplus m_i}}{V_i^{\oplus m_i}}\\
		&= \bigoplus_{i=1}^t\End{\KG}{V_i^{\oplus m_i}} \simeq \bigoplus_{i=1}^t M_{m_i}\lrbracket{\End{\KG}{V_i}} \simeq \bigoplus_{i=1}^t M_{m_i}\lrbracket{\KK} \simeq \bigoplus_{i=1}^t \KK^{m_i\times m_i}.\,
	\end{split}
	\end{align}
\end{proof}

\begin{proposition}
We have
	\begin{align}
		\dim\lrbracket{\End{\KG}{V}}=\sum_{i=1}^t m_i^2.\,
	\end{align}
\end{proposition}

\begin{proof}
	Since $\End{\KG}{V}\simeq  \bigoplus_{i=1}^t\KK^{m_i\times m_i}$ and $\dim\lrbracket{ \bigoplus_{i=1}^t\KK^{m_i\times m_i}}=\sum_{i=1}^t\dim\lrbracket{\KK^{m_i\times m_i}} =\sum_{i=1}^t m_i^2$.   
\end{proof}

From now on, let $\KK=\CC$.

\begin{proposition}[\cite{fulton2013representation}]\label{prop:fulton_harris_1}
	Let $G$ be a finite group and $V$ a finite-dimensional right $\CC\lrrec{G}$-module. If $U$ is a simple left $\CC\lrrec{G}$-module, then $V\otimes_{\CC\lrrec{G}}U$ defined as the $\CC$ vector space
	\begin{align}
		\lrbracket{V \otimes_{\CC} U}/ \left\langle vg \otimes u - v\otimes gu \,:\, v\in V, u\in U, g\in\CC\lrrec{G} \right\rangle
	\end{align}
	is a simple left $\End{\CG}{V}$-module.
\end{proposition}

We are interested in extracting from $V$ a subset of elements such that under $\CG$-action we recover all of $V$. Ultimately, we will reduce any problem formulated in terms of $V$ to be formulated in terms of this representative set. In other words, due to linearity we are ultimately interested in the set of representative elements from a $G$-orbit partitioning of a basis of $V$. The decomposition of $V$ given in \autoref{prop:V_decomp} into simple $\CG$-modules yields such a representative set. Concretely, for
\begin{align}
	V=\bigoplus_{i=1}^{k}\bigoplus_{j=1}^{m_i}V_{i,j}
\end{align}
with $V_{i,j}\simeq_{\CG}V_{i,j'}$ for all $i,j,j'$ and $V_{i,j}$ not isomorphic to $V_{i',j}$ for $i\neq i'$, choose $u_{i,j}\in V_{i,j}$, then
\begin{align}
		\lrbrace{\lrbracket{u_{i,1},\ldots, u_{i,m_i}}\, :\, i\in\lrrec{k}}
	\end{align} is a representative set. Importantly, the following proposition establishes when a given subset from $V$ corresponds to a representative set.
    
\begin{proposition}[Representative set, \cite{polak2020new} Proposition 2.4.3]
\label{prop:polak_certification_of_representative_set}
	Let $G$ be a finite group with $\CG$-module $V$. Let $k,m_1,\ldots, m_k\in\mathbb{N}$ and $u_{i,j}\in V$ for $i\in\lrrec{k}$, $j\in \lrrec{m_i}$. Then the set 
	\begin{align}
		\lrbrace{\lrbracket{u_{i,1},\ldots, u_{i,m_i}}\, :\, i\in\lrrec{k}}
	\end{align}
	is a representative set for the action of $G$ on $V$ if and only if:
	\begin{enumerate}
		\item $V=\bigoplus_{i=1}^k\bigoplus_{j=1}^{m_i}\CC\lrrec{G}\cdot u_{i,j}$, 
		\item $\forall i\in \lrrec{k}$ and $j,j'\in \lrrec{m_i}$, there exists a $\CG$-isomorphism \begin{align}
			\begin{split}
				\psi\,:\, \CC\lrrec{G}\cdot u_{i,j} &\rightarrow \CC\lrrec{G} \cdot u_{i,j'},\,\\
				&u_{i,j} \mapsto u_{i,j'},\,
			\end{split}
		\end{align}
		\item $\sum_{i=1}^km_i^2 \geq \dim\lrbracket{\End{G}{V}}$.
		\end{enumerate}
\end{proposition}

The notion of the representative set together with a $G$-invariant inner product\footnote{Such an inner product can be constructed for any $\CG$-module by averaging an arbitrary inner product over the group; cf.\ the remark below for $V\simeq\CC^Z$.} $\left\langle \cdot , \cdot \right\rangle_G$ on $V$, i.e.
\begin{align}
	\left\langle g v , g w \right\rangle_G := \left\langle v , w \right\rangle_G \hspace{1cm} \forall g\in G, v,w\in V
\end{align}
allows us to make the isomorphism in \autoref{prop:decomp_of_End_KG_V} concrete.

\begin{remark}
	While introducing such a $G$-invariant inner product at first seems arbitrary, our proposed isomorphism needs to yield the desired block structure. While the isotypical components of $V$ are orthogonal w.r.t. a $G$-invariant inner product, the isomorphic simple modules within each isotypical component need not be. However, we can always find a decomposition of each isotypical component into isomorphic simple $\CG$-modules, which are pairwise orthogonal w.r.t. $\left\langle \cdot, \cdot\right\rangle_G$.
\end{remark}

\begin{proposition}[\cite{polak2020new, gijswijt2009block}, Proposition 2.4.4.]\label{prop:polak_iso_gen}
Let $G$ be a finite group, $V$ a $\CG$-module with $\lrbrace{V_i}_{i=1}^t$ pairwise non-isomorphic simple $\CG$-modules and $m_i=\dim_{\CC}\lrbracket{\Hom{G}{V_i}{V}}$ for all $i\in\lrrec{t}$. The map  
	\begin{align}\label{eqn:general_iso}
		\begin{split}
			\psi \,:\, &\End{\CC\lrrec{G}}{V} \rightarrow \bigoplus_{i=1}^t \CC^{m_i\times m_i}\\
			& A \mapsto \bigoplus_{i=1}^t\lrbracket{\left\langle A u_{i,j'}, u_{i,j} \right\rangle_{G}}_{j,j'=1}^{m_i}
		\end{split}
	\end{align}
	is a bijection preserving positive-semidefiniteness. 
\end{proposition}

\begin{remark}
	Importantly, checking each $m_i\times m_i$- matrix $\lrbracket{\left\langle A u_{i,j'}, u_{i,j} \right\rangle_{G}}_{j,j'=1}^{m_i}$ for positive semidefiniteness can be substantially simpler in terms of complexity, or in practice, in terms of computational cost, than checking $A$ itself. 
\end{remark}

Furthermore, for $V\simeq\CC^Z$ for some finite set $Z$ or some $Z\in\mathbb{N}$ a $G$-invariant inner product is given by $\left\langle v,w\right\rangle := w^{\dagger}v$ where $(\cdot)^{\dagger}$ is the conjugate transpose operation.
\begin{remark}
	For a finite group $G$ acting on $\CC^Z$, any permutation representation
	\begin{align}
	\begin{split}
		R\,:\, &G\rightarrow \operatorname{GL}\lrbracket{\CC^Z}\\
		&g \mapsto  R_g,\,
	\end{split}
	\end{align}
	 is unitary \cite[Theorem 3.11]{etingof2009introduction} and thus,
	 \begin{align}
	 	 \left\langle R_g v, R_gw \right\rangle = \lrbracket{R_g w}^{\dagger}R_gv = w^{\dagger}R_g^{\dagger}R_gv= w^{\dagger}v = \left\langle v, w \right\rangle
	 \end{align}
	 proves the $G$-invariance of the inner product.
\end{remark}

Thus, restructuring the representative sets to isotypical components as column vectors in a $Z\times m_i$-matrix, i.e. for all $i\in\lrrec{t}$
	\begin{align}
		U_i=\lrrec{u_{i,1}, \ldots, u_{i,m_i}} \,\in\, \CC^{Z\times m_i}
	\end{align}
	allows us to rewrite \autoref{eqn:general_iso} with the specified inner product. By \cite[Definition 2.4.5]{polak2020new} the collection of all such $U_i$ for each isotypical component in $\CC^Z$ is a representative matrix set.
    
\begin{proposition}[\cite{polak2020new}]
The map
	\begin{align}
		\begin{split}
			\psi \,:\, \End{\CC\lrrec{G}}{\CC^Z} &\rightarrow \bigoplus_{i=1}^t \CC^{m_i \times m_i}\\
			A &\mapsto \bigoplus_{i=1}^t U_i^*AU_i
		\end{split}
	\end{align}
	is a bijection preserving positive semidefiniteness. 
\end{proposition}

\begin{remark}
	Note that \begin{align}
		\End{\CC\lrrec{G}}{\CC^Z}\simeq\lrbrace{A\in\CC^{Z\times Z}\,:\, AR_g=R_gA,\, \forall g\in G}.
	\end{align}
\end{remark}
\begin{remark}
	For $V\simeq \mathbb{R}^Z$, $U_i^*=U_i^T$. 
\end{remark}
As a direct consequence we have the following proposition.

\begin{proposition}
Let $G$ be a finite group, $V$ a finite-dimensional $\CC\lrrec{G}$-module and $U$ a simple $\CC\lrrec{G}$-module, with $U(V)$ the $U$-homogeneous component of $V$ and $i$ the corresponding index. Then the map
	\begin{align}
		\begin{split}
			\psi \,:\, &\End{\CC\lrrec{G}}{U(V)} \rightarrow \CC^{m\times m}\\
			& A \mapsto \lrbracket{\left\langle A u_{i,j'}, u_{i,j} \right\rangle_{G}}_{j,j'=1}^{m},\,
		\end{split}
	\end{align}
	or, for $V\simeq \CC^Z$ for a set $Z$ of finite cardinality,
	\begin{align}
		\begin{split}
			\psi \,:\, &\End{\CC\lrrec{G}}{U(V)} \rightarrow \CC^{m\times m}\\
			& A \mapsto U_i^{\dagger}AU_i,\,
		\end{split}
	\end{align}
	with $m=\dim_{\CC}\lrbracket{\Hom{G}{U}{V}}$, is a bijection preserving positive semidefiniteness.
\end{proposition}

Importantly, for the $\End{\CSn}{\Vn}$ decomposition, we will exploit Schur-Weyl duality as a consequence of the double centralizer theorem to obtain concrete information on the isomorphism.

\begin{proposition}[Double centralizer theorem]\label{prop:double_centralizer}
	Let $V$ be a finite-dimensional vector space, $\Al$ a semisimple subalgebra of $\End{}{V}$, and $\BS:=\End{\Al}{V}$ the centralizer of the $\Al$-action on $V$ in $\End{}{V}$. Then,
	\begin{enumerate}[label=(\roman*)]
		\item $\BS$ is semisimple,\,
		\item $\Al=\End{\BS}{V}$,\,
		\item As an $\Al\otimes \BS$-module, we have \begin{align}
			V\simeq \bigoplus_i U_i\otimes W_i
		\end{align}
		where $U_i$ are all simple $\Al$-modules and $W_i:=\Hom{\Al}{U_i}{V}$ are all simple $\BS$-modules.
	\end{enumerate}
\end{proposition}

\begin{proposition}[\cite{fulton2013representation}, Lemma 6.22]\label{prop:fulton_harris_lemma_6_22}
    Let $U$ be a finite-dimensional right $\CG$-module.
    \begin{enumerate}[label=(\roman*)]
        \item For any $c\in\CG$, the canonical map $U\otimes_{\CG}\CG\cdot c\;\rightarrow\; U\cdot c$ is an isomorphism of left $\End{\CG}{U}$-modules.
        \item If $W=\CG\cdot c$ is a simple left $\CG$-module, then $U\otimes_{\CG}W$ is a simple left $\End{\CG}{U}$-module.
        \item If $W_i = \CG \cdot c_i$ are the distinct simple left $\CG$-modules, with $m_i$ the $\CC$-vector space dimension of $W_i$, then
        \begin{align}
            U\simeq \bigoplus_i\lrbracket{U\otimes_{\CG}W_i}^{\oplus m_i}\simeq \bigoplus_i\lrbracket{Uc_i}^{\oplus m_i}
        \end{align}
        is the decomposition of $U$ into simple left $\End{\CG}{U}$-modules.
    \end{enumerate}
\end{proposition}


\subsection{Representation theory of $S_n$}\label{sec:Sym_repr_theory}

While for general groups it is quite difficult or in fact impossible to make the arguments from the section above more concrete, in the case of the symmetric group the situation looks different. We briefly present the relevant representation theory for the symmetric group $S_n$. For more information see e.g. \cite{fulton2013representation, sagan2013symmetric}. 
	
Notably, the simple $\CSn$-modules are in bijection with the conjugacy classes of $S_n$, which in turn are indexed by partitions.
	
\begin{definition}
		For $n, k\in \mathbb{N}$, a partition $\lambda = (\lambda_1, \ldots, \lambda_k)$ of $n$ into $k$ parts denoted $\lambda \vdash_k\, n$ is an ordered collection $\lambda_1\geq \ldots \geq \lambda_k > 0$ of natural numbers such that $\lambda_1+\ldots +\lambda_k=n$. We will also refer to $k$ as the height of the partition. 
\end{definition}

    To each partition $\lambda \vdash n$, we can identify a Young diagram or shape, 
	\begin{align}
        Y(\lambda) := \left\lbrace (i,j)\in \mathbb{N}^2 : 1\leq i \leq k, 1 \leq j \leq \lambda_{i}\right\rbrace,
    \end{align}
	i.e. a collection of boxes arranged in rows and columns according to the partition. Concretely, the Young shape of $\lambda = (\lambda_1, \ldots, \lambda_k)$, denoted $Y(\lambda)$, has $\lambda_{i}$ boxes in the $i$-th row for $i\in\lrrec{k}$. A Young tableau is a Young shape with a numbering of the boxes by integers $1,\ldots, n$. We refer to a canonical labeling of a Young diagram if the boxes are indexed from top left to bottom right, e.g. for partition $(3,2,2,1)$
	\begin{align}
	\ytableausetup{centertableaux}
	\begin{ytableau}
		1 & 2 & 3 \\
		4 & 5 \\
		6 & 7 \\
		8 
	\end{ytableau}.\,
	\end{align}
    
	\begin{remark}
		Note that the length of the first column corresponds to the height of the partition. 
	\end{remark}
    
	A standard Young tableau has a labeling which is strictly increasing within each row and column. By \autoref{prop:wedderburn_little}, we can identify the isotypic components of simple $\CSn$-modules by idempotents from $Z\lrbracket{\CSn}$. However, if we discard the condition that the idempotents be central, we can exploit the bijection between $S_n$ and Young tableaux to construct projectors on to simple $\CSn$-modules. Given a standard Young tableau, we identify subgroups of $S_n$ stabilizing rows or columns, i.e. 
	\begin{align}
		P_{\lambda} = \lrbrace{\pi\in S_n \,:\, \pi \text{ preserves each row}}
	\end{align}
	and 
	\begin{align}
		Q_{\lambda} = \lrbrace{\pi\in S_n \,:\, \pi \text{ preserves each column}}.\,
	\end{align}
    
\begin{definition}[Young symmetrizer, \cite{fulton2013representation} p. 46]
	Let $\lambda$ be a partition of $n$. The Young symmetrizer of $Y(\lambda)$ is the element in $\CSn$ constructed from the row and column stabilizer subgroups for a given labeling
	\begin{align}
		c_{\lambda} =  \underbrace{\lrbracket{\sum_{\pi\in P_{\lambda}}\pi}}_{:=a_{\lambda}}\underbrace{\lrbracket{\sum_{\pi\in Q_{\lambda}}\text{sgn}(\pi)\pi}}_{:=b_{\lambda}}.\,
	\end{align}
\end{definition}

Young symmetrizers are minimal idempotents up to scaling \cite[Theorem 4.3]{fulton2013representation}.

\begin{proposition}[\cite{fulton2013representation}, Lemma 4.26.]
We have $c^2_{\lambda}=\alpha_{\lambda} c_{\lambda}$ with
\begin{align}
	\alpha_{\lambda}=\frac{n!}{\dim_{\CC}\lrbracket{V_{\lambda}}}.\,
\end{align}
\end{proposition}

\begin{proposition}[\cite{stevens2016schur}, Theorem 2.1]
	The image $V_{\lambda}=\CSn c_{\lambda}$ of the group algebra under right multiplication by the Young symmetrizer $c_{\lambda}$ is a simple left $\CSn$-module called a Specht module. Furthermore, every $V_{\lambda}$ can be identified isomorphically with a simple module of $\CSn$ for a unique $\lambda\vdash n$.
\end{proposition}

\begin{proposition}[\cite{stevens2016schur}]
We have the following decomposition of the regular $\CSn$-module
	\begin{align}\label{eqn:CSn_decomp}
		\CSn \simeq \bigoplus_{\lambda} \Specht^{\oplus \dim_{\CC}\lrbracket{\Specht}}
	\end{align}
	i.e.\ each Specht module appears with multiplicity equal to its dimension.
\end{proposition}

\begin{proof}
 Follows directly from \autoref{prop:wederburn_artin_wederburn} and Maschke's theorem. 
\end{proof}

\begin{remark}
	Alternatively, this follows from character theory and \cite[Corollary 0.12]{christandl2006structure}. 
\end{remark}

Importantly, while other Young symmetrizers are still minimal projectors, they are in general not central. The minimal central idempotents obtained from the Wedderburn decomposition (\autoref{prop:wedderburn_little}) correspond to isotypic components. Concretely, each minimal central idempotent projects onto the entire isotypic subspace associated with a certain simple module. By contrast, at the cost of losing centrality, the minimal idempotents from Young symmetrizers yield a finer decomposition into Specht modules. In general, only suitable linear combinations of Young symmetrizers are again central. The primary advantage, however, lies in the ability to explicitly construct the Young symmetrizers using Young tableaux.
\begin{proposition}
	The Young symmetrizer $c_{(n)}$ is proportional to a minimal central idempotent. 
\end{proposition}

\begin{proof}
	Since $c_{(n)}$ is (proportional to) the group average, it is invariant under conjugation and thus commutes with every element of $\CSn$, and thus, is central. Applying the group average corresponding to $c_{(n)}$ to $\CSn$ projects onto the one-dimensional subspace associated with the trivial representation. Concretely, when $c_{(n)}$ acts on $\CSn$, it sends any element of the group algebra to a scalar multiple of the sum of all group elements $\sum_{\pi\in S_n}\pi$. Thus, the image of the group average is exactly the subspace spanned by this sum, corresponding to the trivial representation of $S_n$. Since this space is one-dimensional, $c_{(n)}$ is minimal. Concretely, with \autoref{eqn:idempotents_characters} and $\chi^{\text{trivial}}(\pi)=1$ for all $\pi\in S_n$ we have 
	\begin{align}
		e_{\lambda =(n)} = \frac{1}{n!}\sum_{\pi\in S_n} \pi
	\end{align}
	which is proportional to $c_{(n)}$.
\end{proof}

\begin{proposition}[\cite{stevens2016schur}, Proposition 2.10]\label{prop:specht_dim}
Let $\lambda \vdash n$. Then
	\begin{align}
		\dim_{\CC}\lrbracket{\Specht} = \lrvert{\text{SYT}\lrbracket{\lambda}},\,
	\end{align}
	where $\text{SYT}\lrbracket{\lambda}$ denotes the set of standard Young tableaux for $\lambda$.
\end{proposition}

\begin{proposition}[\cite{fulton2013representation}, Hook Length Formula 4.12.]
	Let $\lambda \vdash n$. Then 
	\begin{align}
		\dim_{\CC}\lrbracket{\Specht} = \frac{n!}{\prod_{(i,j)\in Y(\lambda)} h(i,j)},\,
	\end{align}
	where $h(i,j)$ is the hook length of cell $(i,j)$. 
\end{proposition}

\begin{remark}
	All modules of $\CSn$ are self-dual, since any two $\CSn$-modules with the same character are isomorphic as modules and $\forall\pi\in S_n$ 
	\begin{align}
		\chi_{V^*}(\pi) = \overline{\chi_V(\pi)}=\chi_V\lrbracket{\pi^{-1}}=\chi_{V}(\pi).\,
	\end{align}
\end{remark}


\subsection{Additional side information}\label{sec:additional_side_info}

\begin{proposition}
Consider the group $G=\lrbrace{e}$ with $\CG$-modules $W_i$ of finite dimension $d_{W_i}=\dim_{\CC}\lrbracket{W_i}$ for $i=1,2$. Furthermore, set $d:=\dim_{\CC}(V)$, let $\lambda\vdash_{d} n$ run over the partitions of $n$ into at most $d$ parts, and set $m_{\lambda}:=\dim_{\CC}\lrbracket{\Schurf{\lambda}V}$. Then,
	\begin{align}
		\begin{split}
			\End{\CC\lrrec{G}\times\CSn\times\CG}{W_1\otimes\Vn\otimes W_2}\simeq \bigoplus_{\lambda \vdash_{d}\, n} \CC^{d_{W_1}\cdot d_{W_2}\cdot m_{\lambda} \times d_{W_1}\cdot d_{W_2}\cdot m_{\lambda}}.\,
		\end{split}
	\end{align}
\end{proposition}

\begin{proof}
Clearly, $W_1,W_2$ are simple $\CG$-modules, and thus, $W_1\otimes\Specht\otimes W_2$ is a simple $\CG\times\CSn\times\CG$-module. With \cite[Theorem 10]{serre1977linear} and Schur's lemma, 
\begin{align}
	\Hom{\CG\times\CSn\times\CG}{W_1\otimes\Specht\otimes W_2}{W_1\otimes V_{\tilde{\lambda}}\otimes W_2}\end{align}
	is either one-dimensional (spanned by an isomorphism) for $\lambda = \tilde{\lambda}$ or zero. Thus,
\begin{align}
		\begin{split}
			&\End{\CC\lrrec{G}\times\CSn\times\CG}{W_1\otimes\Vn\otimes W_2} \\
			&\simeq \Hom{\CG\times\CSn\times\CG}{W_1\otimes\lrbracket{\bigoplus_{\lambda \vdash_{d}\, n} \Specht^{\oplus \dim_{\CC}\lrbracket{\Schurf{\lambda}V}}}\otimes W_2}{W_1\otimes\Vn\otimes W_2}\\
			&\simeq \bigoplus_{\lambda \vdash_{d}\, n} \Hom{\CG\times\CSn\times\CG}{W_1\otimes\Specht^{\oplus \dim_{\CC}\lrbracket{\Schurf{\lambda}V}}\otimes W_2}{W_1\otimes\Vn\otimes W_2}\\
			&\simeq \bigoplus_{\lambda \vdash_{d}\, n} \Hom{\CG\times\CSn\times\CG}{W_1\otimes \Specht^{\oplus \dim_{\CC}\lrbracket{\Schurf{\lambda}V}}\otimes W_2}{W_1\otimes\lrbracket{\bigoplus_{\tilde{\lambda}\vdash_{d}\, n}V_{\tilde{\lambda}}^{\oplus \dim_{\CC}\lrbracket{\Schurf{\tilde{\lambda}}V}}}\otimes W_2}\\
			&\simeq \bigoplus_{\lambda \vdash_{d}\, n} \End{\CC\lrrec{G}\times\CSn\times\CG}{W_1\otimes \Specht^{\oplus \dim_{\CC}\lrbracket{\Schurf{\lambda}V}}\otimes W_2}\\
			& \simeq \bigoplus_{\lambda \vdash_{d}\, n} \End{\CC\lrrec{G}\times\CSn\times\CG}{\CC^{\oplus \dim_{\CC}\lrbracket{W_1}}\otimes \Specht^{\oplus \dim_{\CC}\lrbracket{\Schurf{\lambda}V}}\otimes \CC^{\oplus \dim_{\CC}\lrbracket{W_2}}}\\ 
			&\simeq \bigoplus_{\lambda \vdash_{d}\, n} \text{Mat}_{\dim_{\CC}\lrbracket{W_1}\cdot\dim_{\CC}\lrbracket{\Schurf{\lambda}V}\cdot \dim_{\CC}\lrbracket{W_2}}\lrbracket{\CC}.
		\end{split}
	\end{align}
\end{proof}

\begin{proposition}
The map
	\begin{align}
		\begin{split}
			\psi \,:\, \End{\CC\lrrec{G}\times\CSn\times\CG}{W_1\otimes\Vn\otimes W_2} &\rightarrow \bigoplus_{\lambda \vdash_{d}\, n} \CC^{d_{W_1}\cdot d_{W_2}\cdot m_{\lambda} \times d_{W_1}\cdot d_{W_2}\cdot m_{\lambda}}\\
			A &\mapsto \bigoplus_{\lambda \vdash_{d}\, n} \lrbracket{\mathcal{I}_{W_1}\otimes U_{\lambda}\otimes \mathcal{I}_{W_2}}^*A\lrbracket{\mathcal{I}_{W_1}\otimes U_{\lambda}\otimes \mathcal{I}_{W_{2}}}
		\end{split}
	\end{align}
	is a bijection preserving positive semidefiniteness.
\end{proposition}

\begin{proof}
	The representative matrix sets for the $\CG$-action on $W_1, W_2$ are simply the corresponding identities $\mathcal{I}_{W_1}, \mathcal{I}_{W_2}$. The rest follows from \cite[Theorem 10]{serre1977linear} and \cite{chee2023efficient} together with the arguments from the previous section.
\end{proof}
The restriction to $\Schurf{\lambda}V$ as a subspace of $\Vn$ follows accordingly.

\section{Examples}
\label{sec:examples}

This section contains numerous examples that are referenced throughout this work.

\subsection{Bose-symmetry}

\begin{example}\label{ex:Schur_weyl_duality}
	Let $V$ be a vector space of dimension $d\geq n$. Then, for $n=2$,
	\begin{align}
		V^{\otimes 2} \simeq \vee^{2}\lrbracket{V}\oplus\wedge^{2}\lrbracket{V}
	\end{align}
 and for $n=3$
	\begin{align}
		V^{\otimes 3} \simeq \vee^{3}\lrbracket{V}\oplus\wedge^{3}\lrbracket{V}\oplus\lrbracket{\Schurf{(2,1)}V}^{\oplus 2}
	\end{align}
	where we have identified $\Schurf{(n)}V\simeq \vee^{n}\lrbracket{V}$ with the symmetric subspace, appearing with multiplicity $\dim_{\CC}\lrbracket{V_{(n)}}=1$ (trivial representation), and $\Schurf{(1,\ldots, 1)}V\simeq \wedge^{n}\lrbracket{V}$ with the antisymmetric (alternating) subspace or $n$-th exterior power of $V$, appearing with multiplicity $\dim_{\CC}\lrbracket{V_{(1,\ldots,1)}}=1$ (sign representation). Note that the space 
	\begin{align}
		\Schurf{(2,1)}V \simeq \ker\lrbracket{V\otimes\wedge^2\lrbracket{V} \rightarrow \wedge^3\lrbracket{V}}
	\end{align}
	with 
	\begin{align}\begin{split}
		V\otimes\wedge^2\lrbracket{V} &\rightarrow \wedge^3\lrbracket{V} \\
		v_1\otimes \lrbracket{v_2\wedge v_3} &\mapsto v_1\wedge v_2 \wedge v_{3}
	\end{split}
	\end{align}
	captures the elements mapped to zero in subsequent steps in the graded structure of the exterior power. By \autoref{prop:specht_dim}, the dimension of the Specht module $V_{(2,1)}$ is given by $\dim_{\CC}\lrbracket{V_{(2,1)}}=2$, which can be determined by the cardinality of the set of standard $(2,1)$ Young tableaux
	\begin{align}
		\ytableausetup{centertableaux}
	\begin{ytableau}
		1 & 2\\
		3
	\end{ytableau},\, \hspace{1cm} \ytableausetup{centertableaux}
	\begin{ytableau}
		1 & 3\\
		2
	\end{ytableau}
	\end{align}
	which, in general, is computable via the hook length formula. Furthermore, for $\lambda=(n)$, the Young symmetrizer $c_{(n)}$ is proportional to the orthogonal projector onto the symmetric subspace. 
\end{example}

\begin{example}[Weyl module decomposition into non-isomorphic Specht modules]\label{ex:complex_Weyl_decomp}
Let $n=4$ and $d_{\HS}=4$. Then, $\dim_{\CC}\lrbracket{\Schurf{(2,2)}\CC^{4}}=20$ and thus $\Schurf{(2,2)}\CC^{4}$ is a 20-dimensional simple $\text{GL}_4\lrbracket{\CC}$-module. It is the image of some idempotent from $\CC\lrrec{S_4}$ in $\lrbracket{\CC^{4}}^{\otimes 4}$. The Specht modules of $\CC\lrrec{S_4}$ are indexed by partitions of $4$, i.e.\ $(4), (3,1), (2,2), (2,1,1), (1,1,1,1)$. Their dimensions are $1,3,2,3,1$, respectively. Restricting the $\text{GL}_4(\CC)$-module $\Schurf{(2,2)}\CC^{4}$ along the embedding of $S_4$ into $\text{GL}_4(\CC)$ as permutation matrices yields the decomposition into Specht modules
\begin{align}
    \Schurf{(2,2)}\CC^{4}\big\lvert_{S_4}\simeq V_{(4)}^{\oplus 2}\oplus V_{(3,1)}^{\oplus 3}\oplus V_{(2,2)}^{\oplus 3}\oplus V_{(2,1,1)}.\,
\end{align}
\end{example}

\begin{example}\label{ex:basis_comparison}
		For  $d=n=2$, consider the basis element
		\begin{align}C_{(1,1),(1,1)}^{\vee^2} =
		\lrbracket{\ket{10}+\ket{01}}\lrbracket{\bra{10}+\bra{01}} = \ket{10}\bra{10} + \ket{10}\bra{01} + \ket{01}\bra{10} + \ket{01}\bra{01}  =
			\begin{bmatrix}
				0 & 0 & 0 & 0\\
				0 & 1 & 1 & 0\\
				0 & 1 & 1 & 0\\
				0 & 0 & 0 & 0\\
			\end{bmatrix}
		\end{align}
		of $\End{\CC\lrrec{S_2}}{\vee^2\lrbracket{\CC^2}}$ in comparison to the two basis elements
		\begin{align}
			\ket{10}\bra{10}+ \ket{01}\bra{01} = \begin{bmatrix}
				0 & 0 & 0 & 0\\
				0 & 1 & 0 & 0\\
				0 & 0 & 1 & 0\\
				0 & 0 & 0 & 0\\
			\end{bmatrix},\, \ket{10}\bra{01}+ \ket{01}\bra{10} =\begin{bmatrix}
				0 & 0 & 0 & 0\\
				0 & 0 & 1 & 0\\
				0 & 1 & 0 & 0\\
				0 & 0 & 0 & 0\\
			\end{bmatrix}
		\end{align}
		of $\End{\CC\lrrec{S_2}}{\lrbracket{\CC^2}^{\otimes 2}}$. While the first matrix is clearly Bose-symmetric, the second and third are not. However, their sum yields the Bose-symmetric matrix $C_{(1,1),(1,1)}^{\vee^2}$. Recall the antisymmetric space $\land^{2}\lrbracket{\CC^2}$ is given by $\text{span}_{\CC}\lrbracket{\ket{01}-\ket{10}}$. Accordingly, the space of operators with support and range on that space is given by
		\begin{align}
			\text{span}_{\CC}\lrbracket{\ket{01}\bra{01}-\ket{01}\bra{10}-\ket{10}\bra{01}+\ket{10}\bra{10}},
		\end{align}
		which clearly is orthogonal to $\End{\CC\lrrec{S_2}}{\vee^2\lrbracket{{\CC^2}}}$.
	\end{example}

\begin{example}\label{ex:change_of_basis_bose}
	Let $d_{\HS}=n=2$. Consider Young tableau $\tau$
	\begin{align}
		\ytableausetup{centertableaux}
	\begin{ytableau}
		1 & 2\\
	\end{ytableau}
	\end{align}
	yielding $u_{\tau}=\ket{01}+\ket{10}$. Thus,
	\begin{align}
		u_{\tau}^T\CBose{(1,1)}{(1,1)}{2}u_{\tau} = \begin{bmatrix}
			0 \\
			1 \\
			1 \\
			0\\
		\end{bmatrix}^T\begin{bmatrix}
				0 & 0 & 0 & 0\\
				0 & 1 & 1 & 0\\
				0 & 1 & 1 & 0\\
				0 & 0 & 0 & 0\\
			\end{bmatrix}\begin{bmatrix}
			0 \\
			1 \\
			1 \\
			0\\
		\end{bmatrix} = 4.\,
	\end{align}
	
\end{example}


\subsection{SDP symmetry reduction}

\begin{example}\label{ex:sdp_sym_reduction}
   Consider $\lrvert{\HS}=2$, $n=2$. Then, 
    \begin{align}
	\begin{split}
		\ket{1} &=  \lrbracket{\ket{1_1}\otimes \ket{1_2}},\, \quad\ket{2} =  \lrbracket{\ket{2_1}\otimes \ket{1_2}},\,\\
		 \ket{3} &=\lrbracket{\ket{1_1}\otimes \ket{2_2}},\, \quad \ket{4} = \lrbracket{\ket{2_1}\otimes \ket{2_2}},
	\end{split}
    \end{align}

\begin{align}
    E^{(i,j)} = \begin{bmatrix}
\lrvert{\lrbrace{k\in\lrrec{n}\,:\, \phi_k(i)= 1, \phi_k(j)= 1}} & \lrvert{\lrbrace{k\in\lrrec{n}\,:\, \phi_k(i)= 1, \phi_k(j)= 2}} \\
\lrvert{\lrbrace{k\in\lrrec{n}\,:\, \phi_k(i)= 2, \phi_k(j)= 1}} & \lrvert{\lrbrace{k\in\lrrec{n}\,:\, \phi_k(i)= 2, \phi_k(j)= 2}} 
\end{bmatrix}\in\mathbb{N}_0^{\lrvert{\HS}\times \lrvert{\HS}},
\end{align}
\begin{align}
\begin{split}
    \phi \times \phi^* \, : \, (1,1) \mapsto  \lrbracket{\ket{\phi_1(1)}\otimes \ket{\phi_2(1)}}\lrbracket{\bra{\phi_1(1)}\otimes \bra{\phi_2(1)}} &= \lrbracket{\ket{1_1}\otimes \ket{1_2}}\lrbracket{\bra{1_1}\otimes \bra{1_2}}\\
    &= \lrbracket{\ket{1}_1\otimes \ket{1}_2}\lrbracket{\bra{1}_1\otimes \bra{1}_2},
\end{split}
\end{align}
\begin{align}
    E^{(1,1)}= \begin{bmatrix}
\lrvert{\lrbrace{1,2}} & 0 \\
0 & 0
\end{bmatrix},
\end{align}

    \begin{align}
        \begin{array}{cc}
            E^{(1,1)}= \begin{bmatrix}
2 & 0 \\
0 & 0 
\end{bmatrix},\, E^{(4,4)}=\begin{bmatrix}
0 & 0 \\
0 & 2 
\end{bmatrix},\, & E^{(2,2)}= \begin{bmatrix}
1 & 0 \\
0 & 1 
\end{bmatrix} = E^{(3,3)},\, \\
\\
             E^{(1,4)}= \begin{bmatrix}
0 & 2 \\
0 & 0 
\end{bmatrix},\, E^{(4,1)} = \begin{bmatrix}
0 & 0 \\
2 & 0 
\end{bmatrix} ,\, & E^{(2,3)}= \begin{bmatrix}
0 & 1 \\
1 & 0 
\end{bmatrix} = E^{(3,2)},\, \\
\\
E^{(1,3)}= \begin{bmatrix}
1 & 1 \\
0 & 0 
\end{bmatrix},\, E^{(4,2)} = \begin{bmatrix}
0 & 0 \\
1 & 1 
\end{bmatrix} ,\, & E^{(3,1)}= \begin{bmatrix}
1 & 0 \\
1 & 0 
\end{bmatrix},\, E^{(2,4)} = \begin{bmatrix}
0 & 1 \\
0 & 1 
\end{bmatrix},\, \\
\\
E^{(3,4)}= \begin{bmatrix}
0 & 1 \\
0 & 1 
\end{bmatrix},\, E^{(2,1)} = \begin{bmatrix}
1 & 0 \\
1 & 0 
\end{bmatrix},\, & E^{(4,3)}= \begin{bmatrix}
0 & 0 \\
1 & 1 
\end{bmatrix},\, E^{(1,2)}= \begin{bmatrix}
1 & 1 \\
0 & 0 
\end{bmatrix}.\,
        \end{array}
    \end{align}
\end{example}
Note that $E^{(i,j)}=E^{(\pi (i),\pi (j))}$ and $E^{(i,j)}=\lrbracket{E^{(j, i)}}^T$. Furthermore,
\begin{align}
        \begin{split}
        \begin{array}{cccc}
            O_1^{\HS^n} =\lrbrace{(1,1)},\, &  O_2^{\HS^n} =\lrbrace{(4,4)},\, & O_3^{\HS^n} =\lrbrace{(2,2), (3,3)},\, & O_4^{\HS^n} =\lrbrace{(1,2), (1,3)},\,\\
            \\
            O_5^{\HS^n} =\lrbrace{(2,1), (3,1)},\, & O_6^{\HS^n} =\lrbrace{(1,4)},\, & O_7^{\HS^n} =\lrbrace{(4,1)},\, & O_8^{\HS^n} =\lrbrace{(2,3), (3,2)},\, \\
            \\
            O_9^{\HS^n} =\lrbrace{(3,4), (2,4)},\, & O_{10}^{\HS^n} =\lrbrace{(4,2), (4,3)},\, \\
        \end{array}
        \end{split}
    \end{align}
\begin{align}
\begin{split}
    \begin{array}{ccc}
         C_1^{\HS^n}=\begin{bmatrix}
1 & 0  & 0 & 0\\
0 & 0 & 0& 0\\
0 & 0 & 0& 0\\
0 & 0 & 0& 0\\
\end{bmatrix} & C_2^{\HS^n}=\begin{bmatrix}
0 & 0  & 0 & 0\\
0 & 0 & 0& 0\\
0 & 0 & 0& 0\\
0 & 0 & 0& 1\\
\end{bmatrix} & 
C_3^{\HS^n}=\begin{bmatrix}
0 & 0  & 0 & 0\\
0 & 1 & 0& 0\\
0 & 0 & 1& 0\\
0 & 0 & 0& 0\\
\end{bmatrix}  \\
\\
C_4^{\HS^n}=\begin{bmatrix}
0 & 1  & 1 & 0\\
0 & 0 & 0 & 0\\
0 & 0 & 0& 0\\
0 & 0 & 0& 0\\
\end{bmatrix} & 
C_5^{\HS^n}=\begin{bmatrix}
0 & 0  & 0 & 0\\
1 & 0 & 0& 0\\
1 & 0 & 0& 0\\
0 & 0 & 0& 0\\
\end{bmatrix}  & C_6^{\HS^n}=\begin{bmatrix}
0 & 0  & 0 & 1\\
0 & 0 & 0 & 0\\
0 & 0 & 0 & 0\\
0 & 0 & 0& 0\\
\end{bmatrix}\\
\\
C_7^{\HS^n}=\begin{bmatrix}
0 & 0  & 0 & 0\\
0 & 0 & 0& 0\\
0 & 0 & 0& 0\\
1 & 0 & 0& 0\\
\end{bmatrix}  & C_8^{\HS^n}=\begin{bmatrix}
0 & 0  & 0 & 0\\
0 & 0 & 1 & 0\\
0 & 1 & 0& 0\\
0 & 0 & 0& 0\\
\end{bmatrix} & 
C_9^{\HS^n}=\begin{bmatrix}
0 & 0  & 0 & 0\\
0 & 0 & 0& 1\\
0 & 0 & 0& 1\\
0 & 0 & 0& 0\\
\end{bmatrix}\\
\\
C_{10}^{\HS^n}=\begin{bmatrix}
0 & 0  & 0 & 0\\
0 & 0 & 0 & 0\\
0 & 0 & 0& 0\\
0 & 1 & 1& 0\\
\end{bmatrix} & \sum_{i=1}^{10} C_i =\begin{bmatrix}
        1 & 1 & 1 & 1\\
        1 & 1 & 1 & 1 \\
        1 & 1 & 1 & 1 \\
        1 & 1 & 1 & 1 \\
    \end{bmatrix}.\\
\\
    \end{array}
\end{split}
\end{align}
Then, 
\begin{align}
    \begin{split}
        \begin{array}{cc}
            E_{O_1^{\HS^n}} = \lrbrace{E^{(1,1)}= \begin{bmatrix}
2 & 0 \\
0 & 0 
\end{bmatrix}},\, &   E_{O_2^{\HS^n}} = \lrbrace{E^{(4,4)}= \begin{bmatrix}
0 & 0 \\
0 & 2 
\end{bmatrix}},\,\\
             \\
             E_{O_3^{\HS^n}} = \lrbrace{E^{(2,2)}= \begin{bmatrix}
1 & 0 \\
0 & 1 
\end{bmatrix} = E^{(3,3)}},\, &   E_{O_4^{\HS^n}} = \lrbrace{E^{(1,2)}= \begin{bmatrix}
1 & 1 \\
0 & 0 
\end{bmatrix}= E^{(1,3)}},\,\\
             \\
             E_{O_5^{\HS^n}} = \lrbrace{E^{(2,1)}= \begin{bmatrix}
1 & 0 \\
1 & 0 
\end{bmatrix} = E^{(3,1)}},\, &   E_{O_6^{\HS^n}} = \lrbrace{E^{(1,4)}= \begin{bmatrix}
0 & 2 \\
0 & 0 
\end{bmatrix}},\,\\
             \\
              E_{O_7^{\HS^n}} = \lrbrace{E^{(4,1)}= \begin{bmatrix}
0 & 0 \\
2 & 0 
\end{bmatrix}},\, & 
             E_{O_8^{\HS^n}} = \lrbrace{E^{(2,3)}= \begin{bmatrix}
0 & 1 \\
1 & 0 
\end{bmatrix} = E^{(3,2)}},\,\\
E_{O_9^{\HS^n}} = \lrbrace{E^{(3,4)}= \begin{bmatrix}
0 & 1 \\
0 & 1 
\end{bmatrix} = E^{(2,4)}},\, &   E_{O_{10}^{\HS^n}} = \lrbrace{E^{(4,2)}= \begin{bmatrix}
0 & 0 \\
1 & 1 
\end{bmatrix}= E^{(4,3)}}.\,\\
             \\
        \end{array}
    \end{split}
\end{align}
For the $S_n$ subgroup structure we observe:
\begin{align}
    \begin{split}
        \begin{array}{|c|c|}
            \hline
            S_3 & S_2 \\
            \hline
            \hline
             1 - 2 - 3 & 1 - \underbrace{2 - 3}_{\pi\in S_2}\\
             \downarrow (1)(23) & \downarrow (23)\\
             \\
             1 - 3 - 2 & 1 - 3 - 2\\
             \cline{2-2}
             2 - 3 - 1 &  2 - 3 - 1\\
             2 - 1 - 3 & 2 - 1 - 3\\
             \cline{2-2}
             3 - 1 - 2 &  3 - 1 - 2\\
             3 - 2 - 1 & 3 - 2 - 1\\
             \hline
             \text{1 Orbit} & \text{3 Orbits}\\
             \hline
        \end{array}
    \end{split}
\end{align}

For $\lrvert{\HS}=2, n=3$ and $ S=\lrbrace{(1,1), (2,1), (1,2), (2,2)},\,$ the $S_2$ splits the orbits of$S_3$, e.g.
\begin{align}
    \lrbrace{(1,2), (1,3), (1,5)} =  \underbrace{\lrbrace{(1,3), (1,5)}}_{:=O^t_{(1,1)}} \cup \underbrace{\lrbrace{(1,2)}}_{:=O^t_{(1,2)}}.
\end{align}
Then, 
\begin{align}\label{eqn:example_cano_basis_redu}
\begin{split}
    C_t^{\HS^3}&= \sum_{(a,b)\in O_t^{\HS^{3}}}\bigotimes_{i=1}^3 \ket{\phi_i(a)}\bra{\phi_i(b)} = \sum_{(c,d)\in S} \sum_{(a,b)\in O_{(c,d)}^t} \ket{c}\bra{d}_1\otimes \lrbracket{\bigotimes_{i=2}^3 \ket{\phi_i(a)}\bra{\phi_i(b)}}\\
    &= \sum_{(a,b)\in O^t_{(1,1)}} \ket{1}\bra{1}_1\otimes \lrbracket{\bigotimes_{i=2}^3 \ket{\phi_i(a)}\bra{\phi_i(b)}}\\
    & + \sum_{(a,b)\in O^t_{(1,2)}} \ket{1}\bra{2}_1\otimes \lrbracket{\bigotimes_{i=2}^3 \ket{\phi_i(a)}\bra{\phi_i(b)}}\\
    &=\lrbracket{\ket{1}\bra{3} + \ket{1}\bra{5}} + \ket{1}\bra{2}.
\end{split}
\end{align}
Now let 
\begin{align}
    O^{\HS_2^3}_{t'=1}=\lrbrace{(1,1)}=\lrbrace{\ket{1}\bra{1}_2\otimes \ket{1}\bra{1}_3},\,\hspace{0.5cm}\, O^{\HS_2^3}_{t'=2}=\lrbrace{(1,2), (1,3)}=\lrbrace{\ket{1}\bra{2}_2\otimes \ket{1}\bra{1}_3,\, \ket{1}\bra{1}_2\otimes \ket{1}\bra{2}_3}
\end{align}
be orbits under the $S_{2}$ action on $\HS_2^3$ with
\begin{align}
    K^{\HS_2^3}_{t'=1}=\begin{bmatrix}
        1  & 0 & 0 & 0  \\
0  & 0 & 0 & 0 \\
0  & 0 & 0 & 0  \\
0  & 0 & 0 & 0  \\
    \end{bmatrix},\,\hspace{0.5cm} K^{\HS_2^3}_{t'=2}=\begin{bmatrix}
        0  & 1 & 1 & 0   \\
0  & 0 & 0 & 0   \\
0  & 0 & 0 & 0  \\
0  & 0 & 0 & 0  \\
    \end{bmatrix},\,
\end{align}
such that we can express $C_t^{\HS^3}$ from \autoref{eqn:example_cano_basis_redu} as
\begin{align}
    C_t^{\HS^3} = \ket{1}\bra{1}_{1}\otimes K^{\HS_2^3}_{t'=2} + \ket{1}\bra{2}_{1}\otimes K^{\HS_2^3}_{t'=1}
\end{align}
by decomposing along the $\HS_1$-system.
\begin{align}
\begin{split}
    &\ket{1}\bra{2}_1\otimes\underbrace{\ket{1}\bra{1}_2\otimes\ket{1}\bra{1}_3}_{\text{1 element from } S}\\ 
    &\ket{1}\bra{1}_1\otimes\underbrace{\ket{1}\bra{2}_2\otimes\ket{1}\bra{1}_3}_{\text{2 elements from } S}\\
    & \ket{1}\bra{1}_1\otimes\underbrace{\ket{1}\bra{1}_2\otimes\ket{1}\bra{2}_3}_{\text{2 elements from } S}
\end{split}
\end{align}
and thus $\lrvert{O^t_{(1,1)}}=2$ and $\lrvert{O^t_{(1,2)}}=1$. Returning to the $n=2$ example, under $S_1$ the orbits of $S_2$ get partitioned into 16 orbits of length one, i.e.
\begin{itemize}
    \item $O_1^{\HS^n} =\lrbrace{(1,1)}:$
    \begin{align}
\begin{split}
    \begin{array}{cccc}
        O_{(1, 1)}^{1}= \lrbrace{(1,1)},\,  & O_{(1, 2)}^{1}= \emptyset ,\, & O_{(2, 1)}^{1}= \emptyset ,\, &  O_{(2, 2)}^{1} =\emptyset,\,\\
    \end{array}
\end{split}
\end{align}
    \item $O_2^{\HS^n} =\lrbrace{(4,4)}:$
    \begin{align}
\begin{split}
    \begin{array}{cccc}
         O_{(1, 1)}^{2}= \emptyset,\,  & O_{(1, 2)}^{2}= \emptyset ,\, & O_{(2, 1)}^{2}= \emptyset ,\, &  O_{(2, 2)}^{2} =\lrbrace{(4,4)},\,\\
    \end{array}
\end{split}
\end{align}
    \item $O_3^{\HS^n} =\lrbrace{(2,2), (3,3)}:$
    \begin{align}
\begin{split}
    \begin{array}{cccc}
         O_{(1, 1)}^{3}= \lrbrace{(3,3)},\,  & O_{(1, 2)}^{3}= \emptyset ,\, & O_{(2, 1)}^{3}= \emptyset ,\, &  O_{(2, 2)}^{3} =\lrbrace{(2,2)},\,\\
    \end{array}
\end{split}
\end{align}
    \item $O_4^{\HS^n} =\lrbrace{(1,2), (1,3)}:$
    \begin{align}
\begin{split}
    \begin{array}{cccc}
       O_{(1, 1)}^{4}= \lrbrace{(1,3)},\,  & O_{(1, 2)}^{4}= \lrbrace{(1,2)} ,\, & O_{(2, 1)}^{4}= \emptyset ,\, &  O_{(2, 2)}^{4} =\emptyset ,\,\\
    \end{array}
\end{split}
\end{align}
    \item $O_5^{\HS^n} =\lrbrace{(2,1), (3,1)}:$
    \begin{align}
\begin{split}
    \begin{array}{cccc}
       O_{(1, 1)}^{5}= \lrbrace{(3,1)},\,  & O_{(1, 2)}^{5}= \emptyset ,\, & O_{(2, 1)}^{5}= \lrbrace{(2,1)} ,\, &  O_{(2, 2)}^{5} =\emptyset ,\,\\
    \end{array}
\end{split}
\end{align}
    \item $O_6^{\HS^n} =\lrbrace{(1,4)}:$ 
    \begin{align}
\begin{split}
    \begin{array}{cccc}
       O_{(1, 1)}^{6}= \emptyset,\,  & O_{(1, 2)}^{6}= \lrbrace{(1,4)} ,\, & O_{(2, 1)}^{6}= \emptyset ,\, &  O_{(2, 2)}^{6} =\emptyset ,\,\\
    \end{array}
\end{split}
\end{align}
    \item $O_7^{\HS^n} =\lrbrace{(4,1)}:$ 
    \begin{align}
\begin{split}
    \begin{array}{cccc}
       O_{(1, 1)}^{7}= \emptyset,\,  & O_{(1, 2)}^{7}= \emptyset ,\, & O_{(2, 1)}^{7}= \lrbrace{(4,1)} ,\, &  O_{(2, 2)}^{7} =\emptyset ,\,\\
    \end{array}
\end{split}
\end{align}
    \item $O_8^{\HS^n} =\lrbrace{(2,3), (3,2)}:$
    \begin{align}
\begin{split}
    \begin{array}{cccc}
       O_{(1, 1)}^{8}= \emptyset,\,  & O_{(1, 2)}^{8}= \lrbrace{(3,2)} ,\, & O_{(2, 1)}^{8}= \lrbrace{(2,3)} ,\, &  O_{(2, 2)}^{8} =\emptyset ,\,\\
    \end{array}
\end{split}
\end{align}
    \item $O_9^{\HS^n} =\lrbrace{(3,4), (2,4)}:$
    \begin{align}
\begin{split}
    \begin{array}{cccc}
       O_{(1, 1)}^{9}= \emptyset,\,  & O_{(1, 2)}^{9}= \lrbrace{(3, 4)} ,\, & O_{(2, 1)}^{9}= \emptyset ,\, &  O_{(2, 2)}^{9} =\lrbrace{(2,4)} ,\,\\
    \end{array}
\end{split}
\end{align}
    \item $O_{10}^{\HS^n} =\lrbrace{(4,2), (4,3)}:$
    \begin{align}
\begin{split}
    \begin{array}{cccc}
       O_{(1, 1)}^{10}= \emptyset,\,  & O_{(1, 2)}^{10}= \emptyset ,\, & O_{(2, 1)}^{10}= \lrbrace{(4, 3)} ,\, &  O_{(2, 2)}^{10} =\lrbrace{(4, 2)} .\,\\
    \end{array}
\end{split}
\end{align}
\end{itemize}
\pagebreak


\end{document}